\documentclass[11pt]{statsoc}
\usepackage{amsmath,amssymb}
\usepackage{multirow}
\usepackage{natbib,url}
\usepackage{array}
\newcolumntype{P}[1]{>{\raggedright\arraybackslash}p{#1}}
\usepackage[a4paper]{geometry} 

\usepackage{tabularx}
\usepackage{graphicx}
\usepackage{caption} 
\usepackage{color}

\newtheorem{thm}{Theorem} 
\newtheorem{cor}{Corollary}

\newtheorem{proposition}{Proposition}
\newtheorem{lemma}{Lemma}
\newtheorem{condition}{Condition}
\newtheorem{definition}{Definition}

\newtheorem{remark}{Remark}

\newcommand{\bed}{\begin{definition}}
	\newcommand{\eed}{\end{definition}}

\newcommand{\rom}[1]{\uppercase\expandafter{\romannumeral #1\relax}}

\usepackage{bbm}

\newcommand{\bitem}{\begin{itemize}}
	\newcommand{\eitem}{\end{itemize}}

\newcommand{\beqn}{\begin{equation}}
	\newcommand{\eeqn}{\end{equation}}
\newcommand{\balign}{\begin{align}}
	\newcommand{\ealign}{\end{align}}

\usepackage{amssymb}
\newcommand{\beq}{\begin{equation}}
	\newcommand{\eeq}{\end{equation}}

\newcommand{\diag}{\mathrm{diag}}

\allowdisplaybreaks
\usepackage{enumitem}

\newcommand{\pr}{\mathbb{P}}

\newcommand{\R}{\mathbb{R}}

\newcommand\num{\addtocounter{equation}{1}\tag{\theequation}}
\newcommand{\var}{\mathrm{Var}}
\newcommand{\E}{\mathbb{E}} 
\newcommand{\mf}[1]{\mathbf{#1}} 
\newcommand{\mc}[1]{\mathcal{#1}}
\newcommand{\VV}{\mathbb{V}} 
\newcommand{\les}{\lesssim}
\newcommand{\bo}{\mathbf{1}_p}
\newcommand{\rp}[1]{^{(#1)}} 

\usepackage{xr-hyper}
\usepackage{hyperref}
\title{Testing High-dimensional Multinomials with Applications to Text Analysis}

\author[T. Tony Cai]{T. Tony Cai}
\address{University of Pennsylvania,
Philadelphia, Pennsylvania,
United States.}

\author[Zheng Tracy Ke]{Zheng Tracy Ke}
\address{Harvard University,
Cambridge, Massachusetts,
United States.}

\author[]{Paxton Turner}
\coaddress{Harvard University, Statistics.
1 Oxford St \# 7
Cambridge, MA, USA 02138, 
Email: paxtonturner@g.harvard.edu
}
\address{Harvard University,
Cambridge, Massachusetts,
United States.} 

\begin{document}

\maketitle

\begin{abstract}
Motivated by applications in text mining and discrete distribution inference, we test for equality of probability mass functions of $K$ groups of high-dimensional multinomial distributions. Special cases of this problem include global testing for topic models, two-sample testing in authorship attribution, and closeness testing for discrete distributions. A test statistic, which is shown to have an asymptotic standard normal distribution under the null hypothesis, is proposed. This parameter-free limiting null distribution holds true without requiring identical multinomial parameters within each group or equal group sizes.
	The optimal detection boundary for this testing problem is established, and the proposed test is shown to achieve this optimal detection boundary across the entire parameter space of interest.  
	The proposed method is demonstrated in simulation studies and applied to analyze two real-world datasets to examine, respectively, variation among customer reviews of Amazon movies and the diversity of statistical paper abstracts.
\end{abstract}


\section{Introduction}\label{sec:Intro}

Statistical inference for multinomial data has garnered considerable recent interest \citep{diakonikolas2016new,balakrishnan2018hypothesis}. One important application is in text mining. It is common to model the word counts in a text document by a multinomial distribution \citep{blei2003latent}. 
As a motivating example, the study of online customer ratings and reviews is a trending topic in marketing research. 
Customer reviews are a good proxy to the overall ``word of mouth" and can significantly influence customers' decisions. Research works aim to  understand the patterns in online reviews and their impacts on sales. 
Classical studies only use numerical ratings but ignore the rich text reviews because of their unstructured nature. More recent works have revealed the importance of analyzing text reviews, especially for hedonic products such as books, movies, and hotels  \citep{chevalier2006effect}. 
A question of interest is to detect the heterogeneity in reviewers' response styles. For example, 
\cite{leung2020all} discovered that younger travelers, women, and travelers with less review expertise tend to give more positive reviews and that guests staying in high-class hotels tend to have more extreme response styles than those staying in low-class hotels. 
Knowing such differences will offer valuable insights for hotel managers and online rating/review sites. 

The aforementioned heterogeneity detection can be cast as a hypothesis test on multinomial data. Suppose reviews are written using a vocabulary of $p$ distinct words. Let $X_i\in\mathbb{R}^p$ contain the word counts in review $i$. We assume  $X_i$'s are independent, and 
\beq \label{Mod1-data}
X_i\sim \mathrm{Multinomial}(N_i, \Omega_i), \qquad 1\leq i\leq n,
\eeq
where $N_i$ is the total length of review $i$ and $\Omega_i\in\mathbb{R}^p$ is a probability mass function (PMF) containing the population word frequencies. 
These reviews are divided into $K$ groups by reviewer characteristics (e.g., age, gender, new/returning customer), product characteristics (e.g., high-class versus low-class hotels), and numeric ratings (e.g., from 1 star to 5 stars), where $K$ can be presumably large. 
We view $\Omega_i$ as representing the `true response' of review $i$. The ``average response" of a group $k$ is defined by a weighted average of the PMFs: 
\beq \label{Mod2-groupMean}
\mu_k = (n_k\bar{N}_k)^{-1} \sum_{i\in S_k} N_i\Omega_i, \qquad 1\leq k\leq K. 
\eeq 
Here $S_k\subset\{1,2,\ldots,n\}$ is the index set of group $k$,  $n_k=|S_k|$ is the total number of reviews in group $k$, and $\bar{N}_k=n_k^{-1}\sum_{i\in S_k}N_i$ is the average length of reviews in group $k$. 
We would like to test
\beq \label{Mod3-null}
H_0:\quad \mu_1=\mu_2=\ldots = \mu_K. 
\eeq
When the null hypothesis is rejected, it means there exist statistically significant differences among the group-wise ``average responses".

%

We call \eqref{Mod1-data}-\eqref{Mod3-null} the ``$K$-sample testing for equality of average PMFs in multinomials" or ``$K$-sample testing for multinomials" for short. As $K$ varies, it includes several well-defined problems in text mining and discrete distribution inference as special cases. 

\begin{enumerate}
	
	\item {\it Global testing for topic models}. Topic modeling \citep{blei2003latent} is a popular text mining tool.
	In a topic model, each $\Omega_i$ in \eqref{Mod1-data} is a convex combination of $M$ topic vectors. 
	Before fitting a topic model to a corpus, it is often desirable to determine if the corpus indeed contains multiple topics.  
	This boils down to the global testing problem, which tests $M=1$ versus $M>1$. In this case, we set $K=n$ and view each document as a separate group, so that $\Omega_i$ itself is the within-group average.
	Under the null hypothesis, all these $\Omega_i$'s are equal to a single topic vector.
	Under the alternative, the $\Omega_i$'s  are not all equal. 
	This is thus a special case of our problem with $K=n$ and $n_k=1$. 
	
	
	\item {\it Authorship attribution} \citep{mosteller1963inference,kipnis2022higher}. In these applications, the goal is to determine the unknown authorship of an article from other articles with known authors. A famous example \citep{mosteller2012applied} is to determine the actual authors of a few Federalist Papers written by three authors but published under a single pseudonym. It can be formulated \citep{mosteller1963inference,kipnis2022higher} as testing the equality of population word frequencies between the article of interest and the corpus from a known author, a special case of our problem with $K=2$. 
	
	\item {\it Closeness between discrete distributions} \citep{chan2014optimal,bhattacharya2015testing,balakrishnan2019hypothesis}. 
	There has been a surge of interest in discrete distribution inference. Closeness testing is one of most studied problems. 
	The data from two discrete distributions are summarized in two multinomial vectors $\mathrm{Multinomial}(N_1, \mu)$ and $\mathrm{Multinomial}(N_2, \theta)$. 
	The goal is to test $\mu=\theta$. It is a special case of our testing problem with $K=2$ and $n_1=n_2=1$. 
\end{enumerate}

In this paper, we provide a unified solution to all the aforementioned problems. 
The key to our methodology is a flexible statistic called DELVE (DE-biased and Length-assisted Variability Estimator). 
It provides a general similarity measure for comparing groups of discrete distributions such as count vectors associated with text corpora. Similarity measures (such as the classical cosine similarity, log-likelihood ratio statistic, and others) are fundamental in text mining and have been applied to problems in distribution testing \citep{kim2022minimax}, computational linguistics  \citep{gomaa2013survey}, econometrics \citep{hansen2018transparency}, and computational biology \citep{kolodziejczyk2015technology}. Our method is a new and flexible similarity measure that is potentially useful in these areas. 

We emphasize that our setting does not require that the $X_i$'s in the same group are drawn from the same distribution. Under the null hypothesis \eqref{Mod3-null}, the group-wise means are equal, but the $\Omega_i$'s within each group can still be different from each other. As a result, the null hypothesis is composite and designing a proper test statistic is non-trivial.

\subsection{Our results and contributions} \label{subsec:OurResults}

The dimensionality of the testing problem is captured by $(n, p, K)$ and $\bar{N}:=n^{-1}\sum_{i=1}^n N_i$. 
We are interested in a high-dimensional setting where 
\beq \label{Cond-scaling}
n\bar{N}\to\infty, \quad p\to\infty, \quad \mbox{and}\quad n^2\bar{N}^2/(Kp) \to\infty. 
\eeq
In most places of this paper, we use a subscript $n$ to indicate asymptotics, but our method and theory do apply to the case where $n$ is finite and $\bar{N}\to\infty$. In text applications,  $n\bar{N}$ is the total count of words in the corpus, and
a large $n\bar{N}$ means either there are sufficiently many documents, or the documents are sufficiently long. 
Given that $n\bar{N}\to\infty$, we further allow $(p,K)$ to grow with $n$ at a speed such that $Kp\ll n^2\bar{N}^2$. 
In particular, our settings allow $K$ to range from $2$ to $n$, so as to cover all the application examples. 

We propose a test that enjoys the following properties:
\begin{enumerate} 
	\item[(a)] {\it Parameter-free null distribution}: 
	We shall define a test statistic $\psi$ in \eqref{define:psi} and show that $\psi\to N(0,1)$ under the null $H_0$ in \eqref{Mod3-null}.
	Even under $H_0$, the model contains a large number of free parameters because the null hypothesis is only about the equality of ``average" PMFs but still allows $(N_i,\Omega_i)$ to differ within each group. As an appealing property, the null distribution of $\psi$ does not depend on these individual multinomial parameters; hence, we can always conveniently obtain the asymptotic $p$-value for our proposed test. 
	
	\item[(b)] {\it Minimax optimal detection boundary}: We define a quantity $\omega_n:=\omega_n(\mu_1,\mu_2,\ldots,\mu_K) $ in \eqref{def:omega_n} that measures the difference among the $K$ group-wise mean PMFs. It satisfies that $\omega_n=0$ if and only if the null hypothesis holds, and it has been properly normalized so that $\omega_n$ is bounded under the alternative hypothesis (provided some mild regularity conditions hold). 
	We show that the proposed test has an asymptotic full power if
	$
	\omega_n^4n^2{\bar N}^2/(Kp) \to\infty. 
	$
	We also provide a matching lower bound by showing that the null hypothesis and the alternative hypothesis are asymptotically indistinguishable if 
	$
	\omega_n^4n^2\bar{N}^2/(Kp)\to 0. 
	$
	Therefore, the proposed test is minimax optimal. Furthermore, in the boundary case where $\omega_n^4n^2\bar{N}^2/(Kp)\to c_0$ for a constant $c_0>0$, 
	we show that  $\psi\to N(0,1)$ under $H_0$, and $\psi\to N(c_1, 1)$, under a specific alternative hypothesis $H_1$ in \eqref{Boundary1}, 
	with $c_1$ being an explicit function of $c_0$. 
\end{enumerate}

To the best of our knowledge, this testing problem for a general $K$ has not been studied before. The existing works primarily focused on closeness testing and authorship attribution (see Section~\ref{subsec:literature}), which are special cases with $K=2$. In comparison, our test is applicable to any value of $K$, offering a unified solution to multiple applications. 
Even for $K=2$, the existing works do not provide a test statistic that has a tractable null distribution. They determined the rejection region and calculated $p$-values using either a (conservative) large-deviation bound or a permutation procedure. Our test is the first one equipped with a tractable null distribution. 
Our results about the optimal detection boundary for a general $K$ are also new to the literature.  
By varying $K$ in our theory, we obtain the optimal detection boundary for different sub-problems. For some of them (e.g., global testing for topic models, authorship attribution with moderate sparsity), the optimal detection boundary was not known before; hence, our results help advance the understanding of the statistical limits of these problems.

\subsection{Related literature} \label{subsec:literature}


First, we make a connection to discrete distribution inference. 
Let $X\sim \mathrm{Multinomial}(N,\Omega)$ represent a size-$N$ sample from a discrete distribution with $p$ categories. The one-sample closeness testing aims to test $H_0: \Omega=\mu$, for a given PMF $\mu$. 
Existing works focus on finding the minimum separation condition in terms of the $\ell^1$-norm or $\ell^2$-norm of $\Omega-\mu$. 
\cite{balakrishnan2019hypothesis} derived the minimum $\ell^1$-separation condition and 
proposed a truncated chi-square test to achieve it. 
\cite{valiant2017automatic} studied the ``local critical radius", a local separation condition that depends on the ``effective sparsity" of $\mu$, and they proposed a ``2/3rd + tail" test to achieve it. 
In the two-sample closeness testing problem, given $X_1\sim \mathrm{Multinomial}(N_1,\Omega_1)$ and $X_2\sim \mathrm{Multinomial}(N_2,\Omega_2)$, it aims to test $H_0: \Omega_1=\Omega_2$.
Again, this  literature focuses on finding the minimum separation condition in terms of the $\ell^1$-norm or $\ell^2$-norm of $\Omega_1-\Omega_2$. 
When $N_1=N_2$, \cite{chan2014optimal} derived the minimum $\ell^1$-separation condition and proposed a weighted chi-square test to attain it. 
\cite{bhattacharya2015testing} extended their results to the unbalanced case where $N_1\neq N_2$, assuming $\|\Omega_1-\Omega_2\|_1\geq p^{-1/12}$. 
This assumption was later removed by \cite{diakonikolas2016new}, who established the minimum $\ell^1$-separation condition in full generality. \cite{kim2022minimax} proposed a two-sample kernel $U$-statistic and showed that it attains the minimum $\ell^2$-separation condition. 

Since the two-sample closeness testing is a special case of our problem with $K=2$ and $n_1=n_2=1$, our test is directly applicable. 
An appealing property of our test is its tractable asymptotic null distribution of $N(0,1)$. 
In contrast, for the chi-square statistic in \cite{chan2014optimal}  or the $U$-statistic in \citep{kim2022minimax}, the rejection region is determined by either an upper bound from concentration inequalities or a permutation procedure, which may lead to a conservative threshold or need additional computational costs. 
Regarding the testing power, we show in Section~\ref{subsec:closeTesting} that our test achieves the minimum $\ell^2$-separation condition, i.e., our method is an optimal ``$\ell^2$ testor." 
Our test can also be turned into an optimal ``$\ell^1$ testor" (a test that achieves the minimum $\ell^1$-separation condition) by re-weighting terms in the test statistic (see Section~\ref{subsec:closeTesting}).

Another related problem is the independence testing \citep{diakonikolas2016new,berrett2019nonparametric}. 
Given i.i.d. bivariate samples from the joint distribution of discrete variables $I$ and $J$, it aims to test if $I$ and $J$ are independent. This is connected to our testing problem with $K=n$, as in this case our null hypothesis implies that the word distribution is independent of the document label. However, the data generating processes in two problems are not the same. In independence testing, it is assumed that the vectorization of $X$ follows a multinomial distribution with $n\bar{N}$ trials and $np$ possible outcomes. In our problem, each $X_i$ follows a multinomial distribution with $N_i$ trials and $p$ possible outcomes. Hence, we cannot directly apply existing results from independence testing. In addition, we allow $K$ to be any integer in $[2, n]$. When $K\neq n$, it is unknown how to relate independence testing to our problem.

Next, we make a connection to text mining. In this literature, a multinomial vector $X\sim \mathrm{Multinomial}(N,\Omega)$ represents the word counts for a document of length $N$ written with a dictionary containing $p$ words. 
In a topic model, each $\Omega_i$ is a convex combination of $M$ ``topic vectors": $\Omega_i=\sum_{k=1}^M w_i(k)A_k$, where each $A_k\in\mathbb{R}^p$ is a PMF and the combination coefficient vector $w_i\in\mathbb{R}^K$ is called the ``topic weight" vector for document $i$.  Given a collection of documents $X_1,X_2,\ldots,X_n$, the global testing problem aims to test $M=1$ versus $M>1$. Interestingly, the optimal detection boundary for this problem has never been rigorously studied. 
As we have explained, this problem is a special case of our testing problem with $K=n$. 
Our results (a) provide a test statistic that has a tractable null distribution and (b) reveal that the optimal detection boundary is $\omega^2_n\asymp(\sqrt{n}\bar{N})^{-1}\sqrt{p}$. 
Both (a) and (b) are new results. 
When comparing our results with those about estimation of $A_k$'s \citep{ke2017new}, it suggests that global testing requires a strictly lower signal strength than topic estimation. 

For authorship attribution, \cite{kipnis2022higher} treats the corpus from a known author as a single document and tests the null hypothesis that this combined document and a new document have the same population word frequencies. It is a two-sample closeness testing problem, except that sparsity is imposed on the difference of two PMFs.  \cite{kipnis2022higher} proposed a test which applies an ``exact binomial test" to obtain a $p$-value for each word and combines these $p$-values using Higher Criticism \citep{DJ04}. \cite{donoho2022higher} analyzed this test when the number of ``useful words" is $o(\sqrt{p})$, and they derived a sharp phase diagram  
(a related one-sample setting was studied in \cite{arias2015sparse}). 
In Section~\ref{subsec:authorChallenge}, we show that our test is applicable to this problem and has some nice properties: (a) tractable null distribution; (b) allows for $s\geq c\sqrt{p}$, where $s$ is the number of useful words; and (c) does not require
documents from the known author to have identical population word frequencies, making the setting more realistic.  
On the other hand, when $s=o(\sqrt{p})$, our test is less powerful than the one in \cite{kipnis2022higher,donoho2022higher}, as our test does not utilize sparsity explicitly. 
We can further improve our test in this regime by modifying the DELVE statistic to incorporate sparsity (see the remark in Section~\ref{subsec:authorChallenge}).


\medskip

The rest of this paper is arranged as follows. 
In Section~\ref{sec:Method}, we introduce the test statistic and explain the rationale behind it. We then present in Section~\ref{sec:Theory} the main theoretical results, including the asymptotic null distribution, power analysis, a matching lower bound, the study of two special cases ($K=n$ and $K=2$), and a discussion of the contiguity regime. Section~\ref{sec:Applications} applies our results to text mining and discrete distribution testing.  Simulations are  in Section~\ref{sec:Simu} and real data analysis is  in Section~\ref{sec:RealData}.  
The paper is concluded with a discussion in Section~\ref{sec:Discuss}. All proofs are in 
\cite{DELVE-supp}. 

\section{The DELVE Test} \label{sec:Method}

Recall that $X_1,\ldots,X_n$ are independent, and $X_i\sim \mathrm{Multinomial}(N_i, \Omega_i)$ for $1\leq i\leq n$. There is a known partition $\{1,2,\ldots,n\}=\cup_{k=1}^K S_k$. Write $n_k=|S_k|$, $\bar{N}_k=n_k^{-1}\sum_{i\in S_k}N_i$, and $\bar{N}=n^{-1}\sum_{i=1}^n N_i$. In \eqref{Mod2-groupMean}, we have defined the group-wise mean PMF $\mu_k=(n_k\bar{N}_k)^{-1}\sum_{i\in S_k}N_i\Omega_i$. We further define the overall mean PMF $\mu \in\mathbb{R}^p$ by 
\beq \label{define:eta}
\mu: =\frac{1}{n\bar{N}}\sum_{k=1}^Kn_k\bar{N}_k\mu_k=\frac{1}{n\bar{N}}\sum_{i=1}^n N_i\Omega_i.  
\eeq
We introduce a quantity $\rho^2=\rho^2(\mu_1,\ldots,\mu_K)$ by 
\beq \label{def:rhoSquare}
\rho^2: = \sum_{k=1}^Kn_k\bar{N}_k \|\mu_{k}-\mu\|^2. 
\eeq
This quantity measures the variations across $K$ group-wise mean PMFs. It is true that
the null hypothesis \eqref{Mod3-null} holds if and only if $\rho^2=0$. 
Inspired by this observation, we hope to construct an unbiased estimator of $\rho^2$ and develop it to a test statistic. 

We can easily obtain the minimum variance unbiased estimators of $\mu_k$ and $\mu$:
\beq \label{define:etaHat}
\hat{\mu}_k = \frac{1}{n_k\bar{N}_k}\sum_{i\in S_k}X_i, \qquad\mbox{and}\qquad \hat{\mu} = \frac{1}{n\bar{N}}\sum_{k=1}^K n_k\bar{N}_k\hat{\mu}_k =  \frac{1}{n\bar{N}}\sum_{i=1}^n X_i. 
\eeq
For each $1\leq j\leq p$, let $\mu_{kj}$, $\mu_j$, $\hat{\mu}_{kj}$ and $\hat{\mu}_j$ represent the $j$th entry of $\mu_k$, $\mu$, $\hat{\mu}_k$ and $\hat{\mu}$, respectively.  
A naive estimator of $\rho^2$ is
\beq \label{define:tildeT}
\widetilde{T}=\sum_{j=1}^p \widetilde{T}_j, \qquad\mbox{where}\quad \widetilde{T}_j = \sum_{k=1}^K n_k\bar{N}_k(\hat{\mu}_{kj}-\hat{\mu}_j)^2. 
\eeq
This estimator is biased. In Section \ref{subsec:proof-prop-unbiased} of \cite{DELVE-supp}, we show that
$
\E[\widetilde{T}_j]=\sum_{k=1}^K \bigl[n_k\bar{N}_k (\mu_{kj}-\mu_j)^2 +\bigl( \frac{1}{n_k\bar{N}_k}-\frac{1}{n\bar{N}} \bigr)\sum_{i\in S_k}N_i\Omega_{ij}(1-\Omega_{ij})\bigr]. 
$
It motivates us to debias $\widetilde{T}_j$ by using an unbiased estimate of $\Omega_{ij}(1-\Omega_{ij})$. By basic properties of multinomial distributions, $\E[X_{ij}(N_i-X_{ij})] = N_i(N_i-1)\Omega_{ij}(1-\Omega_{ij})$. We thereby use $\frac{1}{N_i(N_i-1)}X_{ij}(N_i-X_{ij})$ to estimate $\Omega_{ij}(1-\Omega_{ij})$. It yields an unbiased estimator of $\rho^2$:
\beq \label{DELAC}
T=\sum_{j=1}^p T_j, \quad T_j=  \sum_{k=1}^K\biggl[ n_k\bar{N}_k(\hat{\mu}_{kj}-\hat{\mu}_{j})^2 -   \Bigl( \frac{1}{n_k\bar{N}_k}-\frac{1}{n\bar{N}} \Bigr)\sum_{i\in S_k}\frac{X_{ij}(N_i-X_{ij})}{N_i-1}  \biggr]. 
\eeq 
\begin{lemma} \label{prop:unbiased}
	Under Models~\eqref{Mod1-data}-\eqref{Mod2-groupMean},
	the estimator in \eqref{DELAC} satisfies that $\E[T]=\rho^2$. 
\end{lemma}

To use $T$ for hypothesis testing, we need a proper standardization of this statistic. In Sections \ref{subsec:T-decompose}-\ref{subsec:T-var} of \cite{DELVE-supp}, we study $\VV(T)$, the variance of $T$. Under mild regularity conditions, it can be shown that $\VV(T)=\Theta_n\cdot [1+o(1)]$, where
\begin{align} \label{define:Theta_n}
	&\Theta_n: = 4\sum_{k=1}^K\sum_{j=1}^p n_k\bar{N}_k (\mu_{kj}-\mu_j)^2\mu_{kj} + 2\sum_{k=1}^K\sum_{i\in S_k}\sum_{j=1}^p  \Bigl(\frac{1}{n_k\bar{N}_k}-\frac{1}{n\bar{N}}\Bigr)^2\frac{N_i^3}{N_i-1} \Omega_{ij}^2\\
	&+ \frac{2}{n^2\bar{N}^2}\sum_{1\leq k\neq \ell\leq K}\sum_{i\in S_k}\sum_{m\in S_\ell} \sum_{j=1}^pN_iN_m\Omega_{ij}\Omega_{mj}  + 2\sum_{k=1}^K \sum_{\substack{i\in S_k, m\in S_k,\\ i\neq m}}\sum_{j=1}^p \Bigl(\frac{1}{n_k\bar{N}_k}-\frac{1}{n\bar{N}}\Bigr)^2 N_iN_m\Omega_{ij}\Omega_{mj}. \nonumber
\end{align}
In $\Theta_n$, the first term vanishes under the null, so it suffices to estimate the other three terms in $\Theta_n$. By  properties of multinomial distributions, $\E[X_{ij}X_{mj}]=N_iN_m\Omega_{ij}\Omega_{mj}$, $\E[X^2_{ij}] = N_i^2\Omega_{ij}^2 +N_i\Omega_{ij}(1-\Omega_{ij})$, and $\E[X_{ij}(N_i-X_{ij})]=N_i(N_i-1)\Omega_{ij}(1-\Omega_{ij})$.  It inspires us to estimate $\Omega_{ij}\Omega_{mj}$ by $\frac{X_{ij}X_{mj}}{N_iN_m}$ and estimate $\Omega_{ij}^2$ by  $\frac{X_{ij}^2}{N_i^2} - \frac{X_{ij}(N_i-X_{ij})}{N^2_i(N_i-1)}=\frac{X_{ij}^2-X_{ij}}{N_i(N_i-1)}$. Define
\begin{align} \label{define:V}
	&V =  2\sum_{k=1}^K\sum_{i\in S_k}\sum_{j=1}^p  \Bigl(\frac{1}{n_k\bar{N}_k}-\frac{1}{n\bar{N}}\Bigr)^2\frac{X_{ij}^2-X_{ij}}{N_i(N_i-1)}  
	+\frac{2}{n^2\bar{N}^2}\sum_{k\neq\ell}\sum_{i\in S_k}\sum_{m\in S_\ell} \sum_{j=1}^p X_{ij}X_{mj} \cr
	&\qquad +  2\sum_{k=1}^K \sum_{\substack{i\in S_k, m\in S_k,\\ i\neq m}}\sum_{j=1}^p \Bigl(\frac{1}{n_k\bar{N}_k}-\frac{1}{n\bar{N}}\Bigr)^2 X_{ij}X_{mj}.
\end{align}
The test statistic we propose is as follows (in the rate event $V<0$, we simply set $\psi=0$):
\beq \label{define:psi}
\psi = T/\sqrt{V}. 
\eeq
We call $\psi$ the {\it DEbiased and Length-adjusted Variability Estimator (DELVE)}. In Section~\ref{subsec:Main-null}, we show that under mild regularity conditions, $\psi\to N(0,1)$ under the null hypothesis. For any fixed $\kappa\in (0,1)$, the asymptotic level-$\kappa$ DELVE test rejects $H_0$ if  
\beq \label{define:Reject}
\psi > z_{\kappa}, \qquad\mbox{where $z_{\kappa}$ is the $(1-\kappa)$-quantile of $N(0,1)$}.  
\eeq

\begin{remark}[Other testing ideas]{\rm
The likelihood ratio (LR) test can only be applied when $\Omega_i$'s are equal within each group (in this case, the null/alternative hypotheses have much fewer free parameters). Moreover, the DELVE test attains the minimax optimal detection boundary in high-dimensional settings, but there is no such guarantee for the LR test. 
From simulations  in Section~\ref{sec:Simu}, when $p$ is large, DELVE has better power than LR. 
Another idea is to use the ANOVA statistic $\widetilde{T}$ in \eqref{define:tildeT} without de-biasing and apply 
a chi-square approximation or permutation procedure to compute the $p$-value. This test is unfortunately suboptimal. 
There are settings in which the bias term dominates the ``signal" term in $\widetilde{T}$, causing the test to lose power (see Remark \ref{debiasing-effect-remark} for details). 
}\end{remark}

\begin{remark}{\rm
We have assumed $X_1, \ldots, X_n$ are independent. This is better interpreted as the conditional independence given $\Omega_i$'s. When $\Omega_i$'s are random and have some dependence structure, $X_i$'s can be (marginally) dependent. We will see in Section~\ref{sec:Theory} that the asymptotic null distribution of $\psi$ does not depend on $\Omega_i$'s; then, the same asymptotic distribution also holds for random and dependent $\Omega_i$. We have also assumed that the distribution of $X_i$ is multinomial. However, our test only uses the first two moments of multinomials, not the likelihood. As a result,  our method is relatively robust to model misspecification, and it is extendable to settings with under/over dispersion.
}\end{remark}

\subsection{The special cases of $K=n$ and $K=2$} \label{subsec:test-2cases} 

As seen in Section~\ref{sec:Intro}, the application examples of $K=n$ and $K=2$ are particularly intriguing. In these cases, we give more explicit expressions of our test statistic. 

When $K=n$,  we have $S_k=\{i\}$ and $\hat{\mu}_{kj}=N_i^{-1}X_{ij}$. The null hypothesis becomes
$
H_0: \Omega_1=\Omega_2=\ldots=\Omega_n. 
$
The statistic in \eqref{DELAC} reduces to
\beq \label{DELAC(K=n)}
T =\sum_{j=1}^p  \sum_{i=1}^n \biggl[ \frac{(X_{ij}-N_i\hat{\mu}_j)^2}{N_i} -    \Bigl( 1-\frac{N_i}{n\bar{N}} \Bigr)\frac{X_{ij}(N_i-X_{ij})}{N_i(N_i-1)}  \biggr]. 
\eeq
Moreover, in the variance estimate \eqref{define:V}, the last term is exactly zero, and it can be shown that the third term is negligible compared to the first term. We thereby consider a simpler variance estimator by only retaining the first term in \eqref{define:V}:
\beq \label{define:V(K=n)}
V^* =2 \sum_{i=1}^n\sum_{j=1}^p \Bigl(\frac{1}{N_i}-\frac{1}{n\bar{N}}\Bigr)^2 \frac{X_{ij}^2-X_{ij}}{N_i(N_i-1)}. 
\eeq 
The simplified DELVE test statistic is $\psi^*=T/\sqrt{V^*}$.

When $K=2$, we observe two collections of multinomial vectors, denoted by $\{X_i\}_{1\leq i\leq n}$ and $\{G_i\}_{1\leq i\leq m}$. We assume for $1\leq i\leq n$ and $1\leq j\leq m$, 
\beq \label{model-2sample-0}
X_i \sim \mathrm{Multinomial}(N_i, \Omega_i), \qquad 
G_j \sim \mathrm{Multinomial}(M_j, \Gamma_j). 
\eeq
Write $\bar{N}=n^{-1}\sum_{i=1}^n N_i$ and $\bar{M}=m^{-1}\sum_{i=1}^m M_i$. The null hypothesis becomes
\beq\label{define:eta(K=2)-0}
H_0:\quad \eta=\theta, \qquad \mbox{where }\eta = \frac{1}{n\bar{N}}\sum_{i=1}^n N_i \Omega_{i}, \mbox{ and }  \theta = \frac{1}{m\bar{M}}\sum_{i=1}^m M_i \Gamma_{i},
\eeq 
where $\theta$ and $\eta$ are the two group-wise mean PMFs.  
We estimate them by $\hat{\eta} = (n\bar{N})^{-1}\sum_{i=1}^n X_i$ and $\hat{\theta} = (m\bar{M})^{-1}\sum_{i=1}^m G_i$.   
The statistic in \eqref{DELAC} has an equivalent form as follows: 
\beq \label{DELAC(K=2)}
T =\frac{n\bar{N}m\bar{M}}{n\bar{N}+m\bar{M}}\biggl[ \|\hat{\eta}-\hat{\theta}\|^2 - \sum_{i=1}^n\sum_{j=1}^p \frac{X_{ij}(N_i-X_{ij})}{n^2\bar{N}^2 (N_i-1)}-\sum_{i=1}^m \sum_{j=1}^p \frac{G_{ij}(M_i-G_{ij})}{m^2\bar{M}^2(M_i-1)}\biggr]. 
\eeq
The variance estimate \eqref{define:V} has an equivalent form as follows:
\begin{align} \label{define:V(K=2)}
	&V = \frac{4\sum_{i=1}^n\sum_{i'=1}^m\sum_{j=1}^p  X_{ij} G_{i'j}}{(n\bar{N}+m\bar{M})^2} +  \frac{2m^2\bar{M}^2 \big[\sum_{i=1}^n \frac{X_{ij}^2-X_{ij}}{N_i(N_i-1)} +  \sum_{1\leq i \neq i'\leq n}  X_{ij} X_{i'j} \big]}{n^2 \bar{N}^2(n\bar{N}+m\bar{M})^2} \cr 
	&\quad   +  \frac{2n^2\bar{N}^2\big[ \sum_{i=1}^m \frac{G_{ij}^2-G_{ij}}{M_i(M_i-1)} +\sum_{1\leq i \neq i'\leq m}  G_{ij} G_{i'j} \big]}{m^2 \bar{M}^2(n\bar{N}+m\bar{M})^2} . 
\end{align}
The DELVE test statistic is $\psi=T/\sqrt{V}$.

\subsection{A variant: DELVE+} \label{subsec:modification}
We introduce a variant of the DELVE test statistic to better suit real data.  Let $\hat{\mu}$, $T$ and $V$ be as in \eqref{define:etaHat}, \eqref{DELAC} and \eqref{define:V}. Define
\beq \label{define:psi+}
\psi^+ = T/\sqrt{V^+}, \qquad \mbox{where}\quad V^+=V\cdot \bigl( 1+ \| \hat \mu \|_2 T/\sqrt{V}\bigr). 
\eeq
We call \eqref{define:psi+} the DELVE+ test statistic. In theory, this modification has little effect on the key properties of the test. To see this, we note that $\|\hat{\mu}\|_2=o_{\mathbb{P}}(1)$ in high-dimensional settings. Suppose $T/\sqrt{V}\to N(0,1)$ under $H_0$. Since  $\|\hat{\mu}\|_2\to 0$, it is seen immediately that $V^+/V\to 1$; hence, the asymptotic normality also holds for $\psi^+$. Suppose $T/\sqrt{V}\to \infty$ under the alternative hypothesis. It follows that $V^+\leq 2\max\{V, \|\hat{\mu}\|_2\cdot T\sqrt{V}\}$ and $\psi^+\geq \frac{1}{\sqrt{2}}\min\{T/\sqrt{V}, \, \|\hat{\mu}\|_2^{-1}(T/\sqrt{V})^{1/2}\}\to\infty$. We have proved the following lemma:

\begin{lemma} \label{prop:modification}
	As $n\bar{N}\to\infty$, suppose $\|\hat{\mu}\|_2\to 0$ in probability. Under $H_0$, if $T/\sqrt{V}\to N(0,1)$, then $T/\sqrt{V^+}\to N(0,1)$. Under $H_1$, if $T/\sqrt{V}\to \infty$, then $T/\sqrt{V^+}\to\infty$. 
\end{lemma}

\noindent
In practice, this modification avoids  extremely small $p$-values. In some real datasets, $V$ is very small and leads to an extremely small $p$-value in the original DELVE test.  In DELVE+,  as long as $T$ is positive, $\psi^+$ is  smaller than $\psi$, so that the $p$-value is adjusted. 

In the numerical experiments, we consider both DELVE and DELVE+. For theoretical analysis, since these two versions  have almost identical theoretical properties, we only focus on the original DELVE test statistic.

\section{Theoretical Properties} \label{sec:Theory}
We first present the regularity conditions. For a constant $c_0\in (0,1)$, we assume
\beq \label{cond1-basic}
\min_{1\leq i\leq n} N_i \geq 2, \qquad  \max_{1\leq i\leq n} \| \Omega_i \|_\infty \leq 1 - c_0, \qquad 
\max_{1\leq k\leq K} \frac{n_k \bar{N}_k}{n \bar{N}} \leq 1 - c_0. 
\eeq 
In \eqref{cond1-basic}, the first condition is mild. Noting that $\|\Omega_i\|_1=1$, the second condition excludes those cases where one of the $p$ categories has an extremely dominating probability in the PMF $\Omega_i$, which is also mild. In the third condition, $n_k\bar{N}_k$ is the total number of counts  in all multinomials of group $k$, and this condition excludes the extremely unbalanced case where one group occupies the majority of counts (in the special case of $K = 2$, we further relax this condition to allow for severely unbalanced groups (see Section \ref{subsec:K=2})). 

Recall that $\mu_k=\frac{1}{n_k\bar{N}_k}\sum_{i\in S_k} N_i\Omega_i$ is the mean PMF within group $k$. We also define a `covariance' matrix of PMFs for group $k$ by $\Sigma_k = \frac{1}{n_k\bar{N}_k}\sum_{i\in S_k} N_i\Omega_i\Omega_i'$. 
Let
\beq \label{def:alpha_n}
\alpha_n:= \max\left\{ \sum_{k = 1}^K \frac{\| \mu_k \|_3^3}{n_k \bar{N}_k} ,  \quad  \sum_{k = 1}^K \frac{\| \mu_k \|^2}{n_k^2 \bar{N}_k^2}\right \} \bigg / \bigg( \sum_{k = 1}^K \| \mu_k \|^2 \bigg)^2,
\eeq
and
\beq \label{def:beta_n}
\beta_n: =  \max \biggl\{ \sum_{k=1}^K  \sum_{i\in S_k}\frac{N^2_i}{n_k^2 \bar{N}_k^2}\|\Omega_i\|_3^3, \quad \sum_{k=1}^K \| \Sigma_k \|_F^2
\bigg\}\bigg/ (K \| \mu \|^2). 
\eeq
We assume that as $n\bar{N}\to\infty$, 
\beq \label{cond2-regular}
\alpha_n
=o(1), \qquad \beta_n=o(1), \qquad \mbox{and} \quad \frac{\| \mu \|_4^4}{K\|\mu\|^4} = o(1). 
\eeq
Here $\alpha_n$ and $\beta_n$ only depend on group-wise quantities, such as $\mu_k$, $\Sigma_k$ and $\sum_{i\in S_k}N^2_i\|\Omega_i\|_3^3$; hence, a small number of  `outliers' (i.e., extremely large entries) in $\Omega$ has little effect on $\alpha_n$ and $\beta_n$. Furthermore, in a simple case where $\max_k n_k\leq C\min_k n_k$,  $\max_{k}\bar{N}_k\leq C\min_k \bar{N}_k$ and $\|\Omega\|_{\max}=O(1/p)$, it holds that $\alpha_n=O(\max\{\frac{1}{n\bar{N}}, \frac{Kp}{n^2\bar{N}^2}\})$, $\beta_n=O(\max\{\frac{K^2}{n^2p}, \frac{1}{p}\})$ and $\frac{\| \mu \|_4^4}{K\|\mu\|^4} =O(\frac{1}{Kp})$. 
When $n\bar{N}\to\infty$ and $p\to\infty$, \eqref{cond2-regular} reduces to $n^2\bar{N}^2/(Kp)\to\infty$. 
This condition is necessary for successful testing, because our lower bound in Section~\ref{subsec:Main-LB} implies that the two hypotheses are asymptotically indistinguishable if $n^2\bar{N}^2/(Kp)\to 0$. 

\subsection{The asymptotic null distribution} \label{subsec:Main-null}
Under the null hypothesis, the $K$ group-wise mean PMFs $\mu_1,\mu_2,\ldots,\mu_K$, are equal to each other, but this hypothesis is still highly composite, as $(N_i,\Omega_i)$ are not necessarily the same within each group. We show that the DELVE test statistic always enjoys a parameter-free asymptotic null distribution. Let $T$, $\Theta_n$ and $V$ be as in \eqref{DELAC}-\eqref{define:V}. 
The next two theorems are proved in \cite{DELVE-supp}.

\begin{thm} \label{thm:null}
	Consider Models~\eqref{Mod1-data}-\eqref{Mod2-groupMean}, where the null hypothesis \eqref{Mod3-null} holds. Suppose \eqref{cond1-basic} and \eqref{cond2-regular} are satisfied.  
	As $n\bar{N}\to\infty$, $T/\sqrt{\Theta_n} \to N(0,1)$ in distribution. 
\end{thm}
\begin{thm} \label{thm:null2}
	Under the conditions of Theorem~\ref{thm:null}, as $n{\bar N}\to\infty$, $V/\Theta_n \to 1$ in probability, and 
	$
	\psi:=T/\sqrt{V} \to   N(0,1)
	$ in distribution. 
\end{thm}

By Theorem~\ref{thm:null2}, the asymptotic $p$-value is $1-\Phi(\psi)$, where $\Phi(\cdot)$ is the CDF of $N(0,1)$. For any $\kappa\in (0,1)$, the rejection region of the asymptotic level-$\kappa$ test is as given in \eqref{define:Reject}. 

The proofs of Theorems~\ref{thm:null}-\ref{thm:null2} contain two key steps. In the first step, we decompose $T$ into mutually uncorrelated terms. Define a set of independent, mean-zero random vectors $
\{Z_{ir}\}_{1\leq i\leq n, 1\leq r\leq N_i}$, where  $Z_{ir}\sim\mathrm{Multinomial}(1, \Omega_{i})-\Omega_i$. Then, $X_i=N_i\Omega_i+\sum_{r=1}^{N_i}Z_{ir}$ (in distribution). 
We plug it into \eqref{DELAC} to get $T=T_1+T_2+T_3+T_4$, where $T_1$ is a linear form of $\{Z_{ir}\}$, $T_2$-$T_4$ are quadratic forms of $\{Z_{ir}\}$, and $T_1$-$T_4$ are uncorrelated (see Section~\ref{sec:decomposition} of \cite{DELVE-supp}). In the second step, we construct a martingale for each term $T_j$. 
This is accomplished by re-arranging the double-index sequence $Z_{ir}$ to a single-index sequence and successively adding terms in this sequence to $T_j$.  We then apply the martingale central limit theorem (CLT) \citep{hall2014martingale} 
to prove asymptotic normality of each $T_j$. The asymptotic normality of $T$ follows by identifying the dominating terms in $T_1$-$T_4$ (as model parameters change, the dominating terms also change) and studying their joint distribution. 
This step involves extensive calculations to bound conditional variances and verify the Lindeberg conditions of martingale CLT, as well as subtle uses of the Cauchy-Schwarz inequality to simplify moment bounds. 

\begin{remark}[An adjustment when $p=O(1)$]
{\rm
		While we focus on high-dimensional settings, the case of $p=O(1)$ is still of interest. In this case,  the variance estimator $V$ may not be consistent. We propose a refined estimator $\widetilde{V}$ in Section~\ref{sec:Finite-p} of \cite{DELVE-supp}. When $V$ is replaced by $\widetilde{V}$, $\psi\to N(0,1)$ continues to hold.
}\end{remark}

\subsection{Power analysis} \label{subsec:Main-power}
Under the alternative hypothesis, the PMFs $\mu_1,\mu_2,\ldots,\mu_K$ are not the same. In Section~\ref{sec:Method}, we introduce a quantity $\rho^2$ (see \eqref{def:rhoSquare}) to capture the total variation in $\mu_k$'s, but this quantity is not scale-free. We define a scaled version of $\rho^2$ as 
\beq \label{def:omega_n}
\omega_n=\omega_n(\mu_1,\mu_2,\ldots,\mu_K): = \frac{1}{n\bar{N}\|\mu\|^2} \sum_{k=1}^Kn_k\bar{N}_k \|\mu_{k}-\mu\|^2. 
\eeq
It is seen that $\omega_n\leq \max_{k}\{\frac{\|\mu_k-\mu\|^2}{\|\mu\|^2}\}$, which is properly scaled. 

\begin{thm} \label{thm:alt}
	Consider Models~\eqref{Mod1-data}-\eqref{Mod2-groupMean}, where \eqref{cond1-basic} and \eqref{cond2-regular} are satisfied. 
	Then, $\E[T]=n\bar{N}\|\mu\|^2\omega_n^2$, and $\VV(T)= O\bigl(\sum_{k=1}^K\|\mu_k\|^2\bigr) + \E[T]\cdot O\bigl(\max_{1\leq k\leq K}\|\mu_k\|_\infty\bigr)$. 
\end{thm}

For the DELVE test to have an asymptotically full power, we need $\E[T]\gg \sqrt{\VV(T)}$. By Theorem~\ref{thm:alt}, this is satisfied if $\E[T]\gg \sqrt{\sum_k \|\mu_k\|^2}$ and $\E[T]\gg \max_k\|\mu_k\|_\infty$. Between these two requirements, the latter one is weaker; hence, we only need $\E[T]\gg \sqrt{\sum_{k=1}^K \|\mu_k\|^2}$. It gives rise to the following theorem:

\begin{thm} \label{thm:alt2}
	Under the conditions of Theorem~\ref{thm:alt}, we further assume that under the alternative hypothesis, as $n\bar{N}\to\infty$, 
	\beq \label{SNR}
	\mathrm{SNR}_n:= 
	\frac{n\bar{N}\|\mu\|^2\omega_n^2}{\sqrt{\sum_{k=1}^K \|\mu_k\|^2}}\;\; \to\;\;\infty. 
	\eeq  
	Under the alternative hypothesis, $\psi\to\infty$ in probability. For any fixed $\kappa\in (0,1)$, the level-$\kappa$ DELVE test has an asymptotic level of $\kappa$ and an asymptotic power of $1$. If we choose $\kappa=\kappa_n$ such that $\kappa_n\to 0$ and $1 - \Phi(\mathrm{SNR}_n) = o(\kappa_n)$, where $\Phi$ is the CDF of $N(0,1)$,
	then the sum of type I and type II errors of the DELVE test converges to $0$.  
\end{thm}

The detection boundary in \eqref{SNR} has simpler forms in some special cases. For example, if $\|\mu_k\|\asymp \|\mu\|$ for $1\leq k\leq K$, then $\mathrm{SRN}_n\asymp n\bar{N}\omega_n^2\|\mu\|/\sqrt{K}$. If, furthermore, all entries of $\mu$ are at the same order, which implies $\|\mu\|\asymp p^{-1/2}$, then $\mathrm{SRN}_n\asymp n^2\bar{N}^2\omega_n^2/\sqrt{Kp}$. In this case, the detection boundary simplifies to
$
\omega_n^4n^2\bar{N}^2/(Kp) \to\infty. 
$

\begin{remark}[The effect of de-biasing on power]
\label{debiasing-effect-remark}
{\rm
Let $\widetilde{T}$ be the statistic in \eqref{define:tildeT} without bias correction. 
Under $H_1$, when $\mathrm{SNR}_n\to\infty$ but $n\bar{N}\ll Kp$, the bias in $\widetilde{T}$ can dominate the ``signal" $\rho^2$. Consequentely, any test based on $\widetilde{T}$ has no power (details and examples are in Section~\ref{supp:ANOVA} of \cite{DELVE-supp}). This shows that de-biasing is critical for achieving not only parameter-free limiting null but also good power.
}\end{remark}

\subsection{A matching lower bound} \label{subsec:Main-LB}
We have seen that the DELVE test successfully separates two hypotheses if $\mathrm{SNR}_n\to\infty$, where $\mathrm{SNR}_n$ is as defined in \eqref{SNR}. We now present a lower bound to show that the two hypotheses are asymptotically indistinguishable if $\mathrm{SNR}_n\to 0$. 

Let $\ell_i\in \{1,2,\ldots,K\}$ denote the group label of $X_i$.  
Write $\xi=\{(N_i, \Omega_i, \ell_i)\}_{1\leq i\leq n}$. 
Let $\mu_k$, $\alpha_n$, $\beta_n$, and $\omega_n$ be the same as defined in \eqref{Mod2-groupMean}, \eqref{def:alpha_n}, \eqref{def:beta_n}, and \eqref{def:omega_n}, respectively. For each given $(n,p,K,\bar{N})$, we write $\mu_k=\mu_k(\xi)$ to emphasize its dependence on parameters, and similarly for $\alpha_n,\beta_n,\omega_n$. For any $c_0\in (0,1)$ and sequence $\epsilon_n$, define
\beq \label{LB-param-class}
{\cal Q}_n(c_0, \epsilon_n) := \Big\{ \xi=\{(N_i, \Omega_i, \ell_i)\}_{i = 1 }^n: \,  \mbox{\eqref{cond1-basic} holds for $c_0$}, \, \, \max (\alpha_n(\xi), \beta_n(\xi) ) \leq \epsilon_n  \Big\}
\eeq
Furthermore, for any sequence $\delta_n$, we define a parameter class for the null hypothesis and a parameter class for the alternative hypothesis: 
\begin{align} \label{LB-param-class2}
	{\cal Q}_{0n}^*(c_0,\epsilon_n)& ={\cal Q}_n(c_0, \epsilon_n)\cap \left\{\xi: \omega_n(\xi)=0 \right\}, \cr
	{\cal Q}_{1n}^*(\delta_n; c_0,\epsilon_n)& ={\cal Q}_n(c_0, \epsilon_n)\cap \left\{\xi:   \frac{n\bar{N}\|\mu(\xi)\|^2 \omega^2_n(\xi)}{\sqrt{\sum_{k=1}^K\|\mu_k(\xi)\|^2}} \geq \delta_n\right \}. 
\end{align}
\begin{thm} \label{thm:LB}
	Fix a constant $c_0\in (0,1)$ and positive sequences $\epsilon_n$ and $\delta_n$ such that $\epsilon_n\to 0$ as $n\to\infty$. 
	For any sequence of $(n, p,K, \bar{N})$ indexed by $n$, consider 
	Models~\eqref{Mod1-data}-\eqref{Mod2-groupMean} for $\Omega\in {\cal Q}_n(c_0, \epsilon_n)$. Let ${\cal Q}_{0n}^*(c_0,\epsilon_n)$ and ${\cal Q}_{1n}^*(\delta_n; c_0,\epsilon_n)$ be as in \eqref{LB-param-class2}. If $\delta_n\to 0$, then 
	$
	\limsup_{n\to\infty} \inf_{\Psi\in\{0,1\}}\bigl\{ \sup_{\xi\in {\cal Q}_{0n}^*(c_0,\epsilon_n)} \mathbb{P}_{\xi}(\Psi=1)+\sup_{\xi\in {\cal Q}_{1n}^*(\delta_n; c_0,\epsilon_n)} \mathbb{P}_{\xi}(\Psi=0)\bigr\}=1.
	$
\end{thm}

\subsection{The special case of $K=2$} \label{subsec:K=2}
The special case of $K=2$ is found in closeness testing and authorship attribution. 
We study this case more carefully. 
Given $\{X_i\}_{1\leq i\leq n}$ and $\{G_i\}_{1\leq i\leq m}$, we assume
\beq \label{model-2sample}
X_i \sim \mathrm{Multinomial}(N_i, \Omega_i), \qquad 
G_j \sim \mathrm{Multinomial}(M_j, \Gamma_j). 
\eeq
Write $\bar{N}=n^{-1}\sum_{i=1}^n N_i$ and $\bar{M}=m^{-1}\sum_{i=1}^m M_i$. The null hypothesis becomes
\beq\label{define:eta(K=2)}
H_0:\quad \eta=\theta, \qquad \mbox{where }\eta = \frac{1}{n\bar{N}}\sum_{i=1}^n N_i \Omega_{i}, \mbox{ and }  \theta = \frac{1}{m\bar{M}}\sum_{i=1}^m M_i \Gamma_{i},
\eeq 
where $\theta$ and $\eta$ are the two group-wise mean PMFs. In this case, the test statistic $\psi$ has a more explicit form as in \eqref{DELAC(K=2)}-\eqref{define:V(K=2)}.

In our previous results for a general $K$, the regularity conditions (e.g., \eqref{cond1-basic}) impose restrictions on the balance of sample sizes among groups.
For $K=2$, the severely unbalanced setting is interesting (e.g., in authorship attribution, $n=1$ and $m$ can be large). We relax the regularity conditions to the following ones:

\begin{condition} \label{cond:K=2}
	Let $\theta$ and $\eta$ be as in \eqref{define:eta(K=2)} and define two matrices $\Sigma_1 = \frac{1}{n\bar{N}}\sum_{i=1}^n N_i\Omega_i\Omega_i'$ and $\Sigma_2 = \frac{1}{m\bar{M}}\sum_{i=1}^m M_i\Gamma_i\Gamma_i'$. We assume that the following statements are true (a) For $1\leq i\leq n$ and $1\leq j\leq m$, $N_i \geq 2$, $\| \Omega_i \|_\infty \leq 1 - c_0$, $M_j\geq 2$, and $\| \Gamma_j \|_\infty \leq 1 - c_0$, where $c_0\in (0,1)$ is a contant, (b) $\max\big\{ 
	\big(  \frac{\| \eta \|_3^3}{ n \bar{N}} + \frac{\| \theta \|_3^3}{ m \bar{M}}   \big), \,\big(\frac{ \| \eta \|_2^2 }{ n^2 \bar{N}^2}+ \frac{ \| \theta \|_2^2 }{ m^2 \bar{M}_2^2}\big)\big\} 
	\big /   \bigl\| \frac{m \bar{M}}{ n\bar{N}+ m \bar{M}} \eta +  
	\frac{n \bar{N}}{ n\bar{N}+ m \bar{M}} \theta  \bigr\|^4=o(1)$, (c) $\max \big\{ \sum_i \frac{N_i^2}{n^2 \bar{N}^2} \| \Omega_i \|_3^3,\,  \sum_i \frac{M_i^2}{m^2 \bar{M}^2} \| \Gamma_i \|_3^3, \, \|\Sigma_1\|_F^2 + \| \Sigma_2 \|_F^2\big\}
	\big/ \| \mu \|^2 =o(1)$, and (d) $\| \mu \|_4^4 / \| \mu \|^4=o(1)$. 
\end{condition}
\noindent
Condition (a) is similar to \eqref{cond1-basic}, except that we drop the sample size balance requirement. Conditions (b)-(d) are equivalent to \eqref{cond2-regular} but have more explicit expressions for $K=2$.

\begin{thm} \label{thm:K=2} 
	In Model \eqref{model-2sample}, we test the null hypothesis $H_0$: $\theta=\mu$. As $\min\{n{\bar N}, m\bar{M}\}\to\infty$, suppose Condition~\ref{cond:K=2} is satisfied.  Under the alternative hypothesis, we further assume 
	\beq \label{SNR(K=2)}
	\frac{ \| \eta - \theta \|^2 }{  \big( \frac{1}{n\bar{N}} + \frac{1}{m \bar{M}}  \big) 
		\max\{\| \eta \|,\, \| \theta \|\} } \to \infty. 
	\eeq
	Consider the DELVE test statistic $\psi = T/\sqrt{V}$. The following statements are true. Under the null hypothesis, $\psi \to N(0,1)$ in distribution. Under the alternative hypothesis, $\psi \to\infty$ in probability. Moreover for any fixed $\kappa\in (0,1)$, the level-$\kappa$ DELVE test has an asymptotic level of $\kappa$ and an asymptotic power of $1$. 
\end{thm}

Compared with the theorems for a general $K$, first, Theorem~\ref{thm:K=2} allows the two groups to be severely unbalanced and reveals that the detection boundary depends on the harmonic mean of $n\bar{N}$ and $m\bar{M}$. Second, the detection boundary is expressed using $\|\eta-\theta\|$, which is easier to interpret. We also note that, when $K = 2$, straightforward calculation yields $\mathbb{E}[T] = \rho^2 = ( \frac{1}{n\bar{N}} + \frac{1}{m\bar{M}} )^{-1} \| \eta - \theta \|^2$, which explains the appearance of the harmonic means in the detection boundary \eqref{SNR(K=2)}.


\subsection{The special case of $K=n$} \label{subsec:K=n}
The special case of $K=n$ is interesting for two reasons. First, the application example of global testing in topic models corresponds to $K=n$. Second, for any $K$, when $\Omega_i$'s within each group are assumed to be the same (e.g., this is the case in closeness testing of discrete distributions), it suffices to aggregate the counts in each group, i.e., let $Y_k=\sum_{i\in S_k}X_i$ and operate on $Y_1,\ldots,Y_K$ instead of the original $X_i$'s; this reduces to the case of $K=n$. 

When $K=n$, the null hypothesis has a simpler form:
\beq \label{null(K=n)}
H_0: \quad \Omega_i=\mu, \qquad 1\leq i\leq n. 
\eeq
Moreover, under the alternative hypothesis, the quantity $\omega_n^2$ in \eqref{def:omega_n} simplifies to 
\beq \label{omega(K=n)}
\omega_n=\omega_n(\Omega_1,\Omega_2,\ldots,\Omega_n)=\frac{1}{n\bar{N}\|\mu\|^2}\sum_{i=1}^n N_i\|\Omega_i-\mu\|^2.
\eeq
The DELVE test statistic also has a simplified form as in \eqref{DELAC(K=n)}-\eqref{define:V(K=n)}. 
We can prove the same theoretical results under {\it weaker conditions}: 

\begin{condition} \label{cond:K=n}
	We assume that the following statements are true: (a) For a constant $c_0\in (0,1)$, $2\leq N_i \leq (1-c_0)n\bar{N}$ and $\| \Omega_i \|_\infty \leq 1 - c_0$, $1\leq i\leq n$, and (b) $\max\big\{ \sum_i \frac{ \| \Omega_i \|_3^3 }{ N_i}  , \, \sum_i \frac{ \| \Omega_i \|^2 }{ N_i^2}  \big\} 	\big/ ( \sum_i \| \Omega_i \|^2)^2=o(1)$, and  $(\sum_i \| \Omega_i \|_3^3)/(n \| \mu \|^2) =o(1)$
\end{condition}

\noindent
When $K=n$, Condition (a) is equivalent to \eqref{cond1-basic}; and Condition (b) is weaker than \eqref{cond2-regular}, as
we have dropped the requirement  $\frac{\|\mu\|_4^4}{K\|\mu\|^4}=o(1)$. 
We obtain weaker conditions for $K=n$ because the dominant terms in $T$ differ from those for $K<n$.


\begin{thm} \label{thm:K=n} 
	In Model \eqref{Mod1-data}, we test the null hypothesis \eqref{null(K=n)}. As $n\to\infty$, we assume that Condition~\ref{cond:K=n} is satisfied. Under the alternative, we further assume that 
	\beq \label{SNR(K=n)}
	\frac{n\bar{N}\|\mu\|^2\omega_n^2}{ \sqrt{\sum_{i=1}^n \|\Omega_i \|^2}} \to\infty. 
	\eeq
	Let $T$ and $V^*$ be the same as in \eqref{DELAC(K=n)}-\eqref{define:V(K=n)}. Consider the simplified DELVE test statistic $\psi^* = T/\sqrt{V^*}$. 
Under the null hypothesis, $\psi^*\to N(0,1)$ in distribution. Under the alternative hypothesis, $\psi^*\to\infty$ in probability. Moreover, for any fixed $\kappa\in (0,1)$, the level-$\kappa$ DELVE test has an asymptotic level of $\kappa$ and an asymptotic power of $1$. 
\end{thm}

The detection boundary in \eqref{SNR(K=n)} has a simpler form if $\sum_i\|\Omega_i\|^2\asymp n\|\mu\|^2$. In this case, \eqref{SNR(K=n)} is equivalent to 
$
\sqrt{n}\bar{N}\|\mu\|\omega_n^2\to\infty. 
$
Additionally, if all entries of $\mu$ are at the same order, then $\|\mu\|\asymp 1/\sqrt{p}$, and \eqref{SNR(K=n)} further reduces to $\sqrt{n\bar{N}^2/p}\cdot \omega_n^2\to\infty$.


\subsection{A discussion of the contiguity regime} \label{subsec:boundary}

Our power analysis in Section~\ref{subsec:Main-power} concerns $\mathrm{SNR}_n\to \infty$, and our lower bound in Section~\ref{subsec:Main-LB} concerns $\mathrm{SNR}_n\to 0$. We now study the contiguity regime where $\mathrm{SNR}_n$ tends to a constant. 
For illustration, we consider a special choice of parameters, which allows us to obtain a simple expression of the testing risk. 

Suppose $K=n$ and $N_i=N$ for all $1\leq i\leq n$. Consider the pair of hypotheses:
\beq \label{Boundary1}
H_0:\;\; \Omega_{ij}=p^{-1}, \qquad\mbox{v.s.}\qquad H_1: \;\; \Omega_{ij}=p^{-1}(1+\nu_n\delta_{ij}), 
\eeq
where $\{\delta_{ij}\}_{1\leq i\leq n,1\leq j\leq p}$ satisfy that $|\delta_{ij}|= 1$, $\sum_{j=1}^p\delta_{ij}=0$ and $\sum_{i=1}^n\delta_{ij}=0$. Such $\delta_{ij}$ always exist.\footnote{For example, we can first partition the dictionary into two halves and then partition all the documents into two halves; this divides $\{1,2,\ldots,p\}\times\{1,2,\ldots,n\}$ into four subsets; we construct $\delta_{ij}$'s freely on one subset and then specify the $\delta_{ij}$'s on the other three subsets by symmetry.}
The $\mathrm{SNR}_n$ in \eqref{SNR} satisfies that $\mathrm{SNR}_n \asymp (N\sqrt{n}/\sqrt{p})\nu_n^2$. We thereby set
\beq \label{Boundary2}
\nu_n^2 = \frac{\sqrt{2p}}{N\sqrt{n}}\cdot a, \qquad\mbox{for a constant } a>0. 
\eeq
Since $K=n$ here, we consider the simplified DELVE test statistic $\psi^*$ as in Section~\ref{subsec:K=n}.

\begin{thm} \label{thm:boundary}
	Consider Model~\eqref{Mod1-data} with $N_i=N$. For a constant $a>0$, let the null and alternative hypotheses be specified as in \eqref{Boundary1}-\eqref{Boundary2}.  As $n\to\infty$, if $p=o(N^2n)$, then  $\psi^* \to N(0,1)$ under $H_0$ and $\psi^* \to N(a,1)$ under $H_1$. 
\end{thm}

Let $\Phi$ be the cumulative distribution function of the standard normal. 
By Theorem~\ref{thm:boundary}, for any fixed constant $t\in (0,a)$, if we reject the null hypothesis when $\psi^*>t$, then the sum of type I and type II errors converges to
$[1-\Phi(t)] + [1-\Phi(a-t)]$.

\section{Applications to other statistical problems } \label{sec:Applications}


As mentioned in Section~\ref{sec:Intro}, our testing problem includes global testing for topic models, authorship attribution, and closeness testing for discrete distributions as special examples. In this section, the DELVE test is applied separately to these three problems.  

\subsection{Global testing for topic models} \label{subsec:TM}
Topic modeling \citep{blei2003latent} is a popular tool in text mining. It aims to learn a small number of ``topics" from a large corpus. Given $n$ documents written using a dictionary of $p$ words, let $X_i\sim \mathrm{Multinomial}(N_i,\Omega_i)$ denote the word counts of document $i$, where $N_i$ is the length of this document and $\Omega_i\in\mathbb{R}^p$ contains the population word frequencies. 
In a topic model, there exist $M$ topic vectors $A_1, A_2, \ldots, A_M\in\mathbb{R}^p$, where each $A_k$ is a PMF. 
Let $w_i\in\mathbb{R}^M$ be a nonnegative vector whose entries sum up to $1$, where $w_i(k)$ is the ``weight" document $i$ puts on topic $k$. It assumes 
\beq \label{TopicModel}
\Omega_i=\sum_{k=1}^M w_i(k)A_k,  \qquad 1\leq i\leq n. 
\eeq
Under \eqref{TopicModel}, the matrix $\Omega=[\Omega_1,\Omega_2,\ldots,\Omega_n]$ admits a low-rank nonnegative factorization. 

Before fitting a topic model, we would like to know whether the corpus indeed involves multiple topics. This is the global testing problem: $H_0:  M=1$ v.s. $ H_1:  M>1$.
When $M=1$, by writing $A_1=\mu$,  the topic model reduces to the null hypothesis in \eqref{null(K=n)}. We can apply the DELVE test by treating each $X_i$ as a separate group (i.e., $K=n$).  

\begin{cor} \label{cor:TM}
	Consider Model~\eqref{Mod1-data} and define a vector $\xi\in\mathbb{R}^n$ by $\xi_i=\bar{N}^{-1}N_i$. Suppose that $\Omega=\mu{\bf 1}_n'$ under the null hypothesis, with $\mu=n^{-1}\Omega\xi $, and that $\Omega$ satisfies \eqref{TopicModel} under the alternative hypothesis, with $r:=\mathrm{rank}(\Omega)\geq 2$. Suppose $\bar{N}/(\min_{i}N_i)=O(1)$. Denote by $\lambda_1,\lambda_2,\ldots,\lambda_r>0$ the singular values of $\Omega [\diag(\xi)]^{1/2}$, arranged in the descending order. We further assume that under the alternative hypothesis, 
	\beq \label{SNR-TM}
	\bar{N}\cdot \frac{\sum_{k=2}^r \lambda_k^2}{\sqrt{\sum_{k=1}^r\lambda_k^2}} \to\infty.
	\eeq
	For any fixed $\kappa\in (0,1)$, the level-$\kappa$ DELVE test has an asymptotic level $\kappa$ and an asymptotic power $1$.   
\end{cor}

The least-favorable configuration in the proof of Theorem~\ref{thm:LB} is in fact a topic model that follows \eqref{TopicModel} with $M=2$. Transferring the argument yields the following lower bound that confirms the optimality of DELVE for the global testing of topic models.


\begin{cor} \label{cor:TM-LB}
	Let ${\cal R}_{n,M}(\epsilon_n, \delta_n)$ be the collection of $\{(N_i, \Omega_i)\}_{i=1}^n$ satisfying the following conditions: 1) $\Omega$ follows the topic model \eqref{TopicModel} with $M$ topics; 2) Condition~\ref{cond:K=n} holds with $o(1)$ replaced by $\leq \epsilon_n$; 3) $\bar{N}(\sum_{k=2}^r \lambda_k^2)/(\sum_{k=1}^r\lambda_k^2)^{1/2}\geq \delta_n$. If $\epsilon_n\to 0$ and $\delta_n\to 0$, then 
	$
	\limsup_{n\to\infty} \inf_{\Psi\in\{0,1\}}\Bigl\{ \sup_{{\cal R}_{n,1}(\epsilon_n, 0)} \mathbb{P}(\Psi=1)+\sup_{\cup_{M\geq 2} {\cal R}_{n,M}(\epsilon_n, \delta_n)} \mathbb{P}(\Psi=0)\Bigr\}=1.
	$ 
\end{cor}
\noindent

The detection boundary \eqref{SNR-TM} can be simplified when $M=O(1)$. Following \cite{ke2017new}, we define $\Sigma_A=A'H^{-1}A$ and $\Sigma_W=n^{-1}WW'$, where $A=[A_1,A_2,\ldots,A_M]$, $W=[w_1,w_2,\ldots,w_n]$ and $H=\mathrm{diag}(A{\bf 1}_M)$. \cite{ke2017new} argued that it is reasonable to assume that eigenvalues of these two matrices are at the constant order. If this is true, with some mild additional regularity conditions, each $\lambda_k$ is at the order of $\sqrt{n/p}$. Hence, \eqref{SNR-TM} reduces to
$
\sqrt{n}\bar{N}/\sqrt{p} \to \infty. 
$
In comparison, \cite{ke2017new} showed that a necessary condition for any estimator $\hat{A}=[\hat{A}_1,\hat{A}_2,\ldots,\hat{A}_M]$ to achieve $\frac{1}{M}\sum_{k=1}^M \|\hat{A}_k-A_k\|_1=o(1)$ is $\sqrt{n\bar{N}/p}\to\infty$. 
We conclude that consistent estimation of topic vectors requires strictly stronger conditions than successful testing.

\vspace{-.2cm}

\subsection{Authorship attribution} \label{subsec:authorChallenge}

In authorship attribution, 
given a corpus from a known author, we want to test whether a new document is from the same author. 
It is a special case of our testing problem with $K=2$. We can directly apply the results in Section~\ref{subsec:K=2}. However, the setting in Section~\ref{subsec:K=2} has no sparsity. \cite{kipnis2022higher,donoho2022higher} point out that the number of words with discriminating power is often much smaller than $p$. 
To see how our test  performs under sparsity, we consider a sparse model. As in Section~\ref{subsec:K=2}, let 
\beq \label{model-2sample-author}
X_i \sim \mathrm{Multinomial}(N_i, \Omega_i), \; 1\leq i\leq n, \quad\mbox{and}\quad
G_i \sim \mathrm{Multinomial}(M_i, \Gamma_i),\;  1\leq i\leq m. 
\eeq
Let $\bar{N}$ and $\bar{M}$ be the average of $N_i$'s and $M_i$'s, respectively. 
Write $\eta = \frac{1}{n\bar{N}}\sum_{i=1}^n N_i \Omega_{i}$ and $\theta = \frac{1}{m\bar{M}}\sum_{i=1}^m M_i \Gamma_{i}$.  We assume for some $\zeta_n>0$, 
\beq \label{SparseModel}
\eta_j=\theta_j, \;\; \mbox{for }j\notin S, \qquad \mbox{and}\qquad \bigl|\sqrt{\eta_{j}}-\sqrt{\theta_j}\bigr|\geq \zeta_n, \;\; \mbox{for }j\in S. 
\eeq

\begin{cor} \label{cor:author}
	Under the model \eqref{model-2sample-author}-\eqref{SparseModel}, consider testing $H_0: S=\emptyset$ v.s. $ H_1: S\neq\emptyset$,
	where Condition~\ref{cond:K=2} is satisfied. 
	Let $\eta_S$ and $\theta_S$ be the sub-vectors of $\eta$ and $\theta$ restricted to the coordinates in $S$. Suppose that under the alternative hypothesis, 
	\beq \label{SNR-author}
	\frac{ \zeta_n^2\cdot (\|\eta_S\|_1+\|\theta_S\|_1) }{  \big( \frac{1}{n\bar{N}} + \frac{1}{m \bar{M}}  \big) 
		\max\{\| \eta \|,\, \| \theta \|\} } \to \infty. 
	\eeq
	As $\min\{n\bar{N},m\bar{M}\}\to\infty$, the level-$\kappa$ DELVE test has an asymptotic level $\kappa$ and an asymptotic power $1$. Furthermore, if $n\bar{N}\asymp m\bar{M}$ and $\min_{j\in S}(\eta_j+\theta_j)\geq cp^{-1}$ for a constant $c>0$, then 
	\eqref{SNR-author} reduces to $n\bar{N}\zeta_n^2|S|/\sqrt{p}\to\infty$. 
\end{cor}

\cite{donoho2022higher} studied a case where $N=M$, $n=m=1$, $p\to\infty$,  
\beq \label{setting-DonohoKipnis}
|S|=p^{1-\vartheta}, \qquad \mbox{and}\qquad \zeta_n= c\cdot  N^{-1/2}\sqrt{\log(p)}. 
\eeq
When $\vartheta>1/2$ (i.e., $|S|=o(\sqrt{p})$), they derived a phase diagram for the aforementioned testing problem (under a slightly different setting where the data distributions are Poisson instead of multinomial). They showed that when $\vartheta>1/2$ and $c$ is a properly large constant, a Higher-Criticism-based test has an asymptotically full power. 
\cite{donoho2022higher}  did not study the case of $\vartheta\leq 1/2$. 
By Corollary~\ref{cor:author}, when $\vartheta\leq 1/2$ (i.e., $|S|\geq C\sqrt{p}$),  the DELVE test has  asymptotically full power. 

\vspace{-.2cm}

\begin{remark}{\rm
When $\vartheta > 1/2$ in \eqref{setting-DonohoKipnis}, the DELVE test loses power. However, we can borrow the idea of maximum test or Higher Criticism test \citep{DJ04}. For example, recalling $T_j$ in \eqref{DELAC}, we may use $\max_{1\leq j\leq p}\{T_j/\sqrt{V_j}\}$ as the test statistic, where $V_j$ is a proper estimator of the variance of $T_j$. We leave this to future work. 
}
\end{remark}

\vspace{-.7cm}
\subsection{Closeness testing between discrete distributions} \label{subsec:closeTesting}

Two-sample closeness testing is a subject of intensive study in discrete distribution inference \citep{bhattacharya2015testing,chan2014optimal,diakonikolas2016new,kim2022minimax}. It is a special case of our problem with $K=2$ and $n_1=n_2=1$. 
We thereby apply both Theorem~\ref{thm:K=2} and Theorem~\ref{thm:K=n}. 

\begin{cor} \label{cor:closeness}
	Let $Y_1$ and $Y_2$ be two discrete variables taking values on the same $p$ outcomes.
	Let $\Omega_1\in\mathbb{R}^p$ and $\Omega_2\in\mathbb{R}^p$ be their corresponding PMFs. Suppose we have $N_1$ samples of $Y_1$ and $N_2$ samples of $Y_2$. The data are summarized in two multinomial vectors: 
	$
	X_1\sim \mathrm{Multinomial}(N_1, \Omega_1),  X_2\sim \mathrm{Multinomial}(N_2, \Omega_2).
	$
	We test 
	$
	H_0: \Omega_1=\Omega_2. 
	$
	Write $\mu=\frac{1}{N_1+N_2}(N_1\Omega_1+N_2\Omega_2)$. Suppose $\min\{N_1, N_2\}\geq 2$, $\max\{\|\Omega_1\|_\infty, \|\Omega_2\|_\infty\}\leq 1-c_0$, for a constant $c_0\in (0,1)$. Suppose  $\frac{1}{( \sum_{k=1}^2  \| \Omega_k \|^2)^2}\max\big\{ \sum_{k=1}^2 \frac{ \| \Omega_k \|_3^3 }{ N_k}  , \sum_{k=1}^2 \frac{ \| \Omega_k \|^2 }{ N_k^2}  \big\} =o(1)$, and  $\frac{1}{n\|\mu\|^2}\sum_{k=1}^2 \| \Omega_k \|_3^3 =o(1)$.
	We  assume that under the alternative hypothesis, 
	\beq \label{SNR(closeTesting)}
	\frac{ \| \Omega_1 - \Omega_2 \|^2 }{  \big( N_1^{-1} +N_2^{-1}  \big) 
		\max\{\| \Omega_1 \|,\, \| \Omega_2 \|\} } \to \infty. 
	\eeq
	As $\min\{N_1, N_2\}\to\infty$, the level-$\kappa$ DELVE test has 
	level $\kappa$ and power $1$, asymptotically.
\end{cor}

The requirement \eqref{SNR(closeTesting)} matches with the minimum $\ell^2$-separation condition for two-sample closeness testing \citep[Proposition 4.4]{kim2022minimax}. Hence, our test is an optimal $\ell^2$-testor. 
Other optimal $\ell^2$-testors \citep{chan2014optimal,bhattacharya2015testing,diakonikolas2016new} are not equipped with tractable null distributions. 

\vspace{-.3cm}

\begin{remark}{\rm
		We can modify DELVE to incorporate frequency-dependent weights. 
		Define
		$
		T(w):=\sum_{j=1}^p w_j T_j$,  where $T_j$ is the same as in \eqref{DELAC} and let 
		$w_j =\bigl(\max\{1/p,\; \hat{\mu}_j\}\bigr)^{-1}$.  
		Such weights were used in discrete distribution inference \citep{balakrishnan2019hypothesis,chan2014optimal} to turn an optimal $\ell^2$ testor to an optimal $\ell^1$ testor.  
		We can similarly study the power of the test based on $T(w)$, except that we need an additional assumption $n\bar{N}\gg p$ to guarantee that $\hat{\mu}_j$ is a sufficiently accurate estimator of $\mu_j$. 
}\end{remark}

\section{Simulations} \label{sec:Simu}

We investigate the numerical performance of DELVE in simulations. Recall that we introduced a variant of DELVE, DELVE+, in Section~\ref{subsec:modification}. DELVE+ has similar theoretical properties but is more suitable for real data. We include both versions in simulations.

\begin{figure}[!tb]
	\includegraphics[width=0.32\textwidth]{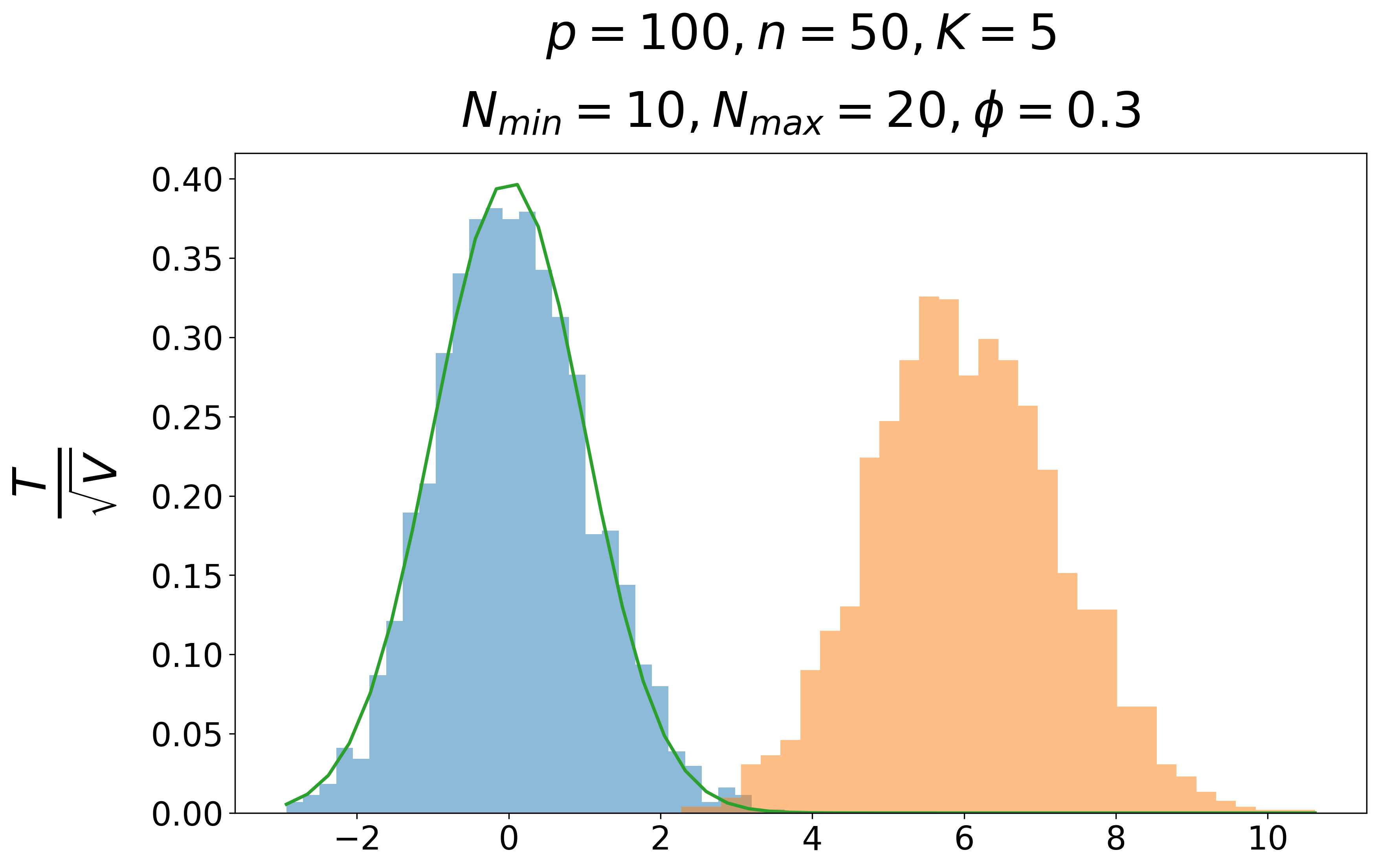}
	\includegraphics[width=0.32\textwidth]{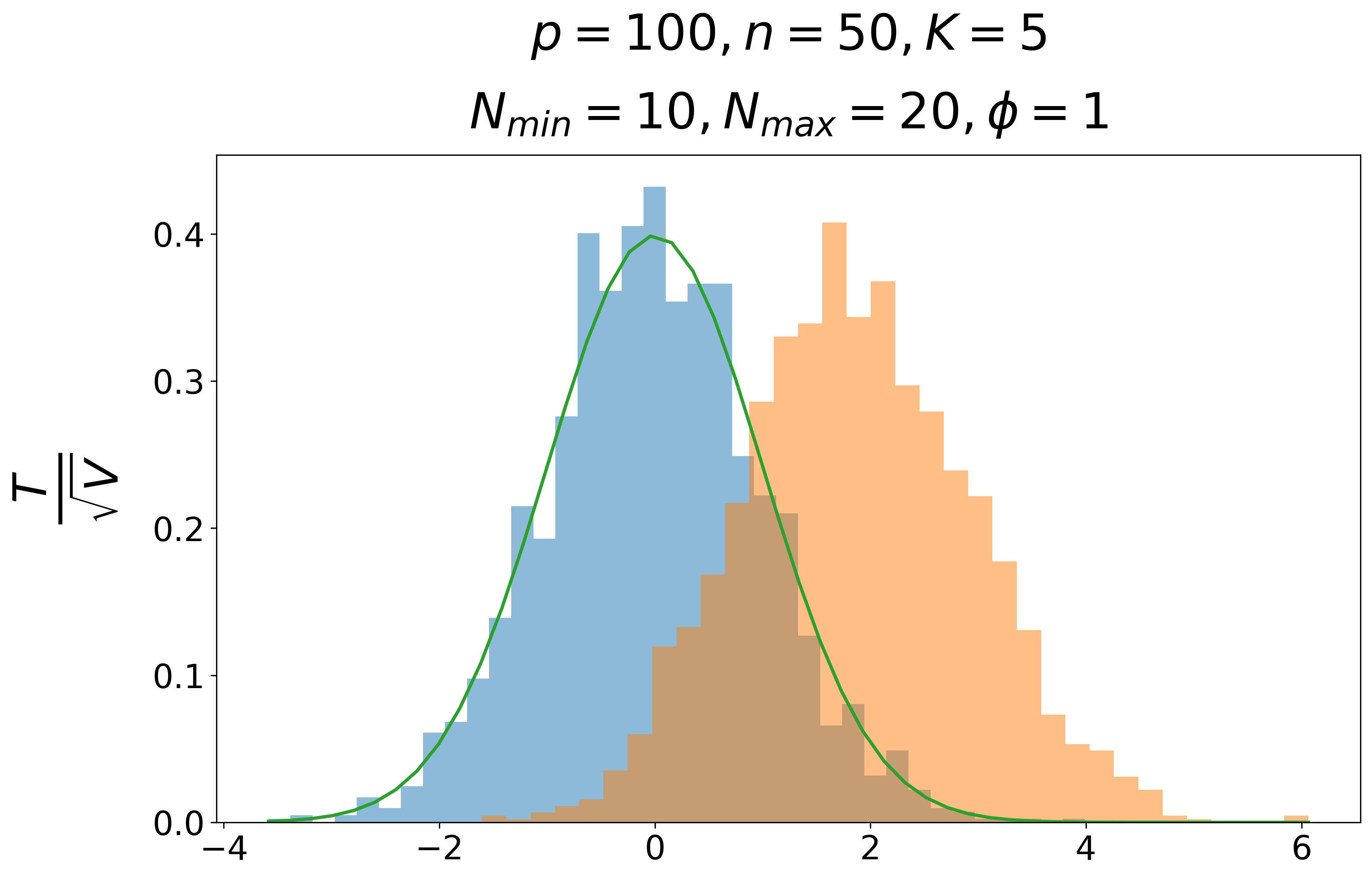}
	\includegraphics[width=0.32\textwidth]{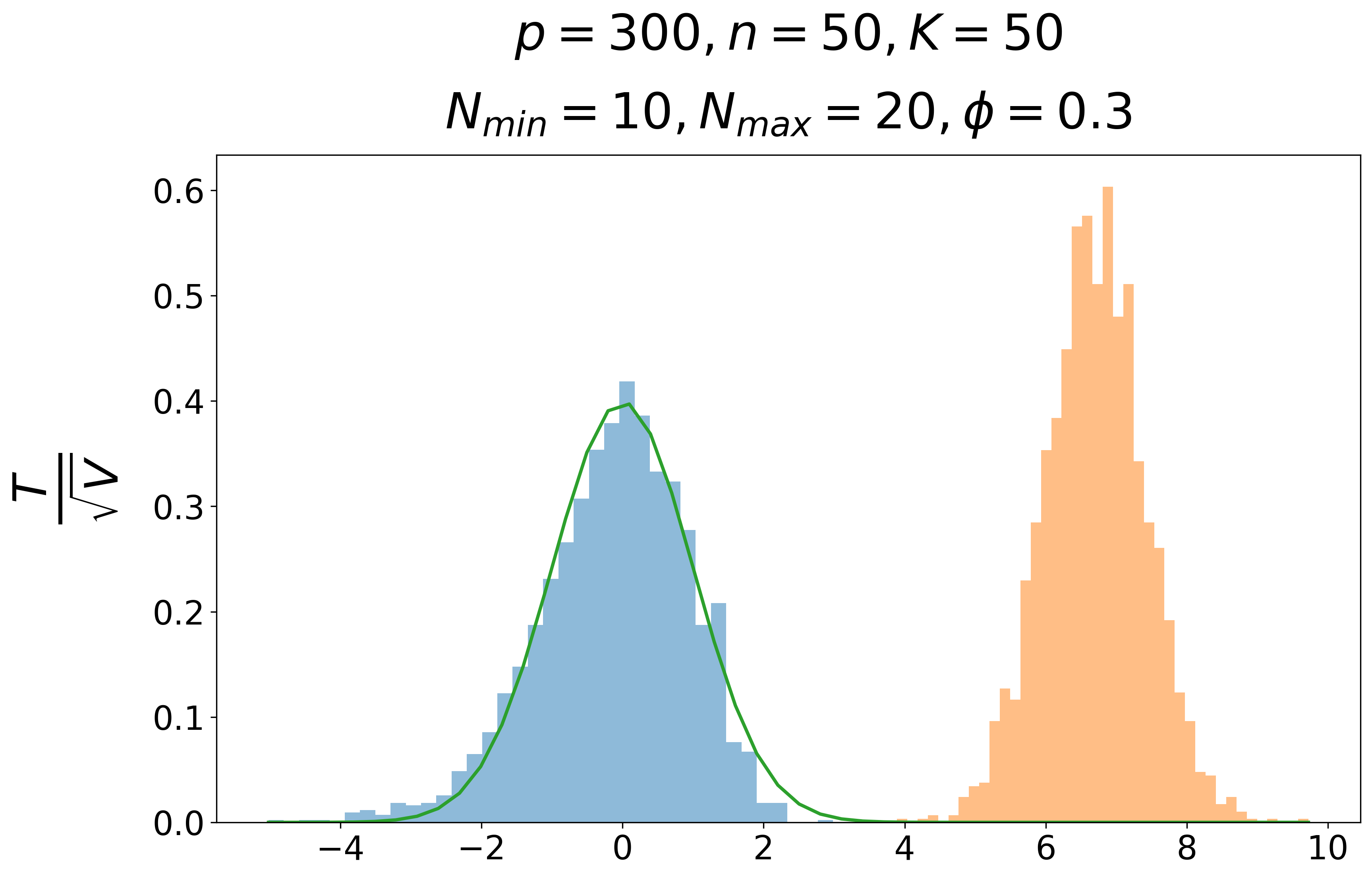}\\
	
	\includegraphics[width=0.32\textwidth]{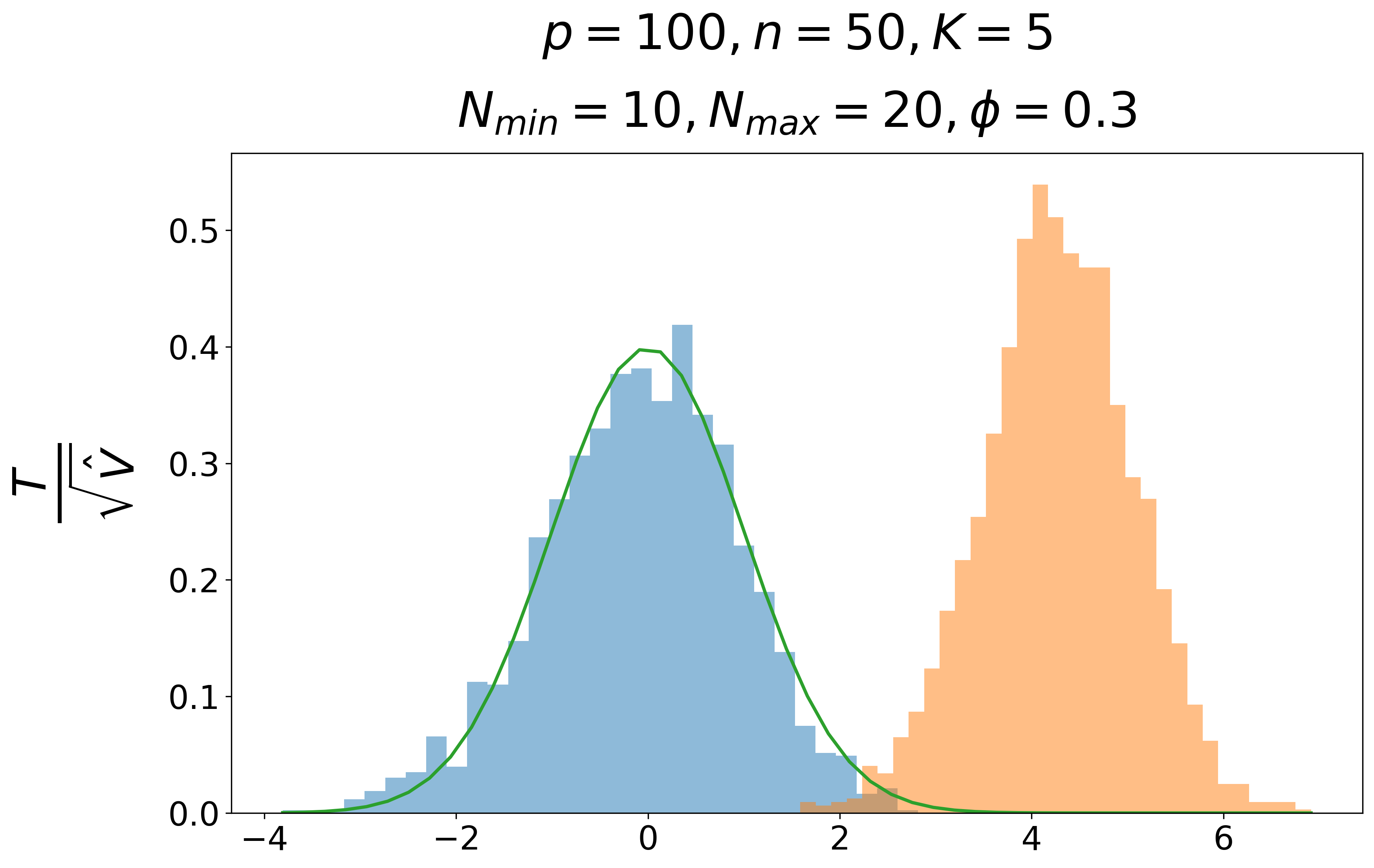}
	\includegraphics[width=0.32\textwidth]{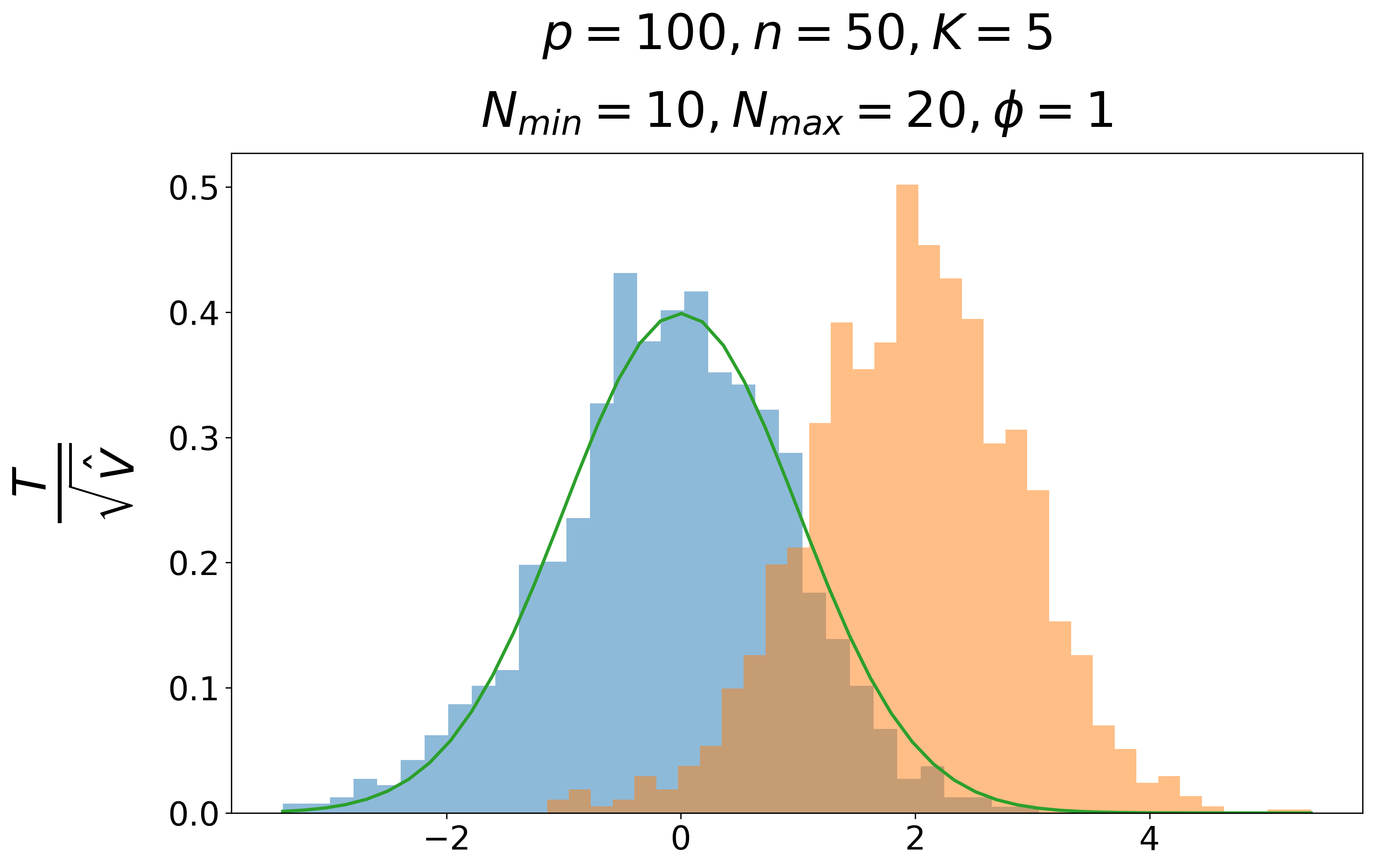}
	\includegraphics[width=0.32\textwidth]{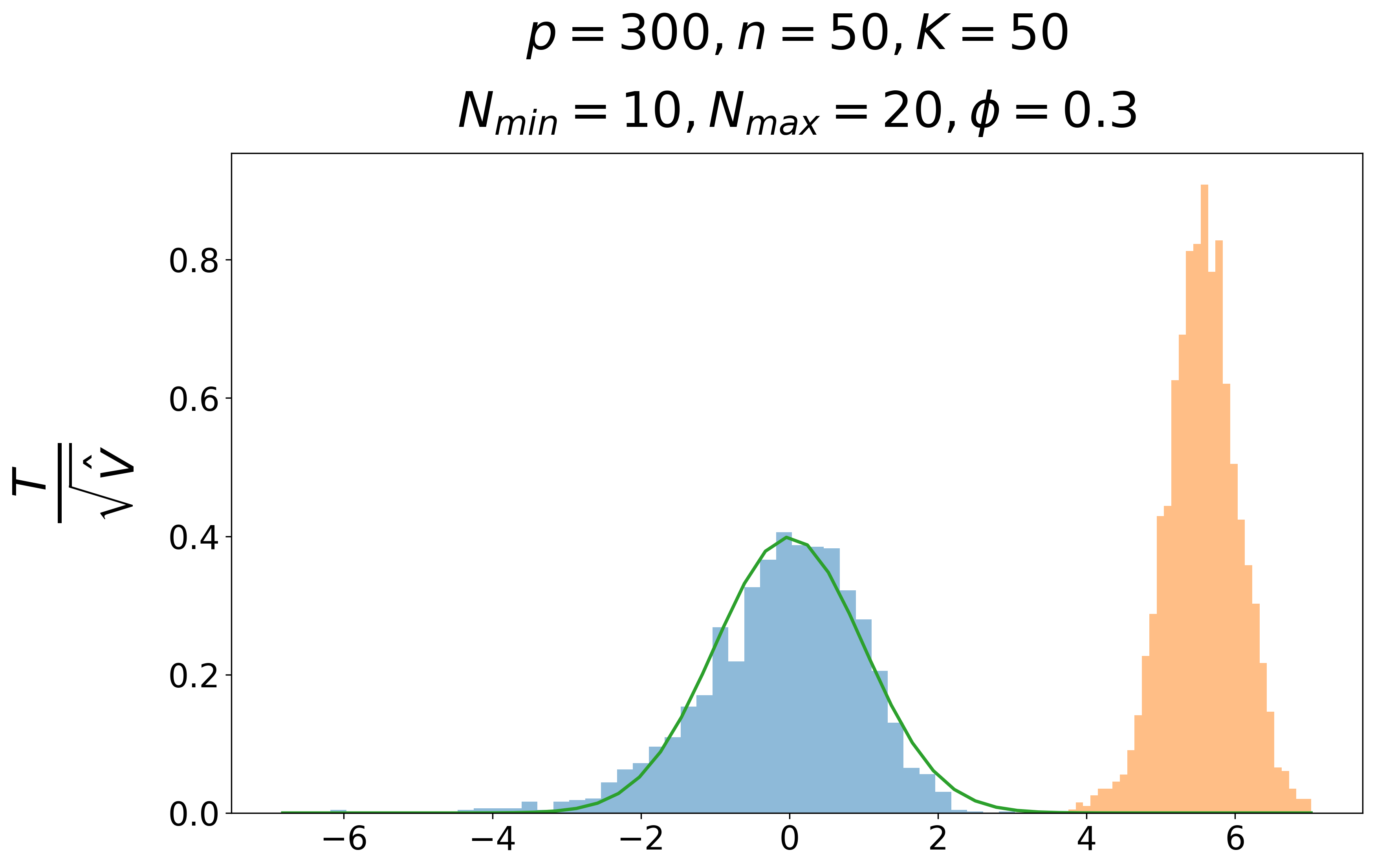}
	\caption{Histograms of DELVE (top panels) and DELVE+ (bottom panels) statistics in Experiments 1.1-1.3. In each plot, the blue and orange histograms correspond to the null and alternative hypotheses, respectively;  and the green curve is the density of $N(0,1)$. 
	}
	\label{fig:Experiment1}
\end{figure}

{\it Experiment 1 (Asymptotic normality)}.
Given $(n, p, K, N_{\min}, N_{\max}, \phi)$, we generate data as follows: first, divide $\{1,\ldots,n\}$ into $K$ equal-size groups. Next, we draw $\Omega_1^{alt}, \ldots, \Omega_n^{alt}$ i.i.d. from $\text{Dirichlet}(p, \phi \mathbf{1}_p)$. 
Third, we draw $N_i \stackrel{iid}{\sim} \text{Uniform}[N_{\min}, N_{\max}]$ and set $\Omega_i^{null} = \mu$, where $\mu := \frac{1}{n \bar{N}} \sum_i N_i \Omega_i^{alt}$. Last, we generate $X_1, \ldots, X_n$ using Model~\eqref{Mod1-data}. 
We consider three sub-experiments. In Experiment 1.1, $(n, p, K, N_{\min}, N_{\max}, \phi)=(50, 100, 5, 10, 20, 0.3)$. 
In Experiment 1.2, $\phi$ is changed to $1$, and the other parameters are the same. When $\phi=1$, $\Omega_i^{alt}$ are drawn from the uniform distribution of the standard probability simplex; in comparison, $\phi=0.3$ puts more mass near the boundary of the standard probability simplex. In Experiment 1.3, we keep all parameters the same as in Experiment 1.1, except that $(p, K)$ are changed to $(300, 50)$.  
For each sub-experiment, we generate 2000 data sets under the null hypothesis and plot the histogram of the DELVE test statistic $\psi$ (in blue); similarly, we generate 2000 data sets under the alternative hypothesis and plot the histogram of $\psi$ (in orange). The results are contained in Figure~\ref{fig:Experiment1}. 

In all  sub-experiments, when the null hypothesis holds, the histograms of DELVE and DELVE+ fit the standard normal density reasonably well. This supports our theory in Section~\ref{subsec:Main-null}. Second, when $(p, K)$ increase, the finite sample effect becomes slightly more pronounced (c.f., Experiment 1.3 versus Experiment 1.1). Third, the tests have power in differentiating two hypotheses. As $\phi$ decreases or $K$ increases, the power increases, and the two histograms become further apart. Last, in the alternative hypothesis, DELVE+ has smaller mean and variance than DELVE. 
By Lemma~\ref{prop:modification}, they have similar asymptotic behaviors. The simulations suggest that they have noticeable finite-sample differences.

\begin{figure}[!tb]
	\centering
	\includegraphics[width=0.3\textwidth]{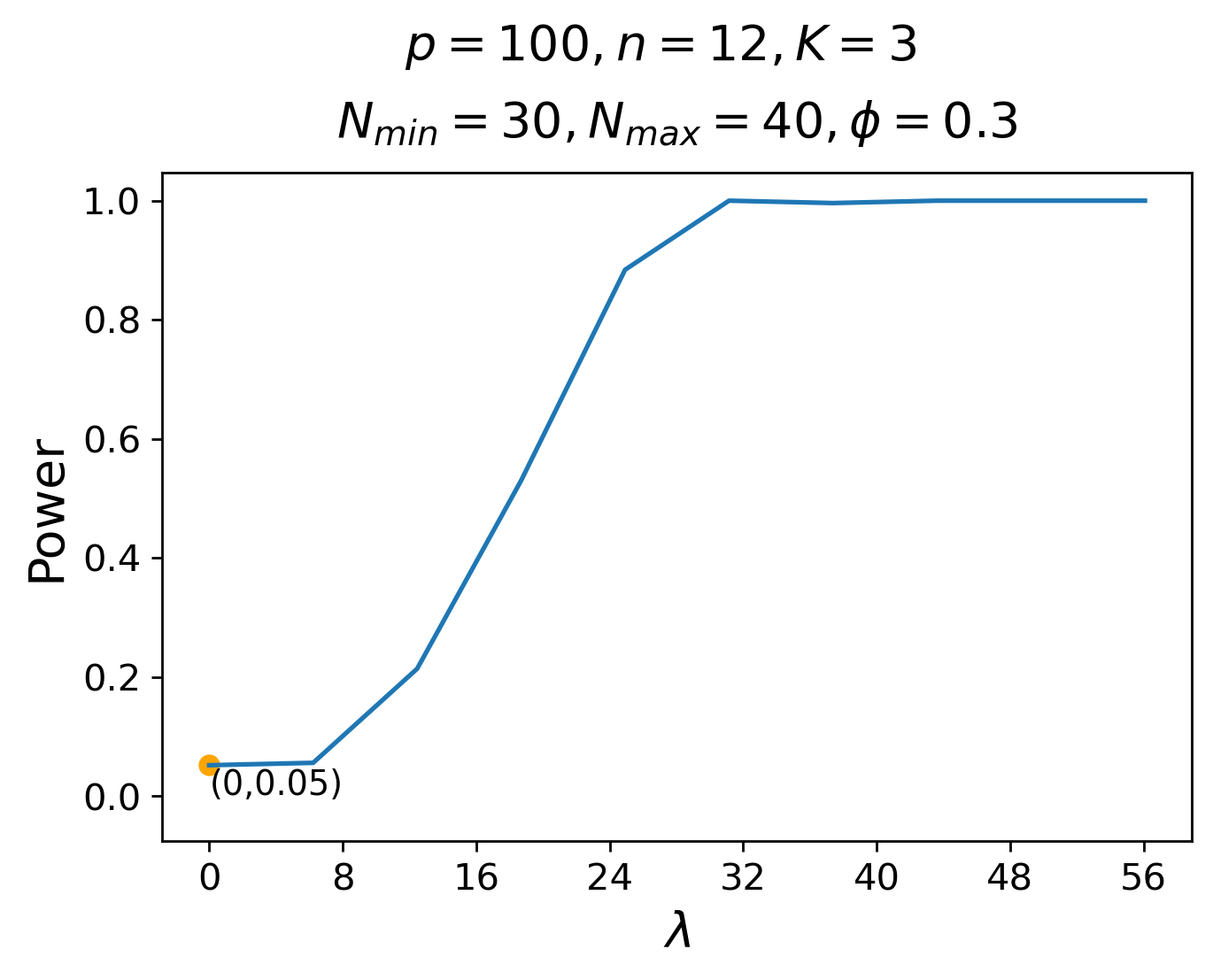}
	\includegraphics[width=0.3\textwidth]{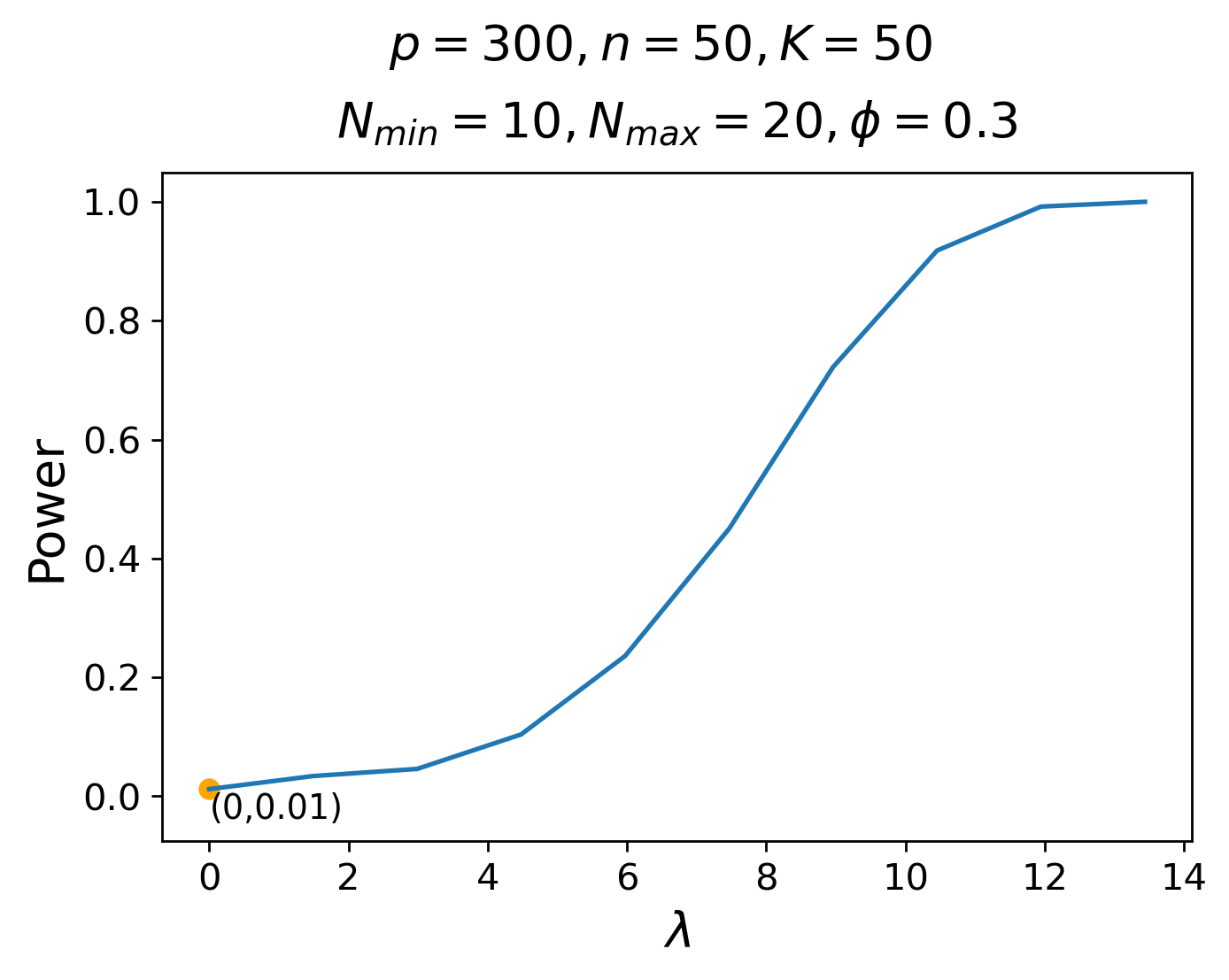}
	\includegraphics[width=0.3\textwidth]{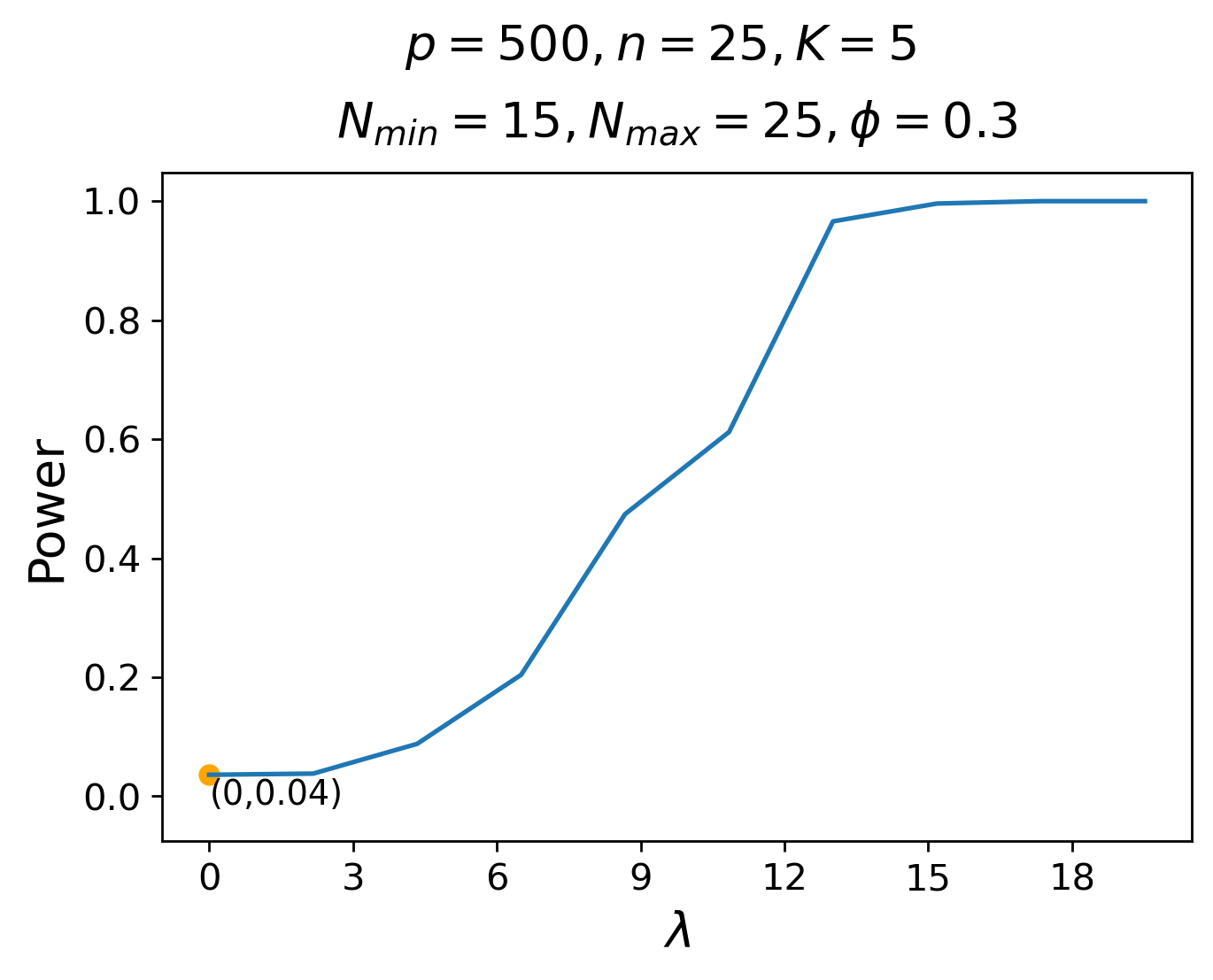}

	\caption{
		Power of the level-$5\%$ DELVE test ($x$-axis represents the SNR $\lambda(\tau_n) =\frac{ n \bar{N} \| \mu \|\tau_n^2}{\sqrt{K}}$).
	}
	\label{fig:Experiment2}
\end{figure}

{\it Experiment 2 (Power curve)}.  
Similarly as in Experiment 1, we divide $\{1,2,\ldots,n\}$ into $K$ equal-size groups and draw $N_i \sim \text{Uniform}[N_{\min}, N_{\max}]$. In this experiment, $\Omega_i$'s are generated in a different way. 
Under $H_0$, we draw $\mu \sim \text{Dirichlet}(p/2, \phi \mathbf{1}_{p/2})$ and set $\Omega_i^{null} = \tilde \mu$, where $\tilde \mu_j = \frac{1}{2} \mu_j$ for $j \leq p/2$ and $\tilde \mu_j = \frac{1}{2} \mu_{j-p/2}$ for $j\geq p/2 + 1$. 
Under $H_1$, fixing some $\tau_n\in [0,1]$, we draw $z_1, \ldots, z_K$, $b_1, \ldots, b_{p/2} \stackrel{iid}{\sim} \text{Rademacher}(1/2)$ and let $\Omega_{ij}^{alt} = \tilde{ \mu}_j (1 + \tau_n z_k b_j)$, for $i$ in group $k$ and $1\leq j\leq p/2$,  and let $\Omega_{ij}^{alt} =  \tilde{\mu}_j (1 - \tau_n z_k b_{j-p/2})$ for $p/2+1\leq j\leq p$. 
By applying our theory in Section~\ref{subsec:Main-power} together with some calculations, the signal-to-noise ratio is captured by 
$
\lambda(\tau_n) := K^{-1/2} n \bar{N} \| \mu \| \tau_n^2. 
$
In particular, it holds that $\omega_n^2(\Omega^{alt}) = \tau_n^2$, for the $\omega_n^2$ defined in \eqref{def:omega_n}. 
We consider three sub-experiments, Experiment 2.1-2.3, where the parameter values of $(n,p,K,N_{\min}, N_{\max},\phi)$ are the same as in Experiments 1.1-1.3. For each sub-experiment, we consider a grid of 10 equally-spaced values of $\lambda$. When $\lambda=0$, it corresponds to $H_0$; when $\lambda>0$, it corresponds to $H_1$. 
For each $\lambda$, we generate $500$ data sets and compute the fraction of rejections of the level-$5\%$ DELVE test. This gives a power curve for the level-$5\%$ DELVE test, in which the first point associated with $\lambda=0$ is the actual level of the test. The results are in Figure~\ref{fig:Experiment2}. We repeat the same experiments for the DELVE+ test; owing to space limit, the plots are in \cite{DELVE-supp}. 
In all three experiments, the actual level of our proposed tests is $\leq 5\%$, suggesting that our tests perform well at controlling the type-I error. As $\lambda$ increases, the power gradually increased to $1$, suggesting that $\lambda$ is a good metric of the signal-to-noise ratio. This supports our theory in Section~\ref{subsec:Main-power}. 

\begin{figure}[tbp]
	\centering
	\includegraphics[width=0.248\textwidth, height=0.22\textwidth, trim=0 15 0 0, clip=true]{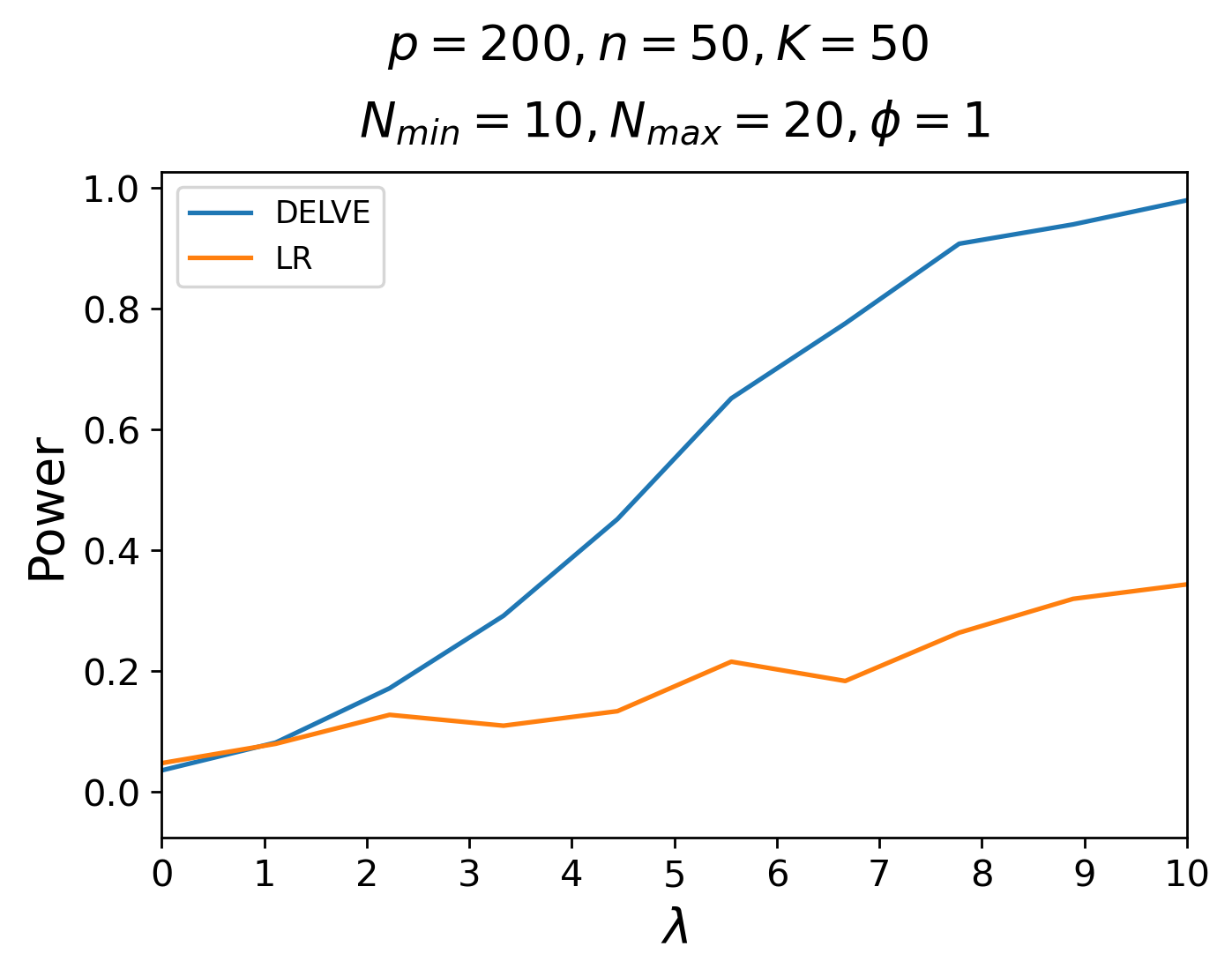}
	\includegraphics[width=0.248\textwidth, height=0.22\textwidth,  trim=0 15 0 0, clip=true]{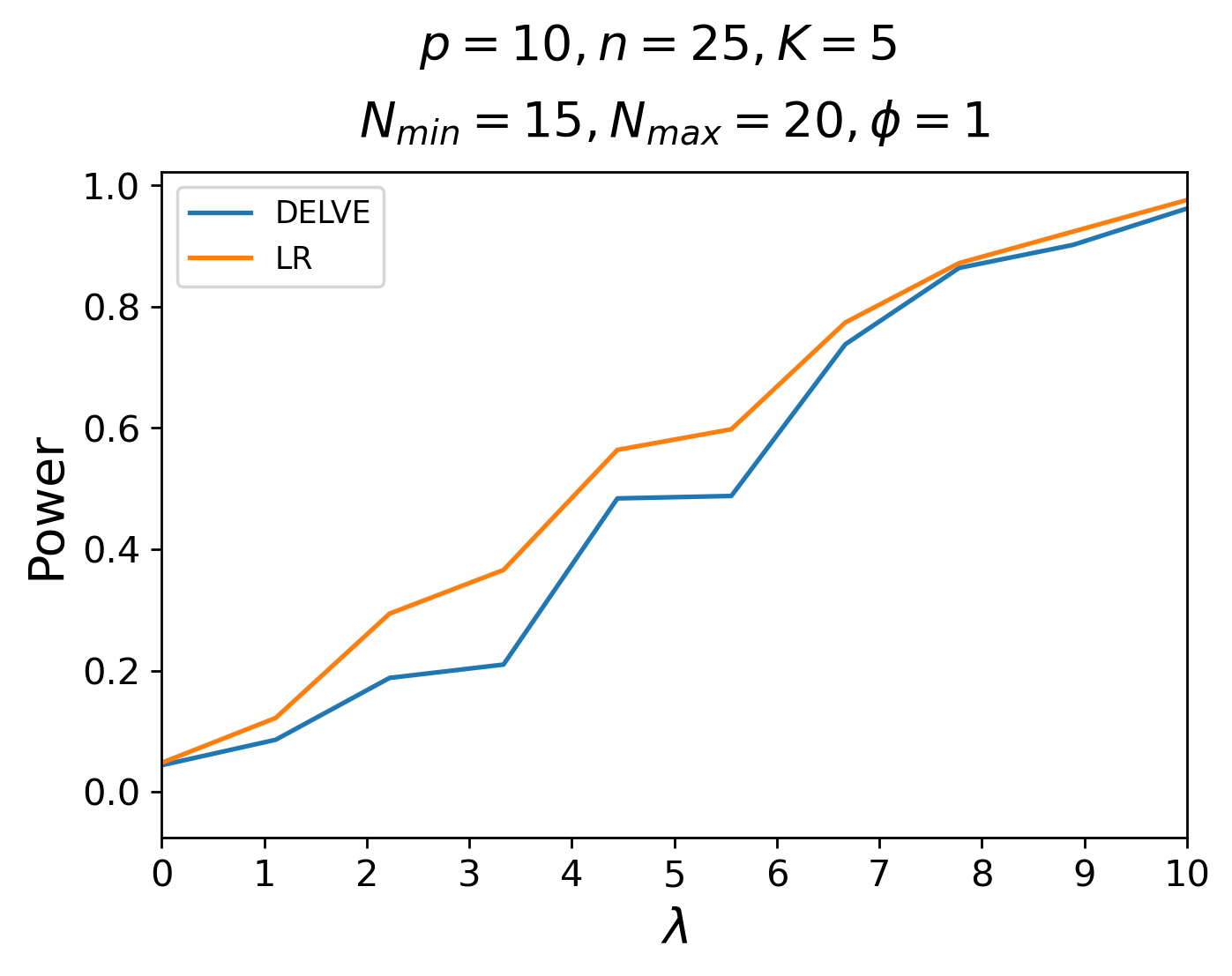}
		\includegraphics[width=0.245\textwidth, height=0.22\textwidth]{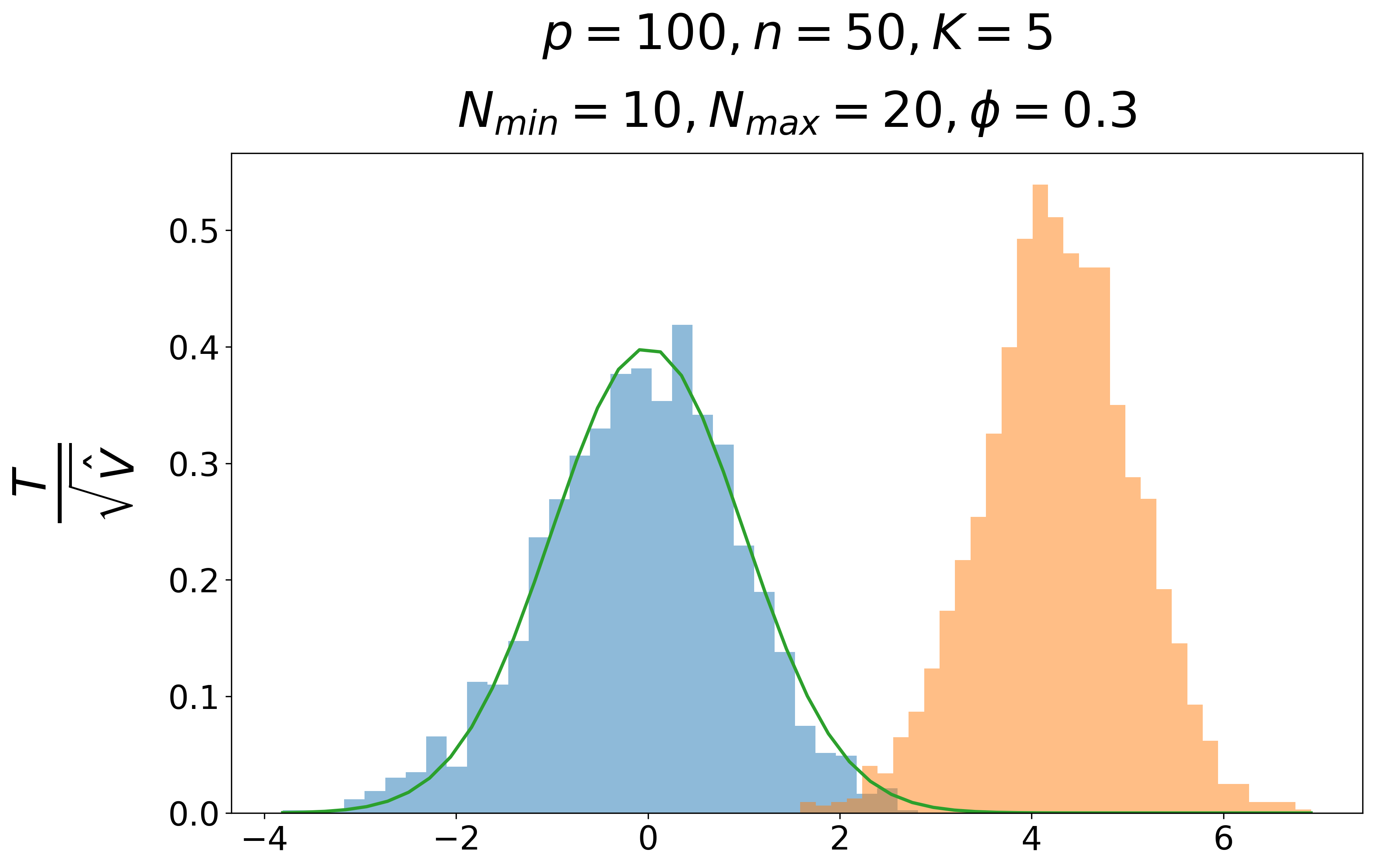} 
	\includegraphics[width=0.235\textwidth, height=0.22\textwidth]{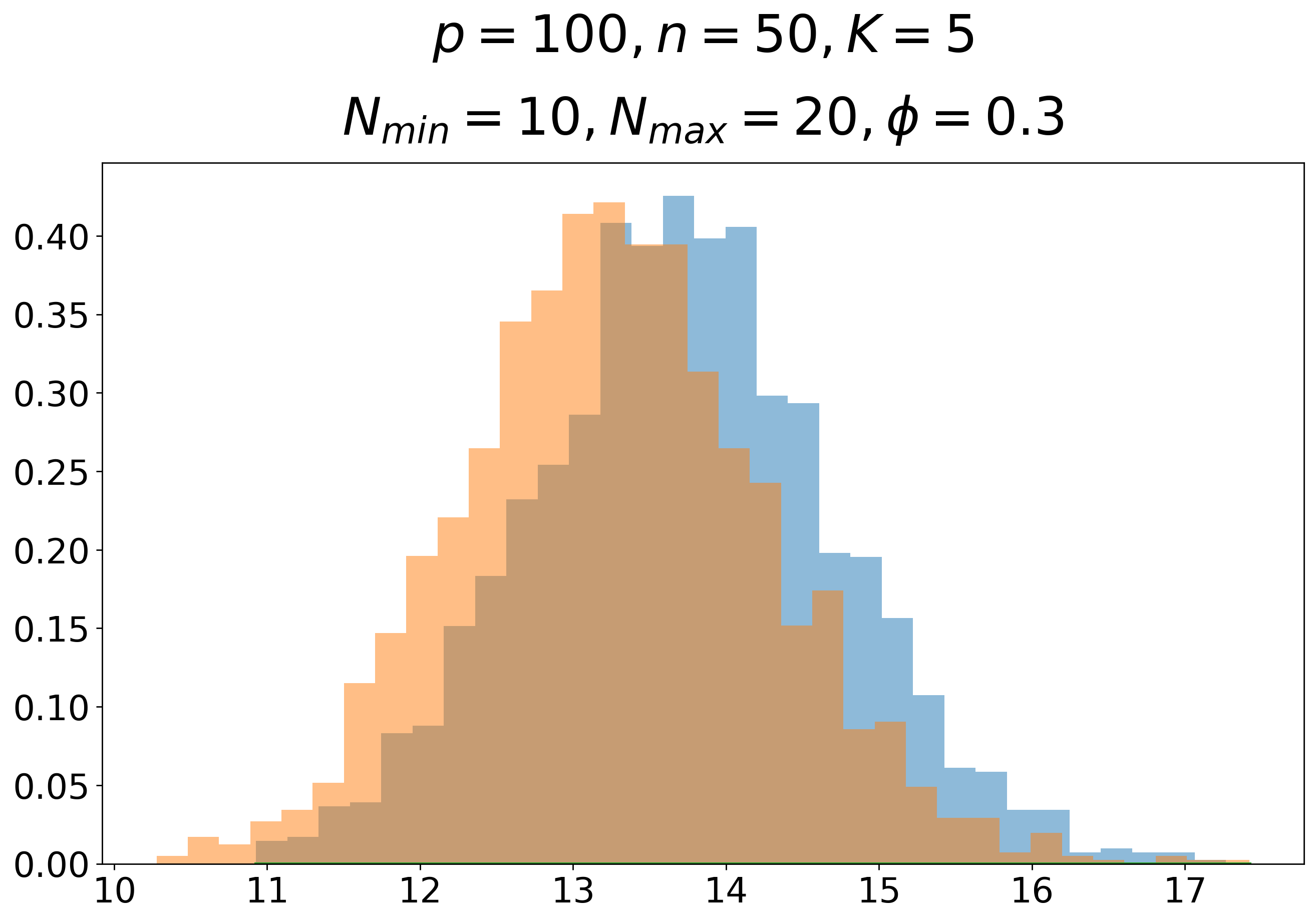}
	\caption{Comparison of DELVE+, LR, and ANOVA (details are in Experiment~3).}
	\label{fig:Experiment3}
\end{figure}

{\it Experiment 3 {\it (Comparison with the LR and ANOVA tests)}.} 
This experiment contains two sub-experiments. 
In Experiment 3.1, we compare DELVE+ with the likelihood ratio (LR) test. The LR test is only well-defined in the special case where $\Omega_i$'s are equal within each group. In this case, $T^{\mathrm{LR}}=\sum_{k}    n_k \bar{N}_k \sum_j \hat \mu_{kj}  \log\big( \frac{ \hat \mu_{kj} }{ \hat \mu_j } \big)$, where $\hat{\mu}_k$ and $\hat{\mu}$ are the same as in \eqref{define:eta}, and $\log(0/0) = 0$. Given $(n,p,K,N_{\min}, N_{\max},\phi)$, we generate data in the same way as in Experiment~2 (these settings guarantee that $\Omega_i$'s are equal within-group, hence favoring the LR test). Since no asymptotic normality result is known for $T^{\mathrm{LR}}$, we use an ideal threshold for the LR test - drawing 500 data sets from the null model ($\lambda=0$) and computing the empirical $95\%$-quantile of $T^{\mathrm{LR}}$. The power curves for two representative settings ($p=200$ and $p=10$) are shown in the left two panels of Figure~\ref{fig:Experiment3}. 
More settings can be found in Section~\ref{supp:LR} of \cite{DELVE-supp}.
We observe that DELVE+ significantly outperforms LR when $p$ is large/moderate compared to $n$, and they perform similarly (with LR being slightly better) when $p$ is small.
In Experiment 3.2, we compare DELVE+ with the ANOVA test that uses $\widetilde{T}$ in \eqref{define:tildeT} as the test statistic. 
The simulation settings are the same as in Experiment 1.1. The third panel of Figure~\ref{fig:Experiment3} is a replication of the bottom left panel of Figure~\ref{fig:Experiment1} and shows the histograms of DELVE+ test statistics under two hypotheses. The fourth panel of  Figure~\ref{fig:Experiment3} contains the histograms of $\widetilde{T}$. We see that 
$\widetilde{T}$ fails to distinguish two hypotheses while DELVE+ is able to do so. 
As explained in Remark \ref{debiasing-effect-remark}, the naive ANOVA test can lose power due to the lack of de-biasing.

\section{Real Data Analysis} \label{sec:RealData}

We consider two real corpora consisting of statistical paper abstracts and Amazon movie reviews, respectively.  
We use them to showcase: Although testing the null hypothesis \eqref{Mod3-null} is only a binary decision problem, it can be used to answer various questions of interest by simply varying the definition of ``groups" in \eqref{Mod3-null}. 
For example, we may define ``groups" of movie reviews by movie title, star rating, posting time, reviewer characteristics, etc.. 
Then, our test can detect many different kinds of heterogeneity in movie reviews (the same holds for other product reviews).
In Section~\ref{sec:Method}, we proposed DELVE and DELVE+ and explained that the latter is more suitable for real data; hence, we use DELVE+ here.   


\vspace{-.5cm}

\subsection{Abstracts of statisticians} \label{subsec:StatAbstracts}

The data set from \cite{ji2016coauthorship} contains the bibtex information of published papers in four top-tier statistics journals, {\it Annals of Statistics}, {\it Biometrika}, {\it Journal of the American Statistical Association}, and {\it Journal of the Royal Statistical Society - Series B}, from 2003 to the first half of 2012. In the pre-processing step, we first remove common stop words such as ``for", ``also", ``can", and ``the", and  common domain-specific words such as ``statistician", ``estimate", and ``sample". We then perform \textit{stemming}, which maps together words with a common prefix such as ``play", ``player", and ``playing". Finally, we perform tokenization, which maps each abstract to its vector of word (stem) counts.

We conduct two experiments. In the first one, we fix an author and treat the collection of his/her co-authored abstracts as a corpus. We apply DELVE+ with $K = n$, where $n$ is the number of abstracts written by this author. The $Z$-score measures the ``diversity" or ``variability" of this authors' abstracts. An author with a high $Z$-score possesses either diverse research interests or a variable writing style. A number of authors have only 1--2 papers, and the variance estimator $V$ is often negative; we remove all those authors.  In Figure~\ref{fig:author_diversity} (left), we plot the histogram of $Z$-scores of retained authors. The mean is $4.52$ and the standard deviation is $2.94$. In Figure~\ref{fig:author_diversity} (middle), we show the plot of $Z$-score versus logarithm of the number of abstracts written by this author. The most prolific author has 82 papers and a $Z$-score larger than $20$, implying a huge diversity in his/her abstracts. There is also a positive association between $Z$-score and number of papers. It suggests that senior authors have more diversity in their abstracts, which is intuitive.

\begin{figure}[!tb]
	\centering
	\includegraphics[width=0.3\textwidth, height=.24\textwidth, trim=10 0 10 9.2, clip=true]{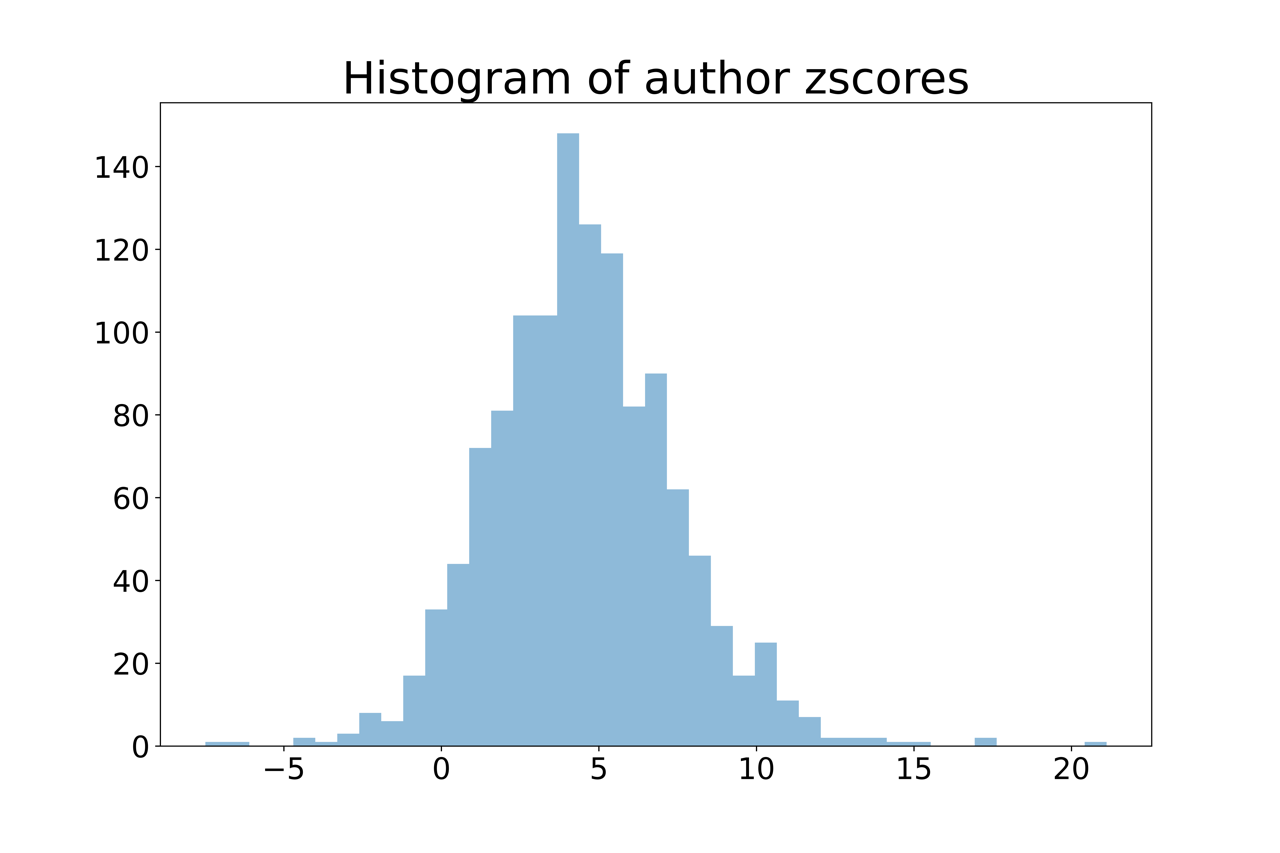}
	\includegraphics[width=0.35\textwidth, height=.255\textwidth, trim=15 0 10 50, clip=true]{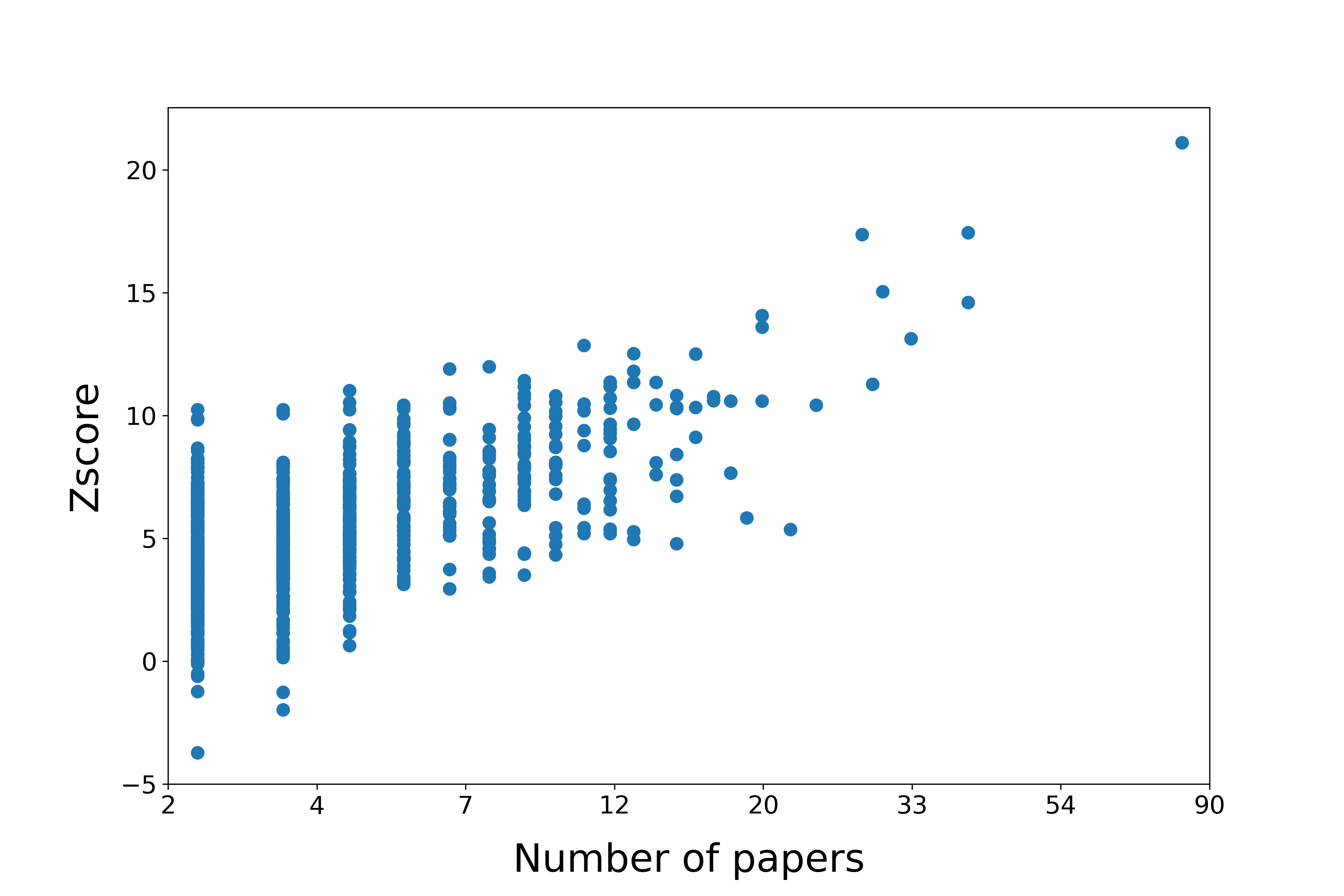}
	\includegraphics[width=0.33\textwidth, height=.24\textwidth, trim=0 0 0 25, clip=true]{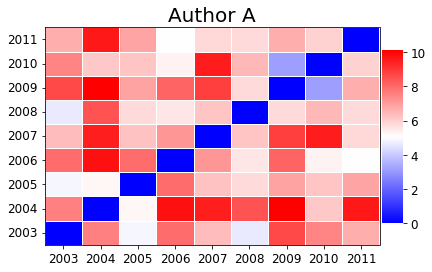}
	\caption{Results about statistical abstracts. Left: Histogram of author $Z$-scores (mean is $4.52$, and standard deviation is $2.94$). Middle: Author $Z$-score versus number of papers. Right: Pairwise $Z$-score plot for a representative author. 
	}
	\label{fig:author_diversity}
\end{figure}

In the second experiment, we further divide an author's abstracts into smaller groups by publication year. 
Owing to space limit, we only show the results for the most prolific author who has 82 papers, but we keep in mind that the same analysis can be done for each author in the data set (see \cite{DELVE-supp}).  
We divide this author's abstracts into 9 groups, each group corresponding to one year. For each pair of groups, we implement DELVE+ with $K=2$. This yields a pairwise plot of $Z$-scores, as shown in Figure~\ref{fig:author_diversity} (right). 
It reveals the temporal patterns of this author in abstract writing. The group consisting of 2004-2005 abstracts has comparably large $Z$-scores in the pairwise comparison with other groups. To interpret the results, we read titles and abstracts of all of this author's papers and found that in 2004-2005 he/she extensively studied topics related to bandwidth selection in the context of nonparametric estimation. 

\begin{remark}{\rm
The asymptotic normality in Section~\ref{subsec:Main-null} is established under the condition $n^2\bar{N}^2\gg Kp$. It is worth checking if this holds in real data. We compute $DR:= n^2 \bar{N}^2/(K p)$ for all the corpora analyzed in the above two experiments (see Section \ref{supp:DR} of \cite{DELVE-supp}). These DR values are quite large. Therefore, it would be appropriate to apply the asymptotic normality result, and we think the $Z$-scores and $p$-values are trustworthy. 
}\end{remark}

\vspace{-1cm}

\subsection{Amazon movie reviews} 
\label{subsec:Amazon}

\begin{figure}[!tb]
\centering
		\hspace{-15pt}	
		\includegraphics[width=0.3\textwidth, height=.2\textwidth, trim=0 30 0 20, clip=true]{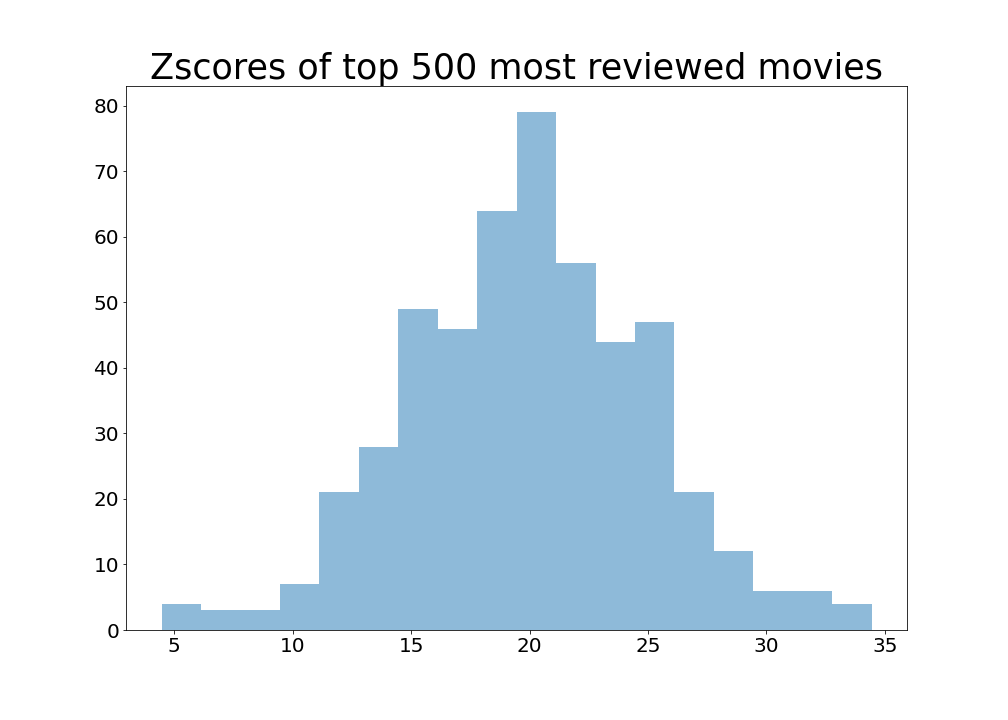} 
		\scalebox{.65}{
			\begin{tabular}[b]{r  p{6.6cm} r r}
				Rank & Title                                         & $Z$-Score & Total reviews   \\
				\hline 
				1 & Prometheus                                         & 34.44  & 813                \\
				2 & Expelled: No Intelligence Allowed                  & 34.17  & 830             \\
				3 & V for Vendetta                                     & 32.24  & 815                 \\
				4 & Sin City                                           & 31.72  & 828                \\
				$\vdots$ & $\vdots$ & $\vdots$ & $\vdots$  \\
				17 & Cars                                               & 19.98  & 902                \\
				18 & Food, Inc.                                         & 17.81  & 876        \\
				19 & Jeff Dunham: Arguing with Myself                   & 4.96   & 860                \\
				20 & Jeff Dunham: Spark of Insanity                     & 4.46   & 877              \\
		\end{tabular}}

	\includegraphics[width=0.30\textwidth, height=.28\textwidth]{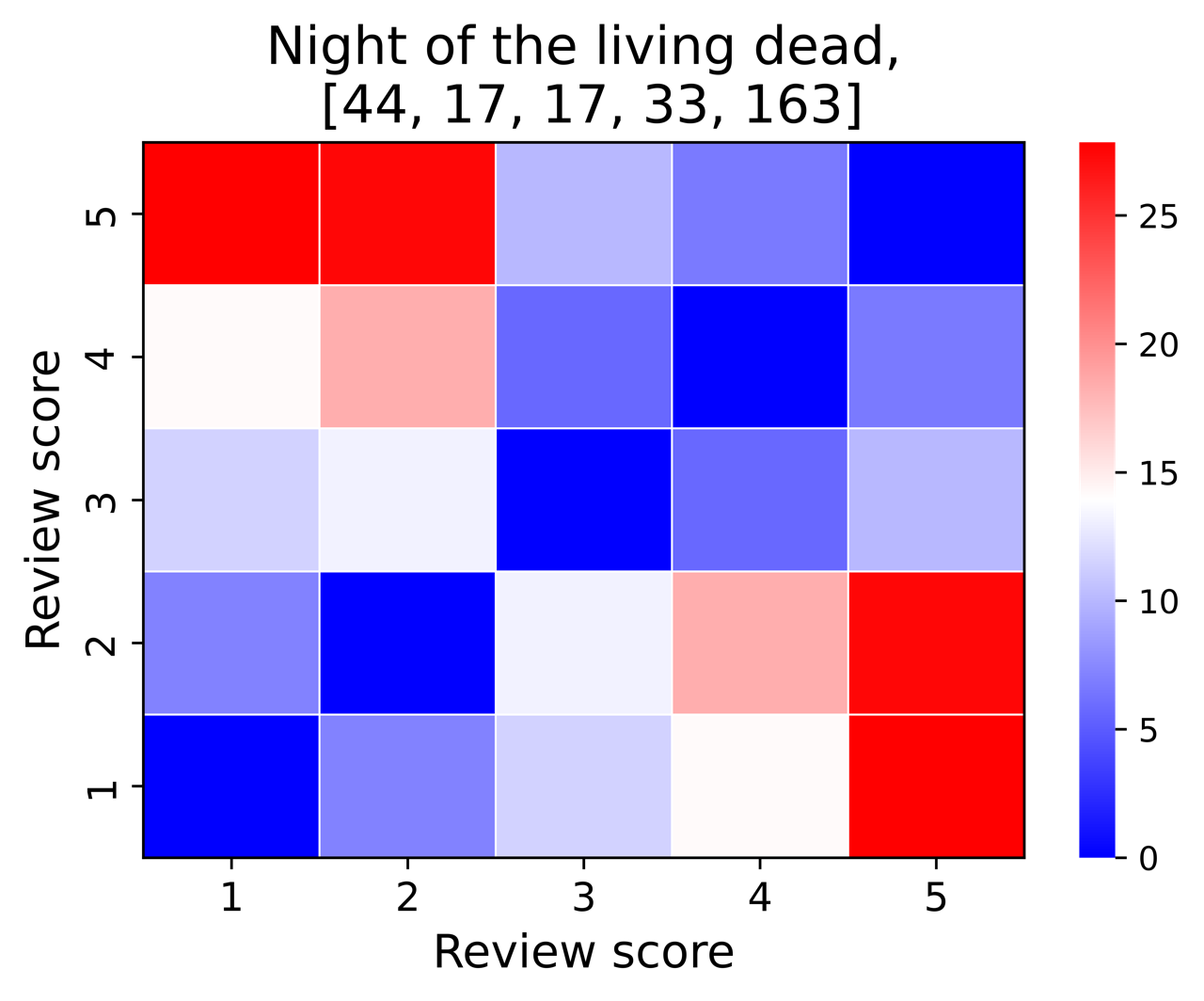}	\hspace{0.3cm}
	\includegraphics[width=0.30\textwidth, height=.28\textwidth]{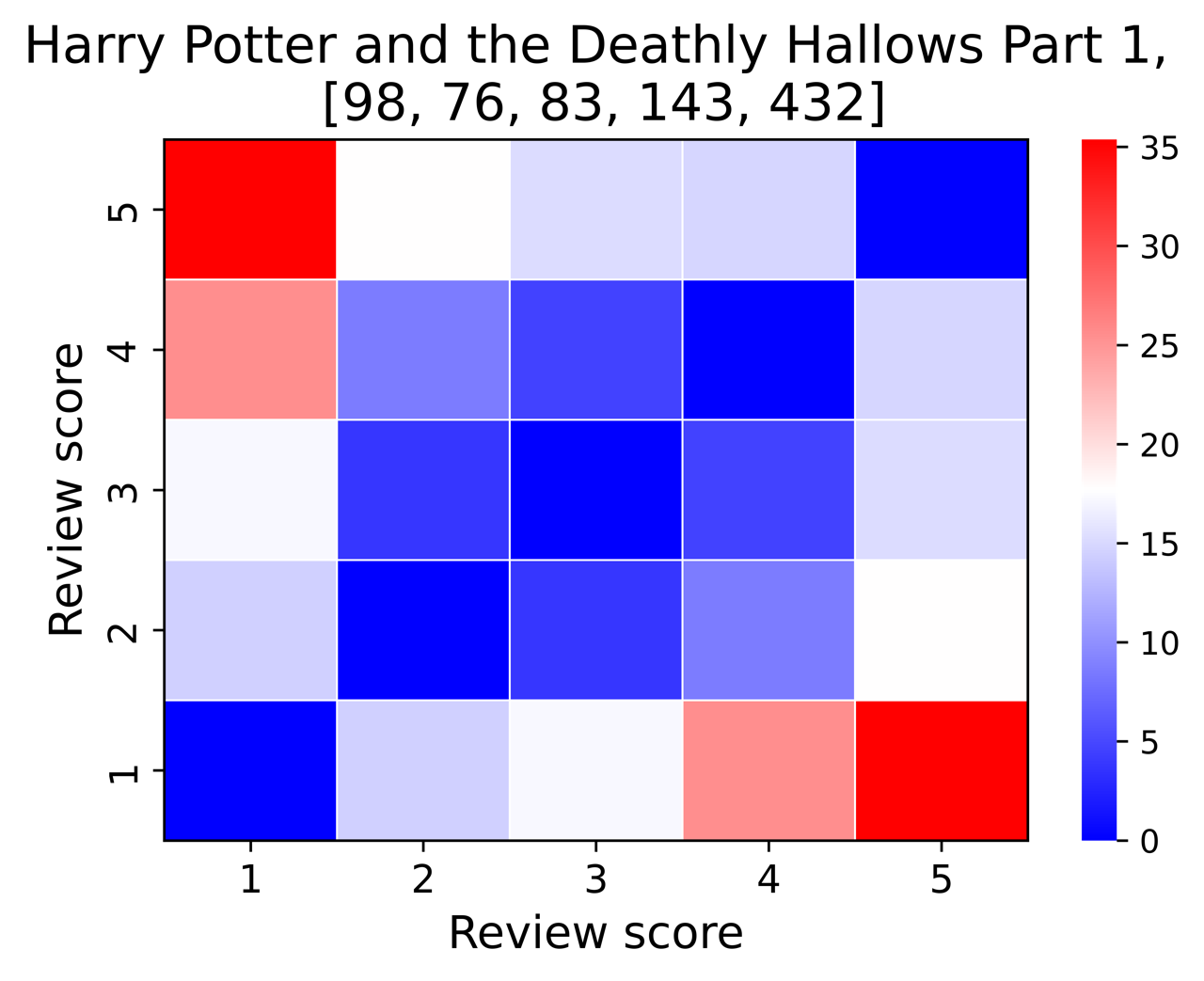}
	\hspace{0.3cm}
	\includegraphics[width=0.30\textwidth, height=.28\textwidth]{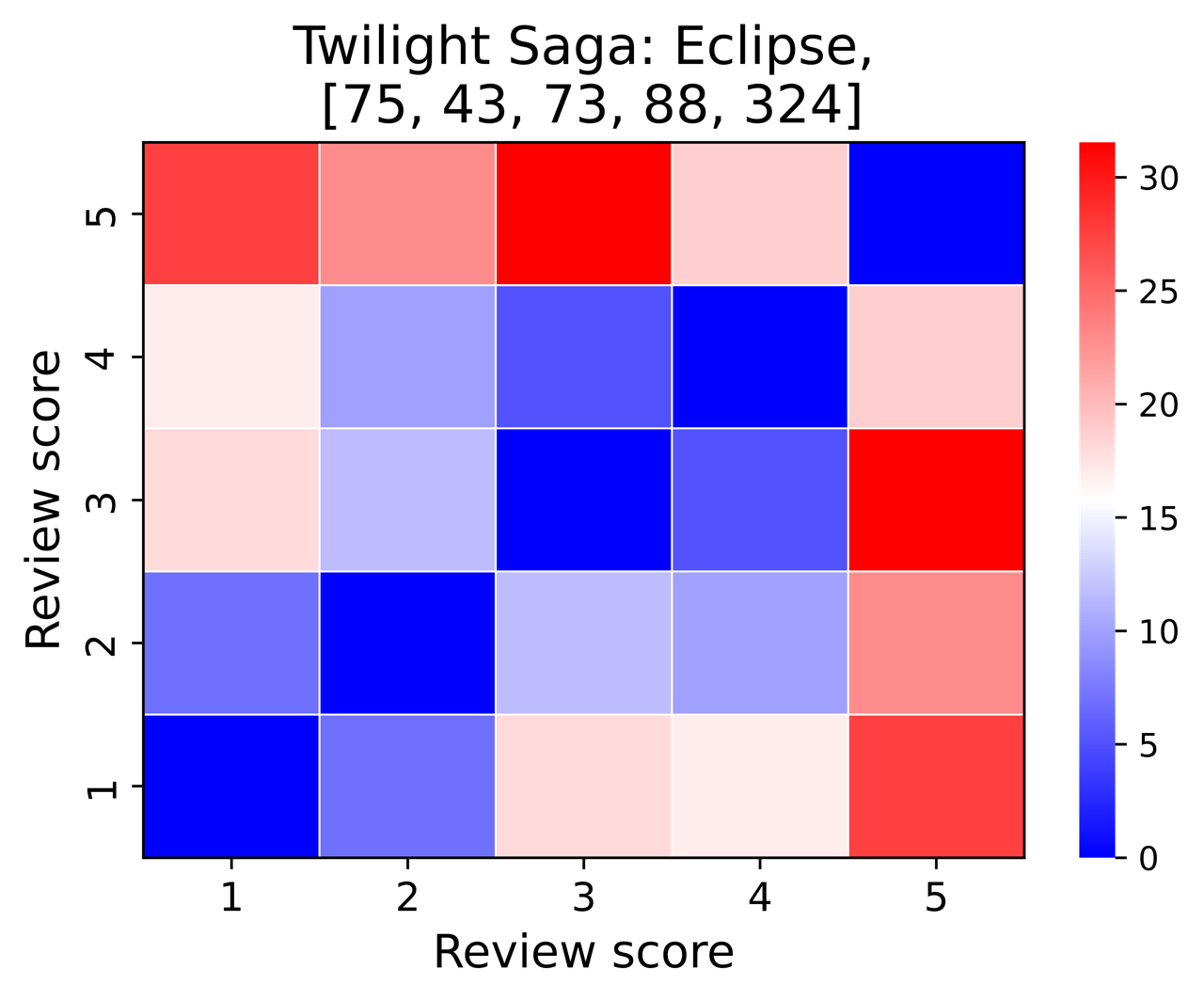}
	\caption{Results about movie reviews. Top left: Histogram of $Z$-scores for the 500 most-reviewed movies (mean is $19.97$, and standard deviation is $5.07$). Top right: Information and $Z$-scores for the top 20 most reviewed movies. Bottom: Pairwise $Z$-score plots for 3 representative movies (the title lists the number of reviews of each rating from 1--5). 	
	} 
	\label{fig:movie_heatmaps}
\end{figure}

The dataset in \cite{maurya18} contains 1,924,471 reviews of 143,007 visual media products (ie, DVDs, Bluray, or streams). We cleaned and stemmed these review text similarly as in Section~\ref{subsec:StatAbstracts}. In the first experiment, given a movie, we consider the corpus consisting of all reviews of this movie and apply DELVE+ with $K=n$. The results are in the top panels of Figure~\ref{fig:movie_heatmaps}. First, we plot the histogram of $Z$-scores for the top 500 most reviewed movies. The mean is $19.97$ and the standard deviation is $5.07$. Compared with the histogram of $Z$-scores for statistics paper abstracts, there is much larger diversity in movie reviews.  Next, we list the 4 movies with the highest $Z$-scores and lowest $Z$-scores out of the 20 most reviewed movies. Each movie has more than 800 reviews, but some have surprisingly low $Z$-scores. The works by the comedian Jeff Dunham have the lowest $Z$-scores,  suggesting  strong homogeneity among the reviews. The 2012 horror film {\it Prometheus}  has the highest degree of review diversity among the 20 most reviewed movies.  
In the second experiment, we further divide each movie's reviews into 5 groups by star rating. We compare each pair of groups using DELVE+  with $K=2$, resulting in a pairwise $Z$-score plot. In the bottom panels of Figure \ref{fig:movie_heatmaps}, we plot this for 3 popular movies. We see a variety of polarization patterns among the scores. In {\it Harry Potter and the Deathly Hallows Part I}, DELVE+ signifies that the reviews with ratings in the range 2--4 stars are all similar. We see a smooth gradation in how the 1-star reviews differ from those from 2--4 stars, and similarly for  5-star reviews versus those from 2--4 stars. {\it Twilight Saga: Eclipse} shows three clusters: 1--2 stars, 3--4 stars, and 5 star, while {\it Night of the living dead} shows two  clusters: 1--2 stars and  3--5 stars. 

As mentioned in Section \ref{sec:Intro}, the marketing research aims to understand patterns of online customer reviews. 
Our DELVE testing framework is a flexible approach to detecting many kinds of heterogeneity in review text. If reviewer characteristics (e.g., gender) are available, we can group reviews by these characteristics and answer questions such as if female and male reviewers have different styles in writing review text. In the experiments here, we showcase how to use DELVE to find patters in movie ratings.   
Although many literature works have studied patterns of movie reviews \citep{baek2012helpfulness}, most are based on the distribution of numeric ratings. The three movies in Figure \ref{fig:movie_heatmaps} have similar distributions of numerical ratings, but the patters in text reviews are considerably different. Such plots will be useful for improving rating systems, recommending movies to customers, and detecting fake reviews.


\vspace{-0.3cm}

\section{Discussions} \label{sec:Discuss}

We examine the  testing for equality of PMFs of $K$ groups of high-dimensional multinomial distributions. The proposed DELVE statistic has a parameter-free limiting null that allows for computation of $Z$-scores and $p$-values on real data. DELVE achieves the optimal detection boundary over the whole range of parameters  $(n, p, K, \bar{N})$, including the high-dimensional case  $p\to\infty$, which is very relevant to applications in text mining.

This work leads to interesting questions for future study. Recall that the $\rho^2$ defined in \eqref{def:rhoSquare} is a measure of heterogeneity among the group-wise means.  So far, the focus is on testing $\rho^2=0$, but we may also consider estimation and inference of $\rho^2$.  Assuming $\rho^2=0$, we have obtained a consistent variance estimator for the DELVE metric in \eqref{DELAC} and established it asymptotic normality. To construct a confidence interval for $\rho^2$, we will need such results under the alternative hypothesis (where $\rho^2\neq 0$). 
From Figure~\ref{fig:Experiment1}, the asymptotic normality still holds when $\rho^2\neq 0$, except that stronger regularity conditions may be required. Inspired by the authorship attribution problem \citep{kipnis2021two,kipnis2022higher}, it is interesting to consider a sparse alternative hypothesis where the group mean vectors are equal except on a small set of ``giveaway words". 
As discussed in Section \ref{subsec:authorChallenge}, we may combine DELVE with the idea of higher criticism.

Another exciting future direction is to extend our methods from the `bag-of-words' model to more realistic sequence-based models. One approach is to consider the counts of adjacent words (bi-grams) instead of raw word counts. More generally, one can consider the counts of short sequences of words, which are known as $m$-grams. It is possible that a suitably modified version of DELVE would perform well in a setting where the next word is generated according to a Markov transition kernel whose input is the previous $m-1$ observed words \citep{jurafsky2000speech}. A final idea is to combine words that have similar meanings or are close in a word embedding into `superwords' and to use these superword counts as the basis for DELVE. We leave them to future work.

\medskip
\footnotesize  
\noindent \textbf{Acknowledgments} The research of T. Tony Cai was supported in part by NSF Grant DMS-2015259 and NIH grant R01-GM129781. The research of Zheng Tracy Ke was supported in part by NSF CAREER Grant DMS-1943902.

\vspace{-0.3cm}

\newpage 

\appendix 

{\noindent \Large\textbf{Appendix}}

\section{Additional simulation results} \label{supp:Simu}

\noindent \textbf{Notational conventions:} We write $A \les B$ (respectively, $A \gtrsim B$) if there exists an absolute constant $C > 0$ such that $A \leq C \cdot B$ (respectively $A \geq C \cdot B$). If both $A \les B$ and $B \les A$, we write $A \asymp B$. The implicit constant $C$ may vary from line to line. For sequences $a_t, b_t$ indexed by an integer $t \in \mathbb{N}$, we write $a_t \ll b_t $ if $b_t/a_t \to \infty$ as $t \to \infty$, and we write $a_t \gg b_t$ if $a_t/b_t \to \infty$ as $t \to \infty$. We also may write $a_t = o(b_t)$ to denote $a_t \ll b_t$. In particular, we write $a_t = (1 + o(1)) b_t$ if $a_t/b_t \to 1$ as $t \to \infty$. Given a positive integer  $T $, define $[T] = \{1, 2, \ldots, T \}$. 

We present some simulation results that are not included in the main paper for space constraint. 

\subsection{Power diagrams of DELVE+} \label{supp:Figure2omitted}

In Experiment~2 of Section~\ref{sec:Simu}, we investigate the power of the DELVE test. We now present the power diagrams for DELVE+. Please see Figure~\ref{fig:Experiment2-supp}, where the simulation settings are the same as those in Figure~\ref{fig:Experiment2}. 
Comparing these two figures, we observe that DELVE+ and DELVE have similar power on simulated data. This is consistent with our theory in Section~\ref{subsec:modification}.

\begin{figure}[htb]
	\centering

	\includegraphics[width=0.32\textwidth]{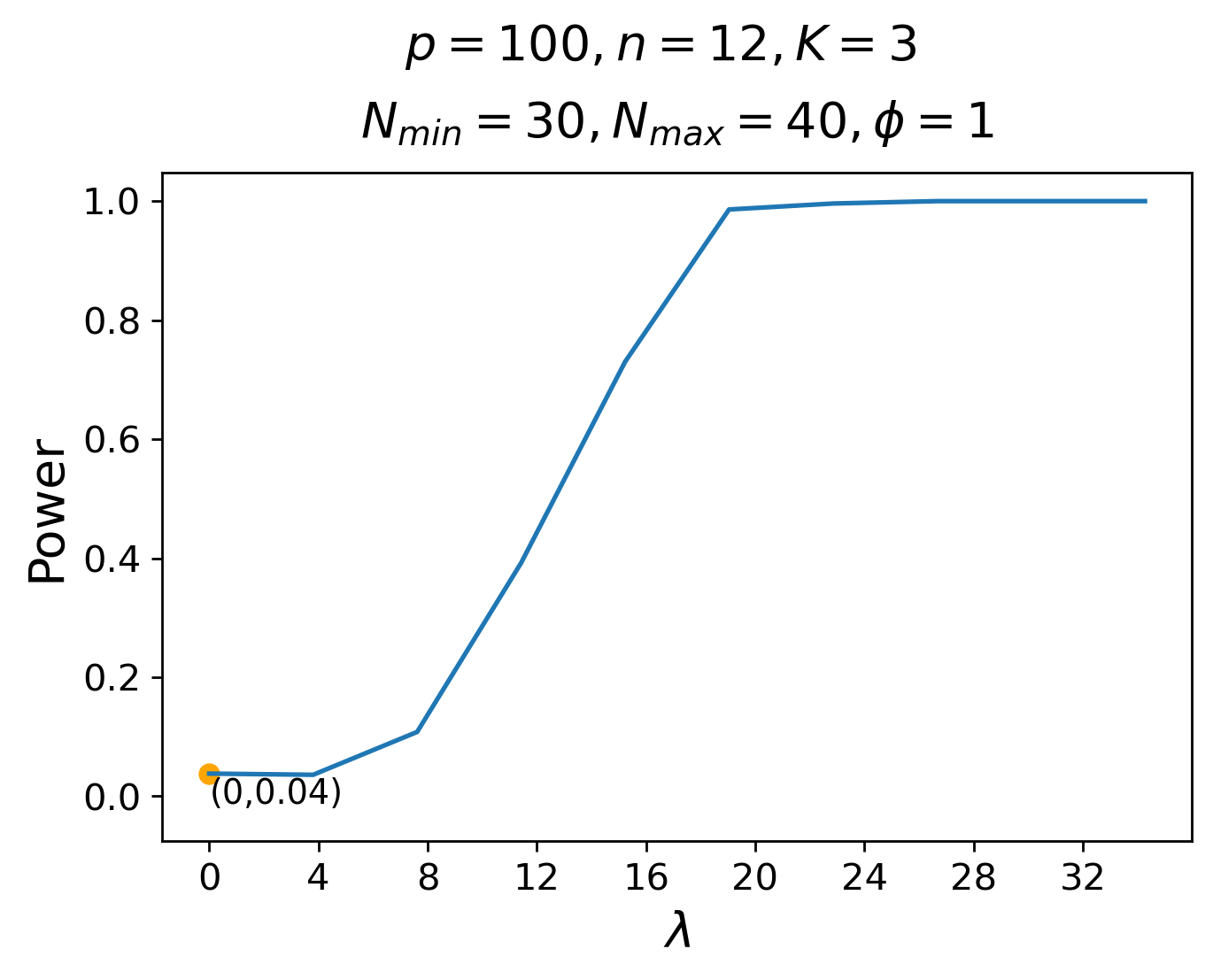}
	\includegraphics[width=0.32\textwidth]{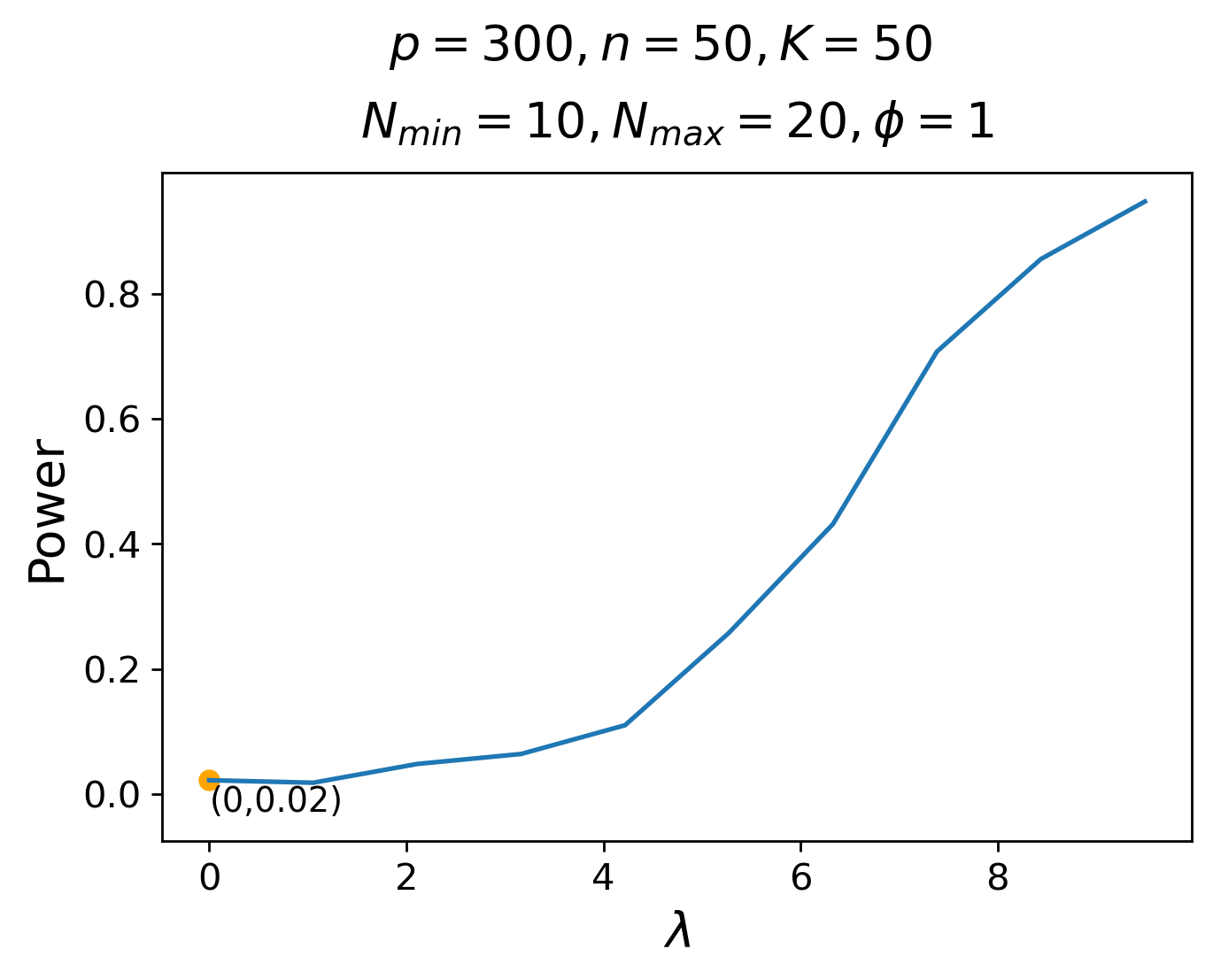}
	\includegraphics[width=0.32\textwidth]{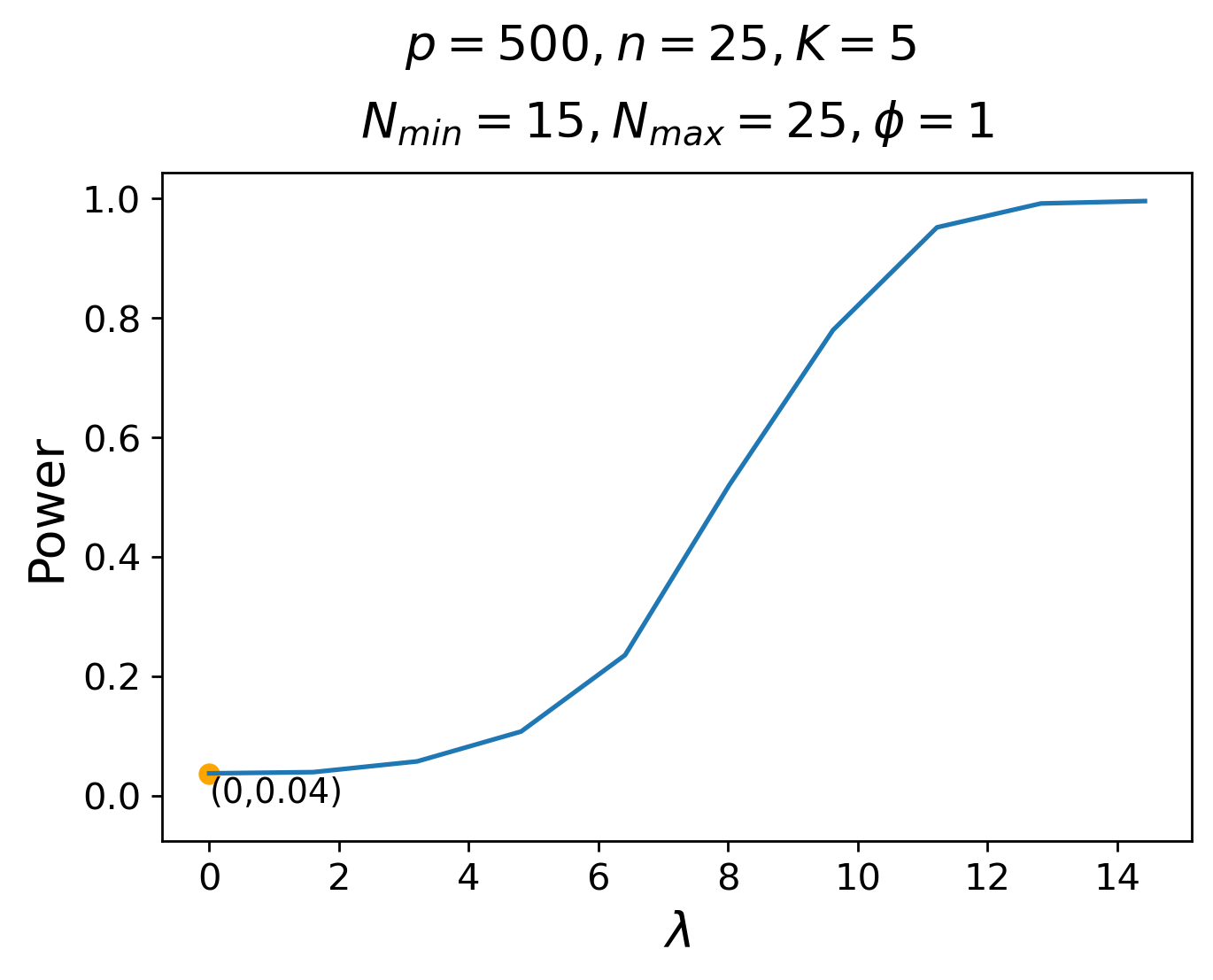}
	\caption{
		Power of the level-$5\%$ DELVE+ test ($x$-axis represents the SNR $\lambda(\tau_n) =\frac{ n \bar{N} \| \mu \|\tau_n^2}{\sqrt{K}}$).
	}
	\label{fig:Experiment2-supp}
\end{figure}

\subsection{More comparison between LR and DELVE+} \label{supp:LR}

In Experiment~3 of Section~\ref{sec:Simu}, we compare the power of DELVE+ with that of the likelihood ratio (LR) test. We recall that in our general setting \eqref{Mod3-null}, both the null and alternative hypotheses are highly composite, because $\Omega_i$'s are allowed to be unequal within each group. It is impossible to compute the LR test statistic, except in the special setting where all of the $\Omega_i$'s in group $k$ are equal to $\mu_k$. In this special setting, the LR test statistic takes the form 
\begin{equation}
	\label{eqn:LR_test}
	LR :=  \sum_{k}    n_k \bar{N}_k \sum_j \hat \mu_{kj}  \log\big( \frac{ \hat \mu_{kj} }{ \hat \mu_j } \big),
\end{equation}
where 
\begin{equation}\label{define:etaHat-supp}
	\hat{\mu}_k = \frac{1}{n_k\bar{N}_k}\sum_{i\in S_k}X_i, \qquad\mbox{and}\qquad \hat{\mu} = \frac{1}{n\bar{N}}\sum_{k=1}^K n_k\bar{N}_k\hat{\mu}_k =  \frac{1}{n\bar{N}}\sum_{i=1}^n X_i. 
\end{equation}
To ensure that LR is well-defined in the case of zero-counts (ie, $\hat \mu_{kj} = 0$ ), we define $\log(0/0) = 0$.

\begin{figure}[htb]
	\centering
	\includegraphics[width=0.32\textwidth]{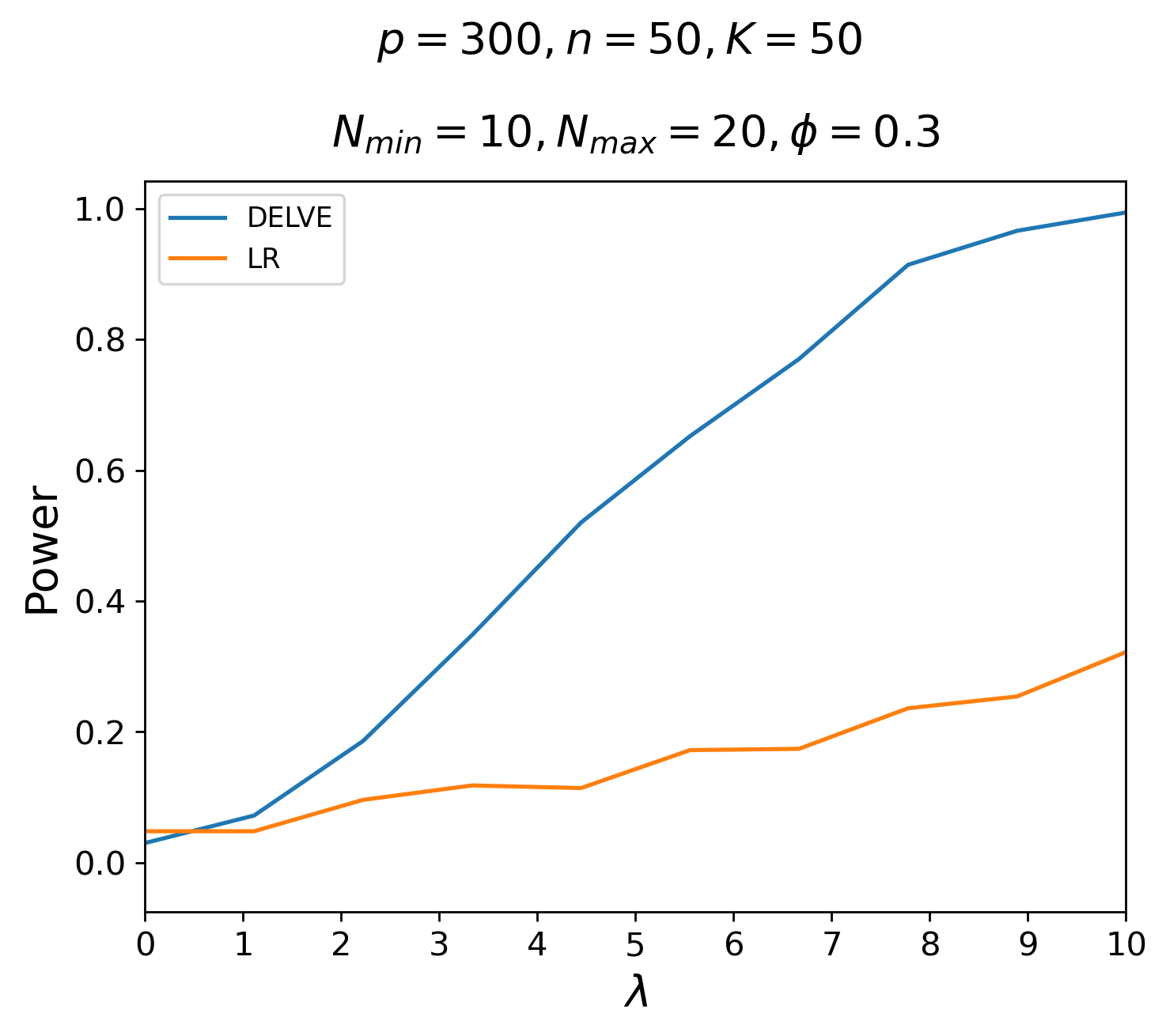}
	\includegraphics[width=0.32\textwidth]{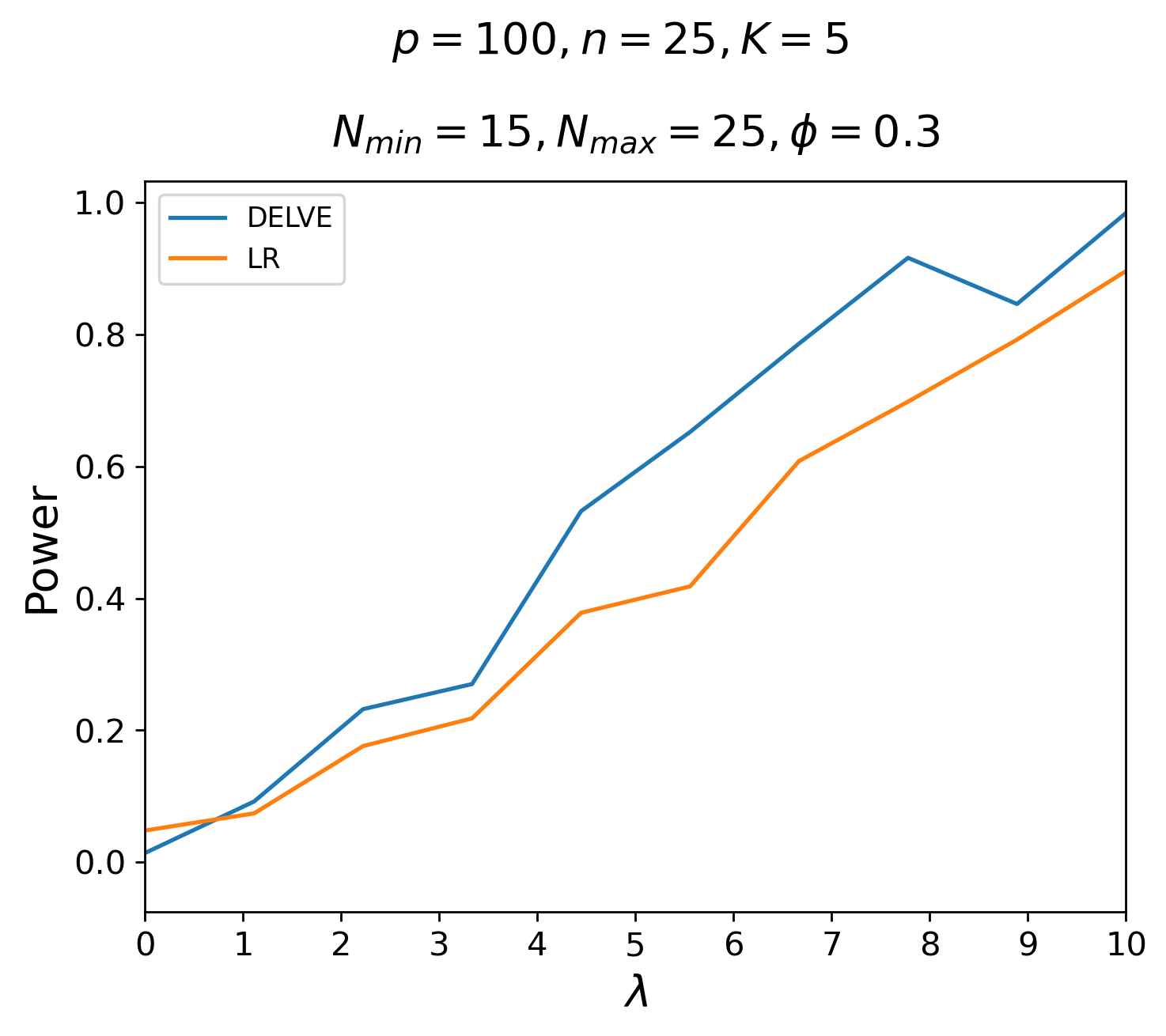}
	\includegraphics[width=0.32\textwidth]{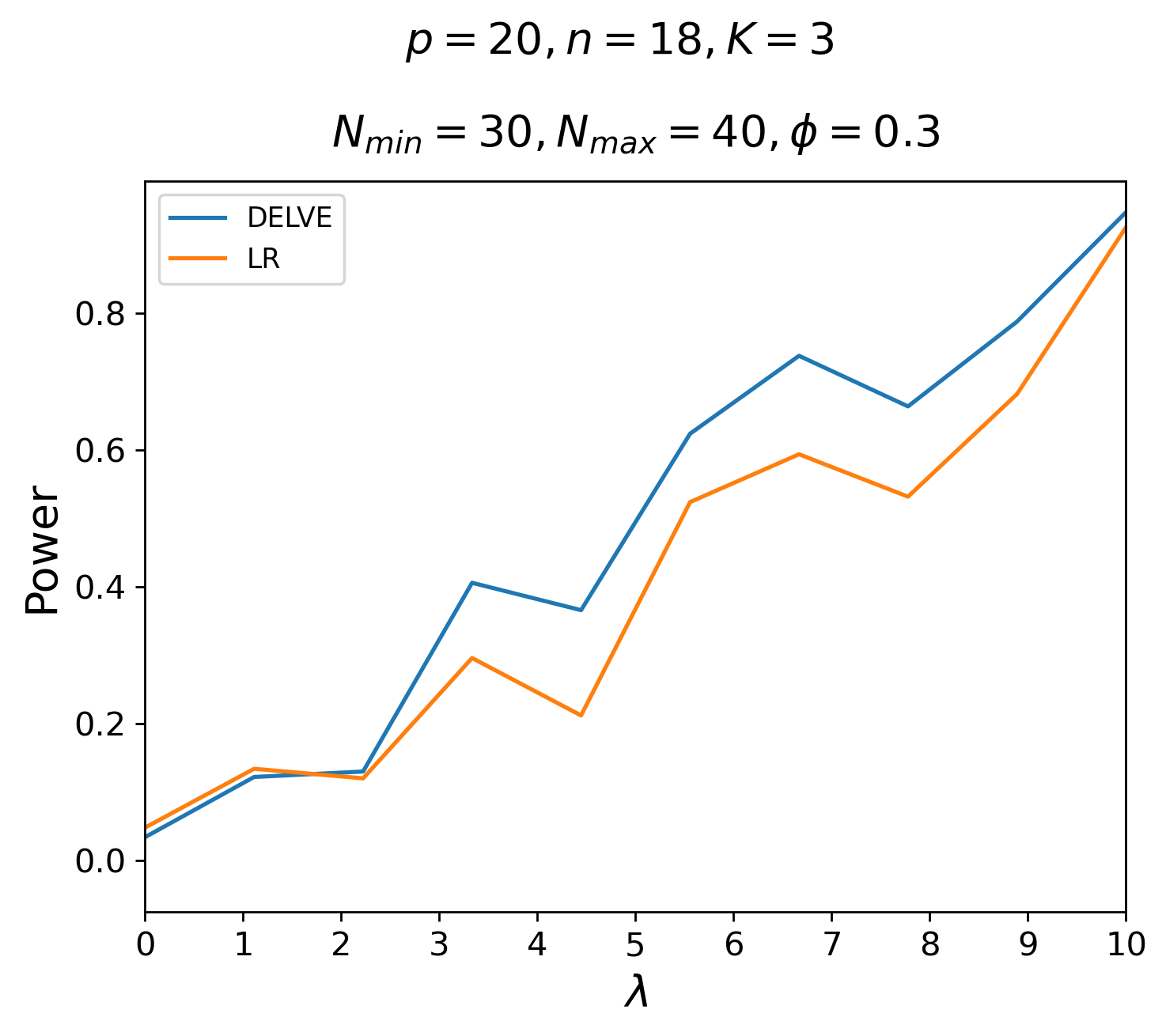}
	\caption{Power curves for DELVE+ (blue) and LR (orange) versus SNR $\lambda$ for two different settings of $(p, n, K, N_{\min}, N_{\max}, \phi )$. }
	\label{fig:LR_vs_DELVE-supp}
\end{figure}

In Figure~\ref{fig:Experiment3} of the main paper, we have seen the power diagrams of LR and DELVE+ for two values of $(p, n, K, N_{\min}, N_{\max}, \phi)$. Results for some other values of  $(p, n, K, N_{\min}, N_{\max}, \phi)$ are in  Figure~\ref{fig:LR_vs_DELVE-supp}. These results suggest that when $p$ is relatively large, DELVE+ outperforms LR in terms of power. In theory, DELVE+ attains the optimal detection boundary, but the asymptotic behavior of LR for large-$p$ is unclear.  There are cases where LR performs somewhat better than DELVE+, but they seem to be limited to the smaller-$p$ regime.

\section{Supplementary results from real data} \label{supp:Real}

\subsection{The pairwise $Z$-score of another author} \label{supp:fan}

In Section~6.1, we give a pair-wise $Z$-score plot for a representative author (denoted by Author~A). We can produce such a plot for any author in our data set. Here we show another example (this author is denoted by Author~B). 
Compared to Author~A, the publication years of Author~B's papers are less evenly distributed. We divide Author~B's abstracts into 6 groups, and the time window sizes for 6 groups are unequal, to guarantee that  all groups have roughly equal numbers of abstracts. The pairwise $Z$-score plot for Author~B is in the right panel of Figure~\ref{fig:author_AB-supp}. We also include the pairwise $Z$-score plot for Author~A in the left panel of this figure (which is the same as the right panel of Figure~\ref{fig:author_diversity}). 

\begin{figure}[htb]
	\centering
	\includegraphics[width=0.42\textwidth]{PDF/Figure4_Plot1.png}
	\hspace{1cm}
	\includegraphics[width=0.4\textwidth]{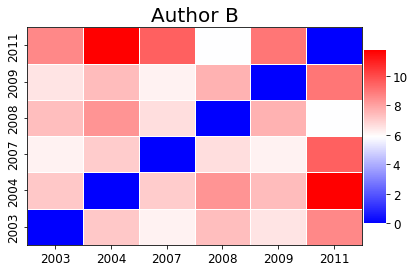}
	\caption{Pairwise $Z$-score plots for 
		Author A (left) and 
		Author B (right). In the cell $(x,y)$, we compare the corpus of an author's abstracts from time $x$ with the corpus of that author's abstracts from time $y$. The heatmap shows the value of DELVE+ with $K = 2$ for each cell.
	} 	
	\label{fig:author_AB-supp}
\end{figure}

There are some interesting temporal patterns. For Author A, the group consisting of 2004-2005 abstracts has comparably large $Z$-scores in the pairwise comparison with other groups, and similarly for Author B, the group of 2011-2012 abstracts have relatively large $Z$-scores. To gain further insight, we collected the titles and abstracts of each author's papers and manually inspected them. We found that Author A extensively studied topics related to bandwidth selection in the context of nonparametric estimation. For Author B, the time period 2011-2012 reveals a more intense focus on variable selection, compared to this author's papers in other years within this data set.

\subsection{Checking the applicability of our asymptotic result on real data} \label{supp:DR}

The properties of the DELVE test are established in the asymptotic regime of $n^2\bar{N}^2/(Kp)\to\infty$ (see  Section~\ref{sec:Theory}). We check if this ``asymptotics" is reasonable for real applications. To this end, define the {\it dimension ratio} as
\begin{equation}
	\label{eqn:dimension_ratio}
	DR := n^2 \bar{N}^2 K^{-1} p^{-1}.
\end{equation}
The larger $DR$, the more appropriate to apply our asymptotic theory. We report the DR values of all the corpora used in the analysis of statistics abstracts. In the first experiment of Section~\ref{subsec:StatAbstracts}, for each author, we take all his/her abstracts as the corpus and apply DELVE with $K=n$. Each author is associated with a corpus. Figure~\ref{fig:DR-supp1} displays the DR values for the corpora of the 15 most prolific authors. In the second experiment of Section~\ref{subsec:StatAbstracts}, we take the abstracts written by an author (Author A), divide them by year into 9 groups, and apply DELVE with $K=2$ to each pair of groups. There are a total of $(9\times 8)/2=36$ corpora for this experiment, whose DR values are shown in the left panel of Figure~\ref{fig:DR-supp2}. In Section~\ref{supp:fan}, we conduct similar analysis for another author (Author B). The DR values in this experiment are in the right panel of Figure~\ref{fig:DR-supp2}. 
These DR values are large, suggesting that our asymptotic setting is relevant for real applications and that the $Z$-scores obtained in these experiments are trustworthy. 

\begin{figure}[t]
	\centering
	\scalebox{0.7}{
		\begin{tabular}{llllll}
			Author & Total papers& Average abstract & Vocab size & $DR$      \\
			& ($n$) & length ($\bar{N}$)  & ($p$)  & ($n \bar{N}^2/p$) \\ 
			\hline 
			1               & 81                     & 75.90                   & 1103            & 423.07   \\
			2               & 40                     & 81.78                   & 801             & 333.94   \\
			3               & 39                     & 75.38                   & 758             & 292.39 \\
			4               & 32                     & 68.66                   & 562             & 268.39  \\
			5               & 30                     & 98.77                   & 672             & 435.48  \\
			6               & 27                     & 85.74                   & 698             & 284.37   \\
			7               & 27                     & 72.59                   & 592             & 240.34   \\
			8               & 24                     & 65.58                   & 471             & 219.17   \\
			9               & 22                     & 61.23                   & 415             & 198.73   \\
			10              & 20                     & 73.55                   & 463             & 233.68 \\
			11 & 20 & 84.15  & 502 & 282.12 \\
			12 & 19 & 114.53 & 617 & 403.90 \\
			13 & 19 & 52.47  & 361 & 144.92 \\
			14 & 18 & 77.06  & 459 & 232.85 \\
			15 & 18 & 59.17  & 369 & 170.77 \\
		\end{tabular}
	}

	\caption{Summary statistics and DR values of the corpora of the top 15 most prolific authors.} 
	\label{fig:DR-supp1}
\end{figure} 

\begin{figure}[tb]
	\centering
	\begin{tabular}{cc}
		\scalebox{0.7}{
			
			\begin{tabular}{c|ccccccccc}
				& 2003    & 2004    & 2005    & 2006    & 2007    & 2008    & 2009    & 2010    \\
				\hline 
				2003 & ---    \\
				2004 & 1411 & ---  \\
				2005 & 1313 & 1518 & ---   \\
				2006 & 1986 & 2208 & 2107 & ---  \\
				2007 & 1408 & 1541 & 1470 & 2216 & \\
				2008 & 1448 & 1615 & 1547 & 2263 & 1615 & \\
				2009 & 1887 & 2088 & 1981 & 2753 & 2065 & 2223 \\
				2010 & 1506 & 1714 & 1650 & 2395 & 1714 & 1758 & 2293 \\
				2011 & 1393 & 1576 & 1499 & 2213 & 1617 & 1631 & 2160 & 1762   
			\end{tabular}
		}
		&
		\scalebox{0.8}{
			\begin{tabular}{c|cccccc}
				Time $\backslash$ Time  & 2003    & 2004    & 2007    & 2008    & 2009      \\
				\hline 
				2003 &\\
				2004 & 1145 & \\
				2007 & 859  & 1636 &  \\
				2008 & 784  & 1548 & 1263 &  \\
				2009 & 1226 & 2064 & 1675 & 1597 &  \\
				2011 & 963  & 1694 & 1347 & 1358 & 1843    
			\end{tabular}
		}
	\end{tabular}
	\caption{The DR values for cells the pairwise $Z$-score plots in Figure~\ref{fig:author_AB-supp}, where the left table is for  Author~A  and the right table is for Author~B.} 
	\label{fig:DR-supp2}
\end{figure}

\section{Some analysis of the naive ANOVA test} \label{supp:ANOVA}

In Section~\ref{sec:Method}, we introduced a native estimator of $\rho^2$ as 
\[
\widetilde{T} = \sum_{k=1}^K n_k\bar{N}_k\|\hat{\mu}_{k}-\hat{\mu}\|^2.
\]
Consider a $K\times p$ ``contingency table" whose $(k,j)$th cell is $\sum_{i\in S_k}X_i(j)$. Then, $\widetilde{T}$ is an ANOVA-type statistic associated with this contingency table. It is interesting to investigate the test based on $\widetilde{T}$ and compare it with our proposed DELVE test. 

In the proof of Lemma \ref{lem:decompose}, we will show that
\beq \label{tildeT-mean}
\mathbb{E}[\widetilde{T}] = \rho^2 + 
J_5, \qquad\mbox{where}\quad J_5 = \sum_{k=1}^K\sum_{i\in S_k} \sum_j \Bigl(1-\frac{n_k\bar{N}_k}{n\bar{N}}\Bigr)\frac{N_i\Omega_{ij}(1-\Omega_{ij})}{n_k\bar{N}_k}. 
\eeq
Here, $\rho^2$ is the signal of interest, and $J_5$ characterizes the bias in $\widetilde{T}$. To gain some insight about the order of these two terms, we consider a simple case where (i) groups have equal size,  (ii) $N_i$'s are equal, (iii) $\Omega_{ij} = O(p^{-1})$, (iv) under $H_1$, $\min_k \|\mu_k-\mu\|\geq c_0\|\mu\|$, for a constant $c_0>0$. It holds that 
\beq \label{tildeT-order}
J_5 \asymp K/p\;\; \mbox{under $H_0$ and $H_1$}, \quad\mbox{and}\quad  \rho^2 \asymp n \bar{N} /p^2\;\; \mbox{under $H_1$}. 
\eeq
The bias term is negligible if $n\bar{N}\ll Kp$. This is a stronger condition than the optimal detection boundary, which only requires $n^2 N^2 \gg K p$. In particular, when
\[
Kp \;\;\ll \;\; n^2\bar{N}^2 \;\; \ll\;\; K^2p^2,
\]
the bias term dominates the ``signal" term, so the test based on $\widetilde{T}$ may lose power. In comparison, the DELVE statistic $T$ in \eqref{DELAC} is a de-biased version of $\widetilde{T}$, hence, it has no such issue.

{\bf An example where $\widetilde{T}$ is powerless}. Suppose $K = n$, both $n$ and $p$ are even, and $N_i \equiv N$. Take two vectors $\sigma \in \{-1, 1 \}^p$ and $\varepsilon \in \{-1,1\}^n$ such that $\sum_{j=1}^p\sigma_j=0$ and $\sum_{i=1}^n \varepsilon_i = 0$. Under $H_0$, let $\Omega=p^{-1}{\bf 1}_p{\bf 1}_n'$. Under $H_1$, let $\Omega_{ij} = p^{-1} + \alpha p^{-1} \varepsilon_i \sigma_j$, for some $\alpha \in (0,1)$. We can easily check that each $\Omega_i$ is indeed a PMF. 
For this example, 
\begin{align*}
	J_5^{alt} - J_5^{null} 
	&= (1 - \frac{1}{n}) \sum_{i,j} \frac{1}{p}( 1+ \alpha \varepsilon_i \sigma_j )
	(1 - \frac{1}{p} - \frac{1}{p} \alpha \varepsilon_i \sigma_j )  - (1 - \frac{1}{n}) \sum_{i,j} \frac{1}{p} (1 - \frac{1}{p} )
	\\&
	= - (1 - \frac{1}{n}) \sum_{i,j} \frac{1}{p^2} \alpha^2 \varepsilon_i^2 \sigma_j^2 = - (1 - \frac{1}{n}) \frac{\alpha^2 n}{p}.
\end{align*}
Moreover, $\rho_{null}^2=0$, and $\rho_{alt}^2 =O( nN/p^2)$. When $p \gg N$ and $\alpha$ is lower bounded by a constant, 
\begin{align*}
	\mathbb{E}_1[\widetilde{T}] - \mathbb{E}_0[\widetilde{T}] = \rho_{alt}^2+
	J_5^{alt} - J_5^{null} = O\Bigl(\frac{nN}{p^2} \Bigr)
	- (1 - \frac{1}{n}) \frac{\alpha^2 n}{p} 
	\leq  - \frac{\alpha^2 n}{2p}. 
\end{align*}
Since $\mathbb{E}_1[\widetilde{T}]$ is smaller than $\mathbb{E}_0[\widetilde{T}]$, the test based on $\widetilde{T}$ is powerless.

\section{Properties of $T$ and $V$} \label{sec:decomposition}

This section is a preparation for the proofs of our main theorems. 
We recall that 
\beq \label{relaxedMod1}
X_i\sim \mathrm{Multinomial}(N_i, \Omega_i), \qquad 1\leq i\leq n. 
\eeq
For each $1\leq k\leq K$, define
\beq \label{relaxedMod2}
\mu_k = \frac{1}{n_k\bar{N}_k}\sum_{i\in S_k}N_i\Omega_i\;\in\; \mathbb{R}^p, \qquad \Sigma_k=\frac{1}{n_k\bar{N}_k}\sum_{i\in S_k}N_i\Omega_i\Omega_i'\;\in\;\mathbb{R}^{p\times p}. 
\eeq
Moreover, let
\beq \label{relaxedMod3}
\mu=\frac{1}{n\bar{N}}\sum_{k=1}^Kn_k\bar{N}_k\mu_k = \frac{1}{n\bar{N}}\sum_{i=1}^n N_i\Omega_i \, , 
\quad \Sigma = \frac{1}{n\bar{N}} \sum_{k = 1}^n n_k \bar{N}_k \Sigma_k
= \frac{1}{n \bar{N}} \sum_{i=1}^n N_i \Omega_i \Omega_i' 
\eeq
The DELVE test statistic is $\psi=T/\sqrt{V}$, where $T$ is as in \eqref{DELAC} and $V$ is as in \eqref{define:V}. As a preparation for the main proofs, 
in this section, we study $T$ and $V$ separately. 

\subsection{The decomposition of $T$} \label{subsec:T-decompose}

It is well-known that a multinomial with the number of trials equal to $N$ can be equivalently written as the sum of $N$ independent multinomials each with the number of trials equal to $1$. This inspires us to introduce  a set of independent, mean-zero random vectors:
\beq \label{Multinomials}
\{Z_{ir}\}_{1\leq i\leq n, 1\leq r\leq N_i}, \qquad\mbox{with }Z_{ir}=B_{ir}-\mathbb{E}B_{ir}, \;\; \mbox{and}\;\; B_{ir}\sim\mathrm{Multinomial}(1, \Omega_{i}). 
\eeq
We use them to get a decomposition of $T$ into mutually uncorrelated terms:

\begin{lemma} \label{lem:decompose}
	Let $\{Z_{ir}\}_{1\leq i\leq n, 1\leq r\leq N_i}$ be as in \eqref{Multinomials}. For each $Z_{ir}\in\mathbb{R}^p$, let $\{Z_{ijr}\}_{1\leq j\leq p}$ denote its $p$ coordinates.  
	Recall that $\rho^2=\sum_{k=1}^Kn_k\bar{N}_k \|\mu_{k}-\mu\|^2$. 
	For $1\leq j\leq p$, define
	\begin{eqnarray*}
		U_{1j} &=&  2\sum_{k=1}^K\sum_{i\in S_k}\sum_{r=1}^{N_i} (\mu_{kj}-\mu_j)Z_{ijr},\cr
		U_{2j} &=& \sum_{k=1}^K \sum_{i\in S_k}  \sum_{1\leq r\neq s\leq N_i}\Bigl(\frac{1}{n_k\bar{N}_k} -\frac{1}{n\bar{N}} \Bigr)\frac{N_i}{N_i-1} Z_{ijr}Z_{ijs},\cr
		U_{3j} &=& - \frac{1}{n\bar{N}}\sum_{1\leq k\neq\ell\leq K} \sum_{i\in S_k}\sum_{m \in S_\ell} \sum_{r=1}^{N_i}\sum_{s=1}^{N_m} Z_{ijr}Z_{mjs},\cr
		U_{4j} &=&  \sum_{k=1}^K \sum_{\substack{i\in S_k,m\in S_k\\i\neq m}} \sum_{r=1}^{N_i}\sum_{s=1}^{N_m}\Bigl(\frac{1}{n_k\bar{N}_k} -\frac{1}{n\bar{N}} \Bigr)  Z_{ijr}Z_{mjs}. 
	\end{eqnarray*}
	Then, $T= \rho^2+\sum_{\kappa=1}^4 {\bf 1}_p' U_{\kappa}$. 
	Moreover, $\mathbb{E}[U_\kappa]={\bf 0}_p$ and $\mathbb{E}[U_{\kappa}U_{\zeta}']={\bf 0}_{p\times p}$ for $1\leq \kappa\neq \zeta \leq 4$. 
\end{lemma}

\subsection{The variance of $T$} \label{subsec:T-var}

By Lemma~\ref{lem:decompose},  the four terms $\{{\bf 1}_p' U_{\kappa}\}_{1\leq \kappa\leq 4}$ are uncorrelated with each other. Therefore, 
\[
\mathrm{Var}(T) = \mathrm{Var}({\bf 1}_p'U_1)+\mathrm{Var}({\bf 1}_p'U_2)+\mathrm{Var}({\bf 1}_p'U_3)+\mathrm{Var}({\bf 1}_p'U_4). 
\]
It suffices to study the variance of each of these four terms.

\begin{lemma} \label{lem:var1}
	Let $U_1$ be the same as in Lemma~\ref{lem:decompose}. Define
	\begin{align}
		\Theta_{n1} &= 4\sum_{k=1}^K n_k\bar{N}_k\bigl\|\diag(\mu_k)^{1/2}(\mu_k-\mu)\bigr\|^2 \label{eqn:Theta_n1} \\
		L_n &= 4\sum_{k=1}^K n_k\bar{N}_k \bigl\|\Sigma_k^{1/2}(\mu_k-\mu)\bigr\|^2
		\label{eqn:Ln}
	\end{align}
	Then $\mathrm{Var}({\bf 1}_p'U_1) = \Theta_{n1} - L_n$. 
	Furthermore, if $\max_{1\leq k\leq K}\|\mu_k\|_\infty=o(1)$, then $\mathrm{Var}({\bf 1}_p'U_1)=o(\rho^2)$. 
\end{lemma}


\begin{lemma} \label{lem:var2}
	Let $U_2$ be the same as in Lemma~\ref{lem:decompose}. Define
	\begin{align}
		\Theta_{n2} &= 	2\sum_{k=1}^K \Bigl(\frac{1}{n_k\bar{N}_k}-\frac{1}{n\bar{N}}\Bigr)^2\sum_{i\in S_k}\frac{N_i^3}{N_i-1} \|\Omega_i\|^2 \label{eqn:Theta_n2}\\
		A_n &= 	2\sum_{k=1}^K \Bigl(\frac{1}{n_k\bar{N}_k}-\frac{1}{n\bar{N}}\Bigr)^2\sum_{i\in S_k}\frac{N_i^3}{N_i-1} \|\Omega_{i}\|_3^3 \label{eqn:An}
	\end{align}
	Then $$\Theta_{n2} - A_n \leq \mathrm{Var}({\bf 1}_p'U_2)  \leq  \Theta_{n2}.$$ 
	Furthermore, if 
	\begin{align}
		\label{eqn:var_condition_K=n}
		\max_{1\leq k\leq K}\big\{\frac{\sum_{i\in S_k}N^2_i\|\Omega_i\|_3^3}{\sum_{i\in S_k}N_i^2\|\Omega_i\|^2} \bigr\}=o(1),
	\end{align}
	then $\mathrm{Var}({\bf 1}_p'U_2) = [1+ o(1)] \cdot \Theta_{n2}$. 
\end{lemma}


\begin{lemma} \label{lem:var3}
	Let $U_3$ be the same as in Lemma~\ref{lem:decompose}. Define
	\begin{align}
		\Theta_{n3} &= \frac{2}{n^2\bar{N}^2}\sum_{ k\neq \ell }\sum_{i\in S_k}\sum_{m\in S_\ell} \sum_j N_iN_m\Omega_{ij}\Omega_{mj}  \label{eqn:Theta_n3} \\
		B_n &= 2\sum_{k\neq \ell}\frac{n_kn_{\ell}\bar{N}_k\bar{N}_\ell}{n^2\bar{N}^2}{\bf 1}_p'(\Sigma_k\circ\Sigma_\ell){\bf 1}_p
		\label{eqn:Bn}
	\end{align}
	Then $$ \Theta_{n3} - B_n\leq \mathrm{Var}({\bf 1}_p'U_3) \leq  \Theta_{n3} + B_n.$$
\end{lemma}

\begin{lemma} \label{lem:var4}
	Let $U_4$ be the same as in Lemma~\ref{lem:decompose}. Define 
	\begin{align}
		\Theta_{n4} &= 2\sum_{k=1}^K \sum_{\substack{i\in S_k,m\in S_k\\i\neq m}}\sum_j \Bigl(\frac{1}{n_k\bar{N}_k}-\frac{1}{n\bar{N}}\Bigr)^2 N_i N_m \Omega_{ij}\Omega_{mj}. \label{eqn:Theta_n4} \\
		E_n &= 2\sum_k
		\sum_{\substack{i\in S_k, m\in S_k,\\ i\neq m}} \sum_{1 \leq j, j' \leq p}  \Bigl(\frac{1}{n_k\bar{N}_k}-\frac{1}{n\bar{N}}\Bigr)^2  N_i N_m 
		\Omega_{ij} \Omega_{ij'} \Omega_{mj} 
		\Omega_{mj'}
		\label{eqn:En}
	\end{align}
	Then $$\Theta_{n4} - E_n \leq \mathrm{Var}({\bf 1}_p'U_4) \leq  \Theta_{n4} + E_n$$. 
\end{lemma}

Using Lemmas~\ref{lem:var1}-\ref{lem:var4}, we derive regularity conditions such that the first term in $\mathrm{Var}({\bf 1}_p'U_\kappa)$ is the dominating term. Observe that $\Theta_n = \Theta_{n1}+ \Theta_{n2} + \Theta_{n3} + \Theta_{n4}$, where the quantity $\Theta_n$ is defined in \eqref{define:Theta_n}. The following intermediate result is useful. 

\begin{lemma}
	\label{lem:Theta_n2+n3+n4}
	Suppose that  \eqref{cond1-basic} holds. Then
	\begin{align}
		\Theta_{n2} + \Theta_{n3} + \Theta_{n4} \asymp \sum_k \| \mu_k \|^2. 
	\end{align}
	Moreover, under the null hypothesis, $\Theta_n \asymp  K \| \mu \|^2$. 
\end{lemma}

The next result is useful in proving that our variance estimator $V$ is asymptotically unbiased.

\begin{lemma}
	\label{lem:var-null}
	Suppose that \eqref{cond1-basic} holds, and recall the definition of $\Theta_n$ in \eqref{define:Theta_n}. 
	Define 
	\begin{align}
		\beta_n = \frac{ \max \bigg\{ 
			\sum_k  \sum_{i\in S_k}\frac{N^2_i}{n_k^2 \bar{N}_k^2}\|\Omega_i\|_3^3 \, ,  \, \,
			\sum_k \| \Sigma_k \|_F^2
			\bigg\} }{ K \| \mu \|^2 }. \label{eqn:beta_n}
	\end{align}
	If $\beta_n = o(1)$, then under the null hypothesis, $\var(T) = [1+o(1)] \cdot \Theta_n$. 
\end{lemma}

We also study the case of $K = 2$ more explicitly. In the lemmas below we use the notation from Section \ref{subsec:K=2}.  First we have an intermediate result analogous to Lemma \ref{lem:Theta_n2+n3+n4} that holds under weaker conditions. 

\begin{lemma}
	\label{lem:Theta_n2+n3+n4-K=2}
	Consider $K = 2$ and suppose that $\min N_i \geq 2$, $\min M_i \geq 2$
	Then
	\begin{align*}
		\Theta_{n2} + \Theta_{n3} + \Theta_{n4} \asymp  \bigg \|  \frac{m \bar{M}}{ n\bar{N}+ m \bar{M}} \eta +  
		\frac{n \bar{N}}{ n\bar{N}+ m \bar{M}} \theta  \bigg \|^2.
	\end{align*}
	Moreover, under the null hypothesis, $	\Theta_n \asymp \|  \mu  \|^2$.
\end{lemma}

The next result is a version of Lemma \ref{lem:var-null} for the case $K = 2$ that holds under weaker conditions. 
\begin{lemma}
	\label{lem:var-null-K=2}
	Suppose that $\min_i N_i \geq 2$ and $\min_i M_i \geq 2$.  Define 
	\begin{align}
		\beta_n\rp{2}  = \frac{ \max \bigg\{ 
			\sum_i N_i^2 \| \Omega_i \|^3 , \, \,
			\sum_i M_i^2 \| \Gamma_i \|^3 \, ,  \, \,
			\| \Sigma_1\|_F^2 + \| \Sigma_2\|_F^2
			\bigg\} }{  \| \mu \|^2 }. \label{eqn:beta_n2}
	\end{align}
	If $\beta_n\rp{2}  = o(1)$, then under the null hypothesis, $\var(T) = [1+o(1)] \cdot \Theta_n$. 
\end{lemma}

\subsection{The decomposition of $V$} \label{subsec:V-decompose}

\begin{lemma}
	\label{lem:Vdecompose}
	Let $\{Z_{ir}\}_{1\leq i\leq n, 1\leq r\leq N_i}$ be as in \eqref{Multinomials}. Recall that 
	\begin{align} 
		V &=  2\sum_{k=1}^K\sum_{i\in S_k}\sum_{j=1}^p  \Bigl(\frac{1}{n_k\bar{N}_k}-\frac{1}{n\bar{N}}\Bigr)^2\biggl[  \frac{N_iX_{ij}^2}{N_i-1} - \frac{N_iX_{ij}(N_i-X_{ij})}{(N_i-1)^2}\biggr]\\
		& +\frac{2}{n^2\bar{N}^2}\sum_{1\leq k\neq \ell\leq K}\sum_{i\in S_k}\sum_{m\in S_\ell} \sum_{j=1}^p X_{ij}X_{mj} +  2\sum_{k=1}^K \sum_{\substack{i\in S_k, m\in S_k,\\ i\neq m}}\sum_{j=1}^p \Bigl(\frac{1}{n_k\bar{N}_k}-\frac{1}{n\bar{N}}\Bigr)^2 X_{ij}X_{mj}. \nonumber
	\end{align}
	Define
	\begin{align*}
		\theta_i &= \big( \frac{1}{n_k \bar{N}_k} - \frac{1}{n\bar{N}} )^2 \frac{N_i^3}{N_i  -1} \quad \text{for } i \in S_k \, \,,\quad \text{and let} \, \,  \\
		\alpha_{im} 
		&= \begin{cases}
			\frac{2}{n^2\bar{N}^2} &\quad \text{ if } i \in S_k, m \in S_\ell, k \neq \ell \\
			2\big( \frac{1}{n_k \bar{N}_k} - \frac{1}{n\bar{N}} )^2 &\quad \text{ if } i, m \in S_k 
		\end{cases} 
	\end{align*}
	If we let
	\begin{align}
		A_1 &= 
		\sum_{i} \sum_{r =1}^{N_i} 
		\sum_j \big[ \frac{4\theta_i \Omega_{ij}}{N_i} 
		+ \sum_{m \in [n] \backslash \{i\} } 2 \alpha_{im} N_m \Omega_{mj} \big] Z_{ijr}, \label{eqn:A1_def}
		\\
		A_2
		&=  \sum_i \sum_{r \neq s \in [N_i]} \frac{2 \theta_i}{N_i(N_i -1)} 
		\big( \sum_j Z_{ijr} Z_{ijs} \big) \label{eqn:A2_def}  
		\\
		A_3 &= \sum_{i \neq m} \sum_{r =1 }^{N_i} \sum_{s = 1}^{N_m} 
		\alpha_{im} \big(  \sum_j Z_{ijr} Z_{mjs} \big)  \label{eqn:A3_def}, 
	\end{align}
	then these terms are mean zero, are mutually uncorrelated, and satisfy 
	\begin{align}
		\label{eqn:V_decomp}
		V = A_1  +A_2 +A_3 + \Theta_{n2} + \Theta_{n3} + \Theta_{n4}.
	\end{align}
\end{lemma}

\subsection{Properties of  $V$}

First we control the variance of $V$. 

\begin{lemma}
	\label{lem:varV}
	Let $A_1, A_2,$ and $A_3$ be defined as in Lemma \ref{lem:Vdecompose}. Then 
	\begin{align*}
		\var(A_1)
		&\lesssim \frac{1}{n \bar{N}} \| \mu \|_3^3
		+ \sum_k \frac{ \| \mu_k \|_3^3}{ n_k \bar{N}_k} 
		\lesssim \sum_k \frac{ \| \mu_k \|_3^3}{ n_k \bar{N}_k} 
		\\ 
		\var(A_2) &\lesssim \sum_k \sum_{i \in S_k} \frac{N_i^2 \| \Omega_i \|_2^2 }{n_k^4 \bar{N}_k^4} \lesssim 	\sum_k \frac{\| \mu_k \|^2}{n_k^2 \bar{N}_k^2} 
		\\
		\var(A_3) &\lesssim \sum_k \frac{ \| \mu_k \|^2 }{ n_k^2 \bar{N}_k^2} 
		+ \frac{1}{n^2 \bar{N}^2} \| \mu \|^2
		\lesssim \sum_k \frac{ \| \mu_k \|^2 }{ n_k^2 \bar{N}_k^2} . 
	\end{align*}
\end{lemma}

Next we show consistency of $V$ under the null, which is crucial in properly standardizing our test statistic and establishing asymptotic normality.
\begin{proposition}
	\label{prop:var_estimation_null}
	Recall the definition of $\beta_n$ in \eqref{eqn:beta_n}. Suppose that $\beta_n = o(1)$ and that the condition \eqref{cond1-basic} holds. If under the null hypothesis we have
	\begin{align}
		\label{eqn:null_ell2}
		K^2 \| \mu \|^4 
		&\gg  \sum_k \frac{ \| \mu \|^2}{n_k^2 \bar{N}_k^2} 
		\vee 
		\sum_k \frac{ \| \mu \|_3^3}{ n_k \bar{N}_k} ,
	\end{align}
	then $V / \var{T} \to 1$ in probability.
\end{proposition}

To later control the type II error, we must also show that $V$ does not dominate the true variance under the alternative. We first state an intermediate result that is useful throughout.

\begin{lemma}
	\label{lem:var_lbd}
	Suppose that, under either the null or alternative, $\max_i \| \Omega_i \|_\infty \leq 1 - c_0$ holds for an absolute constant $c_0>0$. Then
	\begin{align}
		\var(T) \gtrsim \Theta_{n2} + \Theta_{n3} + \Theta_{n4} .
	\end{align}
\end{lemma}

\begin{proposition}
	\label{prop:var_estimation_alt}
	Suppose that under the alternative \eqref{cond1-basic} holds and 
	\begin{align}
		\big( \sum_k \| \mu_k \|^2  \big)^2
		&\gg  \sum_k \frac{ \| \mu_k \|^2}{n_k^2 \bar{N}_k^2} 
		\vee \sum_k \frac{ \| \mu_k \|_3^3}{ n_k \bar{N}_k} . 
	\end{align}
	Then $V = O_{\pr}( \var(T))$ under the alternative.
\end{proposition}

We also require  versions of Proposition \ref{prop:var_estimation_null} and Proposition \ref{prop:var_estimation_alt} that hold under weaker conditions in the special case $K = 2$. We omit the proofs as they are  similar. Below we use the notation of Section \ref{subsec:K=2}. 

\begin{proposition}
	\label{prop:var_estimation_null_K=2}
	Suppose that $K = 2$ and recall the definition of $\beta_n\rp{2}$ in \ref{eqn:beta_n2}. Suppose that $\beta_n\rp{2} = o(1)$, $\min_i N_i \geq 2, \min_i M_i \geq 2$, and $\max_i \| \Omega_i \|_\infty \leq 1 - c_0, \max_i \| \Gamma_i \|_\infty \leq 1 - c_0$. If under the null hypothesis 
	\begin{align}
		\| \mu \|^4
		\gg	\max\Big\{ 
		\,\big(\frac{ \| \mu \|_2^2 }{ n^2 \bar{N}^2}+ \frac{ \| \mu \|_2^2 }{ m^2 \bar{M}_2^2}\big), \, 	\big(  \frac{\| \mu \|_3^3}{ n \bar{N}} + \frac{\| \mu  \|_3^3}{ m \bar{M}}   \big) \Big\} ,
	\end{align}
	then $V/\var(T) \to 1$ in probability.
\end{proposition}

%
Under the alternative we have the following. 

\begin{proposition}
	\label{prop:var_estimation_alt_K=2}
	Suppose that $K = 2$, $\min_i N_i \geq 2, \min_i M_i \geq 2$, and $\max_i \| \Omega_i \|_\infty \leq 1 - c_0, \max_i \| \Gamma_i \|_\infty \leq 1 - c_0$. If under the alternative
	\begin{align}
		\bigg \|  \frac{m \bar{M}}{ n\bar{N}+ m \bar{M}} \eta +  
		\frac{n \bar{N}}{ n\bar{N}+ m \bar{M}} \theta  \bigg \|^4
		\gg	\max\Big\{ 
		\,\big(\frac{ \| \eta \|_2^2 }{ n^2 \bar{N}^2}+ \frac{ \| \theta \|_2^2 }{ m^2 \bar{M}_2^2}\big), \, 	\big(  \frac{\| \eta \|_3^3}{ n \bar{N}} + \frac{\| \theta \|_3^3}{ m \bar{M}}   \big) \Big\} ,
	\end{align}
	then $V = O_{\mathbb{P}}(\var(T))$. 
\end{proposition}

In the setting of $K = n$ and utilize the variance estimator $V^*$. The next results capture the behavior of $V^*$ under the null and alternative. The proofs are given later in this section.

\begin{proposition}
	\label{prop:var_estimation_null_K=n}
	Define
	\begin{align}
		\label{eqn:beta_nn_def}
		\beta_n\rp{n} = \frac{ \sum_i \|\Omega_i \|^3 }{ n \| \mu \|^2 }. 
	\end{align}
	Suppose that \eqref{cond1-basic} holds, $\beta_n\rp{n} = o(1)$, and 
	\begin{align}
		n^2 \| \mu \|^4 \gg \sum_i \frac{\| \mu \|^2}{N_i^2 }
		\vee \sum_i \frac{\| \mu \|_3^3}{N_i }.
	\end{align}
	Then $V^*/\var(T) \to 1$ in probability as $n \to \infty$. 
\end{proposition}

\begin{proposition}
	\label{prop:var_estimation_alt_K=n}
	Suppose that under the alternative \eqref{cond1-basic} holds and 
	\begin{align}
		\big( \sum_i \| \Omega_i \|^2 \big)^2 \gg 
		\sum_i \frac{\| \Omega_i \|^2}{N_i^2 }
		\vee \sum_i \frac{\| \Omega_i  \|_3^3}{N_i }.
	\end{align}
	Then $V^* = O_{\mathbb{P}}(\var(T))$ under the alternative. 
\end{proposition}


\subsection{Proof of Lemma~\ref{lem:decompose}}  \label{subsec:proof-LemmaA1}
We first show that $\mathbb{E}[U_\kappa]={\bf 0}_p$ and $\mathbb{E}[U_{\kappa}U_{\zeta}']={\bf 0}_{p\times p}$ for $\kappa\neq \zeta$. Note that $\{Z_{ir}\}_{1\leq i\leq n,1\leq r\leq N_i}$ are independent mean-zero random vectors. It follows that each $U_{\kappa}$ is a mean-zero random vector. We then compute $\mathbb{E}[U_{\kappa j_1}U_{\zeta j_2}]$ for $\kappa\neq \zeta$ and all $1\leq j_1,j_2\leq p$. By direct calculations,
\[
\mathbb{E}[U_{1j}U_{2j_2}] = 2\sum_{(k,i,r,s)} \sum_{(k',i',r')} \Bigl(\frac{1}{n_k\bar{N}_k} -\frac{1}{n\bar{N}} \Bigr)(\mu_{k'j}-\mu_j) \frac{N_i}{N_i-1} \mathbb{E}[Z_{ij_2r}Z_{ij_2s}Z_{i'j_1r'}]. 
\]
If $i'\neq i$, or if $i'=i$ and $r'\notin\{r,s\}$, then $Z_{i'j_1r'}$ is independent of $Z_{ij_2r}Z_{ij_2s}$, and it follows that $\mathbb{E}[Z_{ij_2r}Z_{ij_2s}Z_{i'j_1r'}]=0$. 
If $i'=i$ and $r=r'$, then $\mathbb{E}[Z_{ij_2r}Z_{ij_2s}Z_{i'j_1r'}]=\mathbb{E}[Z_{ij_2r}Z_{ij_1r}]\cdot \mathbb{E}[Z_{ij_2s}]$; since $r\neq s$, we also have $\mathbb{E}[Z_{ij_2r}Z_{ij_2s}Z_{i'j_1r'}]=0$. This proves $\mathbb{E}[U_{1j}U_{2j_*}]=0$. Since this holds for all $1\leq j_1,j_2\leq p$, we immediately have
\[
\mathbb{E}[U_1U_2']={\bf 0}_{p\times p}. 
\]
We can similarly show that $\mathbb{E}[U_\kappa U_\zeta']={\bf 0}_{p\times p}$, for other $\kappa\neq \zeta$. The proof is omitted. 

It remains to prove the desirable decomposition of $T$. Recall that $T=\sum_{j=1}^p T_j$. Write $\rho^2=\sum_{j=1}^p \rho_j^2$, where $\rho_j^2=2\sum_{k=1}^Kn_k\bar{N}_k(\mu_{kj}-\mu_j)^2$. It suffices to show that 
\beq \label{lem-decompose-goal}
T_j = \rho_j^2+U_{1j}+U_{2j}+U_{3j}+U_{4j}, \qquad \mbox{for all } 1\leq j\leq p. 
\eeq

To prove \eqref{lem-decompose-goal}, we need some preparation. Define
\beq \label{lem-decompose-0}
Y_{ij} :=\frac{X_{ij}}{N_i} - \Omega_{ij}= \frac{1}{N_i}\sum_{r=1}^{N_i}Z_{ijr}, \qquad  Q_{ij}: = Y_{ij}^2-\mathbb{E}Y^2_{ij}=Y_{ij}^2-\frac{\Omega_{ij}(1-\Omega_{ij})}{N_i}. 
\eeq
With these notations, $X_{ij}=N_i(\Omega_{ij}+Y_{ij})$ and $N_iY_{ij}^2=N_iQ_{ij}+\Omega_{ij}(1-\Omega_{ij})$. 
Moreover, we can use  \eqref{lem-decompose-0} to re-write $Q_{ij}$ as a function of $\{Z_{ijr}\}_{1\leq r\leq N_i}$ as follows: 
\[
Q_{ij}  = \frac{1}{N^2_i}\sum_{r=1}^{N_i}[Z_{ijr}^2-\Omega_{ij}(1-\Omega_{ij})]+\frac{1}{N^2_i}\sum_{1\leq r\neq s\leq N_i}Z_{ijr}Z_{ijs}.
\]
Note that $Z_{ijr}=B_{ijr}-\Omega_{ij}$, where $B_{ijr}$ can only take values in $\{0,1\}$. Hence, $(Z_{ijr}+\Omega_{ij})^2 = (Z_{ijr}+\Omega_{ij})$ always holds. Re-arranging the terms gives $Z^2_{ijr}-\Omega_{ij}(1-\Omega_{ij})=(1-2\Omega_{ij})Z_{ijr}$. It follows that
\beq \label{lem-decompose-00}
Q_{ij} = (1-2\Omega_{ij}) \frac{Y_{ij}}{N_i} + \frac{1}{N^2_i}\sum_{1\leq r\neq s\leq N_i}Z_{ijr}Z_{ijs}. 
\eeq
This is a useful equality which we will use in the proof below.

We now show \eqref{lem-decompose-goal}. 
Fix $j$ and write $T_j=R_j -D_j$, where
\[
R_j = \sum_{k=1}^K n_k\bar{N}_k(\hat{\mu}_{kj}-\hat{\mu}_{j})^2, \quad\mbox{and}\quad  D_j= \sum_{k=1}^K\sum_{i\in S_k} \xi_k\frac{X_{ij}(N_i-X_{ij})}{n_k\bar{N}_k(N_i-1)}, \quad \mbox{with }\xi_k = 1-\frac{n_k\bar{N}_k}{n\bar{N}}
\]
First, we study $D_j$. Note that $X_{ij}(N_{ij}-X_{ij})=N^2_i(\Omega_{ij}+Y_{ij})(1-\Omega_{ij}-Y_{ij})=N^2_i\Omega_{ij}(1-\Omega_{ij})-N^2_iY_{ij}^2+N_i^2(1-2\Omega_{ij})Y_{ij}$, where $Y_{ij}^2=Q_{ij}+N_i^{-1}\Omega_{ij}(1-\Omega_{ij})$. It follows that 
\[
\frac{X_{ij}(N_{ij}-X_{ij})}{N_i(N_i-1)} =\Omega_{ij}(1-\Omega_{ij})-\frac{N_iQ_{ij}}{N_i-1}+\frac{N_i}{N_i-1} (1-2\Omega_{ij})Y_{ij}.
\]
We apply \eqref{lem-decompose-00} to get 
\beq \label{lem-decompose-1}
\frac{X_{ij}(N_{ij}-X_{ij})}{N_i(N_i-1)} =\Omega_{ij}(1-\Omega_{ij}) +(1-2\Omega_{ij})Y_{ij} - \frac{1}{N_i(N_i-1)}\sum_{1\leq r\neq s\leq N_i}Z_{ijr}Z_{ijs}.
\eeq 
It follows that 
\begin{align} \label{lem-decompose-2}
	D_j &=\sum_{k=1}^K \sum_{i\in S_k} \frac{\xi_kN_i}{n_k\bar{N}_k} \Omega_{ij}(1-\Omega_{ij}) +  \sum_{k=1}^K \sum_{i\in S_k} \frac{\xi_kN_i}{n_k\bar{N}_k} (1-2\Omega_{ij})Y_{ij}\cr
	&\qquad - \sum_{k=1}^K \sum_{i\in S_k} \frac{\xi_k}{n_k\bar{N}_k(N_i-1)}\sum_{1\leq r\neq s\leq N_i}Z_{ijr}Z_{ijs}. 
\end{align}

Next, we study $R_j$. Note that $n_k\bar{N}_k(\hat{\mu}_{kj}-\hat{\mu}_j)=\sum_{i\in S_k}(X_{ij}-\bar{N}_k\hat{\mu}_j)$. It follows that
\[
R_j = \sum_{k=1}^K\frac{1}{n_k\bar{N}_k} \biggl[ \sum_{i\in S_k}(X_{ij}-\bar{N}_k\hat{\mu}_j)\biggr]^2. 
\]
Recall that $X_{ij}=N_i(\Omega_{ij}+Y_{ij})$. By direct calculations,  $\sum_{i\in S_k}X_{ij}=n_k\bar{N}_k\mu_{kj}+\sum_{i\in S_k}N_iY_{ij}$, and $\hat{\mu}_j=\mu_j +(n\bar{N})^{-1}\sum_{m=1}^nN_mY_{mj}$. We then have the following decomposition: 
\begin{align*}
	\sum_{i\in S_k}(X_{ij}-\bar{N}_k\hat{\mu}_j) &= n_k\bar{N}_k(\mu_{kj}-\mu_j) + \sum_{i\in S_k}N_iY_{ij} -\frac{n_k\bar{N}_k}{n\bar{N}}\Bigl(\sum_{m=1}^{n} N_mY_{mj}\Bigr). 
\end{align*}
Using this decomposition, we can expand $[\sum_{i\in S_k}(X_{ij}-\bar{N}_k\hat{\mu}_j)]^2$ to a total of 6 terms, where 3 are quadratic terms and 3 are cross terms. It yields a decomposition of $R_j$ into 6 terms: 
\begin{align} \label{lem-decompose-3}
	R_j & = \sum_{k=1}^Kn_k\bar{N}_k(\mu_{kj}-\mu_j)^2 +    \sum_{k=1}^K  \frac{1}{n_k\bar{N}_k} \Bigl(\sum_{i\in S_k} N_iY_{ij}\Bigr)^2 + \sum_{k=1}^K \frac{n_k\bar{N}_k}{n^2\bar{N}^2}\Bigl(\sum_{m=1}^nN_mY_{mj}\Bigr)^2\cr
	&\qquad +2\sum_{k=1}^K (\mu_{kj}-\mu_j)\Bigl(\sum_{i\in S_k}N_iY_{ij}\Bigr) -2\sum_{k=1}^K\frac{n_k\bar{N}_k}{n\bar{N}} (\mu_{kj}-\mu_j)\Bigl(\sum_{m=1}^nN_mY_{mj}\Bigr)\cr
	&\qquad - \frac{2}{n\bar{N}}\sum_{k=1}^K \Bigl(\sum_{i\in S_k} N_i Y_{ij}\Bigr)\Bigl(\sum_{m=1}^nN_mY_{mj} \Bigr)\cr
	&\equiv I_1+I_2+I_3+I_4+I_5+I_6. 
\end{align}
By definition, $\sum_{k=1}^Kn_k\bar{N}_k=n\bar{N}$ and $\sum_{k=1}^K n_k\bar{N}_k\mu_{kj}=n\bar{N}\mu_j$. It follows that
\[
I_3 = \frac{1}{n\bar{N}}\Bigl(\sum_{m=1}^nN_mY_{mj}\Bigr)^2, \qquad I_5=0, \qquad I_6=-\frac{2}{n\bar{N}}\Bigl(\sum_{m=1}^nN_mY_{mj}\Bigr)^2=-2I_3. 
\] 
It follows that
\beq \label{dem-decompose-3(2)}
R_j = I_1 +I_2-I_3+I_4. 
\eeq
We further simplify $I_3$. Recall that $\xi_k=1-(n\bar{N})^{-1}n_k\bar{N}_k$. 
By direct calculations, 
\begin{align} \label{lem-decompose-4}
	I_3 &=\frac{1}{n\bar{N}}\Bigl(\sum_{m=1}^n N_mY_{mj}\Bigr)^2=\frac{1}{n\bar{N}}\biggl[\sum_{k=1}^K\Bigl(\sum_{i\in S_k} N_iY_{ij}\Bigr)\biggr]^2\cr
	&= \frac{1}{n\bar{N}} \sum_{k=1}^K \Bigl(\sum_{i\in S_k} N_iY_{ij}\Bigr)^2 + \frac{1}{n\bar{N}} \sum_{1\leq k\neq\ell\leq K}\Bigl(\sum_{i\in S_k} N_iY_{ij}\Bigr)\Bigl(\sum_{m\in S_\ell} N_mY_{mj}\Bigr)\cr
	&= \sum_{k=1}^K (1-\xi_k)\frac{1}{n_k\bar{N}_k} \Bigl(\sum_{i\in S_k}N_iY_{ij}\Bigr)^2 + \underbrace{\frac{1}{n\bar{N}}\sum_{k\neq\ell} \sum_{i\in S_k}\sum_{m \in S_\ell} N_iN_mY_{ij}Y_{mj}}_{J_1}\cr
	&= I_2 - \sum_{k=1}^K\sum_{i\in S_k} \frac{\xi_k}{n_k\bar{N}_k}\Bigl(\sum_{i\in S_k}N_iY_{ij}\Bigr)^2+ J_1\cr
	&= I_2 +J_1 - \sum_{k=1}^K\frac{\xi_k}{n_k\bar{N}_k}\Bigl(\sum_{i\in S_k}N_i^2Y_{ij}^2\Bigr) - \underbrace{\sum_{k=1}^K \frac{\xi_k}{n_k\bar{N}_k} \sum_{\substack{i\in S_k,m\in S_k\\i\neq m}}N_iN_mY_{ij}Y_{mj}}_{J_2}.   
\end{align}
By \eqref{lem-decompose-0},  $N_iY_{ij}^2 = N_iQ_i+\Omega_{ij}(1-\Omega_{ij}) $. We further apply \eqref{lem-decompose-00} to get 
\[
N_i^2Y_{ij}^2 =  N_i(1-2\Omega_{ij})Y_{ij} + \sum_{1\leq r\neq s\leq N_i}Z_{ijr}Z_{ijs}+N_i\Omega_{ij}(1-\Omega_{ij}). 
\]
It follows that
\begin{align} \label{lem-decompose-5}
	\sum_{k=1}^K\frac{\xi_k}{n_k\bar{N}_k}&\Bigl(\sum_{i\in S_k}N_i^2Y_{ij}^2\Bigr)= 
	\underbrace{\sum_{k=1}^K\sum_{i\in S_k}\frac{\xi_kN_i}{n_k\bar{N}_k}(1-2\Omega_{ij})Y_{ij}}_{J_3} \cr
	& + \underbrace{\sum_{k=1}^K\sum_{i\in S_k}\frac{\xi_k}{n_k\bar{N}_k}\sum_{r\neq s}Z_{ijr}Z_{ijs}}_{J_4} + \underbrace{\sum_{k=1}^K\sum_{i\in S_k} \frac{\xi_kN_i}{n_k\bar{N}_k} \Omega_{ij}(1-\Omega_{ij})}_{J_5}. 
\end{align}
We plug \eqref{lem-decompose-5} into \eqref{lem-decompose-4} to get $I_3=I_2+J_1-J_2-J_3-J_4-J_5$. Further plugging $I_3$ into the expression of $R_j$ in \eqref{dem-decompose-3(2)}, we have  
\begin{align} \label{lem-decompose-6}
	R_j &= I_1+I_4-J_1+J_2+J_3+J_4+J_5,
\end{align}
where $I_1$ and $I_4$ are defined in \eqref{lem-decompose-3}, $J_1$-$J_2$ are defined in \eqref{lem-decompose-4}, and $J_3$-$J_5$ are defined in \eqref{lem-decompose-5}. 

Finally, we combine the expressions of $D_j$ and $R_j$. By \eqref{lem-decompose-2} and the definitions of $J_1$-$J_5$, 
\begin{align*}
	D_j &= J_5 +J_3 -  \sum_{k=1}^K \sum_{i\in S_k} \frac{\xi_k}{n_k\bar{N}_k(N_i-1)}\sum_{r\neq s}Z_{ijr}Z_{ijs}\cr
	&= J_5 + J_3 + J_4 - \underbrace{\sum_{k=1}^K \sum_{i\in S_k} \frac{\xi_k N_i}{n_k\bar{N}_k(N_i-1)}\sum_{r\neq s}Z_{ijr}Z_{ijs}}_{J_6}. 
\end{align*}
Combining it with \eqref{lem-decompose-6} gives $T_j = R_j-D_j = I_1+I_4-J_1+J_2+J_6$. We further plug in the definition of each term. It follows that
\begin{align}\label{lem-decompose-7}
	T_j 
	&= \sum_{k=1}^Kn_k\bar{N}_k(\mu_{kj}-\mu_j)^2 +2\sum_{k=1}^K\sum_{i\in S_k} (\mu_{kj}-\mu_j)N_iY_{ij} - \frac{1}{n\bar{N}}\sum_{k\neq\ell} \sum_{i\in S_k, m \in S_\ell} N_iN_mY_{ij}Y_{mj}\cr
	& \qquad + \sum_{k=1}^K \sum_{\substack{i\in S_k,m\in S_k\\i\neq m}}\frac{\xi_k}{n_k\bar{N}_k}  N_iN_mY_{ij}Y_{mj}+\sum_{k=1}^K \sum_{i\in S_k} \frac{\xi_k N_i}{n_k\bar{N}_k(N_i-1)}\sum_{r\neq s}Z_{ijr}Z_{ijs}. 
\end{align}
We plug in $Y_{ij}=N_i^{-1}\sum_{r=1}^{N_i}Z_{ijr}$ and take a sum of $1\leq j\leq p$. It gives \eqref{lem-decompose-goal} immediately. The proof is now complete. 
\qed

\subsection{Proof of Lemma~\ref{lem:var1}} 
Recall that $\{Z_{ir}\}_{1\leq i\leq n, 1\leq r\leq N_i}$ are independent random vectors. Write 
\[
{\bf 1}_p'U_1 = 2 \sum_{k=1}^K \sum_{i\in S_k}\sum_{r=1}^{N_i} (\mu_k-\mu)'Z_{ir}.
\]
The covariance matrix of $Z_{ir}$ is $\diag(\Omega_i)-\Omega_i\Omega_i'$. It follows that  
\begin{align} \label{lem-var-1}
	\mathrm{Var}({\bf 1}_p'U_1) &= 4 \sum_{k=1}^K\sum_{i\in S_k} \sum_{r=1}^{N_i} (\mu_k-\mu)'\bigl[\diag(\Omega_i)-\Omega_i\Omega_i'\bigr](\mu_k-\mu)\cr
	&= 4 \sum_{k}(\mu_k-\mu)'\Bigl[\diag\Bigl(\sum_{i\in S_k}N_i\Omega_i\Bigr)-\Bigl(\sum_{i\in S_k} N_i\Omega_i\Omega_i'\Bigr)\Bigr](\mu_k-\mu)\cr
	&=4\sum_k   (\mu_k-\mu)'\Bigl[\diag(n_k\bar{N}_k\mu_k)-n_k\bar{N}_k\Sigma_k\Bigr](\mu_k-\mu)\cr
	&=4\sum_k n_k\bar{N}_k\bigl\|\diag(\mu_k)^{1/2}(\mu_k-\mu)\bigr\|^2 - 4\sum_kn_k\bar{N}_k \bigl\|\Sigma_k^{1/2}(\mu_k-\mu)\bigr\|^2. 
\end{align}
This proves the first claim. Furthermore, by \eqref{lem-var-1},
\[
\mathrm{Var}({\bf 1}_p'U_1)\leq 4\sum_k n_k\bar{N}_k\bigl\|\diag(\mu_k)^{1/2}(\mu_k-\mu)\bigr\|^2\leq 4\sum_k n_k\bar{N}_k\|\diag(\mu_k)\|\|\mu_k-\mu\|^2. 
\]
Note that $\|\diag(\mu_k)\|=\|\mu_k\|_\infty$. Therefore, if $\max_{k}\|\mu_k\|_\infty=o(1)$, the right hand side above is $o(1)\cdot 4\sum_kn_k\bar{N}_k \|\mu_k-\mu\|^2=o(\rho^2)$. This proves the second claim. \qed

\subsection{Proof of Lemma~\ref{lem:var2}} 

For each $1\leq k\leq K$, define a set of index triplets: ${\cal M}_k=\{(i,r,s): i\in S_k, 1\leq r< s\leq N_i\}$. Let ${\cal M}=\cup_{k=1}^K {\cal M}_k$. Write for short $\theta_i=(\frac{1}{n_k\bar{N}_k}-\frac{1}{n\bar{N}})^2\frac{N_i^3}{N_i-1}$, for $i\in S_k$. It is seen that
\[
{\bf 1}_p'U_2 = 2\sum_{(i,r,s)\in {\cal M}} \frac{\sqrt{\theta_i}}{\sqrt{N_i(N_i-1)}} W_{irs}, \qquad\mbox{with} \quad W_{irs} = \sum_{j=1}^p Z_{ijr}Z_{ijs}. 
\]
For $W_{irs}$ and $W_{i'r's'}$, if $i\neq i'$, or if $i=i'$ and $\{r,s\}\cap \{r',s'\}=\empty\emptyset$, then these two variables are independent; if $i=i'$, $r=r'$ and $s\neq s'$, then $\mathbb{E}[W_{irs}W_{irs'}] = \sum_{j, j'} \mathbb{E}[Z_{ijr}Z_{ijs}Z_{ij'r}Z_{ij's'}]=\sum_{j,j'}\mathbb{E}[Z_{ijr}Z_{ij'r}]\cdot\mathbb{E}[Z_{ijs}]\cdot\mathbb{E}[Z_{ij's'}]=0$. Therefore, $\{W_{irs}\}_{(i,r,s)\in {\cal M}}$ is a collection of mutually uncorrelated variables. It follows that 
\[
\mathrm{Var}({\bf 1}_p'U_2) = 4\sum_{(i,r,s)\in {\cal M}} \frac{\theta_i}{N_i(N_i-1)}\mathrm{Var}(W_{irs}). 
\]
It remains to calculate the variance of each $W_{irs}$. By direction calculations, 
\begin{align} \label{lem-var-3-add}
	\mathrm{Var}(W_{irs}) &= \sum_j \mathbb{E}[Z_{ijr}^2Z_{ijs}^2] + 2\sum_{ j<\ell} \mathbb{E}[Z_{ijr}Z_{ijs}Z_{i\ell r}Z_{i\ell s}] \cr
	&= \sum_j [\Omega_{ij}(1-\Omega_{ij})]^2 + 2\sum_{ j< \ell }(-\Omega_{ij}\Omega_{i\ell})^2\cr
	&=  \sum_j\Omega_{ij}^2-2\sum_j \Omega^3_{ij}+ \Bigl(\sum_j\Omega^2_{ij}\Bigr)^2\cr
	&=\|\Omega_i\|^2-2\|\Omega_i\|_3^3 + \|\Omega_i\|^4\cr
\end{align}
Since $\max_{ij} \Omega_{ij} \leq 1$, we have
\begin{align*}
	\| \Omega_i \|^2 - \|\Omega_i\|_3^3	\leq 	\var(W_{irs})  \leq \|\Omega_i\|^2. 
\end{align*}
Therefore, 
\begin{align*}
	\mathrm{Var}({\bf 1}_p'U_2) &= 4\sum_{k=1}^K\sum_{i\in S_k} \sum_{1\leq r<s\leq N_i}\frac{\theta_i}{N_i(N_i-1)}\var(W_{irs})
	\\&= 2 \sum_{k=1}^K \sum_{i\in S_k} \theta_i \var(W_{irs}) 
	\geq 2 \sum_{k=1}^K \sum_{i\in S_k} \theta_i \big[ \| \Omega_i \|^2 - \|\Omega_i\|_3^3 \big] = \Theta_{n2} - A_n,
\end{align*}
and similarly $	\mathrm{Var}({\bf 1}_p'U_2) \leq \Theta_{n2}$, which proves the first claim. To prove the second claim, note that $\mathrm{Var}({\bf 1}_p'U_2) = \Theta_{n2} + O(A_n)$. By \eqref{eqn:var_condition_K=n} and the assumption $\min N_i \geq 2$, we have
\begin{align*}
	A_n &\lesssim \sum_{k} \big( \frac{1}{n_k \bar{N}_k} - \frac{1}{n \bar{N}} \big)^2 \sum_{i \in S_k} N_i^2 \| \Omega_i \|_3^3 
	\\&=  \sum_{k} \big( \frac{1}{n_k \bar{N}_k} - \frac{1}{n \bar{N}} \big)^2 \cdot o\bigg(  \sum_{i \in S_k}  N_i^2 \| \Omega_i \|^2 \bigg)
	= o( \Theta_{n2}),
\end{align*}
which implies that $\var( {\bf 1 }_p U_2) = [1+o(1)] \Theta_{n2}$, as desired. 

\qed

\subsection{Proof of Lemma~\ref{lem:var3}} 

For each $1\leq k<\ell\leq K$, define a set of index quadruples: ${\cal J}_{k\ell}=\{(i,r,m,s): i\in S_k, j\in S_{\ell}, 1\leq r\leq N_i, 1\leq s\leq N_m\}$. Let ${\cal J}=\cup_{(k,\ell): 1\leq k<\ell\leq K}{\cal J}_{k\ell}$. It is seen that 
\[
{\bf 1}_p'U_3 =  -  \frac{2}{n\bar{N}}\sum_{(i,r, m, s)\in {\cal J}} V_{irms}, \qquad\mbox{where}\;\; V_{irms} = \sum_{j=1}^p Z_{ijr}Z_{mjs}. 
\]
For $V_{irms}$ and $V_{i'r'm's'}$, if $\{(i,r), (m,s)\}\cap \{(i',r'), (m',s')\}=\emptyset$, then the two variables are independent of each other. If $(i,r)=(i',r')$ and $(m,s)\neq (m',s')$,  then 
$\mathbb{E}[V_{irms}V_{irm's'}]=\sum_{j,j'}\mathbb{E}[Z_{ijr}Z_{mjs}Z_{ij'r}Z_{m'j's'}] =\sum_{j,j'}\mathbb{E}[Z_{ijr}Z_{ij'r}]\cdot\mathbb{E}[Z_{mjs}]\cdot \mathbb{E}[Z_{m'js'}]=0$. Therefore, the only correlated case is when $(i,r, m, s)=(i',r', m', s')$. This implies that $\{V_{irms}\}_{(i,r,m,s)\in {\cal J}}$ is a collection of mutually uncorrelated variables. 
Therefore, 
\[
\mathrm{Var}({\bf 1}_p'U_3) = \frac{4}{n^2\bar{N}^2}\sum_{(i,r,m,s)\in {\cal J}} \mathrm{Var}(V_{irms}). 
\]
Note that $\mathrm{Var}(V_{irms})=\mathbb{E}[(\sum_j Z_{ijr}Z_{mjs})^2]=\sum_{j,j'}\mathbb{E}[Z_{ijr}Z_{mjs}Z_{ij' r}Z_{mj' s}] $; also, the covariance matrix of $Z_{ir}$ is $\diag(\Omega_i)-\Omega_i\Omega_i'$. It follows that 
\begin{align} \label{Var-of-V(irms)}
	\mathrm{Var}(V_{irms}) 
	&=  \sum_{j} \mathbb{E}[Z^2_{ijr}] \cdot\mathbb{E}[Z^2_{mjs}] +   \sum_{j\neq j'} \mathbb{E}[Z_{ijr}Z_{ij' r}] \cdot\mathbb{E}[Z_{mjs}Z_{mj' s}]\cr
	&= \sum_j \Omega_{ij}(1-\Omega_{ij}) \Omega_{mj}(1-\Omega_{mj})+ \sum_{j\neq j' }\Omega_{ij}\Omega_{ij'}\Omega_{mj}\Omega_{mj'}\cr
	&=\sum_j\Omega_{ij}\Omega_{mj} - 2\sum_{j}\Omega^2_{ij}\Omega^2_{mj} + \sum_{j,j'}\Omega_{ij}\Omega_{ij'}\Omega_{mj}\Omega_{mj'}.
\end{align}
Write for short $\delta_{im} = - 2\sum_{j}\Omega^2_{ij}\Omega^2_{mj} + \sum_{j,j'}\Omega_{ij}\Omega_{ij'}\Omega_{mj}\Omega_{mj'}$.Combining the above gives
\begin{align} \label{lem-var-2}
	\mathrm{Var}&({\bf 1}_p'U_3)  =  \frac{4}{n^2\bar{N}^2}\sum_{k<\ell}\sum_{i\in S_k}\sum_{m\in S_{\ell}}\sum_{r=1}^{N_i}\sum_{s=1}^{N_m}\Bigl(\sum_j \Omega_{ij}\Omega_{mj} +\delta_{im} \Bigr)\cr
	&= \frac{2}{n^2\bar{N}^2}\sum_{ k\neq \ell }\sum_{i\in S_k}\sum_{m\in S_\ell} \sum_j N_iN_m\Omega_{ij}\Omega_{mj} +  \frac{2}{n^2\bar{N}^2}\sum_{k\neq\ell}\sum_{i\in S_k}\sum_{m\in S_{\ell}}N_iN_m\delta_{im} .
\end{align}
It is easy to see that $|\delta_{im}|\leq \sum_{j,j'}\Omega_{ij}\Omega_{ij'}\Omega_{mj}\Omega_{mj'}$. Also, by the definition of $\Sigma_k$ in \eqref{relaxedMod2}, we have $\Sigma_k(j,j')=\frac{1}{n_k\bar{N}_k}\sum_{i\in S_k}N_i\Omega_{ij}\Omega_{ij'}$. 
Using these results, we immediately have
\begin{align} \label{lem-var-2(add)}
	\Bigl|\frac{2}{n^2\bar{N}^2}\sum_{k\neq\ell}\sum_{i\in S_k}\sum_{m\in S_{\ell}}N_iN_m\delta_{im}\Bigr|&\leq \frac{2}{n^2\bar{N}^2} \sum_{k\neq \ell }\sum_{i\in S_k}\sum_{m\in S_{\ell}}\sum_{j,j'}N_iN_m\Omega_{ij}\Omega_{ij'}\Omega_{mj}\Omega_{mj'}\cr
	&=   \frac{2}{n^2\bar{N}^2}\sum_{j,j'}\sum_{k\neq\ell}\Bigl(\sum_{i\in S_k}N_i\Omega_{ij}\Omega_{ij'}\Bigr)\Bigl(\sum_{m\in S_\ell}N_i\Omega_{mj}\Omega_{mj'}\Bigr)\cr
	&=  \frac{2}{n^2\bar{N}^2}\sum_{j,j'}\sum_{k\neq\ell}n_k\bar{N}_k \Sigma_k(j,j')\cdot n_{\ell}\bar{N}_{\ell}\Sigma_\ell(j,j')\cr
	&= 2\sum_{k\neq \ell}\frac{n_kn_{\ell}\bar{N}_k\bar{N}_\ell}{n^2\bar{N}^2}{\bf 1}_p'(\Sigma_k\circ\Sigma_\ell){\bf 1}_p =: B_n
\end{align}
as desired. 

\qed

\subsection{Proof of Lemma~\ref{lem:var4}}

For $1\leq k\leq K$, define a set of index quadruples: ${\cal Q}_k =\{(i,r,m,s): i\in S_k, m\in S_k, i<m, 1\leq r\leq N_i, 1\leq s\leq N_m\}$. Let ${\cal Q}=\cup_{k=1}^K {\cal Q}_k$. Write $\kappa_{im}=(\frac{1}{n_k\bar{N}_k}-\frac{1}{n\bar{N}})^2N_iN_m$, for $i\in S_k$ and $m\in S_k$. It is seen that
\[
{\bf 1}_p'U_4 = 2\sum_{(i, r, m, s)\in {\cal Q}} \frac{\sqrt{\kappa_{im}}}{\sqrt{N_iN_m}}V_{irms}, \qquad\mbox{where}\quad V_{irms} = \sum_{j=1}^p Z_{ijr}Z_{mjs}. 
\]
It is not hard to see that $V_{irms}$ and $V_{i'r'm's'}$ are correlated only if $(i,r,m,s)=(i',r',m',s')$. It follows that
\[
\mathrm{Var}({\bf 1}_p'U_4) = 4\sum_{(i,r,m,s)\in {\cal Q}} \frac{\kappa_{im}}{N_iN_m}\mathrm{Var}(V_{irms}). 
\]
In the proof of Lemma~\ref{lem:var3}, we have studied $\mathrm{Var}(V_{irms})$. In particular, by \eqref{Var-of-V(irms)}, we have
\[
\mathrm{Var}(V_{irms}) =\sum_j \Omega_{ij}\Omega_{mj}+\delta_{im}, \qquad\mbox{with}\quad |\delta_{im}|\leq  \sum_{j,j'}\Omega_{ij}\Omega_{i j'}\Omega_{mj}\Omega_{m j'}.  
\]
Thus
\begin{align*}  
	\mathrm{Var}({\bf 1}_p'U_4) &= 4\sum_{k=1}^K\sum_{\substack{i\in S_k, m\in S_k\\i< m}}\sum_{i=1}^{N_i}\sum_{r=1}^{N_m} \frac{\kappa_{im}}{N_iN_m} \mathrm{Var}(V_{irms})\cr
	&=4\sum_{k=1}^K\sum_{\substack{i\in S_k, m\in S_k\\i< m}}\kappa_{im}\Bigl(\sum_{j}\Omega_{ij}\Omega_{mj} +\delta_{im}\Bigr) \cr
	&= 2\sum_{k=1}^K \sum_{\substack{i\in S_k,m\in S_k\\i\neq m}}\sum_j \kappa_{im}\Omega_{ij}\Omega_{mj} \pm 2 \sum_k\sum_{i \neq  m\in S_k} \kappa_{im} \sum_{j,j'}\Omega_{ij}\Omega_{i j'}\Omega_{mj}\Omega_{m j'} , \cr
	&= \Theta_{n3} \pm E_n.  \num \label{lem-var-4(1)}
\end{align*}
which proves the lemma. 

\qed

\subsection{Proof of Lemma~\ref{lem:Theta_n2+n3+n4}}

By assumption \eqref{cond1-basic}, $N_i^3/(N_i-1) \asymp N_i$ and $\Bigl(\frac{1}{n_k\bar{N}_k}-\frac{1}{n\bar{N}}\Bigr)^2 \asymp \frac{1}{n_k^2 \bar{N}_k^2}$. First, observe that
\begin{align*}
	\Theta_{n2}+\Theta_{n4}
	&= 2\sum_{k=1}^K \Bigl(\frac{1}{n_k\bar{N}_k}-\frac{1}{n\bar{N}}\Bigr)^2\sum_{i\in S_k}\frac{N_i^3}{N_i-1} \|\Omega_i\|^2 
	\\&\quad + 2\sum_{k=1}^K \sum_{\substack{i\in S_k,m\in S_k\\i\neq m}}\sum_j \Bigl(\frac{1}{n_k\bar{N}_k}-\frac{1}{n\bar{N}}\Bigr)^2 N_i N_m \Omega_{ij}\Omega_{mj}
	\\ &\asymp \sum_{k =1} \Bigl(\frac{1}{n_k\bar{N}_k}\Bigr)^2 \sum_j
	\sum_{i, m \in S_k}  N_i \Omega_{ij} \cdot N_m \Omega_{ij} 
	= \sum_k \| \mu_k \|^2. 
	\num \label{eqn:Theta_n2+n4}
\end{align*}
%
Recall the definitions of $\mu_k$ and $\mu$ in \eqref{relaxedMod2}-\eqref{relaxedMod3}. By direct calculations, we have
\begin{align*}
	\Theta_{n3} &=2  \sum_j \sum_{ k\neq \ell }\Bigl( \frac{1}{n\bar{N}}\sum_{i\in S_k}N_i\Omega_{ij}\Bigr)\Bigl(\frac{1}{n\bar{N}}\sum_{m\in S_\ell} N_m\Omega_{mj}\Bigr)\cr
	&=2 \sum_j\sum_{k\neq \ell}\frac{n_k\bar{N}_k}{n\bar{N}}\mu_{kj}\cdot \frac{n_\ell \bar{N}_\ell}{n\bar{N}} \mu_{\ell j}\cr
	&= 2 \sum_{k \neq \ell} \frac{n_kn_{\ell}\bar{N}_k\bar{N}_\ell}{n^2\bar{N}^2} \cdot  \mu_k^{\, \, \prime} \, \mu_\ell  \\
	&\leq 2 \sum_j\Bigl(\sum_k\frac{n_k\bar{N}_k}{n\bar{N}}\mu_{kj}\Bigr)^2 =2\sum_j\mu_j^2=2\|\mu\|^2. 
	\num \label{eqn:Thetan3_bd}
\end{align*}
By Cauchy--Schwarz, 
\begin{align*}
	\| \mu \|^2 &= \sum_j \bigg( \sum_k (\frac{n_k \bar{N_k}}{n \bar{N}}) \mu_{kj} \bigg)^2
	\\&\leq \sum_j \bigg( \sum_k (\frac{n_k \bar{N_k}}{n \bar{N}})^2 \bigg) \cdot 
	\bigg( \sum_k \mu_{kj}^2 \bigg) 
	\\& \leq \sum_j \bigg( \sum_k (\frac{n_k \bar{N_k}}{n \bar{N}}) \bigg) \cdot 
	\bigg( \sum_k \mu_{kj}^2 \bigg) 
	= \sum_j \sum_k \mu_{kj}^2 = \sum_k \| \mu_k \|^2.
	\num \label{eqn:mu_vs_muk}
\end{align*}
Combining \eqref{eqn:Theta_n2+n4}, \eqref{eqn:Thetan3_bd}, and  \eqref{eqn:mu_vs_muk} yields
\begin{align*}
	c \big(\sum_k \| \mu_k \|^2  \big) \leq 	\Theta_{n2} + \Theta_{n3} + \Theta_{n4} \leq C\big( \sum_k \| \mu_k \|^2 \big),
\end{align*}
for absolute constants $c, C>0$. This completes the proof. 
\qed

\subsection{Proof of Lemma~\ref{lem:var-null}}

By \eqref{cond1-basic}, it holds that 
\begin{align}
	\label{eqn:basic1}
	(\frac{1}{n_k\bar{N}_k}-\frac{1}{n\bar{N}})^2 
	\asymp \frac{1}{(n_k \bar{N}_k)^2}, 
\end{align}
and moreover, for all $i \in \{1, 2, \ldots, n\}$, 
\begin{align}
	\label{eqn:basic2}
	\frac{N_i^3}{N_i - 1} \asymp N_i^2. 
\end{align}
Recall the definitions of $A_n,$ $B_n$, and $E_n$ in \eqref{eqn:An}, \eqref{eqn:Bn},  and  \eqref{eqn:En}, respectively. Note that these are the remainder terms in Lemmas \ref{lem:var2}, \ref{lem:var3}, and \ref{lem:var4}, respectively. Under the null hypothesis (recall $\Theta_{n1} \equiv 0$ under the null), 
\begin{align}
	\label{eqn:var_null_main_decomp}
	\var(T) = \Theta_{n2} + \Theta_{n3} + \Theta_{n4}
	+ O(A_n + B_n + E_n). 
\end{align}
It holds that 
\begin{align}
	\label{eqn:An_bd}
	A_n \leq \sum_{k=1}^K \Bigl(\frac{1}{n_k\bar{N}_k}\Bigr)^2\sum_{i\in S_k}N_i^2 \|\Omega_{i}\|_3^3. 
\end{align}
Next,  by linearity and the definition of $\Sigma_k, \Sigma$ in \eqref{relaxedMod2}, \eqref{relaxedMod3}, respectively, 
\begin{align*} 
	B_n &\leq  2\sum_{k, \ell}\frac{n_kn_{\ell}\bar{N}_k\bar{N}_\ell}{n^2\bar{N}^2}{\bf 1}_p'(\Sigma_k\circ\Sigma_\ell){\bf 1}_p 
	\\& \leq 2 {\bf 1}_p' \bigg( \frac{1}{n \bar{N}} \sum_k n_k \bar{N}_k \Sigma_k \bigg) \circ \bigg( \frac{1}{n \bar{N}} \sum_\ell n_\ell \bar{N}_\ell  \Sigma_k\ell   \bigg) 
	{\bf 1}_p 
	\\ &= 2 	{\bf 1}_p ' (\Sigma \circ \Sigma) 	{\bf 1}_p = 2\| \Sigma \|_F^2 
\end{align*} 
By Cauchy--Schwarz,
\begin{align*}
	B_n &\leq \| \Sigma \|_F^2
	= \sum_{j, j'} \bigg( \sum_k (\frac{n_k \bar{N}_k}{n\bar{N}} \Sigma_k(j,j') \bigg)^2 
	\\& \leq \sum_{j,j'} \bigg( \sum_k  (\frac{n_k \bar{N}_k}{n\bar{N}} )^2 \bigg) 
	\cdot \bigg(  \sum_k \Sigma_k(j,j')^2 \bigg)
	\\&\leq \sum_{j,j'} \bigg( \sum_k  \frac{n_k \bar{N}_k}{n\bar{N}}  \bigg) 
	\cdot \bigg(  \sum_k \Sigma_k(j,j')^2 \bigg)
	= \sum_{j,j'}\sum_k \Sigma_k(j,j')^2 =\sum_k \| \Sigma_k\|_F^2.
	\num 	\label{eqn:Bn_bd}
\end{align*}

Next by the definition of $\Sigma_k$ in \eqref{relaxedMod2}, we have $\Sigma_k(j,j')=\frac{1}{n_k\bar{N}_k}\sum_{i\in S_k}N_i\Omega_{ij}\Omega_{ij'}$. It follows that
\begin{align} \label{eqn:En_bd}
	E_n &\leq  \sum_k \sum_{j,j'} \Bigl(  \frac{1}{n_k\bar{N}_k} \sum_{i\in S_k}N_i\Omega_{ij}\Omega_{i j'}\Bigr)\Bigl( \frac{1}{n_k\bar{N}_k}\sum_{m\in S_k}N_m\Omega_{mj}\Omega_{m j'} \Bigr)\cr
	&=\sum_k \sum_{j,j'}\Sigma_k^2(j,j') = \sum_k \|\Sigma_k\|_F^2. 
\end{align}
Next, Lemma \ref{lem:Theta_n2+n3+n4} implies that
\begin{align}
	\label{eqn:app_sumLem}
	\Theta_{n2} + \Theta_{n3} + \Theta_{n4} 
	\asymp \sum_k \| \mu_k \|^2 = K \| \mu \|^2,
\end{align}
where we use that the null hypothesis holds. By assumption of the lemma, we have 
$$\beta_n= \frac{ \max \bigg\{ 
	\sum_k  \sum_{i\in S_k}\frac{N^2_i}{n_k^2 \bar{N}_k^2}\|\Omega_i\|_3^3 \, , \, \, 
	\sum_k \| \Sigma_k \|_F^2
	\bigg\} }{ K \| \mu \|^2 } = o(1)$$
Combining this with   \eqref{eqn:var_null_main_decomp}, \eqref{eqn:An_bd}, \eqref{eqn:Bn_bd}, \eqref{eqn:En_bd},and  \eqref{eqn:app_sumLem}  completes the proof of the first claim. The second claim follows plugging in $\mu_k = \mu$ for all $k \in \{1, 2, \ldots, K\}$. 

%
%
%

\qed

\subsection{Proof of Lemma~\ref{lem:Theta_n2+n3+n4-K=2}}

By assumption, $N_i^3/(N_i-1) \asymp N_i, M_i^3/(M_i-1) \asymp M_i$. By direct calculation,
\begin{align*}
	\Theta_{n2} + \Theta_{n4}
	&\asymp \big[ \frac{m\bar{M}}{(n \bar{N} + m\bar{M}) n \bar{N}}\big]^2
	\sum_{i, m, j} N_i N_m \Omega_{ij} \Omega_{mj}  
	+ \big[ \frac{n\bar{N}}{(n \bar{N} + m\bar{M}) m \bar{M}}\big]^2
	\sum_{i, m} N_i N_m \Gamma_{ij} \Gamma_{mj}  
	\\&= \frac{1}{(n \bar{N} + m \bar{M})^2} \bigg( (m \bar{M})^2 \| \eta \|^2 
	+ n \bar{N}^2 \| \theta \|^2 \bigg). 
	\num 	\label{eqn:theta_n2+n4_K=2}
\end{align*}
Next
\begin{align*}
	\Theta_{n3} 
	&= 
	\frac{4}{(n \bar{N} + m \bar{M})^2} \sum_{i\in S_1}\sum_{m\in S_2} \sum_j N_i \Omega_{ij} \cdot N_m \Gamma_{mj} 
	\\&= \frac{4}{(n \bar{N} + m \bar{M})^2} \cdot
	n \bar{N} m \bar{M} \langle \theta, \eta \rangle. 
	\num	\label{eqn:theta_n3_K=2}
\end{align*}
Combining \eqref{eqn:theta_n2+n4_K=2} and \eqref{eqn:theta_n3_K=2} yields 
\begin{align*}
	\Theta_{n2}+\Theta_{n3}+ \Theta_{n4}
	&\asymp \frac{1}{(n \bar{N} + m \bar{M})^2}
	\big( (m \bar{M})^2 \| \eta \|^2 + 2 n \bar{N} m \bar{M} \langle \theta, \eta \rangle + 
	n \bar{N}^2 \| \theta \|^2\big) \\
	&= \bigg \|  \frac{m \bar{M}}{ n\bar{N}+ m \bar{M}} \eta +  
	\frac{n \bar{N}}{ n\bar{N}+ m \bar{M}} \theta  \bigg \|^2,
\end{align*}
which proves the first claim. The second follows by plugging in $\theta = \eta = \mu$ under the null. \qed

\subsection{Proof of Lemma~\ref{lem:var-null-K=2}}

As in \eqref{eqn:var_null_main_decomp}, we have under the null that 
\begin{align}
	\label{eqn:var_null_main_decomp_K=2}
	\var(T) = \Theta_{n2} + \Theta_{n3} + \Theta_{n4}
	+ O(A_n + B_n + E_n). 
\end{align}
For general $K$, observe that the proofs of the bounds 
\begin{align*}
	A_n &\leq \sum_{k=1}^K \Bigl(\frac{1}{n_k\bar{N}_k}\Bigr)^2\sum_{i\in S_k}N_i^2 \|\Omega_{i}\|_3^3 \\
	B_n &\leq \sum_{k=1}^K \| \Sigma_k \|_F^2 \\
	E_n &\leq \sum_{k=1}^K \| \Sigma_k \|_F^2
\end{align*}
derived in \eqref{eqn:An_bd}, \eqref{eqn:Bn_bd}, and \eqref{eqn:En_bd}, only use the assumption that $N_i, M_i \geq 2$ for all $i$. 

Translating these bounds to the notation of the $K = 2$ case, we have
\begin{align*}
	A_n &\leq 	\sum_i N_i^2 \| \Omega_i \|^3 +
	\sum_i M_i^2 \| \Gamma_i \|^3 \\
	B_n &\leq \| \Sigma_1 \|_F^2 + \| \Sigma_2 \|_F^2 \\
	E_n &\leq \| \Sigma_1 \|_F^2 + \| \Sigma_2 \|_F^2. 
	\num \label{eqn:remainder_bds_K=2}
\end{align*}
Furthermore, we know that $\Theta_n \geq c \| \mu \|^2$ under the null by Lemma \ref{lem:Theta_n2+n3+n4-K=2}, for an absolute constant $c>0$. Combining this with \eqref{eqn:var_null_main_decomp_K=2} and \eqref{eqn:remainder_bds_K=2} completes the proof. 
\qed

\subsection{Proof of Lemma~\ref{lem:Vdecompose}}

Define
\begin{align*}
	V_1 &= 2\sum_{k=1}^K\sum_{i\in S_k}\sum_{j=1}^p  \Bigl(\frac{1}{n_k\bar{N}_k}-\frac{1}{n\bar{N}}\Bigr)^2\biggl[  \frac{N_iX_{ij}^2}{N_i-1} - \frac{N_iX_{ij}(N_i-X_{ij})}{(N_i-1)^2}\biggr]  \\
	V_2 &= \frac{2}{n^2\bar{N}^2}\sum_{1\leq k\neq \ell\leq K}\sum_{i\in S_k}\sum_{m\in S_\ell} \sum_{j=1}^p X_{ij}X_{mj} \\
	V_3 &= 2\sum_{k=1}^K \sum_{\substack{i\in S_k, m\in S_k,\\ i\neq m}}\sum_{j=1}^p \Bigl(\frac{1}{n_k\bar{N}_k}-\frac{1}{n\bar{N}}\Bigr)^2 X_{ij}X_{mj}.
\end{align*}
Observe that $V_1 + V_2 + V_3 = V$. Also define
\begin{align}
	A_{11} &= 
	\sum_{i} \sum_{r =1}^{N_i} 
	\sum_j \big[ \frac{4\theta_i \Omega_{ij}}{N_i}  \big] Z_{ijr} \num \label{eqn:A11_def} \\
	A_{12} &= 
	2\sum_{i} \sum_{r =1}^{N_i} 
	\sum_j \big[\sum_{m \in [n] \backslash \{i\} }  \alpha_{im} N_m \Omega_{mj} \big] Z_{ijr} 
	\num \label{eqn:A12_def}
\end{align}
and observe that $A_{11} + A_{12} = A_1$.

First, we derive the decomposition of $V_1$. Recall that 
\beq 
Y_{ij} :=\frac{X_{ij}}{N_i} - \Omega_{ij}= \frac{1}{N_i}\sum_{r=1}^{N_i}Z_{ijr}, \qquad  Q_{ij}: = Y_{ij}^2-\mathbb{E}Y^2_{ij}=Y_{ij}^2-\frac{\Omega_{ij}(1-\Omega_{ij})}{N_i}. 
\eeq
With these notations, $X_{ij}=N_i(\Omega_{ij}+Y_{ij})$ and $N_iY_{ij}^2=N_iQ_{ij}+\Omega_{ij}(1-\Omega_{ij})$.

Write 
\beq \label{lem-V1decompose-1}
V_1 = 2\sum_{i=1}^n \sum_{i=1}^n\frac{\theta_i}{N_i}\Delta_{ij}, \qquad\mbox{where}\quad\Delta_{ij}:= \frac{X^2_{ij}}{N_i}  - \frac{X_{ij}(N_i-X_{ij})}{N_i(N_i-1)}. 
\eeq

Note that $X_{ij}=N_i(\Omega_{ij}+Y_{ij})$ and $Y_{ij}^2=Q_{ij}+N_i^{-1}\Omega_{ij}(1-\Omega_{ij})$. It follows that
\[
\frac{X_{ij}^2}{N_i} = N_i\Omega_{ij}^2 + 2N_i\Omega_{ij}Y_{ij} +  N_iQ_{ij} + \Omega_{ij}(1-\Omega_{ij}). 
\]
In \eqref{lem-decompose-00}, we have shown that $Q_{ij} = (1-2\Omega_{ij}) \frac{Y_{ij}}{N_i} + \frac{1}{N^2_i}\sum_{1\leq r\neq s\leq N_i}Z_{ijr}Z_{ijs}$. It follows that
\[
\frac{X_{ij}^2}{N_i} = N_i\Omega_{ij}^2 + 2N_i\Omega_{ij}Y_{ij} +  (1-2\Omega_{ij})Y_{ij} + \frac{1}{N_i}\sum_{1\leq r\neq s\leq N_i}Z_{ijr}Z_{ijs} + \Omega_{ij}(1-\Omega_{ij}).
\]
Additionally, by \eqref{lem-decompose-1}, 
\[
\frac{X_{ij}(N_{ij}-X_{ij})}{N_i(N_i-1)} =\Omega_{ij}(1-\Omega_{ij}) +(1-2\Omega_{ij})Y_{ij} - \frac{1}{N_i(N_i-1)}\sum_{1\leq r\neq s\leq N_i}Z_{ijr}Z_{ijs}.
\]
Combining the above gives
\begin{align} \label{lem-Vdecompose-2}
	&\Delta_{ij}  = N_i\Omega_{ij}^2 + 2N_i\Omega_{ij}Y_{ij}   + \frac{1}{N_i-1}\sum_{1\leq r\neq s\leq N_i}Z_{ijr}Z_{ijs} \cr
	&= N_i\Omega_{ij}^2 + 2\Omega_{ij}\sum_{r=1}^{N_i}Z_{ijr}  + \frac{1}{N_i-1}\sum_{1\leq r\neq s\leq N_i}Z_{ijr}Z_{ijs}. 
\end{align}
Recall the definition of $\Theta_{n2}$ in \eqref{eqn:Theta_n2},  $A_2$ in \eqref{eqn:A2_def}, and $A_{11}$ in \eqref{eqn:A11_def}. We have
\begin{align*}
	V_1 &= 2\sum_{k, i \in S_k} \sum_j \frac{\theta_i}{N_i} 
	\big[ N_i\Omega_{ij}^2 + 2\Omega_{ij}\sum_{r=1}^{N_i}Z_{ijr}  + \frac{1}{N_i-1}\sum_{1\leq r\neq s\leq N_i}Z_{ijr}Z_{ijs}  \big].  
	\\ &= \Theta_{n2}
	+ \sum_{k, i \in S_k} \sum_j \frac{4 \theta_i \Omega_{ij} }{N_i} \sum_{r=1}^{N_i}Z_{ijr}
	+ \sum_{k, i \in S_k} \sum_j \frac{2\theta_i}{N_i(N_i-1)}  \sum_{1\leq r\neq s\leq N_i}Z_{ijr}Z_{ijs} 
	\\&= \Theta_{n2} + A_{11} + A_2
	\num \label{eqn:V1_decomp}
\end{align*}

Next, we have
\begin{align*}
	V_2 + V_3 &= 
	\sum_{ i \neq m} \alpha_{im} N_i N_m \sum_j
	\bigg[ (Y_{ij} + \Omega_{ij} ) (Y_{mj} + \Omega_{mj} ) \bigg] 
	\\&= 	\sum_{ i \neq m} \alpha_{im} N_i N_m \sum_j Y_{ij} Y_{mj}
	+ 2 	\sum_{ i \neq m} \alpha_{im} N_i N_m \sum_j  Y_{ij} \Omega_{mj} 
	+ \sum_{i \neq m} \alpha_{im} N_i N_m \sum_j  \Omega_{ij} \Omega_{mj} 
	\\&= \sum_{i \neq m} \sum_{r = 1}^{N_i} \sum_{s = 1}^{N_m}
	\alpha_{im} \big(\sum_j Z_{ijr} Z_{mjs} \big) 
	+	2\sum_{i} \sum_{r =1}^{N_i} 
	\sum_j \big[\sum_{m \in [n] \backslash \{i\} }  \alpha_{im} N_m \Omega_{mj} \big] Z_{ijr} 
	+ \Theta_{n3} + \Theta_{n4} 
	\\&= A_3 + A_{12} + \Theta_{n3} + \Theta_{n4} . 
\end{align*}
Hence
\begin{align*}
	A_1 + A_2 + A_3 + \Theta_{n2} + \Theta_{n3} + \Theta_{n4}
	= V,
\end{align*}
which verifies \eqref{eqn:V_decomp}. By inspection, we also see that $\E A_b = 0$ for $b \in \{1,2,3\}$. That $A_1, A_2, A_3$ are mutually uncorrelated follows immediately from the linearity of expectation and the fact that the random variables $\{ Z_{ijr} \}_{i,r} \cup \{ Z_{ijr} Z_{mjs} \}_{(i,r) \neq (m, s)}$ are mutually uncorrelated. 

\qed

\subsection{Proof of Lemma~\ref{lem:varV}}

Define
\begin{align}
	\label{eqn:gamma_def}
	\gamma_{irj} = \frac{4\theta_i \Omega_{ij}}{N_i} 
	+ \sum_{m \in [n] \backslash \{i\} } 2 \alpha_{im} N_m \Omega_{mj}
\end{align}
and recall that $A_1 = \sum_{i} \sum_{r \in [N_i]}  \sum_j \gamma_{irj} Z_{ijr}$. First we develop a bound on $\gamma_{irj}$.  Suppose that 
$i \in S_k$. Then we  have
\begin{align*}
	\gamma_{irj}
	&\lesssim \frac{N_i \Omega_{ij}}{n_k^2 \bar{N}_k^2}
	+ \sum_{m \in S_k, m \neq i} \frac{N_m \Omega_{mj}}{n_k^2 \bar{N}_k^2}
	+ \sum_{k' \in [K]\backslash \{k\}} \sum_{m \in S_{k'}} \frac{N_m \Omega_{mj}}{n^2 \bar{N}^2} 
	\\&\les  \frac{\mu_{kj}}{n_k \bar{N}_k} 
	+ \frac{ \mu_j}{n \bar{N}}.
\end{align*}
Next  using properties of the covariance matrix of a multinomial vector, we have 
\begin{align*}
	\var(A_1)
	&= \sum_{i, r \in [N_i]} \var( \gamma_{ir:} '  Z_{i:r} ) = \sum_{i, r \in [N_i]} \gamma_{ir:}' \text{Cov}(Z_{i:r}) \gamma_{ir:}
	\\& \leq \sum_{i, r \in [N_i]} \gamma_{ir:}' \text{diag}(\Omega_{i:}) \gamma_{ir:}
	=  \sum_{i, r \in [N_i]}  \sum_j \Omega_{ij} \gamma_{irj}^2
	\\&\les \sum_{k,j} \big(  \frac{\mu_{kj}}{n_k \bar{N}_k}+\frac{ \mu_j}{n \bar{N} } \big)^2
	\sum_{i\in S_k, r \in [N_i]} \Omega_{ij} 
	\\&\les \sum_{k,j} \big(  \frac{\mu_{kj}}{n_k \bar{N}_k}\big)^2 n_k \bar{N}_k \mu_{kj} + \sum_{k,j} \big(  \frac{ \mu_j}{n \bar{N} } \big)^2 n_k \bar{N}_k \mu_{kj} 
	\\&= (\sum_k \frac{\| \mu_k \|_3^3 }{n_k \bar{N}_k} ) 
	+ \frac{\| \mu \|_3^3}{n \bar{N}} \les \sum_k \frac{\| \mu_k \|_3^3 }{n_k \bar{N}_k},
	\num \label{eqn:mu3_vs_muk3}
\end{align*}
which proves the first claim. The last inequality follows because by Jensen's inequality (noting that the function $x \mapsto x^3$ is convex for $x \geq 0$), 
\begin{align*}
	\| \mu \|_3^3
	&= \sum_j \bigg( \sum_k (\frac{n_k \bar{N}_k}{n \bar{N}} ) \mu_{kj} \bigg)^3 
	\leq \sum_j \sum_k (\frac{n_k \bar{N}_k}{n \bar{N}} ) \mu_{kj}^3
	\leq \sum_k \| \mu_k \|_3^3. 
\end{align*}

Next observe that
\begin{align}
	A_2 = \sum_i \sum_{r \neq s} \frac{2 \theta_i}{N_i(N_i -1)} W_{irs}
\end{align}
where recall $W_{irs} = \sum_j Z_{ijr}Z_{ijs}$. Also recall that $W_{irs}$ and $W_{i'r's'}$ are uncorrelated unless $i = i'$ and $\{r,s\} = \{r',s'\}$.
By \eqref{lem-var-3-add}, 
\begin{align*}
	\var(A_2)
	&= \sum_i \sum_{r \neq s} \frac{4 \theta_i^2}{N_i^2 (N_i - 1)^2} \var(W_{irs})
	\\&\les \sum_i \sum_{r \neq s} \frac{4 \theta_i^2}{N_i^2 (N_i - 1)^2}  \| \Omega_i \|^2
	\\&\les \sum_k \sum_{i \in S_k}  \cdot (\frac{1}{n_k\bar{N}_k}-\frac{1}{n\bar{N}})^4\frac{N_i^6}{(N_i-1)^2} \cdot \frac{1}{N_i (N_i  -1)}
	\| \Omega_i \|^2 
	\\&\les \sum_k \sum_{i \in S_k} \frac{N_i^2}{n_k^4 \bar{N}_k^4} \| \Omega_{i} \|^2 
	\num \label{eqn:varA2_bd}
\end{align*}
Also observe that
\begin{align*}
	\sum_k \frac{1}{n_k^4 \bar{N}_k^4} \sum_{i \in S_k} N_i^2  \| \Omega_i \|_2^2
	&\leq 		\sum_k \frac{1}{n_k^2 \bar{N}_k^2} \sum_{i,m \in S_k} \bigg \langle (\frac{N_i}{n_k \bar{N}_k}) \Omega_i, (\frac{N_m}{n_m \bar{N}_m}) \Omega_m \bigg \rangle 
	\\&= 	\sum_k \frac{1}{n_k^2 \bar{N}_k^2} \| \mu_k \|^2.
\end{align*}	This establishes the second claim.

Last we study $A_3$. Observe that
\begin{align*}
	A_3 = \sum_{i \neq m} \sum_{r = 1}^{N_i} \sum_{s = 1}^{N_m} \alpha_{im} V_{irms}
\end{align*}
where recall $V_{irms} = \sum_j Z_{ijr} Z_{mjs}$. Recall that $V_{irms}$ and $V_{i'r'm's'}$ are uncorrelated unless $(r,s) = (r', s')$ and $\{i,m\} = \{i', m'\}$ .By \eqref{Var-of-V(irms)},
\begin{align*}
	\var(A_3) &\lesssim \sum_{i \neq m} \alpha_{im}^2 N_i N_m \sum_j \Omega_{ij}\Omega_{mj} 
	\\&\les \sum_k \sum_{i \neq m \in S_k} \frac{1}{n_k^4 \bar{N}_k^4} \langle N_i \Omega_i, N_m \Omega_m \rangle 
	+ \sum_{k \neq \ell} \sum_{i \in S_k, m \in S_\ell} \frac{1}{n^4 \bar{N}^4} \langle N_i \Omega_i, N_m \Omega_m \rangle 
	\\&\les \sum_k \frac{ \| \mu_k \|^2 }{n_k^2 \bar{N}_k^2}
	+ \sum_{k , \ell} \frac{1}{n^4 \bar{N}^4} \langle n_k \bar{N}_k \mu_k, 
	n_\ell \bar{N}_\ell \mu_\ell \rangle 
	\\&\les \sum_k \frac{ \| \mu_k \|^2 }{n_k^2 \bar{N}_k^2}
	+ \frac{ \| \mu \|^2 }{ n^2 \bar{N}^2} \les \sum_k \frac{ \| \mu_k \|^2 }{n_k^2 \bar{N}_k^2}.
	\num \label{eqn:A3_bd} 
\end{align*}
In the last line we use that $\| \mu \|^2 \leq 2 \sum \| \mu_k \|^2$ as shown in \eqref{eqn:mu_vs_muk}. This proves all required claims. 
\qed

\subsection{Proof of Proposition~\ref{prop:var_estimation_null}}

Under the null hypothesis, we have $\Theta_{n1} \equiv 0$. Thus, $\E V = \Theta_n$ under the null by Lemma \ref{lem:Vdecompose}.  Under \eqref{cond1-basic}, we have $\var(T) = [1+o(1)] \Theta_n$. Therefore, 
\begin{align}
	\label{eqn:var-est-null1}
	\E V = [1+o(1)]\var(T),
\end{align}
so $V$ is asymptotically unbiased under the null.	Furthermore, by Lemma \ref{lem:Theta_n2+n3+n4}, we have
\begin{align}
	\label{eqn:var-est-null2}
	\Theta_n \asymp K \| \mu \|^2. 
\end{align}
In Lemma \ref{lem:varV}, we showed that 
\begin{align*}
	\var(A_2) &\lesssim \sum_k \sum_{i \in S_k} \frac{N_i^2 \| \Omega_i \|_2^2 }{n_k^4 \bar{N}_k^4}
\end{align*}
We conclude by Lemma \ref{lem:varV} that under the null 
\begin{align}
	\var(V) \les \sum_k \frac{ \| \mu \|^2}{n_k^2 \bar{N}_k^2} 
	\vee 
	\sum_k \frac{ \| \mu \|_3^3}{ n_k \bar{N}_k}.
	\label{eqn:var-est-null3}
\end{align}
By Chebyshev's inequality, \eqref{eqn:var-est-null2},  \eqref{eqn:var-est-null3}, and assumption \eqref{eqn:null_ell2} of the theorem statement, we have
\begin{align*}
	\frac{|V-\E V| }{\var(T)} \asymp 	\frac{|V-\E V |}{K \| \mu \|^2} = o_{\mathbb{P}}(1).
\end{align*}
Thus by \eqref{eqn:var-est-null1}, 
\begin{align*}
	\frac{V}{\var(T)}
	&= \frac{(V - \E V)}{\var(T)} + \frac{ \E V }{\var(T) }
	= o_{\mathbb{P}}(1) + [1 + o(1)],
\end{align*}
as desired. 
\qed 

\subsection{Proof of Lemma~\ref{lem:var_lbd}}

By Lemmas \ref{lem:decompose}--\ref{lem:var4}, we have
\begin{align}
	\label{eqn:var_lbd}
	\var(T) &= \sum_{a = 1}^4 \var( \bo' U_a )
	\geq (\sum_{a = 2}^4 \Theta_{na}) - (A_n + B_n + E_n).
\end{align}
Using that $\max_i \| \Omega_i \|_\infty \leq 1 - c_0$, we have $	\| \Omega_i \|^3 \leq (1 - c_0) \| \Omega_i \|^2$, which implies that 
\begin{align*}
	\num \label{eqn:An_tool}
	A_n \leq (1 - c_0) \Theta_{n2}. 
\end{align*}
Again using $\max_i \| \Omega_i \|_\infty \leq 1 - c_0$, as well as $\sum_{j'} \Omega_{ij'} = 1$, we have
\begin{align*}
	B_n &= \frac{2}{n^2\bar{N}^2} \sum_{k\neq \ell }\sum_{i\in S_k}\sum_{m\in S_{\ell}}\sum_{j,j'}N_iN_m\Omega_{ij}\Omega_{ij'}\Omega_{mj}\Omega_{mj'}
	\\&\leq (1-c_0) \cdot  \frac{2}{n^2\bar{N}^2} \sum_{k\neq \ell }\sum_{i\in S_k}\sum_{m\in S_{\ell}}\sum_{j,j'}N_iN_m\Omega_{ij}\Omega_{ij'}\Omega_{mj} 
	\\&= (1-c_0) \cdot  \frac{2}{n^2\bar{N}^2} \sum_{k\neq \ell }\sum_{i\in S_k}\sum_{m\in S_{\ell}}\sum_{j}N_iN_m\Omega_{ij}\Omega_{mj}
	\\&\leq (1-c_0) \cdot \Theta_{n3}. \num \label{eqn:Bn_tool}
\end{align*}
Similarly to control $E_n$, we again use $\max_i \| \Omega_i \|_\infty \leq 1 - c_0$ and obtain 
\begin{align*}
	E_n &= 2\sum_k
	\sum_{\substack{i\in S_k, m\in S_k,\\ i\neq m}} \sum_{1 \leq j, j' \leq p}  \Bigl(\frac{1}{n_k\bar{N}_k}-\frac{1}{n\bar{N}}\Bigr)^2  N_i N_m 
	\Omega_{ij} \Omega_{ij'} \Omega_{mj} 
	\Omega_{mj'}
	\\& \leq (1 - c_0) \cdot  2\sum_k
	\sum_{\substack{i\in S_k, m\in S_k,\\ i\neq m}} \sum_{1 \leq j, j' \leq p}  \Bigl(\frac{1}{n_k\bar{N}_k}-\frac{1}{n\bar{N}}\Bigr)^2  N_i N_m 
	\Omega_{ij} \Omega_{ij'} \Omega_{mj} 
	\\& \leq (1 - c_0) \cdot  2\sum_k
	\sum_{\substack{i\in S_k, m\in S_k,\\ i\neq m}} \sum_{1 \leq j\leq p}  \Bigl(\frac{1}{n_k\bar{N}_k}-\frac{1}{n\bar{N}}\Bigr)^2  N_i N_m 
	\Omega_{ij}  \Omega_{mj}
	\\& \leq (1 - c_0) \cdot \Theta_{n4}. \num  \label{eqn:En_tool}
\end{align*}
Combining \eqref{eqn:var_lbd}, \eqref{eqn:An_tool}, \eqref{eqn:Bn_tool}, and \eqref{eqn:En_tool} finishes the proof. 

\qed 

\subsection{Proof of Proposition~\ref{prop:var_estimation_alt}}

By Lemmas \ref{lem:Theta_n2+n3+n4} and \ref{lem:var_lbd}, 
\begin{align*}
	\num \label{eqn:var_alt1}
	\var(T)  \gtrsim \Theta_{n2} + \Theta_{n3} + \Theta_{n4}
	\gtrsim \sum_k \| \mu_k \|^2. 
\end{align*}
By Lemma \ref{lem:varV}, 
\begin{align*}
	\num \label{eqn:var_alt3}
	\var(V) &\lesssim \sum_k \frac{ \| \mu_k \|^2}{n_k^2 \bar{N}_k^2} 
	\vee \sum_k \frac{ \| \mu_k \|_3^3}{ n_k \bar{N}_k} \\
\end{align*}
Using a similar argument based on Chebyshev's inequality as in the proof of Proposition \ref{prop:var_estimation_null} and applying  \eqref{eqn:var_alt1} and \eqref{eqn:var_alt3}, we have 
\begin{align}
	\label{eqn:var_alt_Chebyshev}
	\frac{|V - \E V|}{\var(T)} \gtrsim \frac{|V - \E V|}{\sum_k \| \mu_k \|^2}
	= o_{\mathbb{P}}(1). 
\end{align}
Next, by Lemma \ref{lem:Vdecompose} and \eqref{eqn:var_alt1},
\begin{align*}
	\num \label{eqn:var_alt2}
	\E V = \Theta_{n2} + \Theta_{n3} + \Theta_{n4} \lesssim \var(T).
\end{align*}
Combining \eqref{eqn:var_alt_Chebyshev} and \eqref{eqn:var_alt2} finishes the proof.
\qed

\subsection{Proof of Proposition~\ref{prop:var_estimation_null_K=n}}

From the proof of Lemma \ref{lem:Vdecompose}, we have
\begin{align*}
	V^* = V_1 = \Theta_{n2} + 
	A_{11} + A_2,
\end{align*}
and the terms on the right-hand-side are mutually uncorrelated. 
From \eqref{eqn:mu3_vs_muk3}, we have 
\begin{align*}
	\var(A_{11}) \les \sum_i \frac{\|\Omega_i \|_3^3}{N_i} \\
	\var(A_2) \les \sum_i \frac{ \| \Omega_i \|^2}{N_i^2}. 
\end{align*}
Hence
\begin{align*}
	\E V^* &= \Theta_{n2} \\
	\var(V^*) &\les \sum_i \frac{\|\Omega_i \|_3^3}{N_i} \vee \sum_i \frac{ \| \Omega_i \|^2}{N_i^2}. 
	\num \label{eqn:V*_properties}
\end{align*}

Since $K = n$ and the null hypothesis holds, we have $\Theta_{n1} \equiv \Theta_{n4} \equiv 0$. Moreover, by \eqref{eqn:Thetan3_bd}, we have
\begin{align*}
	\Theta_{n3} \lesssim \| \mu \|^2 \ll \Theta_{n2} \asymp n \| \mu \|^2. 
\end{align*}
It  follows that 
\begin{align}
	\var(T) = [1+o(1)]\Theta_{n2} \asymp n \| \mu \|^2
	\label{eqn:varT_asymp_K=n}.
\end{align} Thus by \eqref{eqn:V*_properties} and Chebyshev's inequality, we have 
\begin{align*}
	\frac{V^*}{\var(T)} 
	= 	\frac{V^*- \E V^*}{\var(T)}  + 	\frac{\E V^*}{\var(T)} 
	= o_{\mathbb{P}}(1) + 1 + o(1), 
\end{align*}
as desired. 

\qed 

\subsection{Proof of Proposition~\ref{prop:var_estimation_alt_K=n}}

By Lemmas \ref{lem:Theta_n2+n3+n4} and \ref{lem:var_lbd}, 
\begin{align*}
	\num \label{eqn:var_alt1_K=n}
	\var(T)  \gtrsim \Theta_{n2} + \Theta_{n3} 
	\gtrsim \sum_i \| \Omega_i \|^2. 
\end{align*}
By \eqref{eqn:V*_properties}, 
\begin{align*}
	\num \label{eqn:var_alt3_K=n}
	\var(V^*) &\lesssim \sum_i \frac{ \| \Omega_i \|^2}{N_i^2} 
	\vee \sum_i \frac{ \| \Omega_i \|_3^3}{ N_i} \\
\end{align*}
Using a similar argument based on Chebyshev's inequality as in the proof of Proposition \ref{prop:var_estimation_null} and applying  \eqref{eqn:var_alt1_K=n} and \eqref{eqn:var_alt3_K=n}, we have 
\begin{align}
	\label{eqn:var_alt_Chebyshev_K=n}
	\frac{|V^* - \E V^*|}{\var(T)} \gtrsim \frac{|V^* - \E V^*|}{\sum_i \| \Omega_i  \|^2}
	= o_{\mathbb{P}}(1). 
\end{align}
Next, by Lemma \ref{lem:Vdecompose} and \eqref{eqn:var_alt1_K=n},
\begin{align*}
	\num \label{eqn:var_alt2_K=n}
	\E V^* = \Theta_{n2}  \lesssim \var(T).
\end{align*}
Combining \eqref{eqn:var_alt_Chebyshev} and \eqref{eqn:var_alt2_K=n} finishes the proof.
\qed 

\section{Proofs of asymptotic normality results}
\label{sec:asymptotic_normality_appendix}
The goal of this section is to prove Theorems \ref{thm:null} and \ref{thm:null2}. The argument relies on the martingale central limit theorem and the lemmas stated below. As a preliminary, we describe a martingale decomposition of $T$ under the null.  

Define 
\begin{align*}
	U = \bo' (U_3 + U_4), \qquad \text{and } \,\,\, S = \bo' U_2. 
\end{align*}
By Lemma \ref{lem:decompose}, we have $T = U + S$ under the null hypothesis. It holds that 
\begin{align}
	U = \sum_{i < i'} \sigma_{i, i'} \sum_{r= 1}^{N_i} \sum_{s = 1}^{N_{i'}} \big( \sum_j Z_{ijr} Z_{i'js} \big).  
	\label{eqn:U_stat}
\end{align}

where we define 
\[
\sigma_{i,i'} 
= \begin{cases}
	2	\big( \frac{1}{n_k \bar{N}_k} - \frac{1}{n \bar{N}} \big) &\quad \text{ if } i, i' \in S_k \text{ for some } k \\
	- \frac{2}{n \bar{N}} &\quad \text{ else}.
\end{cases}
\]
Define a sequence of random variables
\begin{align}
	D_{\ell, s}
	= \sum_{ i \in [\ell - 1] } 
	\sigma_{i, \ell} \sum_{r = 1}^{N_i} \sum_j Z_{ijr} Z_{\ell j s}
	\label{eqn:Dells_def}
\end{align}
indexed by $(\ell, s ) \in \{ (i, r) \}_{1 \leq i \leq n, 1 \leq r \leq N_i}$, where these tuples are placed in lexicographical order. Precisely, we define
\[
(\ell_1, s_1) \prec (\ell_2, s_2) 
\]
if either
\begin{itemize}
	\item $\ell_1 < \ell_2$, or
	\item $\ell_1 = \ell_2$ and $s_1 < s_2$. 
\end{itemize}

Observe that 
\[
\sum_{\ell, s} D_{\ell,s} = U.
\]
Next define  $\mc{F}_{\prec (\ell, s)} $ to be the $\sigma$-field generated by $\{ Z_{i:r} \}_{(i,r) \prec (\ell, s) }$. Observe that 
\[
\E[ D_{\ell, s} | \mc{F}_{\prec (\ell, s)} ] = 0,
\]
and hence $\{ D_{\ell, s } \}$ is a martingale difference sequence. 
%
Turning to $S$, we have 
\begin{align}
	S =  \sum_{i  =1}^n \sigma_{i} \sum_{r < s} \sum_j Z_{ijr}Z_{ijs}. 
	\label{eqn:S_def}
\end{align}
where we define 
\[
\sigma_i =  2\Bigl(\frac{1}{n_k\bar{N}_k} -\frac{1}{n\bar{N}} \Bigr)\frac{N_i}{N_i-1}
\]
if $i \in S_k$. 
Define
\begin{align}
	E_{\ell, s} = 	\sigma_{\ell} \sum_{ r  \in [s-1] } \sum_j Z_{\ell jr} Z_{\ell js}.  
	\label{eqn:Eells_def}
\end{align}
Note that $E_{\ell, 1} = 0$. Order $(\ell, s)$ lexicographically as above, and recall that $\mc{F}_{\prec (\ell, s)} $ is the $\sigma$-field generated by $\{ Z_{i:r} \}_{(i,r) \prec (\ell, s) }$. Observe that 
\[
\E[ E_{\ell, s} | \mc{F}_{\prec (\ell, s)} ] = 0,
\]
and hence $\{ E_{\ell, s } \}$ is a martingale difference sequence. 
We have 
\begin{align*}
	\sum_{(\ell, s)} \sigma_{\ell} \sum_{ r  \in [s-1] } \sum_j Z_{\ell jr} Z_{\ell js}
	&= \sum_{\ell = 1}^n \sum_{s = 1}^{N_\ell}  \sigma_{\ell} \sum_{ r  \in [s-1] } \sum_j Z_{\ell jr} Z_{\ell js}
	=S. 
\end{align*}
Define
\begin{align*}
	\num \label{eqn:T_mg_terms}
	\mc{M}_{\ell, s} = D_{\ell, s} + E_{\ell, s},
	\quad \widetilde{\mc{M}}_{\ell, s} = 
	\frac{ \mc{M}_{\ell, s} }{\sqrt{\var(T)}}. 
\end{align*}
Thus we obtain the martingale decomposition:
\begin{align*}
	\num \label{eqn:T_mg_decomp}
	T = U + S = \sum_{(\ell, s)} [ D_{\ell, s}   + E_{\ell, s} ] = \sum_{(\ell, s)} \mc{M}_{\ell, s}.  
\end{align*}

The technical results below are crucial to the proof of Theorem \ref{thm:null} given in Section \ref{sec:thm-null-proof}. Theorem \ref{thm:null2} then  follows easily from Theorem \ref{thm:null} and Theorem \ref{prop:var_estimation_null}. 




\begin{lemma}
	\label{lem:condl_var_bias} 
	Let $\widetilde{\mc{M}}_{\ell,s}$ be defined as in \eqref{eqn:T_mg_terms}. It holds that 
	\begin{align*}
		\E \bigg[  \sum_{(\ell, s)} \var\big( \widetilde{\mc{M}}_{\ell, s} \big| \mc{F}_{\prec (\ell, s)} \big) \bigg]  = 1. 
	\end{align*}
\end{lemma}

\begin{lemma}
	\label{lem:condl_var_U} 
	Suppose that $\min N_i \geq 2$ and $\max \| \Omega_i \|_\infty \leq 1 - c_0$. Under the null hypothesis, it holds that 
	\begin{align*}
		\var\bigg( 	\sum_{(\ell, s)} \var(  D_{\ell,s} | \mc{F}_{\prec (\ell, s)} ) \bigg) 
		\lesssim \big( \sum_k \frac{1}{n_k \bar{N}_k} \big) \| \mu \|_3^3 + K \| \mu \|_4^4.
	\end{align*}
\end{lemma}

\begin{lemma}
	\label{lem:fourth_mom_U} 
	Suppose that $\min N_i \geq 2$ and $\max \| \Omega_i \|_\infty \leq 1 - c_0$. Under the null hypothesis, it holds that 
	\begin{align*}
		\sum_{(\ell, s)} \E D_{\ell, s}^4 
		\les  \big(\sum_k \frac{1}{n_k^2 \bar{N}_k^2} \big)\| \mu \|^2 + \big(  \sum_k \frac{ 1 }{ n_k \bar{N}_k} \big) \| \mu \|_3^3 \, , 
	\end{align*}
\end{lemma}

\begin{lemma}
	\label{lem:S_condvar} 
	Suppose that $\min N_i \geq 2$ and and $\max \| \Omega_i \|_\infty \leq 1 - c_0$. Then we have 
	\begin{align}
		\var\bigg(\sum_{(\ell, s)} \var( \tilde E_{\ell,s} | \mc{F}_{\prec (\ell, s)} )\bigg) \les \sum_k \sum_{i \in S_k} \frac{ N_i^3 \| \Omega_i \|_3^3 }{ n_k^4 \bar{N}_k^4 }
		\vee \sum_k \sum_{i  \in S_k } \frac{N_i^4 \| \Omega_i \|_4^4}{n_k^4 \bar{N}_k^4} 
	\end{align}
\end{lemma}

\begin{lemma}
	\label{lem:S_fourth_mom} 
	Suppose that $\min N_i \geq 2$ and and $\max \| \Omega_i \|_\infty \leq 1 - c_0$. Then we have
	\begin{align*}
		\sum_{(\ell,s)} \E \, E_{\ell, s}^4 \les\sum_k 	\sum_{i \in S_k} \frac{ N_i^2 \| \Omega_i \|^2}{ n_k^4 \bar{N}_k^4} 
		\vee \sum_k \sum_{i \in S_k} \frac{ N_i^3 \| \Omega_i \|_3^3 }{ n_k^4 \bar{N}_k^4 }
	\end{align*}
\end{lemma}

\begin{lemma}
	\label{lem:S_bounds_on_bounds}
	Under either the null or alternative, it holds that
	\begin{align*}
		\sum_k \sum_{i \in S_k} \frac{N_i^2 \| \Omega_i \|^2}{n_k^4 \bar{N}_k^4} 
		&\leq \sum_k \frac{1}{n_k^2 \bar{N}_k^2} \| \mu_k \|^2 
		\\ 	\sum_k \sum_{i \in S_k} \frac{N_i^3 \| \Omega_i \|_3^3}{n_k^4 \bar{N}_k^4} 
		&\leq  \sum_k \frac{1}{n_k \bar{N}_k} \| \mu_k \|_3^3
		\\ \sum_k \sum_{i \in S_k} \frac{N_i^4\| \Omega_i \|_4^4}{n_k^4 \bar{N}_k^4}
		&\leq \sum_k \| \mu_k \|_4^4  
	\end{align*}
\end{lemma}

\subsection{Proof of Theorem~\ref{thm:null}}
\label{sec:thm-null-proof}
By the martingale central limit theorem (see e.g. \cite{hall2014martingale}), we have that $T/\sqrt{\var(T)} \Rightarrow N(0,1)$ if the following conditions are satisfied:
\begin{align}
	\sum_{(\ell, s)} \var\big( \widetilde{\mc{M}}_{\ell, s} \big| \mc{F}_{\prec (\ell, s)} \big) &\stackrel{\mathbb{P}}{\to} 1 \label{eqn:mgclt1} \\
	\sum_{(\ell, s)} \E\big[ \widetilde{\mc{M}}_{\ell, s}^2 {\bf 1}_{| \widetilde{\mathcal{M}}_{\ell, s} | > \varepsilon} \big| \mc{F}_{\prec (\ell, s)} \big] &\stackrel{\mathbb{P}}{\to} 0, \quad 
	\text{ for any } \varepsilon > 0. \label{eqn:mgclt2}
\end{align}
It is known that \eqref{eqn:mgclt2}, which is a Lindeberg-type condition, is implied by the Lyapunov-type condition
\begin{align}
	\label{eqn:mgclt2_alt}
	\sum_{(\ell,s)} \E  \widetilde{\mc{M}}_{\ell,s}^4 = o(1). 
\end{align}
See e.g. \cite{jin2018network}. 

Since \eqref{cond1-basic} holds, 
\begin{align}
	\var(T) \gtrsim \Theta = \Theta_{n2} + \Theta_{n3} + \Theta_{n4} 
	\gtrsim K \| \mu \|^2. 
	\label{eqn:varT_lbd} 
\end{align}
Recall that
\[
\widetilde{\mc{M}}_{\ell, s} = \frac{ \mc{M}_{\ell, s}}{ \var(T) }
= \frac{ D_{\ell, s} + E_{\ell, s}}{\var(T)} ,
\]
Note that \eqref{eqn:mgclt1} holds if 
\begin{align}
	\E \bigg[\var\big( \widetilde{\mc{M}}_{\ell, s} \big| \mc{F}_{\prec (\ell, s)} \big)\bigg] &\to 1, \, \, \text{ and} \label{eqn:M_bias} \\
	\var\bigg( \var\big( \widetilde{\mc{M}}_{\ell, s} \big| \mc{F}_{\prec (\ell, s)} \big) \bigg) &\to 0. \label{eqn:M_variance}
\end{align}
Recall that \eqref{eqn:M_bias} holds by Lemma \ref{lem:condl_var_bias}.


Next note that
\begin{align*}
	\E( D_{\ell,s} E_{\ell, s} | \mc{F}_{\prec(\ell, s)} ) = 0,
\end{align*}
by inspection of the expressions for $D_{\ell, s}$ and $E_{\ell, s}$ in \eqref{eqn:Dells_def} and \eqref{eqn:Eells_def}. 
Therefore
\begin{align*}
	\var( \mc{M}_{\ell, s}  | \mc{F}_{\prec (\ell, s)} )
	= 	\var( D_{\ell, s}  | \mc{F}_{\prec (\ell, s)} )
	+ 	\var( E_{\ell, s}  | \mc{F}_{\prec (\ell, s)} ). 
\end{align*}
Hence by \eqref{eqn:varT_lbd}; Lemmas \ref{lem:condl_var_U}, \ref{lem:S_condvar} , and \ref{lem:S_bounds_on_bounds}; and the assumption \eqref{cond2-regular}, under the null hypothesis, we have 
\begin{align*}
	\var\bigg( \var\big( \widetilde{\mc{M}}_{\ell, s} \big| \mc{F}_{\prec (\ell, s)} \big) \bigg) 
	&\leq  \frac{1}{\var(T)^2} \bigg[	\var\bigg( \var\big( D_{\ell, s} \big| \mc{F}_{\prec (\ell, s)} \big) \bigg) 
	+ 	\var\bigg( \var\big( E_{\ell, s} \big| \mc{F}_{\prec (\ell, s)} \big) \bigg)  \bigg] 
	\\&\lesssim \frac{1}{K^2 \| \mu \|^4 }
	\bigg[ \big( \sum_k \frac{1}{n_k \bar{N}_k} \big) \| \mu \|_3^3 + K \| \mu \|_4^4  \big) \| \mu \|^2 \bigg] 
	= o(1).
\end{align*}
This proves \eqref{eqn:M_variance}. Thus, \eqref{eqn:M_bias} and \eqref{eqn:M_variance} are established, which proves \eqref{eqn:mgclt1}. 

Similarly, \eqref{eqn:mgclt2_alt} (and thus \eqref{eqn:mgclt2}) holds by \eqref{eqn:varT_lbd}; Lemmas \eqref{lem:fourth_mom_U}, \eqref{lem:S_fourth_mom}, and \eqref{lem:S_bounds_on_bounds}, and the assumption \eqref{cond2-regular}. Combining \eqref{eqn:mgclt1} and \eqref{eqn:mgclt2} verifies the conditions of the martingale central limit theorem, so we conclude that $T/ \sqrt{\var(T)} \Rightarrow N(0,1)$. Since $\var(T) = [1+o(1)]\Theta_n$ by \eqref{cond2-regular} and Lemma \ref{lem:var-null}, the proof is complete. 
\qed 

\bigskip

We record a useful proposition that records the weaker conditions under which $T/\sqrt{\var(T)}$ is asymptotically normal.

\begin{proposition}
	\label{prop:null} 
	Recall that $\alpha_n$ is defined as
	\beq
	\alpha_n:= \max\left\{ \sum_{k = 1}^K \frac{\| \mu_k \|_3^3}{n_k \bar{N}_k} ,  \quad  \sum_{k = 1}^K \frac{\| \mu_k \|^2}{n_k^2 \bar{N}_k^2}\right \} \bigg / \bigg( \sum_{k = 1}^K \| \mu_k \|^2 \bigg)^2
	\eeq
	in \eqref{def:alpha_n}. If under the null hypothesis, 
	\begin{align}
		\label{eqn:weaker_cond_T_varT}
		\alpha_n = \max\left\{ \sum_{k = 1}^K \frac{\| \mu_k \|_3^3}{n_k \bar{N}_k} ,  \quad  \sum_{k = 1}^K \frac{\| \mu_k \|^2}{n_k^2 \bar{N}_k^2}\right \} \bigg / \bigg( K \| \mu \|^2 \bigg)^2  \to 0, \qquad \text{ and }  \, \, \frac{ \| \mu \|_4^4}{K \| \mu \|^4} \to 0,
	\end{align}
	then $T/\sqrt{\var(T)} \Rightarrow N(0,1)$. 
\end{proposition}

\subsection{Proof of Theorem~\ref{thm:null2}}

By our assumptions, Proposition \ref{prop:var_estimation_null} holds, and $V/\var(T) \to 1$. Thus the variance estimate $V$ is consistent under the null. Theorem \ref{thm:null2} then follows from Slutsky's theorem and Theorem \ref{thm:null}. \qed 

\subsection{Proof of Lemma~\ref{lem:condl_var_bias}}

By Lemma \ref{lem:decompose}, $S$ and $U$ are uncorrelated, and it holds that 
\begin{align*}
	\var(T) = \var(S) + \var(U). 
	\num \label{eqn:condl_var_bias1}
\end{align*}

Next note that
\begin{align*}
	\E( D_{\ell,s} E_{\ell, s} | \mc{F}_{\prec(\ell, s)} ) = 0,
\end{align*}
by inspection of the expressions for $D_{\ell, s}$ and $E_{\ell, s}$ in \eqref{eqn:Dells_def} and \eqref{eqn:Eells_def}. 
Therefore
\begin{align*}
	\var( \mc{M}_{\ell, s}  | \mc{F}_{\prec (\ell, s)} )
	= 	\var( D_{\ell, s}  | \mc{F}_{\prec (\ell, s)} )
	+ 	\var( E_{\ell, s}  | \mc{F}_{\prec (\ell, s)} ). 
\end{align*}

%
Observe that 
\begin{align*}
	\E\bigg[ 	\sum_{(\ell, s)} \var(  E_{\ell,s} | \mc{F}_{\prec (\ell, s)} ) \bigg]
	&= \sum_{(\ell, s)} \E E_{\ell, s}^2 
	= \sum_{(\ell, s)} \sigma_{\ell}^2 
	\sum_{r, r' \in [s-1]} 
	\sum_{j, j'} \E[ Z_{\ell j r } Z_{\ell j s } Z_{\ell j' r'} Z_{\ell j' s } ]
	\\&= \sum_{(\ell, s)} \sigma_{\ell}^2 
	\sum_{r \in [s-1]} 
	\sum_{j, j'} \E[ Z_{\ell j r } Z_{\ell j' r}  Z_{\ell j s } Z_{\ell j' s} ]  
	\\&= \sum_{\ell = 1}^n \sigma_\ell^2  \sum_{s \in [N_\ell]} \sum_{r \in [s-1]}
	\E \big( \sum_j Z_{\ell jr } Z_{\ell j s} \big)^2
	\\&= \var(S). 
	\num \label{eqn:condl_var_bias2}
\end{align*}
The last line is obtained noting that $S$ as defined in \eqref{eqn:S_def} is a sum of uncorrelated terms over $(i, r, s)$. 

Similarly, we have 
\begin{align*}
	\E\bigg[ 	\sum_{(\ell, s)} \var(  D_{\ell,s} | \mc{F}_{\prec (\ell, s)} ) \bigg]
	&= \E\bigg[ 	\sum_{(\ell, s)} \E\big[  D_{\ell,s}^2 | \mc{F}_{\prec (\ell, s)} \big] \bigg] =  	\sum_{(\ell, s)} \E\big[  D_{\ell,s}^2  \big] 
	\\&= \sum_{(\ell, s) } \sum_{i \in [\ell - 1] } 
	\sigma_{i,\ell}^2 \var\big( \sum_{r = 1}^{N_i} \sum_j Z_{ijr} Z_{\ell j s}
	\big) 
	\\&=  \sum_{ \ell } \sum_{ i \in [\ell - 1] } \sigma_{i,\ell}^2 \var\big( \sum_{r = 1}^{N_i} \sum_{s = 1}^{N_\ell}  Z_{ijr} Z_{\ell j s } \big) 
	\\&= \var(U). 
	\num \label{eqn:condl_var_bias3}
\end{align*}

The lemma follows by combining  \eqref{eqn:condl_var_bias1}--\eqref{eqn:condl_var_bias3}. 
\qed 

\subsection{Proof of Lemma~\ref{lem:condl_var_U}}

%
%

Let $M_k = n_k \bar{N}_k$ and $M = n\bar{N}$. Define
\begin{align}
	\Sigma = \frac{1}{M} \sum_k M_k \Sigma_k =	\frac{1}{M} \sum_{\ell \in [n]} N_\ell	\Omega_{\ell j_1} \Omega_{\ell j_2} .
	\label{eqn:Sigma_def} 
\end{align}

Our main goal is to control the conditional variance process. Define
\begin{align*}
	\delta_{j j' \ell} = \E Z_{\ell j r}Z_{\ell j'r} =  \begin{cases}
		\Omega_{\ell j}(1 - \Omega_{\ell j}) &\quad \text{ if } j = j' \\
		-	\Omega_{\ell j} \Omega_{\ell j'} &\quad \text{ else}. 
	\end{cases}
	\num \label{eqn:delta_def_ell}
\end{align*}
Observe that
\begin{align*}
	\var(  D_{\ell,s} | \mc{F}_{\prec (\ell, s)} )
	&= \E[   \sum_{i, i' \in [\ell - 1]} \sum_{r, r'} \sum_{j_1, j_2} \sigma_{i\ell} \sigma_{i'\ell} Z_{ij_1r} Z_{\ell j_1 s}
	Z_{i'j_2r'} Z_{\ell j_2 s}  | \mc{F}_{\prec (\ell,s)} ]
	\\&=  \sum_{i, i' \in [\ell - 1]} \sum_{r, r'} \sum_{j_1, j_2} \sigma_{i\ell} \sigma_{i'\ell}
	Z_{ij_1r}Z_{i'j_2r'} \E[ Z_{\ell j_1 s} Z_{\ell j_2 s} ] 
	\\ &= \sum_{i, i' \in [\ell - 1]} \sum_{r, r'}
	\sigma_{i\ell} \sigma_{i'\ell} \sum_{j_1, j_2} \delta_{j_1 j_2 \ell } Z_{ij_1r}Z_{i'j_2r'}
\end{align*}
Define 
\begin{align}
	\alpha_{i i' j_1 j_2} 
	= \sum_{\ell > i'} N_\ell \sigma_{i \ell} \sigma_{i'\ell} 
	\delta_{j_1 j_2 \ell}.
\end{align}
Thus
\begin{align*}
	\sum_{(\ell, s)} \var(  D_{\ell,s} | \mc{F}_{\prec (\ell, s)} ) &= \sum_{\ell,s} \sum_{i , i' \in [\ell - 1]} \sum_{r = 1}^{N_i} \sum_{r' = 1}^{N_{i'}} \sigma_{i \ell} \sigma_{i'\ell} \sum_{j_1, j_2} \delta_{j_1 j_2 \ell} \, Z_{ij_1r} Z_{i'j_2r'} 
	\\ &= \sum_i \sum_{r = 1}^{N_i} \sum_{r' = 1}^{N_{i}} \,  \sum_{j_1, j_2}
	\bigg(
	\sum_{\ell > i} N_\ell \sigma_{i \ell}^2
	\delta_{j_1 j_2 \ell} \, 
	\bigg) Z_{ij_1r} Z_{i'j_2r'} 
	\\ &\quad+  2\sum_{i < i'} \sum_{r = 1}^{N_i} \sum_{r' = 1}^{N_{i'}} \,  \sum_{j_1, j_2}
	\bigg(
	\sum_{\ell > i'} N_\ell \sigma_{i \ell} \sigma_{i'\ell} 
	\delta_{j_1 j_2 \ell} \, 
	\bigg) Z_{ij_1r} Z_{i'j_2r'} 
	\\ &= \sum_i \sum_{r = 1}^{N_i} \sum_{r' = 1}^{N_{i}} \,  \sum_{j_1, j_2}
	\alpha_{i i j_1 j_2}  Z_{ij_1r} Z_{i'j_2r'} 
	\\ &\quad+  2\sum_{i < i'} \sum_{r = 1}^{N_i} \sum_{r' = 1}^{N_{i'}} \,  \sum_{j_1, j_2}
	\alpha_{i i' j_1 j_2} Z_{ij_1r} Z_{i'j_2r'}.
\end{align*}
Define 
\begin{align}
	\label{eqn:zeta_def_genU}
	\zeta_{iri'r'} 
	= \sum_{j_1, j_2} \alpha_{i i' j_1 j_2} Z_{ij_1r} Z_{i'j_2r'}.
\end{align}
Then
\begin{align*}
	\sum_{(\ell, s)} \var(  D_{\ell,s} | \mc{F}_{\prec (\ell, s)} ) &=
	\sum_i \sum_{r \in [N_i]} \zeta_{irir}
	+ \bigg(  2 \sum_i \sum_{r < r' \in [N_i]} \zeta_{irir'}
	+ 2 \sum_{i < i'} \sum_{r = 1}^{N_i} \sum_{r' = 1}^{N_{i'}} \zeta_{iri'r'} \bigg) 
	\\&=: V_1 + V_2 
\end{align*}

With this decomposition, Lemma \ref{lem:condl_var_U} follows directly from Lemmas \ref{lem:V1} and \ref{lem:V2} stated below and proved in the next remainder of this subsection. 

\begin{lemma}
	\label{lem:V1}	It holds that
	\begin{align*}
		\var(V_1)  \lesssim \big( \sum_k \frac{1}{M_k} \big) \| \mu \|_3^3. 
	\end{align*}
\end{lemma}

\begin{lemma}
	\label{lem:V2}
	It holds that
	\begin{align*}
		\var(V_2) &\lesssim K \| \mu \|_4^4
	\end{align*}
\end{lemma}

\qed 

\subsubsection{Statement and proof of Lemma \ref{lem:alpha_bounds}} 

The proofs of Lemmas \ref{lem:V1} and \ref{lem:V2} heavily rely on the following intermediate result that  bounds the coefficients $\alpha_{i i' j_1 j_2}$ in all cases. 

\begin{lemma}
	\label{lem:alpha_bounds} 
	It holds that
	\begin{align*}
		\alpha_{i i' j_1 j_2}
		&\lesssim \begin{cases}
			\frac{1}{M_k} \mu_{ j_1} &\quad \text{ if } i, i' \in S_k, j_1 = j_2
			\\ \frac{1}{M_k} \Sigma_{k j_1 j_2}
			+\frac{1}{M} \Sigma_{ j_1 j_2} &\quad \text{ if } i, i' \in S_k, j_1 \neq j_2
			\\ 	\frac{1}{M} \mu_{ j_1} 	&\quad \text{ if } i \in S_{ k_1}, i' \in S_{k_2},  k_1 \neq k_2, j_1 = j_2
			\\ \frac{1}{M} \sum_{a = 1}^2 \Sigma_{k_a j_1 j_2} 
			+ \frac{1}{M} \Sigma_{j_1 j_2}
			&\quad \text{ if } i \in S_{ k_1}, i' \in S_{k_2} , k_1 \neq k_2, j_1 \neq j_2
		\end{cases}
	\end{align*}
\end{lemma}

\begin{proof} 
	
	If $j_1 = j_2$ and $i, i' \in S_{k}$, we have 
	\begin{align*}
		|\alpha_{i i' j_1 j_1} |
		&= | \sum_{\ell > i'} N_\ell \sigma_{i \ell} \sigma_{i'\ell}
		\delta_{j_1 j_1 \ell} | 
		\leq \sum_{k' = 1}^K \sum_{\ell  \in S_{k'} } N_\ell \sigma_{i \ell} \sigma_{i'\ell}
		\delta_{j_1 j_1 \ell }
		\\&	\lesssim \frac{1}{M_k} \cdot  \frac{1}{M_k} \sum_{\ell \in S_k} N_\ell \Omega_{\ell j_1} 
		+ \frac{1}{M} \cdot \frac{1}{M} \sum_{\ell \in [n]} N_\ell \Omega_{\ell j_1}
		\lesssim \frac{1}{M_k} \mu_{ j_1} + \frac{1}{M} \mu_{j_1} 
		\lesssim  \frac{1}{M_k} \mu_{ j_1} . 
	\end{align*}
	
	If $j_1 \neq j_2$ and $i, i' \in S_{k}$, we have 
	\begin{align*}
		|\alpha_{i i' j_1 j_2} |
		&=| \sum_{\ell > i'} N_\ell \sigma_{i \ell} \sigma_{i'\ell}
		\delta_{j_1 j_2 \ell} |
		\leq \sum_{\ell \in [n]} N_\ell |\sigma_{i \ell} \sigma_{i'\ell}|
		\Omega_{\ell j_1} \Omega_{\ell j_2} 
		\\&\lesssim \frac{1}{M_k} \cdot \frac{1}{M_k} \sum_{\ell \in S_k} N_\ell 		\Omega_{\ell j_1} \Omega_{\ell j_2} 
		+ 	 \frac{1}{M} \cdot \frac{1}{M} \sum_{\ell \in [n]} N_\ell	\Omega_{\ell j_1} \Omega_{\ell j_2} 
		\lesssim \frac{1}{M_k} \Sigma_{k j_1 j_2}
		+\frac{1}{M} \Sigma_{ j_1 j_2}. 
	\end{align*}
	
	
	%
	%
	
	If $i \neq i'$, $j_1 = j_2$, and $i \in S_{k_1}, i' \in S_{k_2}$ where $k_1 \neq k_2$, we have
	\begin{align*}
		|	\alpha_{i i' j_1 j_1} |
		&= |\sum_{\ell > i'} N_\ell \sigma_{i \ell} \sigma_{i'\ell} 
		\delta_{j_1 j_1 \ell}|
		\leq \sum_\ell N_\ell |\sigma_{i \ell} \sigma_{i'\ell} | \Omega_{\ell j_1} 
		\\& \lesssim \frac{1}{M} \cdot \sum_{a = 1}^2 \frac{1}{M_{k_a}} \sum_{\ell \in S_{k_a}} N_\ell \Omega_{\ell j_1} 
		+\frac{1}{M} \cdot \frac{1}{M} \sum_{\ell \in [n]} N_\ell \Omega_{\ell j_1} 
		= \frac{3}{M} \mu_{j_1}. 
	\end{align*}
	
	If $i \neq i'$, $j_1 \neq j_2$, and $i \in S_{k_1}, i' \in S_{k_2}$ where $k_1 \neq k_2$, we have
	\begin{align*}
		|	\alpha_{i i' j_1 j_2}  |
		&= | \sum_{\ell > i'} N_\ell \sigma_{i \ell} \sigma_{i'\ell} 
		\delta_{j_1 j_2 \ell} |
		\lesssim \sum_\ell  N_\ell \sigma_{i \ell} \sigma_{i'\ell} 
		\Omega_{\ell j_1} \Omega_{\ell j_2} 
		\\&\lesssim  \frac{1}{M} \cdot \sum_{a = 1}^2 \frac{1}{M_{k_a}} \sum_{\ell \in S_{k_a}} N_\ell \Omega_{\ell j_1} \Omega_{\ell j_2} 
		+\frac{1}{M} \cdot \frac{1}{M} \sum_{\ell \in [n]} N_\ell \Omega_{\ell j_1} \Omega_{\ell j_2}  
		\\&\leq   \frac{1}{M} \sum_{a = 1}^2 \Sigma_{k_a j_1 j_2}
		+ \frac{1}{M} \Sigma_{j_1 j_2}. 
	\end{align*}
	
\end{proof} 

%
%
%
%

\subsubsection{Proof of Lemma \ref{lem:V1}}

We have
\begin{align*}
	\var(V_1) 
	&= \sum_{i,r} \E \zeta_{irir}^2. 
\end{align*}
Next by symmetry, 
\begin{align*}
	\E \zeta_{irir}^2 
	&= \sum_{j_1, j_2 ,j_3, j_4} \alpha_{i i j_1 j_2} 
	\alpha_{i i j_3 j_4} \, \E Z_{ij_1r} Z_{i j_3 r}    Z_{ij_2r} Z_{i j_4 r}
	\\ &\lesssim \sum_{j_1} \alpha_{i i j_1 j_1}^2 \Omega_{i j_1} 
	+ \sum_{j_1 \neq j_4} \alpha_{i i j_1 j_1} 
	\alpha_{i i j_1 j_4} \, \Omega_{i j_1} \Omega_{i j_4} 
	\\&\quad +  \sum_{j_1 \neq j_3} \alpha_{i i j_1 j_1} 
	\alpha_{i i j_3 j_3} \, \Omega_{i j_1} \Omega_{i j_3} 
	+ \sum_{j_1 \neq j_2}  \alpha_{i i j_1 j_2}^2 \,  \Omega_{i j_1} \Omega_{i j_2} 
	\\&\quad + \sum_{j_1, j_3, j_4 (dist.)} \alpha_{i i j_1 j_1} 
	\alpha_{i i j_3 j_4} \, \Omega_{i j_1} \Omega_{i j_3} \Omega_{i j_4} 
	+ \sum_{j_1, j_2, j_4 (dist.)}  \alpha_{i i j_1 j_2} 
	\alpha_{i i j_1 j_4} \Omega_{i j_1} \Omega_{i j_2} \Omega_{i j_4}
	\\& \quad + \sum_{j_1, j_2, j_3, j_4 (dist.)} 
	\alpha_{i i j_1 j_2} 
	\alpha_{i i j_3 j_4} \Omega_{ij_1} 	\Omega_{ij_2}  \Omega_{ij_3} 
	\Omega_{ij_4} =: \sum_{a = 1}^7 B_{a, i,r}
\end{align*}
Thus 
\begin{align*}
	\var(V_1) &\lesssim \sum_a \bigg( \underbrace{\sum_{i,r} B_{a, i, r} }_{ =:  B_{a} } \bigg). 
\end{align*}
We analyze $B_1$-- $B_7$ separately, bounding the $\alpha_{ii'j_r j_s}$ coefficients using Lemma \ref{lem:alpha_bounds}.

For $B_1$, 
\begin{align*}
	B_1 &\lesssim \sum_{i,r}  \sum_{j_1} \alpha_{iij_1j_2}^2 \Omega_{i j_1} 
	\lesssim \sum_{k = 1}^k \sum_{i \in S_k} \sum_{r \in [N_i]} \sum_{j_1} (	\frac{1}{M_k} \mu_{ j_1})^2 \Omega_{i j_1} 
	\\&	\lesssim \sum_k  	\sum_{j_1} (	\frac{1}{M_k} \mu_{ j_1})^2  
	M_k \mu_{j_1} 
	\lesssim \big(\sum_k \frac{1}{M_k} \big)  \| \mu \|_3^3. 
	\num \label{eqn:B1_bd}
\end{align*}

For $B_2$, 
\begin{align*}
	B_2 &\lesssim \sum_{i,r}  \sum_{j_1 \neq j_4} \alpha_{i i j_1 j_1} 
	\alpha_{i i j_1 j_4} \, \Omega_{i j_1} \Omega_{i j_4} 
	\\&\lesssim \sum_k \sum_{i \in S_k} \sum_{r \in [N_i]} 
	\sum_{j_1 \neq j_4} 
	\frac{1}{M_k} \mu_{ j_1}  \cdot \big( \frac{1}{M_k} \Sigma_{k j_1 j_4} +\frac{1}{M} \Sigma_{ j_1 j_4} \big) \cdot 
	\Omega_{i j_1} \Omega_{i j_4} 
	\\&\lesssim 
	\sum_k \sum_{j_1 \neq j_4}  \frac{1}{M_k} \mu_{ j_1}  \cdot \big( \frac{1}{M_k} \Sigma_{k j_1 j_4} +\frac{1}{M} \Sigma_{ j_1 j_4} \big) \cdot M_k \Sigma_{k j_1 j_4} 
	\\&\lesssim  \sum_k \frac{1}{M_k}  \sum_{j_1 \neq j_4} \Sigma_{k j_1 j_4}^2 \mu_{j_1} 
	+ \sum_k \frac{1}{M} \sum_{j_1 \neq j_4}  \Sigma_{k j_1 j_4} \Sigma_{ j_1 j_4} \mu_{j_1} 
	\\&\lesssim \sum_k \frac{  \mathbf{1}' \Sigma_{k}^{\circ 2} \mu    }{M_k}
	+ \sum_k \frac{   \mathbf{1}'   ( \Sigma_k \circ \Sigma   ) \mu }{ M}
	= \sum_k \frac{  \mathbf{1}' \Sigma_{k}^{\circ 2} \mu    }{M_k}
\end{align*}

Next,
\begin{align*}
	\sum_{j_1 \neq j_4} \Sigma_{k j_1 j_4}^2 \mu_{j_1} 
	&= \sum_{j_1 \neq j_4} \frac{1}{M_k^2} \sum_{i, i' \in S_k} N_i N_{i'}
	\Omega_{i j_1} \Omega_{i' j_1} \Omega_{ i j_4 } \Omega_{i' j_4} \cdot \mu_{j_1} 
	\\&\leq \sum_{j_1 } \frac{1}{M_k^2} \sum_{i, i' \in S_k} N_i N_{i'}
	\Omega_{i j_1} \Omega_{i' j_1}  \mu_{j_1} \cdot \big( \sum_{j_4} \Omega_{ i j_4 } \Omega_{i' j_4}  \big) 
	\\& \leq \sum_{j_1 } \frac{1}{M_k^2} \sum_{i, i' \in S_k} N_i N_{i'}
	\Omega_{i j_1} \Omega_{i' j_1} \cdot \mu_{j_1} 
	\\&\leq \sum_{j_1} \mu_{j_1}^3 = \| \mu \|_3^3, 
	\num \label{eqn:B21_bd}
\end{align*}
and similarly 
\begin{align*}
	\sum_{j_1 \neq j_4}  \Sigma_{k j_1 j_4} \Sigma_{ j_1 j_4} \mu_{j_1} 
	&= \sum_{j_1 \neq j_4} \frac{1}{M_k M} \sum_{i \in S_k, i' \in [n]} N_i N_{i'} \Omega_{i j_1} \Omega_{i' j_1} \Omega_{ i j_4 } \Omega_{i' j_4} \cdot \mu_{j_1} 
	\\&\leq \sum_{j_1 } \frac{1}{M_k M} \sum_{i \in S_k, i' \in [n]} N_i N_{i'} \Omega_{i j_1} \Omega_{i' j_1}  \mu_{j_1} 
	\\&= \sum_{j_1} \mu_{j_1}^3 = \| \mu \|_3^3. 
\end{align*}
Thus
\begin{align}
	B_2 \lesssim \big(\sum_k \frac{1}{M_k} \big)  \| \mu \|_3^3. 
	\label{eqn:B2_bd}
\end{align}

For $B_3$,
\begin{align*}
	B_3 &\lesssim  \sum_{i,r}  \sum_{j_1 \neq j_3} \alpha_{i i j_1 j_1} 
	\alpha_{i i j_3 j_3} \, \Omega_{i j_1} \Omega_{i j_3} 
	\\&\lesssim  \sum_{k} \sum_{i \in S_k} \sum_{r \in [N_i]} \sum_{j_1 \neq j_3}
	\frac{1}{M_k} \mu_{ j_1}  \cdot  	\frac{1}{M_k} \mu_{ j_3}   \cdot 
	\Omega_{i j_1} \Omega_{i j_3}  
	\\&\lesssim \sum_{k} \sum_{j_1 \neq j_3}
	\frac{1}{M_k} \mu_{ j_1}  \cdot  	\frac{1}{M_k} \mu_{ j_3}   \cdot 
	M_k \Sigma_{k j_1 j_3} \lesssim \sum_k \frac{  \mu' \Sigma_k \mu  }{M_k}. 
\end{align*}
We have by Cauchy-Schwarz,
\begin{align*}
	\mu' \Sigma_k \mu &= \frac{1}{M_k} \sum_{i \in S_k} N_i \mu ' \Omega_i \Omega_{i'}' \mu
	\\&= \frac{1}{M_k} \sum_{i \in S_k} N_i \big(  \sum_j \mu_j \Omega_{ij}  \big)^2
	\\&\leq  \frac{1}{M_k} \sum_{i \in S_k} N_i \big( \sum_j \Omega_{ij} \big) \big( \sum_j \mu_j^2 \Omega_{ij}  \big)
	\\&= \sum_j \mu_j^3 = \| \mu \|_3^3. \num \label{eqn:mu_Sigmak_quad}
\end{align*}
Thus
\begin{align*}
	B_3 \lesssim  \big(\sum_k \frac{1}{M_k} \big)  \| \mu \|_3^3
	\num \label{eqn:B3_bd} 
\end{align*}

For $B_4$,
\begin{align*}
	B_4 &\lesssim  \sum_{i,r} \sum_{j_1 \neq j_2}  \alpha_{i i j_1 j_2}^2 \,  \Omega_{i j_1} \Omega_{i j_2} 
	\lesssim \sum_{k} \sum_{i \in S_k} \sum_{r \in [N_i]} \sum_{j_1 \neq j_2}  
	\big(  \frac{1}{M_k} \Sigma_{k j_1 j_2} +\frac{1}{M} \Sigma_{ j_1 j_2} \big)^2
	\,  \Omega_{i j_1} \Omega_{i j_2} 
	\\&\lesssim \sum_{k} \sum_{j_1 \neq j_2}  
	\big(  \frac{1}{M_k} \Sigma_{k j_1 j_2} +\frac{1}{M} \Sigma_{ j_1 j_2} \big)^2
	\,  \cdot M_k \Sigma_{k j_1 j_2} 
	\lesssim \sum_k \frac{ {\bf{1}}' ( \Sigma_k^{\circ 3} ){\bf{1}}  }{M_k} 
	+ \sum_k \frac{M_k}{ M^2} {\bf{1}}' ( \Sigma_k\circ \Sigma^{\circ 2}  ){\bf{1}}
	\\&\lesssim \big( \sum_k \frac{ {\bf{1}}' ( \Sigma_k^{\circ 3} ){\bf{1}}  }{M_k} \big) 
	+ \frac{1}{M} \mf{1}' (\Sigma^{\circ 3}) \mf{1}. 
\end{align*}
First, 
\begin{align*}
	{\bf{1}}' ( \Sigma_k^{\circ 3} ){\bf{1}}
	&= \frac{1}{M_k^3} \sum_{i_1, i_2, i_3 \in S_k } N_{i_1} N_{i_2} N_{i_3}
	\big( \sum_j \Omega_{i_1 j} \Omega_{i_2 j} \Omega_{i_3 j} \big)^2 
	\\&\leq \frac{1}{M_k^3} \sum_{i_1, i_2, i_3 \in S_k } N_{i_1} N_{i_2} N_{i_3} \cdot
	\sum_j \Omega_{i_1 j} \Omega_{i_2 j} \Omega_{i_3 j} = \sum_j \mu_j^3 = \| \mu \|_3^3,
\end{align*}
and similarly,
\begin{align*}
	{\bf{1}}' ( \Sigma^{\circ 3} ){\bf{1}}
	= \frac{1}{M^3} \sum_{i_1, i_2, i_3 \in [n] } N_{i_1} N_{i_2} N_{i_3}	\big( \sum_j \Omega_{i_1 j} \Omega_{i_2 j} \Omega_{i_3 j} \big)^2 
	\leq \| \mu \|_3^3. 
\end{align*}
Thus
\begin{align*}
	B_4 \lesssim  \big(\sum_k \frac{1}{M_k} \big)  \| \mu \|_3^3
	\num \label{eqn:B4_bd} 
\end{align*}

For $B_5$, 
\begin{align*}
	B_5 &\lesssim  \sum_{i,r} \sum_{j_1, j_3, j_4 (dist.)} \alpha_{i i j_1 j_1} 
	\alpha_{i i j_3 j_4} \, \Omega_{i j_1} \Omega_{i j_3} \Omega_{i j_4}  
	\\&\lesssim  \sum_{k} \sum_{i \in S_k} N_i \sum_{j_1, j_3, j_4 } 
	\frac{1}{M_k} \mu_{ j_1} \cdot 
	(	\frac{1}{M_k} \Sigma_{k j_3 j_4} +\frac{1}{M} \Sigma_{ j_3 j_4}	)\cdot
	\, \Omega_{i j_1} \Omega_{i j_3} \Omega_{i j_4} 
	\\&\les \sum_k \sum_{i \in S_k} \sum_{j_1, j_3, j_4 }  \frac{N_i \mu_{ j_1} \Sigma_{k j_3 j_4} \Omega_{i j_1} \Omega_{i j_3} \Omega_{i j_4} }{M_k^2} 
	+ \sum_k \sum_{i \in S_k} \sum_{j_1, j_3, j_4 }  \frac{N_i \mu_{ j_1} \Sigma_{j_3 j_4} \Omega_{i j_1} \Omega_{i j_3} \Omega_{i j_4} }{M_k M} 
	\\&=: B_{51} + B_{52}.
\end{align*}
We have
\begin{align*}
	B_{51}	&= \sum_k \frac{1}{M_k^3}   \sum_{i_1, i_2 \in S_k} \sum_{j_1, j_3, j_4 } N_{i_1} N_{i_2}  
	\mu_{j_1} \Omega_{i_1 j_1} \Omega_{i_1 j_3} \Omega_{i_2 j_3} \Omega_{i_1 j_4} 
	\Omega_{i_2 j_4} 
	\\&= \sum_k \frac{1}{M_k^3} \sum_{i_1, i_2 \in S_k} N_{i_1} N_{i_2} (\Omega_{i_1}' \mu) \cdot ( \Omega_{i_1}' \Omega_{i_2} )^2 
	\\&\leq \sum_k \frac{1}{M_k^3} \sum_{i_1, i_2 \in S_k} N_{i_1} N_{i_2} \cdot \Omega_{i_1}' \mu \cdot  \Omega_{i_1}' \Omega_{i_2} 
	\\&= \sum_k \frac{1}{M_k^2} \sum_{i_1} N_{i_1} \mu' \Omega_{i_1} \Omega_{i_1}' \mu 
	= \frac{1}{M_k} \mu' \Sigma_k \mu 
	\leq \sum_k \frac{1}{M_k} \| \mu \|_3^3. 
	\num \label{eqn:B51_bd} 
\end{align*}
In the last line we apply \eqref{eqn:mu_Sigmak_quad}. Similarly,
\begin{align*}
	B_{52} &=\sum_k \frac{1}{M_k M^2} \sum_{i_1 \in S_k, i_2 \in [n] } \sum_{j_1, j_3, j_4}
	N_{i_1} N_{i_2} 
	\mu_{j_1} \Omega_{i_1 j_1} \Omega_{i_1 j_3} \Omega_{i_2 j_3} \Omega_{i_1 j_4} \Omega_{i_2 j_4}
	\\&\leq \sum_k \frac{1}{M_k M^2} \sum_{i_1 \in S_k, i_2 \in [n] } N_{i_1} N_{i_2}  \cdot \Omega_{i_1}' \mu \cdot  \Omega_{i_1}' \Omega_{i_2} 
	\\&\leq \sum_k \frac{1}{M_k M} \sum_{i_1 \in S_k} N_{i_1} \mu' \Omega_{i_1} \Omega_{i_1}' \mu 
	\leq \sum_k \frac{1}{M} \| \mu \|_3^3. 
	\num \label{eqn:B52_bd} 
\end{align*}
Thus 
\begin{align*}
	B_5 \lesssim  \big(\sum_k \frac{1}{M_k} \big)  \| \mu \|_3^3
	\num \label{eqn:B5_bd}.  
\end{align*}

For $B_6$, 
\begin{align*}
	B_6 &\lesssim \sum_{k} \sum_{i \in S_k} \sum_{r  \in [N_i]}  \sum_{j_1, j_2, j_4 (dist.)}  \big( \frac{1}{M_k} \Sigma_{k j_1 j_2} +\frac{1}{M} \Sigma_{ j_1 j_2} \big) 
	\big( \frac{1}{M_k} \Sigma_{k j_1 j_4} +\frac{1}{M} \Sigma_{ j_1 j_4} \big) \Omega_{i j_1} \Omega_{i j_2} \Omega_{i j_4}
	\\&\lesssim \sum_{k} \sum_{i \in S_k} \sum_{r  \in [N_i]} \sum_{j_1, j_2, j_4 } \frac{\Sigma_{k j_1 j_2}^2  \Omega_{i j_1} \Omega_{i j_2} \Omega_{i j_4}}{M_k^2} 
	+ 2 \sum_{k} \sum_{i \in S_k} \sum_{r  \in [N_i]} \sum_{j_1, j_2, j_4 }  \frac{\Sigma_{k j_1 j_2} \Sigma_{j_1 j_2}  \Omega_{i j_1} \Omega_{i j_2} \Omega_{i j_4}}{M_k M} 
	\\&\quad  + \sum_{k} \sum_{i \in S_k} \sum_{r  \in [N_i]} \sum_{j_1, j_2, j_4 }  \frac{\Sigma_{ j_1 j_2}^2  \Omega_{i j_1} \Omega_{i j_2} \Omega_{i j_4}}{M^2} 
	=: B_{61} + B_{62} + B_{63}. 
\end{align*}
First,
\begin{align*}
	B_{61} &\leq \sum_{k} \sum_{i \in S_k} \sum_{r  \in [N_i]} \sum_{j_1, j_2, j_4 } \frac{\Sigma_{k j_1 j_2}^2  \Omega_{i j_1} }{M_k^2} 
	= \sum_k \frac{1}{M_k} \mf{1}' \Sigma_k^{\circ 2} \mu 
	\leq  \sum_k  \frac{1}{M_k} \| \mu \|_3^3, 
\end{align*}
where we applied \eqref{eqn:B21_bd}. Similarly, 
\begin{align*}
	B_{62} &\lesssim \sum_k  \frac{1}{M_k} \| \mu \|_3^3, \text{ and} \\
	B_{63} &\lesssim \sum_k  \frac{1}{M_k} \| \mu \|_3^3. 
\end{align*}
Thus
\begin{align*}
	B_6 \lesssim  \big(\sum_k \frac{1}{M_k} \big)  \| \mu \|_3^3
	\num \label{eqn:B6_bd}.  
\end{align*}

For $B_7$, we have
\begin{align*}
	B_7 &\lesssim 
	\sum_{j_1, j_2, j_3, j_4 (dist.)} 
	\big( \frac{1}{M_k} \Sigma_{k j_1 j_2} +\frac{1}{M} \Sigma_{ j_1 j_2} \big) 
	\big( \frac{1}{M_k} \Sigma_{k j_3 j_4} +\frac{1}{M} \Sigma_{ j_3 j_4} \big) \Omega_{ij_1} 	\Omega_{ij_2}  \Omega_{ij_3} 
	\Omega_{ij_4}
	\\&\lesssim \sum_{k} \sum_{i \in S_k} \sum_{r  \in [N_i]} \sum_{j_1, j_2, j_3, j_4 } \frac{\Sigma_{k j_1 j_2} \Sigma_{k j_3 j_4}   \Omega_{i j_1} \Omega_{i j_2} \Omega_{i j_3} \Omega_{i j_4}}{M_k^2} 
	\\&\quad + 2 \sum_{k} \sum_{i \in S_k} \sum_{r  \in [N_i]} \sum_{j_1, j_2,j_3, j_4 }  \frac{\Sigma_{k j_1 j_2} \Sigma_{j_3 j_4}  \Omega_{i j_1} \Omega_{i j_2} \Omega_{i j_3}  \Omega_{i j_4}}{M_k M} 
	\\&\quad  + \sum_{k} \sum_{i \in S_k} \sum_{r  \in [N_i]} \sum_{j_1, j_2,j_3, j_4 }  \frac{\Sigma_{ j_1 j_2} \Sigma_{j_3 j_4}  \Omega_{i j_1} \Omega_{i j_2} \Omega_{i j_3}  \Omega_{i j_4}}{M^2} 
	=: B_{71} + B_{72} + B_{73}.
\end{align*}
Note that
\begin{align*}
	\Sigma_{kj_1 j_2} 
	&= \frac{1}{M_k} \sum_{i \in S_k} N_i \Omega_{i j_1} \Omega_{i j_2} 
	\leq  \frac{1}{M_k} \sum_{i \in S_k} N_i \Omega_{i j_1} = \mu_{j_1}, \text{ and} 
	\\ \Sigma_{j_1 j_2} 
	&= \frac{1}{M} \sum_{i \in [n]} N_i \Omega_{i j_1} \Omega_{i j_2} 
	\leq  \frac{1}{M} \sum_{i \in [n]} N_i \Omega_{i j_1} = \mu_{j_1}.
	\num \label{eqn:Sigma_bd}
\end{align*}
Thus
\begin{align*}
	B_{71} &\leq 
	\sum_{k} \sum_{i \in S_k} \sum_{r  \in [N_i]} \sum_{j_1, j_2,j_3, j_4 } \frac{ \mu_{j_1} \Sigma_{k j_3 j_4}   \Omega_{i j_1} \Omega_{i j_2} \Omega_{i j_3} \Omega_{i j_4}}{M_k^2} 
	\\&\leq \sum_k \sum_{i \in S_k} \sum_{j_1,j_3, j_4 } \frac{N_i \mu_{j_1} \Sigma_{k j_3 j_4}   \Omega_{i j_1}  \Omega_{i j_3} \Omega_{i j_4}}{M_k^2} 
	\leq \sum_k \frac{1}{M_k} \| \mu \|_3^3
\end{align*}
where we applied \eqref{eqn:B51_bd}. Similarly, 
\begin{align*}
	B_{72} &\lesssim \sum_k  \frac{1}{M_k} \| \mu \|_3^3, \text{ and} \\
	B_{73} &\lesssim \sum_k  \frac{1}{M_k} \| \mu \|_3^3. 
\end{align*}
Thus
\begin{align*}
	B_7 \lesssim  \big(\sum_k \frac{1}{M_k} \big)  \| \mu \|_3^3
	\num \label{eqn:B7_bd}.  
\end{align*}
Combining the results for $B_1$--$B_7$ concludes the proof. 

\qed

\subsubsection{Proof of Lemma~\ref{lem:V2}}

We have
\begin{align*}
	\var(V_2 ) \lesssim  4 \sum_{(i,r) \neq (i', r')} \E \zeta_{irir'}^2,
\end{align*}
where $r \in [N_i]$ and $r \in [N_{i'}]$ in the summation above. 

By symmetry, if $(i,r) \neq (i',r')$, 
\begin{align*}
	\E & \zeta_{iri'r'}^2 
	=  \sum_{j_1, j_2 ,j_3, j_4} \alpha_{i i' j_1 j_2} 
	\alpha_{i i' j_3 j_4} \, \E Z_{ij_1r} Z_{i j_3 r}  \, \E  Z_{i'j_2r'} Z_{i' j_4 r'}
	\\&\lesssim   \sum_{j_1} \alpha_{i i' j_1 j_1}^2 
	\, \Omega_{i j_1} \Omega_{i' j_1} 
	+ \sum_{j_1 \neq j_4} \alpha_{i i' j_1 j_1} 
	\alpha_{i i' j_1 j_4} \, \Omega_{i j_1} \Omega_{i' j_1}  \Omega_{i' j_4} 
	\\& \quad +  \sum_{j_1 \neq j_3} \alpha_{i i' j_1 j_1} 
	\alpha_{i i' j_3 j_3} \, \Omega_{ij_1}  \Omega_{ij_3} \Omega_{i'j_1}  \Omega_{i'j_3}
	+ \sum_{j_1 \neq j_2}  \alpha_{i i' j_1 j_2}^2 \, \Omega_{i j_1} \Omega_{i' j_2}
	\\&\quad + \sum_{j_1, j_3, j_4 (dist.)} \alpha_{i i' j_1 j_1} 
	\alpha_{i i' j_3 j_4} \, \Omega_{i j_1} \Omega_{i j_3} \Omega_{i'j_1} 
	\Omega_{i' j_4} 
	+ \sum_{j_1, j_2, j_4 (dist.)}  \alpha_{i i' j_1 j_2} 
	\alpha_{i i' j_1 j_4} \, \Omega_{i j_1} \Omega_{i' j_2} \Omega_{i' j_4} 
	\\& \quad + \sum_{j_1, j_2, j_3, j_4 (dist.)} 
	\alpha_{i i' j_1 j_2} 
	\alpha_{i i' j_3 j_4} \Omega_{ij_1} \Omega_{ij_3} 
	\Omega_{i'j_2} \Omega_{i'j_4} =: \sum_a^7 C_{a,i,r}. 
	\num \label{eqn:V2_expansion2} 
\end{align*}
Thus
\begin{align*}
	\var(V_2) \lesssim \sum_{a = 1}^7 \sum_{(i,r) \neq (i', r')} C_{a,i,r} 
	\lesssim \sum_{a = 1}^7 \underbrace{ \sum_{i,i'} N_i N_{i'} C_{a,i,r} }_{=: C_a}.  
\end{align*}
Next we analyze $C_1, \ldots, C_7$, bounding the $\alpha_{ii'j_r j_s}$ coefficients using Lemma \ref{lem:alpha_bounds}.

For $C_1$, 
\begin{align*}
	C_1 &\lesssim \sum_k \sum_{i, i' \in S_k} \sum_{j_1}
	N_i N_{i'} \alpha_{ii'j_1j_1}^2 \Omega_{i j_1} \Omega_{i' j_1} 
	+ \sum_{k \neq k'} \sum_{i \in S_k, i' \in S_{k'}}   \sum_{j_1}
	N_i N_{i'} \alpha_{ii'j_1j_1}^2 \Omega_{i j_1} \Omega_{i' j_1} 
	\\&\lesssim \sum_k \sum_{i, i' \in S_k}  \sum_{j_1}
	N_i N_{i'} (	\frac{1}{M_k} \mu_{ j_1} )^2 \Omega_{i j_1} \Omega_{i j_1} 
	+ \sum_{k \neq k'} \sum_{i \in S_k, i' \in S_{k'}}   \sum_{j_1}
	(\frac{1}{M} \mu_{ j_1} 	)^2 \Omega_{i j_1} \Omega_{i' j_1} 
	\\&\lesssim \sum_k \sum_{j_1} \mu_{ j_1}^4 
	+ \sum_{k \neq k'} \sum_{j_1} \frac{M_k M_{k'}}{M^2} \mu_{j_1}^4
	\lesssim K \| \mu \|_4^4. 
	\num \label{eqn:C1_bd}
\end{align*}

For $C_2$,
\begin{align*}
	C_2 &\lesssim \sum_k \sum_{i, i' \in S_k} N_i N_{i'}  \sum_{j_1 \neq j_4} \alpha_{i i' j_1 j_1} 
	\alpha_{i i' j_1 j_4} \, \Omega_{i j_1} \Omega_{i' j_1}  \Omega_{i' j_4}
	\\&\quad + \sum_{k \neq k'} \sum_{i \in S_k, i' \in S_{k'}} N_i N_{i'}  \sum_{j_1 \neq j_4} \alpha_{i i' j_1 j_1} 
	\alpha_{i i' j_1 j_4} \, \Omega_{i j_1} \Omega_{i' j_1}  \Omega_{i' j_4}
	\\&\lesssim \sum_k \sum_{i, i' \in S_k} N_i N_{i'}  \sum_{j_1 \neq j_4} 
	\frac{1}{M_k} \mu_{ j_1} 	\cdot 
	\big( \frac{1}{M_k} \Sigma_{k j_1 j_4} +\frac{1}{M} \Sigma_{ j_1 j_4} \big) \, \Omega_{i j_1} \Omega_{i' j_1}  \Omega_{i' j_4}
	\\&\quad + \sum_{k \neq k'} \sum_{i \in S_k, i' \in S_{k'}} N_i N_{i'}  \sum_{j_1 \neq j_4} 	\frac{1}{M} \mu_{ j_1} 	\cdot 
	\big( \frac{1}{M} \sum_{a \in \{k, k'\}} \Sigma_{a j_1 j_4} + \frac{1}{M} \Sigma_{j_1 j_4} \big) \, \Omega_{i j_1} \Omega_{i' j_1}  \Omega_{i' j_4} 
	\\&\lesssim 
	\sum_k \sum_{i, i' \in S_k} N_i N_{i'}  \sum_{j_1 \neq j_4} 
	\frac{1}{M_k} \mu_{ j_1} 	\cdot 
	\big( \frac{1}{M_k} \mu_{j_1} +\frac{1}{M} \mu_{j_1} \big) \, \Omega_{i j_1} \Omega_{i' j_1}  \Omega_{i' j_4} 
	\\&\quad + \sum_{k \neq k'} \sum_{i \in S_k, i' \in S_{k'}} N_i N_{i'}  \sum_{j_1 \neq j_4} 	\frac{1}{M} \mu_{ j_1} 	\cdot 
	\big( \frac{2}{M} \mu_{j_1} + \frac{1}{M} \mu_{j_1} \big) \, \Omega_{i j_1} \Omega_{i' j_1}  \Omega_{i' j_4} 
	\\&\lesssim \sum_k \sum_{j_1} \big( \mu_{j_1}^4 + \frac{M_k}{M} \mu_{ j_1}^4 \big) 
	+ \sum_{k \neq k'} \sum_{j_1} \frac{M_k M_{k'}}{M^2} \mu_{ j_1}^4 
	\lesssim K \| \mu \|_4^4. 
	\num \label{eqn:C2_bd}
\end{align*}
where we applied \eqref{eqn:Sigma_bd}. 

For $C_3$,
\begin{align*}
	C_3 &\lesssim \bigg(\sum_k \sum_{i, i' \in S_k} N_i N_{i'} + \sum_{k \neq k'} \sum_{i \in S_k, i' \in S_{k'}} N_i N_{i'}  \bigg) 
	\sum_{j_1 \neq j_3} \alpha_{i i' j_1 j_1} 
	\alpha_{i i' j_3 j_3} \, \Omega_{ij_1}  \Omega_{ij_3} \Omega_{i'j_1}  \Omega_{i'j_3}
	\\&\lesssim  \sum_k \sum_{i, i' \in S_k} N_i N_{i'}  \sum_{j_1 \neq j_3}
	\frac{1}{M_k} \mu_{ j_1} \cdot 
	\frac{1}{M_k} \mu_{ j_3}  \cdot  \, \Omega_{ij_1}  \Omega_{ij_3} \Omega_{i'j_1}  \Omega_{i'j_3}
	\\&\quad +  \sum_{k \neq k'} \sum_{i \in S_k, i' \in S_{k'}}  N_i N_{i'} 
	\sum_{j_1 \neq j_3} \frac{1}{M} \mu_{ j_1} 	\cdot 
	\frac{1}{M} \mu_{ j_3} 	 \, \Omega_{ij_1}  \Omega_{ij_3} \Omega_{i'j_1}  \Omega_{i'j_3}
	\\&= \sum_k \sum_{j_1 \neq j_3} \mu_{j_1} \mu_{j_3} \Sigma_{k j_1 j_3}^2 +  \sum_{k \neq k'} 
	\sum_{j_1 \neq j_3} \frac{M_k M_{k'}}{M^2} \mu_{j_1} \mu_{j_3} \Sigma_{k j_1 j_3} \Sigma_{k' j_1 j_3} 
	\\&\leq \big(\sum_k \mu' \Sigma_{k}^{\circ 2} \mu \big) + \mu' \Sigma^{\circ 2} \mu. 
\end{align*}
First, by Cauchy--Schwarz,
\begin{align*}
	\mu' \Sigma_k^{\circ 2} \mu 
	&= \frac{1}{M_k^2} \sum_{i, i' \in S_k} N_i N_{i'} \big(\sum_j \mu_j \Omega_{ij}\Omega_{i'j} \big)^2 
	\\&= \frac{1}{M_k^2} \sum_{i, i' \in S_k} N_i N_{i'} \big(\sum_j \Omega_{ij} \Omega_{i'j} \big) \sum_j \mu_j^2 \Omega_{ij} \Omega_{i'j} 
	\\&\leq \frac{1}{M_k^2} \sum_{i, i' \in S_k} N_i N_{i'}  \sum_j \mu_j^2 \Omega_{ij} \Omega_{i'j} 
	= \sum_j \mu_j^4 = \| \mu \|_4^4. 
	\num \label{eqn:C31_bd}
\end{align*}
Similarly 
\begin{align*}
	\mu' \Sigma^{\circ 2} \mu \lesssim \| \mu \|_4^4 .
	\num \label{eqn:C32_bd}
\end{align*}
Hence
\begin{align*}
	C_3 \lesssim K \| \mu \|_4^4. 
	\num \label{eqn:C3_bd}
\end{align*}

For $C_4$,
\begin{align*}
	C_4 &\lesssim \bigg(\sum_k \sum_{i, i' \in S_k} N_i N_{i'}  + \sum_{k \neq k'} \sum_{i \in S_k, i' \in S_{k'}} N_i N_{i'}  \bigg) \sum_{j_1 \neq j_2}  \alpha_{i i' j_1 j_2}^2 \, \Omega_{i j_1} \Omega_{i' j_2}
	\\&\lesssim \sum_k \sum_{i, i' \in S_k} N_i N_{i'}  \sum_{j_1 \neq j_2}  
	\big(  \frac{1}{M_k} \Sigma_{k j_1 j_2} +\frac{1}{M} \Sigma_{ j_1 j_2} \big)^2 \, \Omega_{i j_1} \Omega_{i' j_2} 
	\\&\quad + \sum_{k \neq k'} \sum_{i \in S_k, i' \in S_{k'}} N_i N_{i'} 
	\sum_{j_1 \neq j_2}  \big(\frac{1}{M} \sum_{a \in \{k, k'\}}^2 \Sigma_{a j_1 j_2} + \frac{1}{M} \Sigma_{j_1 j_2}\big)^2 \, \Omega_{i j_1} \Omega_{i' j_2}
	\\&\lesssim \sum_k \sum_{i, i' \in S_k} N_i N_{i'}  \sum_{j_1 \neq j_2}  
	\big(  \frac{1}{M_k^2} \Sigma_{k j_1 j_2}^2 +\frac{1}{M^2} \Sigma_{ j_1 j_2}^2 \big) \, \Omega_{i j_1} \Omega_{i' j_2}	
	\\&\quad + \sum_{k \neq k'} \sum_{i \in S_k, i' \in S_{k'}} N_i N_{i'}  
	\sum_{j_1 \neq j_2}  \big(\frac{1}{M^2} \sum_{a \in \{k, k'\}}^2 \Sigma_{a j_1 j_2}^2 + \frac{1}{M^2} \Sigma_{j_1 j_2}^2 \big) \, \Omega_{i j_1} \Omega_{i' j_2} =: C_{41} + C_{42} 
\end{align*}
First,
\begin{align*}
	C_{41} &\lesssim 
	\sum_k \sum_{i, i' \in S_k} N_i N_{i'}  \sum_{j_1 \neq j_2} 
	\frac{1}{M_k^2} \Sigma_{k j_1 j_2}^2 \Omega_{i j_1} \Omega_{i' j_2}	+ 	 \sum_k \sum_{i, i' \in S_k} N_i N_{i'}  \sum_{j_1 \neq j_2} \frac{1}{M^2} \Sigma_{ j_1 j_2}^2  \Omega_{i j_1} \Omega_{i' j_2}
	\\&\lesssim \sum_k \sum_{j_1 \neq j_2} \Sigma_{k j_1 j_2}^2 \mu_{j_1} \mu_{j_2} 
	+ \sum_k \sum_{j_1 \neq j_2} \frac{M_k^2}{M^2} \Sigma_{j_1 j_2}^2 \mu_{j_1} \mu_{j_2} 
	\leq  \sum_k \mu' \Sigma_{k}^{\circ 2} \mu 
	+ \sum_k \frac{M_k^2}{M^2} \mu' \Sigma^{\circ 2} \mu.
\end{align*}
Similarly, 
\begin{align*}
	C_{42} &\lesssim 
	\sum_{k \neq k'} 
	\sum_{j_1 \neq j_2} \frac{M_k M_{k'}}{M^2} \Sigma_{k j_1 j_2}^2 \mu_{j_1} \mu_{j_2} 
	+ 	 \sum_{k \neq k'} 
	\sum_{j_1 \neq j_2} \frac{M_k M_{k'}}{M^2} \Sigma_{j_1 j_2}^2 \mu_{j_1} \mu_{j_2}
	\\&\lesssim \sum_{k \neq k'} \frac{M_k M_{k'}}{M^2} \big(  \mu' \Sigma_k^{\circ 2} \mu +   \mu' \Sigma^{\circ 2} \mu \big) 
\end{align*}
Combining the previous two displays and applying \eqref{eqn:C31_bd} and \eqref{eqn:C32_bd}, we have
\begin{align}
	C_4 &\lesssim K \| \mu \|_4^4. 
\end{align}

For $C_5$, 
\begin{align*}
	C_5&\lesssim \bigg(\sum_k \sum_{i, i' \in S_k} N_i N_{i'}  + \sum_{k \neq k'} \sum_{i \in S_k, i' \in S_{k'}} N_i N_{i'}  \bigg) \sum_{j_1, j_3, j_4 (dist.)} \alpha_{i i' j_1 j_1} 
	\alpha_{i i' j_3 j_4} \, \Omega_{i j_1} \Omega_{i j_3} \Omega_{i'j_1} 
	\Omega_{i' j_4} 
	\\&\lesssim \sum_k \sum_{i, i' \in S_k} N_i N_{i'}  \sum_{j_1, j_3, j_4 } \frac{1}{M_k} \mu_{ j_1} 	\cdot 
	\big( \frac{1}{M_k} \Sigma_{k j_3 j_4} +\frac{1}{M} \Sigma_{ j_3 j_4} \big)  \, \Omega_{i j_1} \Omega_{i j_3} \Omega_{i'j_1} 
	\Omega_{i' j_4} 
	\\&\quad + \sum_{k \neq k'} \sum_{i \in S_k, i' \in S_{k'}} N_i N_{i'}  \sum_{j_1, j_3, j_4 } 
	\frac{1}{M} \mu_{ j_1} 
	\big( \frac{1}{M} \sum_{a \in \{k, k'\} }^2 \Sigma_{a j_3 j_4} + \frac{1}{M} \Sigma_{j_3 j_4}\big) \, \Omega_{i j_1} \Omega_{i j_3} \Omega_{i'j_1} 
	\Omega_{i' j_4} 
	\\&= \sum_k \sum_{j_1,j_3, j_4 } \mu_{j_1} \Sigma_{k j_3 j_4} \Sigma_{k j_1 j_3} \Sigma_{k j_1 j_4} 
	+ \sum_k \sum_{j_1,j_3, j_4 } \frac{M_k}{M} \mu_{j_1} \Sigma_{ j_3 j_4} \Sigma_{k j_1 j_3} \Sigma_{k j_1 j_4} 
	\\&\quad + 2\sum_{k \neq k'}  \sum_{j_1, j_3, j_4 } \frac{M_k M_{k'}}{M^2} \mu_{j_1} \Sigma_{k j_3 j_4} \Sigma_{k j_1 j_3} \Sigma_{k' j_1 j_4} 
	+  \sum_{k \neq k'}  \sum_{j_1, j_3, j_4 } \frac{M_k M_{k'}}{M^2} \mu_{j_1} \Sigma_{ j_3 j_4} \Sigma_{k j_1 j_3} \Sigma_{k' j_1 j_4}
	\\&= C_{51} + C_{52} + 2C_{53} + C_{54}  
\end{align*}
For $C_{51}$, we have
\begin{align*}
	C_{51} 
	&= \sum_k \frac{1}{M_k^3} \sum_{i_1, i_2, i_3 \in S_k } N_{i_1} N_{i_2} N_{i_3} \langle \mu \circ \Omega_{i_1}, \Omega_{i_2} \rangle
	\langle \Omega_{i_1}, \Omega_{i_3} \rangle
	\langle \Omega_{i_2}, \Omega_{i_3} \rangle 
	\\&=  \sum_k \frac{1}{M_k^2} 
	\sum_{i_1, i_2 \in S_k} N_{i_1} N_{i_2} \langle \mu \circ \Omega_{i_1}, \Omega_{i_2} \rangle \cdot \langle \Omega_{i_1}, \Sigma_k \Omega_{i_2} \rangle 
	\\&\leq  \sum_k  \bigg(\frac{1}{M_k^2} \sum_{i_1, i_2 \in S_k} N_{i_1} N_{i_2} \langle \mu \circ \Omega_{i_1}, \Omega_{i_2} \rangle ^2 \bigg)^{1/2} \bigg( \frac{1}{M_k^2} \sum_{i_1, i_2 \in S_k} N_{i_1} N_{i_2} \langle \Omega_{i_1}, \Sigma_k \Omega_{i_2} \rangle^2 \bigg)^{1/2}
	\\&=: \sum_k C_{511k}^{1/2} \cdot C_{512k}^{1/2}. 
	\num \label{eqn:C51_bd1}
\end{align*}
We have by Cauchy--Schwarz that
\begin{align*}
	C_{511k} &=\frac{1}{M_k^2} \sum_{i_1, i_2 \in S_k} N_{i_1} N_{i_2} \big( \sum_j \mu_j \Omega_{i_1 j} \Omega_{i_2 j} \big)^2 
	\\&	\leq \frac{1}{M_k^2} \sum_{i_1, i_2 \in S_k}   N_{i_1} N_{i_2} \big( \sum_j \mu_j^2 \Omega_{i_1 j} \Omega_{i_2 j} \big) \big( \sum_j \Omega_{i_1 j} \Omega_{i_2 j}  \big) 
	\leq  \|  \mu \|_4^4,
\end{align*}
and similarly
\begin{align*}
	C_{512k} &=\frac{1}{M_k^2} \sum_{i_1, i_2 \in S_k} N_{i_1} N_{i_2}
	\big( \sum_{j_1, j_2} \Omega_{i_1 j_1} \Sigma_{k j_1 j_2} \Omega_{i_2 j_2} \big)^2
	\\&= \frac{1}{M_k^2} \sum_{i_1, i_2} N_{i_1} N_{i_2} \big( \sum_{j_1, j_2} \Omega_{i_1 j_1} \Sigma_{k j_1 j_2}^2 \Omega_{i_2 j_2} \big) 
	\big(
	\sum_{j_1, j_2} \Omega_{i_1 j_1} \Omega_{i_2 j_2} 
	\big) 
	\\&\leq \frac{1}{M_k^2} \sum_{i_1, i_2} N_{i_1} N_{i_2} \big( \sum_{j_1, j_2} \Omega_{i_1 j_1} \Sigma_{k j_1 j_2}^2 \Omega_{i_2 j_2} \big) 
	= \mu' \, \Sigma_k^{\circ 2} \, \mu
	\num \label{eqn:C512_bd}
\end{align*}
Since by Cauchy--Schwarz, 
\begin{align*}
	\mu' \, \Sigma_k^{\circ 2} \, \mu 
	&= \sum_{j_1, j_2} \mu_{ j_1} \mu_{j_2} \big( \frac{1}{M_k} \sum_{i \in S_k} N_i  \Omega_{ij_1} \Omega_{ij_2} \big)^2
	= \frac{1}{M_k^2} \sum_{j_1, j_2} \mu_{j_1} \mu_{j_2} 
	\sum_{i, i' \in S_k} N_i N_{i'} \Omega_{ij_1} \Omega_{i j_2} 
	\Omega_{i'j_1} \Omega_{i'j_2} 
	\\&= \frac{1}{M_k^2} \sum_{i, i' \in S_k} 
	\big( \sum_j \mu_j \Omega_{ij} \Omega_{i' j} \big)^2 
	\leq \frac{1}{M_k^2} \sum_{i, i' \in S_k}  \sum_j \mu_j^2 \Omega_{ij} \Omega_{i'j} \leq \| \mu \|_4^4
	\num	\label{eqn:muSig2_bd}
\end{align*}
we have in total
$
C_{512k} \lesssim K \| \mu \|_4^4
$. Combining the result with the bound for $C_{511k}$ implies that
\begin{align*}
	C_{51} \lesssim K \| \mu \|_4^4. 
\end{align*}
Next we study $C_{52}$ using a similar argument. 
\begin{align*}
	C_{52} &= \sum_k \sum_{j_1,j_3, j_4 } \frac{M_k}{M} \mu_{j_1} \Sigma_{ j_3 j_4} \Sigma_{k j_1 j_3} \Sigma_{k j_1 j_4} 
	\\&=  \sum_k \sum_{j_1,j_3, j_4 } \frac{M_k}{M} \mu_{j_1} \big(
	\frac{1}{M} \sum_{i_1 \in [n]} N_{i_1} \Omega_{i_1j_3} \Omega_{i_1j_4}\big) 
	\big(\frac{1}{M_k} \sum_{i_2 \in S_k} N_{i_2} \Omega_{i_2 j_1} \Omega_{i_2 j_3} \big)
	\big(\frac{1}{M_k} \sum_{i_3 \in S_k} N_{i_3} \Omega_{i_3 j_1} \Omega_{i_3 j_4} \big)
	\\&=\sum_k  \frac{1}{M^2 M_k}
	\sum_{j_1,j_2,j_3} \sum_{\substack{i_1 \in [n] \\ i_2, i_3 \in S_k}}
	N_{i_1} N_{i_2} N_{i_3} \langle \mu \circ \Omega_{i_2}, \Omega_{i_3} \rangle 
	\langle \Omega_{i_1}, \Omega_{i_3} \rangle
	\langle \Omega_{i_1}, \Omega_{i_2} \rangle 
	\\&= \sum_k \frac{1}{M^2} 
	\sum_{i_2, i_3 \in [S_k]} N_{i_2} N_{i_3} \langle \mu \circ \Omega_{i_2}, \Omega_{i_3} \rangle \langle \Omega_{i_3}, \Sigma \Omega_{i_2} \rangle 
	\\&\leq \sum_k \bigg( \frac{1}{M^2} \sum_{i_2, i_3 \in [S_k]} N_{i_2} N_{i_3}  \langle \mu \circ \Omega_{i_2}, \Omega_{i_3} \rangle ^2 \bigg)^{1/2} 
	\bigg( \frac{1}{M^2} \sum_{i_2, i_3 \in [S_k]} N_{i_2} N_{i_3} \langle \Omega_{i_3}, \Sigma \Omega_{i_2} \rangle  \bigg)^{1/2} 
	\\&=: \sum_k C_{521k}^{1/2} C_{522k}^{1/2}.
	\num \label{eqn:C52_bd1}
\end{align*}
Observe that $C_{521k} = C_{511k}$, and thus $C_{521} \les \| \mu \|^4$ by \eqref{eqn:C512_bd}. With a similar argument as in \eqref{eqn:muSig2_bd} we obtain $C_{522k} \les \| \mu \|_4^4$. Hence we obtain
\begin{align*}
	C_{52} \leq \sum_k C_{521k}^{1/2} C_{522k}^{1/2} \les K \| \mu \|_4^4.
\end{align*}

For $C_{53}$, we have 
\begin{align*}
	&C_{53} 
	= \sum_{k \neq k'}  \sum_{j_1, j_3, j_4 } \frac{M_k M_{k'}}{M^2} \mu_{j_1} \Sigma_{k j_3 j_4} \Sigma_{k j_1 j_3} \Sigma_{k' j_1 j_4} 
	\\&\leq \sum_k \sum_{j_1,j_3, j_4 } \frac{M_k}{M} 
	\mu_{j_1} \Sigma_{k j_3 j_4} \Sigma_{k j_1 j_3}  \Sigma_{j_1 j_4}
	\\&=  \sum_k \sum_{j_1,j_3, j_4 } \frac{M_k}{M} 
	\mu_{j_1}
	\big(
	\frac{1}{M_k} \sum_{i_1 \in S_k} N_{i_1} \Omega_{i_1j_3} \Omega_{i_1j_4}\big) 
	\big(\frac{1}{M_k} \sum_{i_2 \in S_k} N_{i_2} \Omega_{i_2 j_1} \Omega_{i_2 j_3} \big)
	\big(\frac{1}{M} \sum_{i_3 \in [n]} N_{i_3} \Omega_{i_3 j_1} \Omega_{i_3 j_4} \big)
	\\&= \sum_k \frac{1}{M^2 M_k} \sum_{\substack{i_1, i_2 \in S_k \\ i_3 \in [n]}} N_{i_1} N_{i_2} N_{i_3} \langle \mu \circ \Omega_{i_2} , \Omega_{i_3} \rangle \langle \Omega_{i_1}, \Omega_{i_2} \rangle
	\langle \Omega_{i_1}, \Omega_{i_3} \rangle 
	\\&=  \sum_k \frac{1}{M^2} \sum_{ i_2 \in S_k, i_3 \in [n] }
	N_{i_2} N_{i_3} \langle \mu \circ \Omega_{i_2} , \Omega_{i_3} \rangle  \langle \Omega_{i_2}, \Sigma_k \Omega_{i_3} \rangle.
	\num \label{eqn:C53_bd1}
\end{align*}
We then upper bound the last line using a similar strategy as in that we used for $C_{51}$ and $C_{52}$, respectively. We omit the details and state the final bound:
\begin{align*}
	\num \label{eqn:C53_bd} 
	C_{53} \les K \| \mu \|_4^4
\end{align*}
Finally for $C_{54}$, summing over $k, k'$ we obtain
\begin{align*}
	C_{54} &\leq  \sum_{j_1,j_3, j_4 } \mu_{j_1} \Sigma_{j_3 j_4} \Sigma_{j_1 j_3} \Sigma_{j_1 j_4}
	= \frac{1}{M^3} \sum_{i_1, i_2, i_3 \in [n]} 
	N_{i_1} N_{i_2} N_{i_3}  \langle \mu \circ \Omega_{i_2} , \Omega_{i_3} \rangle \langle \Omega_{i_1}, \Omega_{i_2} \rangle
	\langle \Omega_{i_1}, \Omega_{i_3} \rangle.
	\num \label{eqn:C54_bd1} 
\end{align*}
We then proceed as in \eqref{eqn:C53_bd1} to control the right-hand side. We omit the details and state the final bound:
\begin{align*}
	\num \label{eqn:C54_bd} 
	C_{54} \les K \| \mu \|_4^4.
\end{align*}
Combining the results for $C_{51}, \ldots, C_{54}$, we see that
\begin{align*}
	C_5 \les K \| \mu \|^4.
\end{align*}
For $C_6$, we have
\begin{align*}
	& C_6 \leq \bigg(\sum_k \sum_{i, i' \in S_k} N_i N_{i'}  + \sum_{k \neq k'} \sum_{i \in S_k, i' \in S_{k'}} N_i N_{i'}  \bigg) \sum_{j_1, j_2, j_4 }  \alpha_{i i' j_1 j_2} 
	\alpha_{i i' j_1 j_4} \, \Omega_{i j_1} \Omega_{i' j_2} \Omega_{i' j_4}
	\\&\lesssim \sum_k \sum_{i, i' \in S_k} N_i N_{i'}  \sum_{j_1, j_2, j_4 }  \big(  \frac{1}{M_k} \Sigma_{k j_1 j_2} +\frac{1}{M} \Sigma_{ j_1 j_2}  \big) 
	\big( \frac{1}{M_k} \Sigma_{k j_1 j_4} +\frac{1}{M} \Sigma_{ j_1 j_4} \big)
	\, \Omega_{i j_1} \Omega_{i' j_2} \Omega_{i' j_4}
	\\&+ \sum_{\substack{k \neq k' \\i \in S_k, i' \in S_{k'} \\ j_1, j_2, j_4 }} N_i N_{i'} 
	\big( \frac{1}{M} \sum_{a \in \{k, k'\} }^2 \Sigma_{a j_1 j_2} + \frac{1}{M} \Sigma_{j_1 j_2} \big) 
	\big( \frac{1}{M} \sum_{a \in \{k, k'\} }^2 \Sigma_{a j_1 j_4} + \frac{1}{M} \Sigma_{j_1 j_4} \big)
	\, \Omega_{i j_1} \Omega_{i' j_2} \Omega_{i' j_4}
	\\&=: C_{61} + C_{62}. 
\end{align*}
For $C_{61}$, we have
\begin{align*}
	C_{61} &= \sum_k \sum_{i' \in S_k} N_{i'} 
	\sum_{j_1, j_2, j_4 } \frac{1}{M_k} 
	\Sigma_{k j_1 j_2} \Sigma_{k j_1 j_4}  \mu_{j_1} \Omega_{i'j_2} \Omega_{i' j_4} 
	\\& \quad + 2\sum_k \sum_{ i' \in S_k} N_{i'} 
	\sum_{j_1, j_2, j_4 } \frac{1}{M} 
	\Sigma_{k j_1 j_2} \Sigma_{ j_1 j_4} \mu_{j_1} \Omega_{i'j_2} \Omega_{i' j_4} 
	\\&+ \quad \sum_k \sum_{i' \in S_k} N_{i'} 
	\sum_{j_1, j_2, j_4 } \frac{M_k}{M^2} 
	\Sigma_{ j_1 j_2} \Sigma_{j_1 j_4} \mu_{j_1} \Omega_{i'j_2} \Omega_{i' j_4} 
	=: C_{611} + 2C_{612} +C_{613}. 
\end{align*}
Relabeling indices, we see that 
\begin{align*}
	C_{611} &= \sum_k \sum_{j_1, j_2, j_4} \mu_{j_1} 
	\Sigma_{k j_1 j_2} \Sigma_{k j_1 j_4}  \Sigma_{k j_2 j_4} 
	= C_{51}
\end{align*}
Hence, $C_{611} \lesssim K \| \mu \|_4^4$. Next,
\begin{align*}
	C_{612} \leq \sum_k \frac{M_k}{M} \sum_{j_1,j_2,j_4} \mu_{j_1} \Sigma_{kj_1j_2} \Sigma_{j_1j_4} \Sigma_{k j_2 j_4} 
	\les K \| \mu \|^4,
\end{align*}
where we applied \eqref{eqn:C53_bd1}. Similarly,
\begin{align*}
	C_{613}
	&= \sum_k \frac{M_k^2}{M^2} \sum_{j_1,j_2,j_4} \mu_{j_1} \Sigma_{j_1 j_2} \Sigma_{j_1 j_4} \Sigma_{k j_2 j_4} 
	\leq \sum_{j_1,j_2,j_4}  \mu_{j_1} \Sigma_{j_1 j_2} \Sigma_{j_1 j_4} \Sigma_{ j_2 j_4} \les K \| \mu \|^4,
\end{align*}
where in the final bound we apply \eqref{eqn:C54_bd1} and \eqref{eqn:C54_bd}. Combining the results above for $C_{611}, C_{612}, C_{613}$, we obtain
\begin{align}
	C_{61} \les K \| \mu \|_4^4
\end{align}
The argument for $C_{62}$ is very similar, so we omit proof and state the final bound. We have 
\begin{align*}
	C_{62} \lesssim K\| \mu \|_4.
\end{align*}
Thus
\begin{align*}
	C_6 \lesssim K \| \mu \|_4^4
\end{align*}

For $C_7$, we have
\begin{align*}
	&C_7 \lesssim \bigg(\sum_k \sum_{i, i' \in S_k} N_i N_{i'}  + \sum_{k \neq k'} \sum_{i \in S_k, i' \in S_{k'}} N_i N_{i'}  \bigg)
	\sum_{j_1, j_2, j_3, j_4 } 
	\alpha_{i i' j_1 j_2} 
	\alpha_{i i' j_3 j_4} \Omega_{ij_1} \Omega_{ij_3} 
	\Omega_{i'j_2} \Omega_{i'j_4}
	\\&\lesssim \sum_k \sum_{i, i' \in S_k} N_i N_{i'} 	\sum_{j_1, j_2, j_3, j_4 } 
	\big(   \frac{1}{M_k} \Sigma_{k j_1 j_2} +\frac{1}{M} \Sigma_{ j_1 j_2}  \big) 
	\big(   \frac{1}{M_k} \Sigma_{k j_3 j_4} +\frac{1}{M} \Sigma_{ j_3 j_4}  \big) 
	\Omega_{ij_1} \Omega_{ij_3} 
	\Omega_{i'j_2} \Omega_{i'j_4}
	\\&+ \sum_{k \neq k'} \sum_{\substack{j_1, j_2, j_3, j_4 \\ i \in S_k, i' \in S_{k'}} } N_i N_{i'} 	
	\big( \frac{1}{M} \sum_{a \in \{k, k'\} }^2 \Sigma_{a j_1 j_2} + \frac{1}{M} \Sigma_{j_1 j_2} \big) 
	\big( \frac{1}{M} \sum_{a \in \{k, k'\} }^2 \Sigma_{a j_3 j_4} + \frac{1}{M} \Sigma_{j_3 j_4} \big) \Omega_{ij_1} \Omega_{ij_3}
	\Omega_{i'j_2} \Omega_{i'j_4}
	\\&= : C_{71} +C_{72}
\end{align*}
Write 
\begin{align*}
	C_{71} 
	&= \sum_k \sum_{i, i' \in S_k} N_i N_{i'} 	\sum_{j_1, j_2, j_3, j_4 } 
	\frac{1}{M_k^2} \Sigma_{k j_1 j_2} \Sigma_{k j_3 j_4} 
	\Omega_{ij_1} \Omega_{ij_3} 
	\Omega_{i'j_2} \Omega_{i'j_4}
	\\&\quad +  2\sum_k \sum_{i, i' \in S_k} N_i N_{i'} 	\sum_{j_1, j_2, j_3, j_4 } 
	\frac{1}{M_k M} \Sigma_{ j_1 j_2} \Sigma_{k j_3 j_4} 
	\Omega_{ij_1} \Omega_{ij_3} 
	\Omega_{i'j_2} \Omega_{i'j_4}
	\\&\quad + \sum_k \sum_{i, i' \in S_k} N_i N_{i'} 	\sum_{j_1, j_2, j_3, j_4 } 
	\frac{1}{M^2} \Sigma_{ j_1 j_2} \Sigma_{ j_3 j_4} 
	\Omega_{ij_1} \Omega_{ij_3} 
	\Omega_{i'j_2} \Omega_{i'j_4} =: C_{711} + 2 C_{712} + C_{713}.
\end{align*}
For $C_{711}$, we have 
\begin{align*}
	C_{711} 
	&= \sum_k \sum_{j_1, j_2, j_3, j_4} 
	\Sigma_{k j_1 j_2} \Sigma_{k j_3 j_4} 
	\Sigma_{k j_1 j_3} \Sigma_{k j_2 j_4} 
	\\& =  \sum_k\frac{1}{M_k^4} 
	\sum_{i_1, i_2, i_3, i_4 \in S_k} 
	N_{i_1} N_{i_2} N_{i_3} N_{i_4} 
	\langle \Omega_{i_1}, \Omega_{i_3} \rangle
	\langle \Omega_{i_1}, \Omega_{i_4} \rangle
	\langle \Omega_{i_2} , \Omega_{i_3} \rangle
	\langle \Omega_{i_2}, \Omega_{i_4} \rangle 
	\\&= \frac{1}{M_k^2}  \sum_k\sum_{i_3, i_4} N_{i_3} N_{i_4}
	\big(\Omega_{i_3}' \Sigma_k \Omega_{i_4} \big)^2 
	=  \sum_k\frac{1}{M_k^2} \sum_{i_3, i_4} N_{i_3} N_{i_4}
	\big( \sum_{j, j'} \Omega_{i_3 j}' \Sigma_{k j j'} \Omega_{i_4 j'} \big)^2 
	\\&\leq  \sum_k \frac{1}{M_k^2} \sum_{i_3, i_4} N_{i_3} N_{i_4}
	\sum_{j, j'} \Omega_{i_3 j}' \Sigma_{k j j'}^2 \Omega_{i_4 j'}
	\leq \sum_k \sum_{j,j'} \mu_{j} \Sigma_{k j j'}^2 \mu_{j'} 
	\lesssim K \| \mu \|_4^4.
	\num \label{eqn:C711_bd}
\end{align*}
In the last line we applied Cauchy--Schwarz and \eqref{eqn:muSig2_bd}. 
For $C_{712}$, we have similarly
\begin{align*}
	C_{712}
	&= \sum_k \frac{M_k}{M} \sum_{j_1, j_2 ,j_3, j_4} \Sigma_{j_1j_2} \Sigma_{kj_3j_4} \Sigma_{kj_1j_3}\Sigma_{kj_2j_4} 
	\\&= \sum_k \frac{1}{M^2M_k} \sum_{\substack{i_1 \in[n] \\ i_2,i_3,i_4 \in S_k}}
	N_{i_1} N_{i_2} N_{i_3} N_{i_4} 
	\langle \Omega_{i_1}, \Omega_{i_3} \rangle
	\langle \Omega_{i_1}, \Omega_{i_4} \rangle
	\langle \Omega_{i_2} , \Omega_{i_3} \rangle
	\langle \Omega_{i_2}, \Omega_{i_4} \rangle
	\\&= \sum_k \frac{M_k}{M^2} \sum_{i_1 \in [n], i_2 \in S_k} N_{i_1} N_{i_2}
	\langle \Omega_{i_1}, \Sigma_k \Omega_{i_2} \rangle^2 
	\leq \sum_k \frac{M_k}{M^2} \sum_{i_1 \in [n], i_2 \in S_k} N_{i_1} N_{i_2} 
	\sum_{j,j'} \Omega_{i_1 j} \Sigma_{kjj'}^2 \Omega_{i_2 j'} 
	\\&\leq  \sum_k \frac{M_k^2}{M^2} \sum_{j,j'} \mu_j \Sigma_{kjj'}^2 \mu_{j'} 
	\les K \| \mu \|_4^4. 
	\num \label{eqn:C712_bd}
\end{align*}
Next, 
\begin{align*}
	C_{713} &= \sum_k \frac{M_k^2}{M^2} \sum_{j_1, j_2 ,j_3, j_4}
	\Sigma_{j_1 j_2} \Sigma_{j_3 j_4} \Sigma_{k j_1 j_3} \Sigma_{kj_2j_4}
	\\&= \sum_k  \frac{1}{M^4} \sum_{\substack{i_1, i_2 \in [n] \\ i_3, i_4 \in S_k}}
	N_{i_1} N_{i_2} N_{i_3} N_{i_4}
	\langle \Omega_{i_1}, \Omega_{i_3} \rangle
	\langle \Omega_{i_1}, \Omega_{i_4} \rangle
	\langle \Omega_{i_2} , \Omega_{i_3} \rangle
	\langle \Omega_{i_2}, \Omega_{i_4} \rangle ,
\end{align*}
and applying a similar strategy as in \eqref{eqn:C711_bd}, \eqref{eqn:C712_bd} leads to the bound $C_{713} \les K \| \mu \|_4^4$. Thus
\begin{align*}
	C_{71} \les K \| \mu \|_4^4. 
\end{align*}

Next , by symmetry and summing over $i \in S_k,i' \in S_{k'}$, we have
\begin{align*}
	C_{72}
	&= \sum_{k \neq k'} \frac{M_k M_{k'}}{M^2}
	\sum_{j_1, j_2 ,j_3, j_4} 
	\bigg[
	2\Sigma_{kj_1j_2} \Sigma_{kj_3j_4} 
	+ 2 \Sigma_{k' j_1 j_2} \Sigma_{k j_3 j_4} 
	+ 4 \Sigma_{kj_1 j_2} \Sigma_{j_3 j_4} 
	+  \Sigma_{j_1j_2}\Sigma_{j_3 j_4} 
	\bigg]
	\Sigma_{kj_1j_3} \Sigma_{k'j_2j_4}
	\\&=: 2C_{721} + 2 C_{722} + 4 C_{723} + C_{724} 
\end{align*}
First, 
\begin{align*}
	C_{721} &\leq \sum_k \frac{M_k}{M} \sum_{j_1, j_2 ,j_3, j_4}
	\Sigma_{kj_1 j_2} \Sigma_{kj_3j_4} \Sigma_{kj_1 j_3} \Sigma_{j_2 j_4} 
	= C_{712} \les K \| \mu \|_4^4
\end{align*}
by \eqref{eqn:C712_bd}. Next,
\begin{align*}
	C_{722} &= \sum_{k \neq k'} \frac{M_k M_{k'}}{M^2} 
	\sum_{j_1, j_2 ,j_3, j_4} \Sigma_{k'j_1j_2} \Sigma_{kj_3j_4} \Sigma_{kj_1 j_3} \Sigma_{k'j_2j_4}
	\\&\leq \sum_{k,k'} \frac{1}{M^2M_k M_{k'}} 
	\sum_{\substack{i_1, i_2 \in S_k \\ i_3, i_4 \in S_{k'}} } N_{i_1} N_{i_2} N_{i_3} N_{i_4}
	\langle \Omega_{i_1}, \Omega_{i_3} \rangle
	\langle \Omega_{i_1}, \Omega_{i_4} \rangle
	\langle \Omega_{i_2} , \Omega_{i_3} \rangle
	\langle \Omega_{i_2}, \Omega_{i_4} \rangle
	\\&= \sum_{k,k'} \frac{M_k}{M^2M_{k'}} 	\sum_{i_3, i_4 \in S_{k'}} N_{i_3} N_{i_4} 
	\langle \Omega_{i_3}, \Sigma_k \Omega_{i_4} \rangle^2
	\leq  \sum_{k,k'} \frac{M_k}{M^2M_{k'}} 	\sum_{i_3, i_4 \in S_{k'}} N_{i_3} N_{i_4}  
	\sum_{j,j'} \Omega_{i_3 j} \Sigma_{kjj'}^2 \Omega_{i_4 j'} 
	\\&\leq  \sum_{k,k'} \frac{M_k M_{k'}}{M^2} \mu' \Sigma_k^{\circ 2} \mu 
	\leq \| \mu \|_4^4,
	\num \label{eqn:C722_bd}
\end{align*}
where we applied Cauchy-Schwarz in the penultimate line and \eqref{eqn:muSig2_bd} in the last line. 

For $C_{723}$, we have
\begin{align*}
	C_{723} &= \sum_{k \neq k'} \frac{M_k M_{k'}}{M^2} \sum_{j_1, j_2 ,j_3, j_4} \Sigma_{kj_1 j_2} \Sigma_{j_3 j_4} \Sigma_{kj_1 j_3} \Sigma_{k'j_2j_4}
	\leq \sum_k \frac{M_k}{M} \sum_{j_1, j_2 ,j_3, j_4} \Sigma_{kj_1 j_2} \Sigma_{j_3 j_4} \Sigma_{kj_1 j_3} \Sigma_{j_2j_4}
	\\&= \sum_k \frac{1}{M^3 M_k} 
	\sum_{\substack{i_1,i_3 \in S_k \\ i_2, i_4 \in [n]}}
	N_{i_1} N_{i_2} N_{i_3} N_{i_4}
	\langle \Omega_{i_1}, \Omega_{i_3} \rangle
	\langle \Omega_{i_1}, \Omega_{i_4} \rangle
	\langle \Omega_{i_2} , \Omega_{i_3} \rangle
	\langle \Omega_{i_2}, \Omega_{i_4} \rangle
	\\&= \sum_k \frac{1}{M^2} \sum_{i_3 \in S_k, i_4 \in [n]}
	N_{i_3} N_{i_4} \langle \Omega_{i_3}, \Sigma_k \Omega_{i_4} \rangle 
	\langle \Omega_{i_3}, \Sigma \Omega_{i_4} \rangle 
	\\&\leq  \frac{1}{2} \sum_k \frac{1}{M^2} \sum_{i_3 \in S_k, i_4 \in [n]}
	N_{i_3} N_{i_4} \big( \langle \Omega_{i_3}, \Sigma_k \Omega_{i_4} \rangle^2 +
	\langle \Omega_{i_3}, \Sigma \Omega_{i_4} \rangle^2  \big)
\end{align*}
Using a similar technique as in \eqref{eqn:C711_bd}--\eqref{eqn:C722_bd} and applying \eqref{eqn:C31_bd}, \eqref{eqn:C32_bd} we obtain
\begin{align*}
	C_{723} \les \| \mu \|_4^4. 
\end{align*}
Finally, for $C_{724}$ we have
\begin{align*}
	C_{724} &=  \sum_{k \neq k'} \frac{M_k M_{k'}}{M^2} \sum_{j_1, j_2 ,j_3, j_4}
	\Sigma_{j_1 j_2} \Sigma_{j_3 j_4} \Sigma_{kj_1 j_3} \Sigma_{k'j_2j_4} 
	\leq \sum_{j_1, j_2 ,j_3, j_4} 
	\Sigma_{j_1 j_2} \Sigma_{j_3 j_4} \Sigma_{j_1 j_3} \Sigma_{j_2j_4}
	\\&= \frac{1}{M^4} \sum_{i_1, i_2, i_3, i_4 \in [n]} 	N_{i_1} N_{i_2} N_{i_3} N_{i_4}
	\langle \Omega_{i_1}, \Omega_{i_3} \rangle
	\langle \Omega_{i_1}, \Omega_{i_4} \rangle
	\langle \Omega_{i_2} , \Omega_{i_3} \rangle
	\langle \Omega_{i_2}, \Omega_{i_4} \rangle
\end{align*}
The details are very similar to \eqref{eqn:C711_bd}--\eqref{eqn:C722_bd}, so we omit them and simply state the final bound:
\begin{align*}
	C_{724} &\les \| \mu \|_4^4
\end{align*}
Combining the bounds for $C_{721}, C_{722}, C_{723}$, and $C_{724}$ yields
\begin{align*}
	C_7 \lesssim K \| \mu \|_4^4. 
\end{align*}
Combining the bounds for $C_1$--$C_7$ proves the result.

\qed

\subsection{Proof of Lemma~\ref{lem:fourth_mom_U}}

We have
\begin{align*}
	&\E D_{\ell,s}^4 
	= \E\bigg[ \big( \sum_{ i \in [\ell - 1] } 
	\sigma_{i, \ell} \sum_{r = 1}^{N_i} \sum_j Z_{ijr} Z_{\ell j s} \big)^4\bigg] 
	\\&= \sum_{i_1, i_2, i_3, i_4 \in [\ell - 1]}
	\sigma_{i_1\ell} \sigma_{i_2 \ell} \sigma_{i_3 \ell} \sigma_{i_4 \ell} 
	\sum_{\substack{r_1, r_2, r_3, r_4 \\  j_1, j_2 ,j_3, j_4}} 
	\E \big[ Z_{i_1 j_1 r_1} Z_{\ell j_1 s}
	Z_{i_2 j_2 r_2} Z_{\ell j_2 s} Z_{i_3 j_3 r_3} Z_{\ell j_3 s}
	Z_{i_4 j_4 r_4} Z_{\ell j_4 s} \big] 
	\\&= \sum_{i_1, i_2, i_3, i_4 \in [\ell - 1]}
	\sigma_{i_1\ell} \sigma_{i_2 \ell} \sigma_{i_3 \ell} \sigma_{i_4 \ell} 
	\sum_{\substack{r_1, r_2, r_3, r_4 \\  j_1, j_2 ,j_3, j_4}} 
	\E \big[ Z_{i_1 j_1 r_1} 
	Z_{i_2 j_2 r_2}  Z_{i_3 j_3 r_3} 
	Z_{i_4 j_4 r_4}  \big] \E \big[
	Z_{\ell j_1 s}Z_{\ell j_2 s}Z_{\ell j_3 s}Z_{\ell j_4 s}
	\big]
	\\&= 
	\sum_{j_1, j_2 ,j_3, j_4} \E[  Z_{\ell j_1 s}Z_{\ell j_2 s}Z_{\ell j_3 s}Z_{\ell j_4 s} ] 
	\sum_{\substack{i_1, i_2, i_3, i_4 \in [\ell - 1]  \\  r_1, r_2, r_3, r_4}}
	\sigma_{i_1\ell} \sigma_{i_2 \ell} \sigma_{i_3 \ell} \sigma_{i_4 \ell} 
	\E[ Z_{i_1 j_1 r_1} 
	Z_{i_2 j_2 r_2}  Z_{i_3 j_3 r_3} 
	Z_{i_4 j_4 r_4} ] 
	\\&=: \sum_{j_1, j_2 ,j_3, j_4} \E[  Z_{\ell j_1 s}Z_{\ell j_2 s}Z_{\ell j_3 s}Z_{\ell j_4 s} ] A_{j_1, j_2, j_3, j_4}
	\num \label{eqn:Dls_fourth_moment}
\end{align*}
In the summations above,  $r_t$ ranges over $[ N_{i_t} ]$. 

Observe that
\begin{align*}
	\num \label{eqn:Z4_bounds}
	| \E[  Z_{\ell j_1 s}Z_{\ell j_2 s}Z_{\ell j_3 s}Z_{\ell j_4 s} ]  |
	\lesssim \begin{cases}
		\Omega_{\ell j_1 }
		&\quad \text{ if } j_1 = j_2 = j_3 = j_4 \\
		\Omega_{\ell j_1 } 	  \Omega_{\ell j_4}
		&\quad \text{ if } j_1 = j_2 = j_3,  j_4 \neq j_1 \\
		\Omega_{\ell j_1 } 	  \Omega_{\ell j_3}
		&\quad \text{ if } j_1 = j_2,  j_3 =  j_4, j_1 \neq j_3 \\
		\Omega_{\ell j_1 } 	  \Omega_{\ell j_3} 	  \Omega_{\ell j_4 } 	 
		&\quad \text{ if } j_1 = j_2,  j_1, j_3, j_4 \, \, dist.  \\
		\Omega_{\ell j_1 } 	  \Omega_{\ell j_2 } 	  \Omega_{\ell j_3} 	  \Omega_{\ell j_4 } 
		&\quad \text{ if }  j_1, j_2, j_3, j_4 \, \, dist. \\
	\end{cases}
\end{align*}
Up to permutation of the indices $j_1, \ldots, j_4$, this accounts for all possible cases. 

To proceed we also bound $A_{j_1, j_2, j_3, j_4}$ by casework on the number of distinct $j$ indices. For brevity we define $\omega_t = (i_t, r_t)$ and slightly abuse notation, letting $Z_{\omega_t, j} = Z_{i_t j r_t}$. Further let $\mc{I}_\ell = \{ \omega = (i,r): i \in [\ell ], 1 \leq r \leq N_i \}$. Our goal is to control
\begin{align}
	A_{j_1, j_2, j_3, j_4}
	= \sum_{\omega_1, \omega_2, \omega_3, \omega_4 \in \mc{I}_{\ell - 1} }
	\sigma_{i_1\ell} \sigma_{i_2 \ell} \sigma_{i_3 \ell} \sigma_{i_4 \ell} 
	\E[ Z_{\omega_1 j_1} Z_{\omega_2 j_2} Z_{\omega_3 j_3} Z_{\omega_4 j_4} ]. 
	\label{eqn:A_main}
\end{align}
To do this, we study \eqref{eqn:A_main} in five cases that cover all possibilities (up to permutation of the indices $j_1, \ldots, j_4$).

\vspace{0.2cm}

\noindent \textit{Case 1:} $j_1 = j_2 = j_3 = j_4$. Define $j = j_1$. It holds that
\begin{align*}
	\sigma_{i_1\ell} \sigma_{i_2 \ell} &\sigma_{i_3 \ell} \sigma_{i_4 \ell} \E[ Z_{\omega_1 j} Z_{\omega_2 j} Z_{\omega_3 j} Z_{\omega_4 j} ]
	\\&= 
	\begin{cases}
		\sigma_{i_1\ell}^4 \, \E Z_{\omega_1 j }^4 \lesssim \sigma_{i_1\ell}^4 \, \Omega_{i_1 j}
		&\quad \text{ if } \omega_1 = \omega_2 = \omega_3 = \omega_4 \\
		\sigma_{i_1\ell}^2  \sigma_{i_3\ell}^2 \E Z_{\omega_1 j}^2 \E Z_{\omega_3 j}^2
		\lesssim \sigma_{i_1\ell}^2  \sigma_{i_3\ell}^2 \Omega_{i_1j}
		\Omega_{i_3 j} 
		&\quad \text{ if } \omega_1 = \omega_2,  \omega_3 = \omega_4, \omega_1 \neq \omega_3 
	\end{cases}
	\num \label{case1_A}
\end{align*}
Up to permutation of the indices $\omega_1, \ldots, \omega_4$, this accounts for all cases such that \eqref{case1_A} is nonvanishing. To be precise, by symmetry, it also holds that for all permutations $\pi:[4] \to [4]$ that
if $\omega_{\pi(1)} = \omega_{\pi(2)},  \omega_{\pi(3)} = \omega_{\pi(4)}, \omega_{\pi(1)} \neq \omega_{\pi(3)} $, then 
\[
\sigma_{i_1\ell} \sigma_{i_2 \ell} \sigma_{i_3 \ell} \sigma_{i_4 \ell} \E[ Z_{\omega_1 j} Z_{\omega_2 j} Z_{\omega_3 j} Z_{\omega_4 j} ] 
\lesssim \sigma_{i_{\pi(1)}\ell}^2  \sigma_{i_{\pi(3)}\ell}^2 \Omega_{i_{\pi(1)} j}
\Omega_{i_{\pi(3)} j}. 
\]
In all other cases besides those considered above, we have
\[
\sigma_{i_1\ell} \sigma_{i_2 \ell} \sigma_{i_3 \ell} \sigma_{i_4 \ell} \E[ Z_{\omega_1 j} Z_{\omega_2 j} Z_{\omega_3 j} Z_{\omega_4 j} ] = 0
\]
by independence.

Therefore,
\begin{align}
	A_{jjjj} &\lesssim
	\sum_{\omega \in \mc{I}_{\ell - 1}} \sigma_{i\ell}^4 \Omega_{i j}
	+ \sum_{\omega_1 \neq \omega_3 \in \mc{I}_{\ell-1}} 
	\sigma_{i_1\ell}^2  \sigma_{i_3\ell}^2 \Omega_{i_1j}
	\Omega_{i_3 j} 
	\label{eqn:case1_A}
\end{align}


In the remaining Cases 2--6, we follow the same strategy of writing out bounds for \[\sigma_{i_1\ell} \sigma_{i_2 \ell} \sigma_{i_3 \ell} \sigma_{i_4 \ell} 
\E[ Z_{\omega_1 j_1} Z_{\omega_2 j_2} Z_{\omega_3 j_3} Z_{\omega_4 j_4} ] \]
that cover all nonzero cases, up to permutation of the indices $\omega_1, \ldots, \omega_4$. 

\vspace{0.1cm}

\noindent \textit{Case 2:} $j_1 = j_2 = j_3, j_1 \neq j_4$. It holds that
\begin{align*}
	\sigma_{i_1\ell} &\sigma_{i_2 \ell} \sigma_{i_3 \ell} \sigma_{i_4 \ell} 
	\E[ Z_{\omega_1 j_1} Z_{\omega_2 j_1} Z_{\omega_3 j_1} Z_{\omega_4 j_4} ]
	\\&\quad = \begin{cases}
		\sigma_{i_1\ell}^4 \E[ Z_{\omega_1 j_1}^3 Z_{\omega_1 j_4} ]
		\lesssim \sigma_{i_1\ell}^4 \, \Omega_{i_1 j_1} \Omega_{i_1 j_4} 
		&\quad \text{ if } \omega_1 = \omega_2 = \omega_3 = \omega_4 \\
		\sigma_{i_1\ell}^2 \sigma_{i_3 \ell}^2
		\E Z_{\omega_1 j_1}^2 \E Z_{\omega_3 j_1} Z_{\omega_3 j_4} 
		\lesssim \sigma_{i_1\ell}^2 \sigma_{i_3 \ell}^2 \, 
		\Omega_{i_1 j_1} \Omega_{i_3 j_1} \Omega_{i_3 j_4}
		&\quad \text{ if } \omega_1 = \omega_2,  \omega_3 = \omega_4, \omega_1 \neq \omega_3 \\
	\end{cases}
	\num \label{case2_A}
\end{align*}
Up to permutation of the indices $\omega_1, \ldots, \omega_4$, this accounts for all cases such that \eqref{case2_A} is nonvanishing. Thus 
\begin{align}
	\num \label{eqn:case2_A}
	A_{j_1, j_1, j_1, j_4} \lesssim 
	\sum_{\omega \in \mc{I}_{\ell-1}}
	\sigma_{i\ell}^4 \Omega_{i_1 j_1} \Omega_{i_1 j_4} 
	+ \sum_{\omega_1 \neq \omega_3 \in \mc{I}_{\ell - 1}}
	\sigma_{i_1\ell}^2 \sigma_{i_3 \ell}^2 \, 
	\Omega_{i_1 j_1} \Omega_{i_3 j_1} \Omega_{i_3 j_4}
\end{align}


\vspace{0.1cm}

\noindent \textit{Case 3:} $j_1 = j_2, j_3 = j_4, j_1 \neq j_3$. It holds that
\begin{align*}
	&\sigma_{i_1\ell}  \sigma_{i_2 \ell} \sigma_{i_3 \ell} \sigma_{i_4 \ell} 
	\E[ Z_{\omega_1 j_1} Z_{\omega_2 j_1} Z_{\omega_3 j_3} Z_{\omega_4 j_3} ]
	\\& = \begin{cases}
		\sigma_{i_1\ell}^4 \E Z_{\omega_1 j_1}^2 Z_{\omega_1 j_3}^2 
		\lesssim \sigma_{i_1\ell}^4 \, 
		\Omega_{i_1 j_1} \Omega_{i_1 j_3} 
		& \text{ if } \omega_1 = \omega_2 = \omega_3 = \omega_4 \\
		\sigma_{i_1\ell}^2 \sigma_{i_3\ell}^2 \E Z_{\omega_1 j_1}^2 
		\E Z_{\omega_{3} j_3}^2 
		\lesssim \sigma_{i_1\ell}^2 \sigma_{i_3\ell}^2 \, 
		\Omega_{i_1 j_1} \Omega_{i_3 j_3} 
		& \text{ if } \omega_1 = \omega_2,  \omega_3 = \omega_4, \omega_1 \neq \omega_3 \\
		\sigma_{i_1\ell}^2 \sigma_{i_3\ell}^2 
		\E Z_{\omega_1 j_1} Z_{\omega_1 j_3} \E Z_{\omega_2 j_1} Z_{\omega_2 j_3}
		\lesssim 	\sigma_{i_1\ell}^2 \sigma_{i_3\ell}^2  
		\Omega_{i_1 j_1 } \Omega_{i_1 j_3} \Omega_{i_2 j_1} \Omega_{i_2 j_3} 
		&\text{ if } \omega_1 = \omega_3,  \omega_2 = \omega_4, \omega_1 \neq \omega_2. 
	\end{cases}
	\num \label{case3_A}
\end{align*}
Up to permutation of the indices $\omega_1, \ldots, \omega_4$, this accounts for all cases such that \eqref{case3_A} is nonvanishing. Thus by symmetry, 
\begin{align*}
	\num  \label{eqn:case3_A}
	A_{j_1, j_1, j_3, j_3} 
	&\les \sum_{\omega \in \mc{I}_{\ell-1}} \sigma_{i_1\ell}^4 \, 
	\Omega_{i_1 j_1} \Omega_{i_1 j_3} 
	+ \sum_{\omega_1 \neq \omega_3 \in \mc{I}_{\ell - 1}}\sigma_{i_1\ell}^2 \sigma_{i_3\ell}^2 \, 
	\Omega_{i_1 j_1} \Omega_{i_3 j_3}
	\\&\quad + \sum_{\omega_1 \neq \omega_3 \in \mc{I}_{\ell - 1}}
	\sigma_{i_1\ell}^2 \sigma_{i_3\ell}^2  
	\Omega_{i_1 j_1 } \Omega_{i_1 j_3} \Omega_{i_3 j_1} \Omega_{i_3 j_3} 
\end{align*}


\vspace{0.1cm}

\noindent \textit{Case 4:} $j_1 = j_2$ and $j_1, j_3, j_4$ distinct. We have
\begin{align*}
	&	\sigma_{i_1\ell} \sigma_{i_2 \ell} \sigma_{i_3 \ell} \sigma_{i_4 \ell} 
	\E[ Z_{\omega_1 j_1} Z_{\omega_2 j_1} Z_{\omega_3 j_3} Z_{\omega_4 j_4} ]\\
	& = \begin{cases}
		\sigma_{i_1\ell}^4 \E Z_{\omega_1 j_1}^2 Z_{\omega_{1} j_3} Z_{\omega_1 j_4}
		\lesssim \sigma_{i_1\ell}^4 \, 
		\Omega_{i_1 j_1} \Omega_{i_1 j_3} \Omega_{i_1 j_4} 
		& \text{ if } \omega_1 = \omega_2 = \omega_3 = \omega_4 \\
		\sigma_{i_1\ell}^2 \sigma_{i_3\ell}^2
		\E Z_{\omega_1 j_1}^2 \E Z_{\omega_3 j_3} Z_{\omega_3 j_4} 
		\lesssim \sigma_{i_1\ell}^2 \sigma_{i_3\ell}^2\,
		\Omega_{i_1 j_1} \Omega_{i_3 j_3} \Omega_{i_3 j_4} 
		& \text{ if } \omega_1 = \omega_2,  \omega_3 = \omega_4, \omega_1 \neq \omega_3 \\
		\sigma_{i_1 \ell}^2 \sigma_{i_2 \ell}^2 
		\E Z_{\omega_1 j_1} Z_{\omega_1 j_3} \E Z_{\omega_2 j_1} Z_{\omega_2 j_4}
		\lesssim 	\sigma_{i_1 \ell}^2 \sigma_{i_2 \ell}^2 \,
		\Omega_{i_1 j_1} \Omega_{i_1 j_3} \Omega_{i_2 j_1} \Omega_{i_2 j_4}
		& \text{ if } \omega_1 = \omega_3,  \omega_2 = \omega_4, \omega_1 \neq \omega_2
	\end{cases}
	\num \label{case4_A} 
\end{align*}
Up to permutation of the indices $\omega_1, \ldots, \omega_4$, this accounts for all cases such that \eqref{case4_A} is nonvanishing. Thus
\begin{align*}
	\num  \label{eqn:case4_A}
	A_{j_1, j_1, j_3, j_4} 
	&\lesssim 
	\sum_{\omega \in \mc{I}_{\ell-1}}  \sigma_{i_1\ell}^4 \, 
	\Omega_{i_1 j_1} \Omega_{i_1 j_3} \Omega_{i_1 j_4} 
	+ 
	\sum_{\omega_1 \neq \omega_3 \in \mc{I}_{\ell - 1}}  \sigma_{i_1\ell}^2 \sigma_{i_3\ell}^2\,
	\Omega_{i_1 j_1} \Omega_{i_3 j_3} \Omega_{i_3 j_4} 
	\\&\quad \sum_{\omega_1 \neq \omega_3 \in \mc{I}_{\ell - 1}}	\sigma_{i_1 \ell}^2 \sigma_{i_2 \ell}^2 \,
	\Omega_{i_1 j_1} \Omega_{i_1 j_3} \Omega_{i_3 j_1} \Omega_{i_3 j_4}. 
\end{align*}

\vspace{0.1cm}

\noindent \textit{Case 5:} $j_1, j_2, j_3, j_4$ distinct. For this final case, it holds that
\begin{align*}
	&	\sigma_{i_1\ell} \sigma_{i_2 \ell} \sigma_{i_3 \ell} \sigma_{i_4 \ell} 
	\E[ Z_{\omega_1 j_1} Z_{\omega_2 j_2} Z_{\omega_3 j_3} Z_{\omega_4 j_4} ]
	\\&= \begin{cases}
		\sigma_{i_1\ell}^4 \E Z_{\omega_1 j_1} Z_{\omega_1 j_2} Z_{\omega_1 j_3} Z_{\omega_1 j_4} \lesssim 	\sigma_{i_1\ell}^4 \,
		\Omega_{i_1 j_1} 		\Omega_{i_1 j_2} 		\Omega_{i_1 j_3} 		\Omega_{i_1 j_4}   
		& \text{ if } \omega_1 = \omega_2 = \omega_3 = \omega_4 \\
		\sigma_{i_1\ell}^2 \sigma_{i_3 \ell}^2
		\E Z_{\omega_1 j_1} Z_{\omega_1 j_2} \E Z_{\omega_3 j_3} Z_{\omega_3 j_4}
		\lesssim 	\sigma_{i_1\ell}^2 \sigma_{i_3 \ell}^2 \,
		\Omega_{i_1 j_1} 		\Omega_{i_1 j_2} 		\Omega_{i_3 j_3} 		\Omega_{i_3 j_4}   
		& \text{ if } \omega_1 = \omega_2,  \omega_3 = \omega_4, \omega_1 \neq \omega_3 \\	
	\end{cases}
\end{align*} 
The above accounts for all nonzero cases, up to permutation of $\omega_1, \omega_2, \omega_3, \omega_4$. Hence
\begin{align*}
	\num  \label{eqn:case5_A}
	A_{j_1, j_2, j_3, j_4} 
	\lesssim 
	\sum_{\omega \in \mc{I}_{\ell-1}} \sigma_{i_1\ell}^4 \,
	\Omega_{i_1 j_1} 		\Omega_{i_1 j_2} 		\Omega_{i_1 j_3} 		\Omega_{i_1 j_4}   
	+\sum_{\omega_1 \neq \omega_3 \in \mc{I}_{\ell - 1}} \sigma_{i_1\ell}^2 \sigma_{i_3 \ell}^2 \,
	\Omega_{i_1 j_1} 		\Omega_{i_1 j_2} 		\Omega_{i_3 j_3} 		\Omega_{i_3 j_4} . 
\end{align*}

Finally we control the fourth moment using the casework above. By \eqref{eqn:Dls_fourth_moment} and symmetry,
\begin{align*}
	\E D_{\ell, s}^4 &\lesssim 
	\sum_j  \E[  Z_{\ell j s}Z_{\ell j s}Z_{\ell j s}Z_{\ell j s} ] A_{j, j, j, j}
	+ \sum_{j_1 \neq j_4}  \E[  Z_{\ell j_1 s}Z_{\ell j_1 s}Z_{\ell j_1 s}Z_{\ell j_4 s} ] A_{j_1, j_1, j_1, j_4}
	\\&\quad+ \sum_{j_1 \neq j_3}  \E[  Z_{\ell j_1 s}Z_{\ell j_1 s}Z_{\ell j_3 s}Z_{\ell j_3 s} ] A_{j_1, j_1, j_3, j_3}
	+ \sum_{j_1, j_3, j_4 \, dist.}  \E[  Z_{\ell j_1 s}Z_{\ell j_1 s}Z_{\ell j_3 s}Z_{\ell j_4 s} ] A_{j_1, j_1, j_3, j_4}
	\\&\quad+ \sum_{j_1, j_2 ,j_3, j_4 \, dist.} \E[  Z_{\ell j_1 s}Z_{\ell j_2 s}Z_{\ell j_3 s}Z_{\ell j_4 s} ] A_{j_1, j_2, j_3, j_4}
	\\&=: F_{1\ell s} + F_{2\ell s} + F_{3\ell s} +F_{4\ell s} +F_{5\ell s}
	\num \label{eqn:F_decomp_U}
\end{align*}
By \eqref{eqn:Z4_bounds},  \eqref{eqn:case1_A}, \eqref{eqn:case2_A} ,\eqref{eqn:case3_A}, \eqref{eqn:case4_A}, and \eqref{eqn:case5_A},
\begin{align*}
	F_{1\ell s}&\lesssim \sum_j \Omega_{\ell j} \bigg( \sum_{\omega \in \mc{I}_{\ell - 1}} \sigma_{i\ell}^4 \Omega_{i j}
	+ \sum_{\omega_1 \neq \omega_3 \in \mc{I}_{\ell-1}} 
	\sigma_{i_1\ell}^2  \sigma_{i_3\ell}^2 \Omega_{i_1j}
	\Omega_{i_3 j} \bigg)  
	\\ F_{2\ell s} & \les \sum_{j_1 \neq j_4} 
	\Omega_{\ell j_1} \Omega_{\ell j_4} \bigg( 
	\sum_{\omega \in \mc{I}_{\ell-1}}
	\sigma_{i\ell}^4 \Omega_{i_1 j_1} \Omega_{i_1 j_4} 
	+ \sum_{\omega_1 \neq \omega_3 \in \mc{I}_{\ell - 1}}
	\sigma_{i_1\ell}^2 \sigma_{i_3 \ell}^2 \, 
	\Omega_{i_1 j_1} \Omega_{i_3 j_1} \Omega_{i_3 j_4} \bigg) 
	\\ F_{3\ell s} & \les \sum_{j_1 \neq j_3} \Omega_{\ell j_1} \Omega_{\ell j_3} 
	\bigg(
	\sum_{\omega \in \mc{I}_{\ell-1}} \sigma_{i_1\ell}^4 \, 
	\Omega_{i_1 j_1} \Omega_{i_1 j_3} 
	+ \sum_{\omega_1 \neq \omega_3 \in \mc{I}_{\ell - 1}}\sigma_{i_1\ell}^2 \sigma_{i_3\ell}^2 \, 
	\Omega_{i_1 j_1} \Omega_{i_3 j_3}
	\\&\qquad + \sum_{\omega_1 \neq \omega_3 \in \mc{I}_{\ell - 1}}
	\sigma_{i_1\ell}^2 \sigma_{i_3\ell}^2  
	\Omega_{i_1 j_1 } \Omega_{i_1 j_3} \Omega_{i_3 j_1} \Omega_{i_3 j_3} 
	\bigg)
	\\ F_{4\ell s} &\les \sum_{j_1, j_3, j_4 \, dist.} 
	\Omega_{\ell j_1} \Omega_{\ell j_3} \Omega_{\ell j_4} 
	\bigg(
	\sum_{\omega \in \mc{I}_{\ell-1}}  \sigma_{i_1\ell}^4 \, 
	\Omega_{i_1 j_1} \Omega_{i_1 j_3} \Omega_{i_1 j_4} 
	+ 
	\sum_{\omega_1 \neq \omega_3 \in \mc{I}_{\ell - 1}}  \sigma_{i_1\ell}^2 \sigma_{i_3\ell}^2\,
	\Omega_{i_1 j_1} \Omega_{i_3 j_3} \Omega_{i_3 j_4} 
	\\&\qquad + \sum_{\omega_1 \neq \omega_3 \in \mc{I}_{\ell - 1}}	\sigma_{i_1 \ell}^2 \sigma_{i_3 \ell}^2 \,
	\Omega_{i_1 j_1} \Omega_{i_1 j_3} \Omega_{i_3 j_1} \Omega_{i_3 j_4}. 
	\bigg)
	\\ F_{5\ell s} &\les 
	\sum_{j_1, j_2 ,j_3, j_4 \, dist.} \Omega_{\ell j_1} \Omega_{\ell j_2} \Omega_{\ell j_3} \Omega_{\ell j_4} \bigg(
	\sum_{\omega \in \mc{I}_{\ell-1}} \sigma_{i_1\ell}^4 \,
	\Omega_{i_1 j_1} 		\Omega_{i_1 j_2} 		\Omega_{i_1 j_3} 		\Omega_{i_1 j_4}   
	\\&\qquad +\sum_{\omega_1 \neq \omega_3 \in \mc{I}_{\ell - 1}} \sigma_{i_1\ell}^2 \sigma_{i_3 \ell}^2 \,
	\Omega_{i_1 j_1} 		\Omega_{i_1 j_2} 		\Omega_{i_3 j_3} 		\Omega_{i_3 j_4}
	\bigg). 
\end{align*}
Define
\begin{align*}
	F_{11\ell s} &=  \sum_{\omega \in \mc{I}_{\ell - 1}} \sigma_{i\ell}^4 \sum_j \Omega_{\ell j} \Omega_{i j}
	\\ F_{21\ell s} &= 
	\sum_{\omega \in \mc{I}_{\ell-1}}
	\sigma_{i\ell}^4 \sum_{j_1 \neq j_4} 
	\Omega_{\ell j_1} \Omega_{\ell j_4}  \Omega_{i_1 j_1} \Omega_{i_1 j_4} 
	\\ F_{31\ell s} &= 
	\sum_{\omega \in \mc{I}_{\ell-1}} \sigma_{i_1\ell}^4 \, 
	\sum_{j_1 \neq j_3} \Omega_{\ell j_1} \Omega_{\ell j_3} \Omega_{i_1 j_1} \Omega_{i_1 j_3} 
	\\ F_{41\ell s} &= 
	\sum_{\omega \in \mc{I}_{\ell-1}}  \sigma_{i_1\ell}^4 \, \sum_{j_1, j_3, j_4 \, dist.} 
	\Omega_{\ell j_1} \Omega_{\ell j_3} \Omega_{\ell j_4} 
	\Omega_{i_1 j_1} \Omega_{i_1 j_3} \Omega_{i_1 j_4} 
	\\ F_{51\ell s} &= 	
	\sum_{\omega \in \mc{I}_{\ell-1}} \sigma_{i_1\ell}^4 \,
	\sum_{j_1, j_2 ,j_3, j_4 \, dist.} \Omega_{\ell j_1} \Omega_{\ell j_2} \Omega_{\ell j_3} \Omega_{\ell j_4}
	\Omega_{i_1 j_1} 		\Omega_{i_1 j_2} 		\Omega_{i_1 j_3} 		\Omega_{i_1 j_4}   
\end{align*}
and 
\begin{align*}
	F_{12\ell s} &= \sum_{\omega_1 \neq \omega_3 \in \mc{I}_{\ell-1}} 
	\sigma_{i_1\ell}^2  \sigma_{i_3\ell}^2 \sum_j \Omega_{\ell j} \Omega_{i_1j}
	\Omega_{i_3 j} 
	\\ F_{22\ell s} &=  \sum_{\omega_1 \neq \omega_3 \in \mc{I}_{\ell - 1}}
	\sigma_{i_1\ell}^2 \sigma_{i_3 \ell}^2  \, \sum_{j_1 \neq j_4} 
	\Omega_{\ell j_1} \Omega_{\ell j_4} 
	\Omega_{i_1 j_1} \Omega_{i_3 j_1} \Omega_{i_3 j_4}
	\\ F_{32\ell s} &= \sum_{\omega_1 \neq \omega_3 \in \mc{I}_{\ell - 1}}\sigma_{i_1\ell}^2 \sigma_{i_3\ell}^2 \,  \sum_{j_1 \neq j_3} \big[ \Omega_{\ell j_1} \Omega_{\ell j_3} 
	\Omega_{i_1 j_1} \Omega_{i_3 j_3} + \Omega_{\ell j_1} \Omega_{\ell j_3} 
	\Omega_{i_1 j_1 } \Omega_{i_1 j_3} \Omega_{i_3 j_1} \Omega_{i_3 j_3} \big]
	\\ F_{42 \ell s} &= 
	\sum_{\omega_1 \neq \omega_3 \in \mc{I}_{\ell - 1}}  \sigma_{i_1\ell}^2 \sigma_{i_3\ell}^2\,
	\sum_{j_1, j_3, j_4 \, dist.} \big[ 
	\Omega_{\ell j_1} \Omega_{\ell j_3} \Omega_{\ell j_4} \Omega_{i_1 j_1} \Omega_{i_3 j_3} \Omega_{i_3 j_4} 
	\\ &\qquad \qquad \qquad \qquad \qquad \qquad \qquad +
	\Omega_{\ell j_1} \Omega_{\ell j_3} \Omega_{\ell j_4} 
	\Omega_{i_1 j_1} \Omega_{i_1 j_3} \Omega_{i_3 j_1} \Omega_{i_3 j_4}\big] 
	\\ F_{52\ell s} &= \sum_{\omega_1 \neq \omega_3 \in \mc{I}_{\ell - 1}} \sigma_{i_1\ell}^2 \sigma_{i_3 \ell}^2 \,
	\sum_{j_1, j_2 ,j_3, j_4 \, dist.} \Omega_{\ell j_1} \Omega_{\ell j_2} \Omega_{\ell j_3} \Omega_{\ell j_4}
	\Omega_{i_1 j_1} 		\Omega_{i_1 j_2} 		\Omega_{i_3 j_3} 		\Omega_{i_3 j_4}
\end{align*}
Note that $\sum_{x = 1}^2 F_{tx\ell s} = F_{t\ell s}$ for all $t \in [5]$. Using the fact that $\sum_j \Omega_{ij} = 1$, we have
\begin{align*}
	\sum_{t} F_{t1\ell s} 
	&\les F_{11\ell s} = \sum_{\omega \in \mc{I}_{\ell - 1}} \sigma_{i\ell}^4 \sum_j \Omega_{\ell j} \Omega_{i j} 
	= \sum_{ \omega \in \mc{I}_{\ell - 1}} \sigma_{i\ell}^4 \langle \Omega_\ell, \Omega_i \rangle. 
	\num \label{eqn:Ft1ells_bd}
\end{align*}
To control 	$\sum_{t} F_{t2\ell s}$ , observe that, since $\Omega_{ij} \leq 1 $ for all $i,j$, 
\begin{align*}
	&\sum_j \Omega_{\ell j} \Omega_{i_1j} = \langle \Omega_\ell ,
	\Omega_{i_1} \circ \Omega_{i_3} \rangle 
	\\ 	&\sum_{j_1 \neq j_4} 
	\Omega_{\ell j_1} \Omega_{\ell j_4} 
	\Omega_{i_1 j_1} \Omega_{i_3 j_1} \Omega_{i_3 j_4}
	\leq \langle \Omega_\ell, \Omega_{i_1} \circ \Omega_{i_3} \rangle \cdot
	\langle \Omega_\ell , \Omega_{i_3} \rangle 
	\\ &\sum_{j_1 \neq j_3} \big[ \Omega_{\ell j_1} \Omega_{\ell j_3} 
	\Omega_{i_1 j_1} \Omega_{i_3 j_3} +\Omega_{\ell j_1} \Omega_{\ell j_3} 
	\Omega_{i_1 j_1 } \Omega_{i_1 j_3} \Omega_{i_3 j_1} \Omega_{i_3 j_3} \big]
	\leq 2 \langle \Omega_\ell, \Omega_{i_1}  \rangle \cdot 
	\langle \Omega_\ell, \Omega_{i_3} \rangle
	\\ &\sum_{j_1, j_3, j_4 \, dist.} \big[ 
	\Omega_{\ell j_1} \Omega_{\ell j_3} \Omega_{\ell j_4} \Omega_{i_1 j_1} \Omega_{i_3 j_3} \Omega_{i_3 j_4} + 
	\Omega_{\ell j_1} \Omega_{\ell j_3} \Omega_{\ell j_4} 
	\Omega_{i_1 j_1} \Omega_{i_1 j_3} \Omega_{i_3 j_1} \Omega_{i_3 j_4}\big] 
	\leq 2 \langle \Omega_\ell, \Omega_{i_1} \rangle 
	\langle \Omega_\ell, \Omega_{i_3} \rangle^2 
	\\ & \sum_{j_1, j_2 ,j_3, j_4 \, dist.} \Omega_{\ell j_1} \Omega_{\ell j_2} \Omega_{\ell j_3} \Omega_{\ell j_4}
	\Omega_{i_1 j_1} 		\Omega_{i_1 j_2} 		\Omega_{i_3 j_3} 		\Omega_{i_3 j_4} \leq  \langle \Omega_{\ell}, \Omega_{i_1} \rangle^2 
	\langle \Omega_{\ell}, \Omega_{i_3} \rangle^2. 
\end{align*}
These bounds are relatively sharp, and it is clear that the first and third lines dominate. Furthermore $as$. Hence,
\begin{align*}
	\sum_{t} F_{t2\ell s} \les 
	F_{12\ell s} + F_{32\ell s} \les 
	\sum_{\omega_1 \neq \omega_3 \in \mc{I}_{\ell-1}} 
	\sigma_{i_1\ell}^2  \sigma_{i_3\ell}^2 \big[ \langle \Omega_\ell ,
	\Omega_{i_1} \circ \Omega_{i_3} \rangle 
	+ 	\langle \Omega_\ell, \Omega_{i_1}  \rangle \cdot 
	\langle \Omega_\ell, \Omega_{i_3} \rangle \big]. 
	\num \label{eqn:Ft2ells_bd}
\end{align*}
Observe that if $\ell \in S_k$, then
\begin{align}
	\sum_\omega \sigma_{i \ell}^4 \Omega_{ij} 
	&\leq  \sum_{i \in S_k} \frac{1}{n_k^4 \bar{N}_k^4} N_i \Omega_{ij} 
	+ \sum_{k' =1}^K \sum_{i \in S_{k'}} 
	\frac{1}{n^4 \bar{N}^4} N_i \Omega_{ij}
	\\&\leq  \frac{1}{n_k^3 \bar{N}_k^3} \mu_{kj} 
	+ \frac{1}{n^3 \bar{N}^3} \mu_j , 
\end{align}
and
\begin{align*}
	\sum_\omega \sigma_{i \ell}^2 \Omega_{ij} 
	&\leq  \sum_{i \in S_k} \frac{1}{n_k^2 \bar{N}_k^2} N_i \Omega_{ij} 
	+ \sum_{k' =1}^K \sum_{i \in S_{k'}} 
	\frac{1}{n \bar{N}} N_i \Omega_{ij} 
	\\&\leq 
	\frac{1}{n_k \bar{N}_k} \mu_{kj} + \frac{1}{n \bar{N}} \mu_j. 
\end{align*}
Next,
\begin{align*}
	\sum_{(\ell, s)} \sum_t F_{t1\ell s} 
	&\les 	\sum_{(\ell, s)} \sum_{ \omega \in \mc{I}_{\ell - 1}} \sigma_{i\ell}^4 \langle \Omega_\ell, \Omega_i \rangle. 
	\\&\les \sum_{(\ell, s)}  \sum_j \Omega_{\ell j} \big(   \frac{1}{n_k^3 \bar{N}_k^3} \mu_{kj} 
	+ \frac{1}{n^3 \bar{N}^3} \mu_j \big) 
	\\&\les  \sum_j \sum_k \frac{1}{n_k^2 \bar{N}_k^2} \mu_{kj}^2 
	+ \sum_j \sum_k \frac{1}{n^2 \bar{N}^2} \mu_j^2
	\les \sum_k \frac{1}{n_k^2 \bar{N}_k^2} \| \mu_k \|^2,
	\num \label{eqn:U_4m_bd1}
\end{align*}
where we applied that $\| \mu \|^2 \les \sum_k \| \mu_k \|^2$ (see \eqref{eqn:mu_vs_muk}). 
Furthermore, 
\begin{align*}
	&\sum_{(\ell, s)} 	\sum_{t} F_{t2\ell s}
	\leq \sum_{k = 1}^K \sum_{\ell \in S_k} N_\ell \sum_{\omega_1 , \omega_3 } 
	\sigma_{i_1\ell}^2  \sigma_{i_3\ell}^2 \big[ \langle \Omega_\ell ,
	\Omega_{i_1} \circ \Omega_{i_3} \rangle 
	+ 	\langle \Omega_\ell, \Omega_{i_1}  \rangle \cdot 
	\langle \Omega_\ell, \Omega_{i_3} \rangle \big] 
	\\& \les \sum_k \sum_{\ell \in S_k} N_\ell 	\bigg[ \sum_j 
	\Omega_{\ell_j} \big( 	\frac{1}{n_k \bar{N}_k} \mu_{kj} + \frac{1}{n \bar{N}} \mu_j \big)^2
	+ \bigg(  \sum_{j}  \Omega_{\ell j}  \cdot \big( 	\frac{1}{n_k \bar{N}_k} \mu_{kj} + \frac{1}{n \bar{N}} \mu_j \big) \bigg)^2
	\bigg]
	\\&\les \sum_k \sum_{\ell \in S_k}  N_\ell \sum_j 
	\Omega_{\ell_j} \big( 	\frac{1}{n_k \bar{N}_k} \mu_{kj} + \frac{1}{n \bar{N}} \mu_j \big)^2
\end{align*}
In the last line we apply Cauchy--Schwarz. Continuing, we have
\begin{align*}
	\sum_{(\ell, s)} 	\sum_{t} F_{t2\ell s} &\les 	\sum_k \sum_{\ell \in S_k}  N_\ell \sum_j 
	\Omega_{\ell_j} \big( 	\frac{1}{n_k \bar{N}_k} \mu_{kj} + \frac{1}{n \bar{N}} \mu_j \big)^2
	\\&\les 	\sum_k \sum_{\ell \in S_k}  N_\ell \sum_j 
	\Omega_{\ell_j}  \big(\frac{1}{n_k \bar{N}_k} \mu_{kj}\big)^2 
	+ 	\sum_k \sum_{\ell \in S_k}  N_\ell \sum_j 
	\Omega_{\ell_j}  \big( \frac{1}{n \bar{N}} \mu_j \big)^2  
	\\&\les \sum_k \frac{ \| \mu_k \|_3^3 }{ n_k \bar{N}_k} 
	+ \sum_k \frac{\| \mu \|_3^3}{n \bar{N}}
	\les \sum_k \frac{ \| \mu_k \|_3^3 }{ n_k \bar{N}_k} ,
	\num \label{eqn:U_4m_bd2}
\end{align*}
where we applied \eqref{eqn:mu3_vs_muk3}. Combining \eqref{eqn:F_decomp_U}, \eqref{eqn:U_4m_bd1} and \eqref{eqn:U_4m_bd2}, we have
\begin{align*}
	\sum_{(\ell, s)} \E D_{\ell, s}^4 &\les 	\sum_{(\ell, s)} \sum_{x = 1}^2 \sum_{t = 1}^5 F_{tx\ell s} 
	\les  \sum_k \frac{\| \mu_k \|^2}{n_k^2 \bar{N}_k^2}  + \sum_k \frac{ \| \mu_k \|_3^3 }{ n_k \bar{N}_k} \, , 
\end{align*}
as desired. 

\subsection{Proof of Lemma~\ref{lem:S_condvar}}

\begin{align}
	\var\bigg[ 	\sum_{(\ell, s)} \var( \tilde E_{\ell,s} | \mc{F}_{\prec (\ell, s)} ) \bigg] &\to 0 
	\label{eqn:mgclt_12_Kn}
\end{align} 
Next we study \eqref{eqn:mgclt_12_Kn}. We have
\begin{align*}
	\var(  E_{\ell,s} | \mc{F}_{\prec (\ell, s)} )
	&= \E[ E_{\ell, s}^2 | \mc{F}_{\prec (\ell, s)}]
	= \sigma_\ell^2 
	\sum_{r, r' \in [s-1]} \sum_{j, j'} 
	\E\big[ Z_{\ell j r } Z_{\ell j s } Z_{\ell j' r'} Z_{\ell j' s }
	\big| \mc{F}_{\prec (\ell, s)} \big  ]
	\\&= \sigma_\ell^2
	\sum_{r, r' \in [s-1]} \sum_{j, j'} Z_{\ell j r } Z_{\ell j' r'} 
	\E[ Z_{\ell j s} Z_{\ell j' s }]
	\\&= 	\sigma_\ell^2
	\sum_{r, r' \in [s-1]} \sum_{j, j'} \delta_{j j' \ell } Z_{\ell j r } Z_{\ell j' r'}, 
	\num \label{eqn:condvar_E}
\end{align*}
where we let 
\begin{align*}
	\delta_{j j' \ell} = \E Z_{\ell j s}Z_{\ell j's} =  \begin{cases}
		\Omega_{\ell j}(1 - \Omega_{\ell j}) &\quad \text{ if } j = j' \\
		-	\Omega_{\ell j} \Omega_{\ell j'} &\quad \text{ else}. 
	\end{cases}
	\num \label{eqn:delta_def_ell_S}
\end{align*}
Define 
\begin{align}
	\label{eqn:new_zeta_def} 
	\varphi_{\ell r \ell r'} = \sum_{j, j'} \delta_{jj'\ell} Z_{\ell j r } Z_{\ell j' r'}.
\end{align}

By \eqref{eqn:condvar_E} we have
\begin{align*}
	\sum_{(\ell, s)} 	\var(  E_{\ell,s} | \mc{F}_{\prec (\ell, s)} )
	&= \sum_{\ell = 1}^n \sum_{s = 1}^{N_\ell} \sum_{r, r' \in [s-1]}
	\sigma_\ell^2 \,  \varphi_{\ell r \ell r'} 
	\\&= \sum_{\ell = 1}^n \sum_{s = 1}^{N_\ell} \big[ \sum_{r \in [s-1]} \sigma_\ell^2 \,  \varphi_{\ell r \ell r}
	+ 2 \sum_{r < r' \in [s-1]} \sigma_\ell^2 \,  \varphi_{\ell r \ell r'}
	\big]  
	\\&= \sum_{\ell = 1}^n
	\sum_{r =1  }^{N_\ell} \sum_{s \in [N_\ell]: s > r} \sigma_\ell^2 \,  \varphi_{\ell r \ell r}
	+ 2 \sum_{\ell = 1}^s \sum_{ r < r' \in [N_\ell] } \sum_{s \in [N_\ell]: s > r'} \sigma_\ell^2 \,  \varphi_{\ell r \ell r'}
	\\&= \sum_{\ell = 1}^n
	\sum_{r =1  }^{N_\ell} (N_\ell - r) \sigma_\ell^2 \,   \varphi_{\ell r \ell r}
	+ 2 \sum_{\ell = 1}^s \sum_{ r < r' \in [N_\ell] } (N_\ell - r') \sigma_\ell^2 \,  \varphi_{\ell r \ell r'}
	\\&\equiv S_1 + S_2.
\end{align*}	
Observe that $S_1$ and $S_2$ are uncorrelated. In addition, the terms in the summation defining $S_1$ are uncorrelated; the same holds for $S_2$ also. 

First we study $S_2$. 
Next,
\begin{align*}
	\E \varphi_{\ell  r \ell  r' }^2
	&= \sum_{j_1, j_2 ,j_3, j_4} \delta_{j_1 j_2, \ell } \delta_{j_3 j_4, \ell } \,
	\E Z_{\ell j_1 r} Z_{\ell j_2 r'} Z_{\ell j_3 r} Z_{\ell j_4 r'} 
	\\&= \sum_{j_1, j_2 ,j_3, j_4} \delta_{j_1 j_2 \ell } \delta_{j_3 j_4 \ell}
	\E Z_{\ell j_1 r} Z_{\ell j_3 r} \E Z_{\ell j_2 r'} Z_{\ell j_4 r'}. 
	\num 	\label{eqn:zeta_expand}
\end{align*}

First we study $V_2$. 
By casework, 
\begin{align*}
	\num \label{eqn:zeta_bd_dist_casework}
	&|\delta_{j_1 j_2 \ell} \delta_{j_3 j_4\ell}
	\E Z_{\ell  j_1 r} Z_{\ell  j_3 r} \E Z_{\ell  j_2 r'} Z_{\ell  j_4 r'}|
	\\&=
	\begin{cases}
		\delta_{jj\ell}^2	\E Z_{\ell jr}^2 \E Z_{\ell jr'}^2 \lesssim \Omega_{\ell j}^4 \quad &\text{ if } j_1 = \cdots = j_4 \\
		\delta_{j_1j_1\ell} \delta_{j_1 j_4\ell} |\E Z_{\ell j_1r}^2 \E Z_{\ell  j_1 r' } Z_{\ell  j_4 r'} |
		\lesssim  \Omega_{\ell j_1}^4 \Omega_{\ell  j_4}^2 
		&\text{ if } j_1 = j_2 = j_3, j_1 \neq j_4 \\
		\delta_{j_1 j_1\ell} \delta_{j_3 j_3\ell}
		\E Z_{\ell  j_1 r} Z_{\ell  j_3 r} \E Z_{\ell  j_1 r'} Z_{\ell  j_3 r'}
		\lesssim \Omega_{\ell  j_1}^3 \Omega_{\ell  j_3}^3 
		&\text{ if } j_1 = j_2, j_3 = j_4, j_1 \neq j_3 \\
		\delta_{j_1 j_2\ell}^2
		\E Z_{\ell  j_1 r}^2 \E Z_{\ell  j_2 r'}^2
		\lesssim \Omega_{\ell  j_1}^3 \Omega_{\ell j_2}^3 
		&\text{ if } j_1 = j_3, j_2 = j_4, j_1 \neq j_2 \\
		\delta_{j_1 j_1\ell} \delta_{j_3 j_4\ell}
		\E Z_{\ell  j_1 r} Z_{\ell  j_3 r} \E Z_{\ell  j_1 r'} Z_{\ell  j_4 r'}
		\lesssim \Omega_{\ell  j_1}^3 \Omega_{\ell j_3}^2 \Omega_{\ell j_4}^2 
		&\text{ if } j_1 = j_2, j_1, j_3, j_4 \, \, dist. \\
		\delta_{j_1 j_2\ell} \delta_{j_1 j_4\ell}
		\E Z_{\ell  j_1 r}^2 \E Z_{\ell  j_2 r'} Z_{\ell  j_4 r'}
		\lesssim \Omega_{\ell  j_1}^3 \Omega_{\ell  j_2}^2 \Omega_{\ell j_4}^2 
		&\text{ if } j_1 = j_3, j_1, j_2, j_4 \,\, dist. \\
		\delta_{j_1 j_2\ell} \delta_{j_3 j_4\ell}
		\E Z_{\ell  j_1 r} Z_{\ell  j_3 r} \E Z_{\ell  j_2 r'} Z_{\ell  j_4 r'}
		\lesssim \Omega_{\ell  j_1}^2 \Omega_{\ell  j_2}^2 \Omega_{\ell j_3}^2 \Omega_{\ell  j_4}^2
		&\text{ if }  j_1, j_2, j_3, j_4 \, \, dist.
	\end{cases}
\end{align*}
Up to permutation of the indices $j_1, \ldots, j_4$, all nonzero terms of \eqref{eqn:zeta_expand} take one of the forms above. By \eqref{eqn:zeta_bd_dist_casework} and Cauchy--Schwarz, we have 
\begin{align*}
	\num \label{eqn:zeta_bd_dist}
	\E \varphi_{\ell r \ell r'}^2 
	&\lesssim \| \Omega_\ell \|_4^4 + \| \Omega_\ell \|_4^4 \| \Omega_\ell \|^2 + 2\| \Omega_\ell \|_3^6
	+ 2 \| \Omega_\ell \|_3^3 \| \Omega_\ell \|^4 + \| \Omega_\ell \|^8
	\lesssim \| \Omega_\ell \|_4^4. 
\end{align*}
%
Recalling that $\{ \varphi_{\ell  r \ell  r' } \}_{\ell, r < r' \in [N_\ell]}$ are mutually uncorrelated, it  follows that
\begin{align*}
	\var(S_2) &\les \sum_{\ell} \sum_{r < r' \in [N_\ell]} (N_\ell - r')^2 \sigma_\ell^2 
	\E \varphi^2_{\ell r \ell r'} 
	\\&\les \sum_{\ell} \sum_{r < r' \in [N_\ell]} (N_\ell - r')^2 \sigma_\ell^4 
	\| \Omega_\ell \|_4^4 
	\\&\les \sum_k \sum_{\ell  \in S_k} N_\ell^4 \cdot \frac{1}{n_k^4 \bar{N}_k^4} \| \Omega_\ell \|_4^4. \num \label{eqn:varS2}
\end{align*}

Next we study $S_1$. We have
\begin{align*}
	\E \varphi_{\ell r \ell r}^2
	&= \sum_{j_1, j_2 ,j_3, j_4} \delta_{j_1 j_2\ell} \delta_{j_3 j_4\ell}
	\E Z_{\ell j_1r} Z_{\ell j_2r} Z_{\ell j_3r} Z_{\ell j_4 r}.  
\end{align*}
We have the following bounds by casework.
\begin{align*}
	\num \label{eqn:zeta_bd_equal_casework}
	&|	\delta_{j_1 j_2\ell} \delta_{j_3 j_4\ell}\E Z_{\ell j_1r} Z_{\ell j_2r} Z_{\ell j_3r} Z_{\ell j_4 r} |
	\\&=
	\begin{cases}
		\delta_{jj\ell}^2	\E Z_{\ell jr}^4 \lesssim \Omega_{\ell j}^3 \quad &\text{ if } j_1 = \cdots = j_4 \\
		\delta_{j_1j_1\ell} \delta_{j_1 j_4\ell} |\E Z_{\ell j_1r}^3 Z_{\ell j_4 r} |
		\lesssim  \Omega_{\ell j_1}^3 \Omega_{\ell j_4}^2 &\text{ if } j_1 = j_2 = j_3, j_1 \neq j_4 \\
		\delta_{j_1j_1\ell} \delta_{j_3j_3\ell} 
		\E Z_{\ell j_1r}^2 Z_{\ell j_3r}^2 
		\lesssim \Omega_{\ell j_1}^2 \Omega_{\ell j_3}^2 
		&\text{ if } j_1 = j_2, j_3 = j_4, j_1 \neq j_3 \\
		\delta_{j_1 j_2\ell}^2 \E Z_{\ell j_1r}^2 Z_{\ell j_2r}^2
		\lesssim \Omega_{\ell j_1}^3 \Omega_{\ell j_2}^3 
		&\text{ if } j_1 = j_3, j_2 = j_4, j_1 \neq j_3 \\
		\delta_{j_1 j_1\ell} \delta_{j_3 j_4\ell} 
		| \E Z_{\ell j_1r}^2 Z_{\ell j_3r} Z_{\ell j_4r}|
		\lesssim \Omega_{\ell j_1}^2 \Omega_{\ell j_3}^2 \Omega_{\ell j_4}^2 
		&\text{ if } j_1 = j_2, j_1, j_3, j_4 \, \, dist. \\
		\delta_{j_1 j_2\ell} \delta_{j_1 j_4\ell}
		| \E Z_{\ell j_1r}^2 Z_{\ell j_2r} Z_{\ell j_4} | 
		\lesssim \Omega_{\ell j_1}^3 \Omega_{\ell j_2}^2 \Omega_{\ell j_4}^2 
		&\text{ if } j_1 = j_3, j_1, j_2, j_4 \,\, dist. \\
		\delta_{j_1 j_2\ell} \delta_{j_3 j_4\ell} 
		| \E Z_{\ell j_1r} Z_{\ell j_2r} Z_{\ell j_3r} Z_{\ell j_4 r} |
		\lesssim \Omega_{\ell j_1}^2 \Omega_{\ell j_2}^2 \Omega_{\ell j_3}^2 \Omega_{\ell j_4}^2
		&\text{ if }  j_1, j_2, j_3, j_4 \, \, dist.
	\end{cases}
\end{align*}
Up to symmetry, this accounts for all possible (nonzero) cases. Hence by Cauchy--Schwarz, 
\begin{align*}
	\num \label{eqn:zeta_bd_equal}
	\E \varphi_{\ell r \ell r}^2 
	&\lesssim \| \Omega_\ell  \|_3^3 
	+ \| \Omega_\ell  \|_3^3 \| \Omega_\ell  \|^2 
	+ \| \Omega_\ell  \|^4 + \| \Omega_\ell  \|_3^6 + \| \Omega_\ell  \|^6 
	+ \| \Omega_\ell  \|_3^3 \| \Omega_\ell  \|^4 + \| \Omega_\ell  \|^8
	\lesssim \| \Omega_\ell  \|_3^3. 
\end{align*}
Recalling that $\{ \varphi_{\ell r \ell r } \}_{\ell, r\in[N_\ell]}$ is an uncorrelated collection of random variables, we have 
\begin{align*}
	\var(S_1) &\les \sum_\ell \sum_{r \in [N_\ell]} 
	(N_\ell - r)^2 \sigma_\ell^4 \E \varphi_{\ell r \ell r}^2 
	\\&\les \sum_\ell \sum_{r \in [N_\ell]} 
	(N_\ell - r)^2 \sigma_\ell^4 \| \Omega_\ell \|_3^3 
	\\&\les \sum_k \sum_{\ell \in S_k} 
	N_\ell^3 \cdot \frac{1}{n_k^4 \bar{N}_k^4} \| \Omega_\ell \|_3^3. 
	\num \label{eqn:varS1}
\end{align*}

Combining \eqref{eqn:varS1} and \eqref{eqn:varS2} proves the result. 
\qed 

\subsection{Proof of Lemma~\ref{lem:S_fourth_mom} } 

We have 
\begin{align*}
	\E E_{\ell,s }^4
	&= \sum_{ r_1, r_2, r_3, r_4 \in [s-1]} \, \,
	\sigma_{\ell}^4 \sum_{j_1, j_2, j_3, j_4} \E Z_{\ell j_1 r_1} Z_{\ell j_1 s}
	Z_{\ell j_2 r_2} Z_{\ell j_2 s}
	Z_{\ell j_3 r_3} Z_{\ell j_3 s}
	Z_{\ell j_4 r_4} Z_{\ell j_4 s}
	\\&= \sigma_\ell^4  
	\sum_{j_1, j_2, j_3, j_4} \bigg[ \, \, \E[ Z_{\ell j_1 s}Z_{\ell j_2 s}
	Z_{\ell j_3 s}
	Z_{\ell j_4 s}] \cdot  \underbrace{\sum_{ r_1, r_2, r_3, r_4 \in [s-1]} \E[ Z_{\ell j_1 r_1} 
		Z_{\ell j_2 r_2} 	Z_{\ell j_3 r_3} Z_{\ell j_4 r_4} ]}_{=: B_{\ell, s; j_1, j_2, j_3, j_4}} \, \,\bigg] 
	\num \label{eqn:S_4thmom_main_bd}
\end{align*}
We have by exhaustive casework that
\begin{align*}
	\num \label{eqn:S_casework}
	&| \E[ Z_{\ell j_1 r_1} 
	Z_{\ell j_2 r_2} 	Z_{\ell j_3 r_3} Z_{\ell j_4 r_4} ]|
	\\&= \begin{cases}
		\E Z_{\ell j_1 r_1}^4 \lesssim \Omega_{\ell j_1}
		& \text{ if  } \substack{j_1 = j_2 = j_3 = j_4; \\ r_1 = r_2 = r_3 = r_4} \\	
		\E Z_{\ell j_1 r_1}^2 \E Z_{\ell j_1 r_3}^2 
		\lesssim \Omega_{\ell j_1}^2 
		& \text{ if  }  \substack{j_1 = j_2 = j_3 = j_4; \\ r_1 = r_2, r_3 = r_4, r_1 \neq r_3 }\\	
		| \E[ Z_{\ell j_1 r_1}^3 Z_{\ell j_4 r_1} ]| 
		\lesssim \Omega_{\ell j_1} \Omega_{\ell j_4} 
		& \text{ if  } \substack{j_1 = j_2 = j_3, j_1 \neq j_4; \\ r_1 = r_2 = r_3 = r_4} \\
		| \E[ Z_{\ell j_1 r_1}^2 \E Z_{\ell j_1 r_3} Z_{\ell j_4 r_3} ]|
		\lesssim \Omega_{\ell j_1}^2 \Omega_{\ell j_4}
		& \text{ if  } \substack{j_1 = j_2 = j_3, j_1 \neq j_4; \\ r_1 = r_2, r_3 = r_4, r_1 \neq r_3} \\
		| \E Z_{\ell j_1 r_1}^2 Z_{\ell j_3 r_1}^2 |
		\lesssim \Omega_{\ell j_1} \Omega_{\ell j_3} 
		& \text{ if  } \substack{j_1 = j_2, j_3 = j_4, j_1 \neq j_3;\\ r_1 = r_2 = r_3 = r_4} \\
		| \E[ Z_{\ell j_1 r_1}^2	Z_{\ell j_3 r_3}^2 ]| 
		\lesssim \Omega_{\ell j_1} \Omega_{\ell j_3}
		& \text{ if  } \substack{j_1 = j_2, j_3 = j_4, j_1 \neq j_3;\\ r_1 = r_2, r_3 = r_4, r_1 \neq r_3} \\
		| \E[ Z_{\ell j_1 r_1} Z_{\ell j_3 r_1} \E 
		Z_{\ell j_1 r_2} 	 Z_{\ell j_3 r_2} ]|
		\lesssim \Omega_{\ell j_1}^2 \Omega_{\ell j_3}^2
		& \text{ if  } \substack{j_1 = j_2, j_3 = j_4, j_1 \neq j_3;\\ r_1 = r_3, r_2 = r_4, r_1 \neq r_2} \\
		| \E[ Z_{\ell j_1 r_1}^2	Z_{\ell j_3 r_1} Z_{\ell j_4 r_1} ]|
		\lesssim \Omega_{\ell j_1} \Omega_{\ell j_3} \Omega_{\ell j_4}
		& \text{ if  } \substack{j_1 = j_2, j_1, j_3, j_4 \, \, dist.;\\ r_1 = r_2 = r_3 = r_4} \\
		| \E[ Z_{\ell j_1 r_1}^2 \E	Z_{\ell j_3 r_3} Z_{\ell j_4 r_3} ]|
		\lesssim \Omega_{\ell j_1} \Omega_{\ell j_3} \Omega_{\ell j_4}	
		& \text{ if  } \substack{j_1 = j_2, j_1, j_3, j_4 \, \, dist.;\\ r_1 = r_2, r_3 = r_4, r_1 \neq r_3} \\
		| \E[ Z_{\ell j_1 r_1} 	Z_{\ell j_3 r_1}  \E
		Z_{\ell j_1 r_2} Z_{\ell j_4 r_2} ]|
		\lesssim \Omega_{\ell j_1}^2 \Omega_{\ell j_3} \Omega_{\ell j_4}
		& \text{ if  } \substack{j_1 = j_2, j_1, j_3, j_4 \, \, dist.; \\ r_1 = r_3, r_2 = r_4, r_1 \neq r_2} \\
		| \E[ Z_{\ell j_1 r_1} 
		Z_{\ell j_2 r_1} 	Z_{\ell j_3 r_1} Z_{\ell j_4 r_1} ]|
		\lesssim \Omega_{\ell j_1} \Omega_{\ell j_2} \Omega_{\ell j_3} \Omega_{\ell j_4}
		& \text{ if } \substack{ j_1, j_2, j_3, j_4 \, \, dist;\\ r_1 = r_2 = r_3 = r_4} \\ 
		| \E[ Z_{\ell j_1 r_1} 
		Z_{\ell j_2 r_1} \E	Z_{\ell j_3 r_3} Z_{\ell j_4 r_3} ]|
		\lesssim \Omega_{\ell j_1} \Omega_{\ell j_2} \Omega_{\ell j_3} \Omega_{\ell j_4}
		& \text{ if } \substack{j_1, j_2, j_3, j_4 \, \, dist;\\ 	r_1 = r_2, r_3 = r_4, r_1 \neq r_3}
	\end{cases}
\end{align*}
Up to permutation of the indices $j_1, j_2, j_3, j_4$ and $r_1, r_2, r_3, r_4$, this accounts for all possible cases such that \eqref{eqn:S_casework} is nonzero.	Therefore,
\begin{align*}
	B_{\ell, s;j_1, j_2, j_3, j_4} &\lesssim \begin{cases}
		s \Omega_{\ell j_1} + s^2 \Omega_{\ell j_1}^2
		&\quad \text{ if  } j_1 = j_2 = j_3 = j_4 \\
		s \Omega_{\ell j_1} \Omega_{\ell j_4} + 
		s^2 \Omega_{\ell j_1}^2 \Omega_{\ell j_4}
		&\quad \text{ if  } j_1 = j_2 = j_3, j_1 \neq j_4 \\
		s \Omega_{\ell j_1} \Omega_{\ell j_3} 
		+ s^2 \Omega_{\ell j_1} \Omega_{\ell j_3} 
		&\quad \text{ if  } j_1 = j_2, j_3 = j_4, j_1 \neq j_3 \\
		s \Omega_{\ell j_1} \Omega_{\ell j_3} \Omega_{\ell j_4}
		+ s^2 \Omega_{\ell j_1} \Omega_{\ell j_3} \Omega_{\ell j_4}
		&\quad \text{ if  } j_1 = j_2, j_1,  j_3,  j_4 \, \, dist. \\
		s \Omega_{\ell j_1} \Omega_{\ell j_2} \Omega_{\ell j_3} \Omega_{\ell j_4}
		+ s^2 \Omega_{\ell j_1} \Omega_{\ell j_2} \Omega_{\ell j_3} \Omega_{\ell j_4}
		&\quad \text{ if  } j_1 , j_2 , j_3 , j_4 \, \, dist. \\
	\end{cases}
\end{align*}
Up to permutation of $j_1, j_2, j_3, j_4$, this accounts for all possible cases. Returning to \eqref{eqn:S_4thmom_main_bd}, we have by applying \eqref{eqn:Z4_bounds} and the previous display that
\begin{align*}
	\E E_{\ell,s}^4 
	&\lesssim \sigma_\ell^4
	\bigg( \sum_{j} \Omega_{\ell j} ( s \Omega_{\ell j} + s^2 \Omega_{\ell j}^2 )
	+ \sum_{j_1 \neq j_4} \Omega_{\ell j_1} \Omega_{\ell j_4} ( s \Omega_{\ell j_1} \Omega_{\ell j_4} 
	+ s^2 \Omega_{\ell j_1}^2 \Omega_{\ell j_4})
	\\&\quad  + \sum_{j_1 \neq j_3} \Omega_{\ell j_1} \Omega_{\ell j_3} ( s \Omega_{\ell j_1} \Omega_{\ell j_3} 
	+ s^2 \Omega_{\ell j_1} \Omega_{\ell j_3}  )
	\\&\quad + \sum_{j_1, j_3, j_4 (dist.)} \Omega_{\ell j_1} \Omega_{\ell j_3} \Omega_{\ell j_4} (s \Omega_{\ell j_1} \Omega_{\ell j_3} \Omega_{\ell j_4}
	+ s^2 \Omega_{\ell j_1} \Omega_{\ell j_3} \Omega_{\ell j_4} )
	\\&\quad + \sum_{j_1, j_2 ,j_3, j_4 \, dist.} \Omega_{\ell j_1} \Omega_{\ell j_2} \Omega_{\ell j_3} \Omega_{\ell j_4} (s \Omega_{\ell j_1} \Omega_{\ell j_2} \Omega_{\ell j_3} \Omega_{\ell j_4}
	+ s^2 \Omega_{\ell j_1} \Omega_{\ell j_2} \Omega_{\ell j_3} \Omega_{\ell j_4}) \bigg)
	\\&\lesssim 
	s \sigma_{\ell}^4 \|\Omega_\ell  \|^2 
	+ s^2 \sigma_{\ell}^4 \| \Omega_\ell \|_3^3.  
\end{align*}
In the third line we group the coefficients of $s$ and $s^2$ and use the fact that $\| \Omega_\ell \|^4 \leq \| \Omega_\ell \|_3^3$ by Cauchy--Schwarz. Therefore
\begin{align*}
	\sum_{(\ell, s)} \E E_{\ell,s}^4
	&\lesssim  \sum_{(\ell, s)} s \sigma_\ell^4 \| \Omega_\ell  \|^2 
	+ \sum_{(\ell, s)} s^2 \sigma_\ell^4\| \Omega_\ell \|_3^3 
	\\&= \sum_k \sum_{\ell \in S_k} \sum_{s \in [N_\ell]} 
	s \sigma_\ell^4 \| \Omega_\ell  \|^2  
	+  \sum_k \sum_{\ell \in S_k} \sum_{s \in [N_\ell]}  
	s^2 \sigma_\ell^4\| \Omega_\ell \|_3^3
	\\&\les \sum_k \sum_{\ell \in S_k} 
	N_\ell^2 \cdot  \frac{1}{n_k^4 \bar{N}_k^4} \| \Omega_\ell  \|^2  
	+ \sum_k \sum_{\ell \in S_k} N_\ell^3 \cdot  \frac{1}{n_k^4 \bar{N}_k^4} \| \Omega_\ell \|_3^3,
\end{align*}
as desired.
\qed 

\subsection{Proof of Lemma~\ref{lem:S_bounds_on_bounds}}

We have
\begin{align*}
	\sum_k \sum_{i \in S_k} \frac{N_i^2 \| \Omega_i \|^2}{n_k^4 \bar{N}_k^4} 
	&\leq \sum_k \frac{1}{n_k^4 \bar{N}_k^4} \sum_{i,m \in S_k}  N_i N_m \langle \Omega_i, \Omega_m \rangle 
	\\&= \sum_k \frac{1}{n_k^2 \bar{N}_k^2} \| \mu_k \|^2,
\end{align*}
which establishes the first claim. 

Similarly, 
\begin{align*}
	\sum_k \sum_{i \in S_k} \frac{N_i^3 \| \Omega_i \|_3^3}{n_k^4 \bar{N}_k^4} 
	&\leq \sum_k \frac{1}{n_k^4 \bar{N}_k^4}  \sum_{i,m,m' \in S_k} N_i N_m N_{m'} \sum_j \Omega_{ij}\Omega_{mj} \Omega_{m'j} 
	\\&\leq \sum_k \frac{1}{n_k \bar{N}_k} \| \mu_k \|_3^3,  
\end{align*}
which proves the second claim.

The third claim follows similarly and we omit the proof. 

\qed 

\section{Proofs of other main lemmas and theorems}

\subsection{Proof of Lemma~\ref{prop:unbiased}} \label{subsec:proof-prop-unbiased}

We start from computing $\mathbb{E}[(\hat{\mu}_{kj}-\hat{\mu}_j)^2]$. 
Write $X_{ij}=N_i(\Omega_{ij}+Y_{ij})$. It follows by elementary calculation that  
\[
\hat{\mu}_{kj}-\hat{\mu}_j = \mu_{kj}-\mu_j +\Bigl( \frac{1}{n_k\bar{N}_k}-\frac{1}{n\bar{N}}\Bigr)\sum_{i\in S_k}N_iY_{ij} - \frac{1}{n\bar{N}} \sum_{\ell: \ell\neq k} \sum_{i\in S_{\ell}}N_iY_{ij}. 
\]
For different $k$, the variables $\sum_{i\in S_k}N_iY_{ij}$ are independent of each other. 
It follows that
\begin{align} \label{lem-unbias-1}
	\mathbb{E}[(\hat{\mu}_{kj}&-\hat{\mu}_j)^2] = (\mu_{kj}-\mu_j)^2 + \Bigl( \frac{1}{n_k\bar{N}_k}-\frac{1}{n\bar{N}}\Bigr)^2\mathbb{E}\Bigl[ \Bigl(\sum_{i\in S_k}N_iY_{ij}\Bigr)^2 \Bigr] +  \sum_{\ell:\ell\neq k} \frac{1}{n^2\bar{N}^2}\mathbb{E}\Bigl[ \Bigl(\sum_{i\in S_\ell}N_iY_{ij}\Bigr)^2 \Bigr]\cr
	&= (\mu_{kj}-\mu_j)^2 + \Bigl( \frac{1}{n_k\bar{N}_k}-\frac{1}{n\bar{N}}\Bigr)^2 \sum_{i\in S_k}N_i\Omega_{ij}(1-\Omega_{ij})  +  \sum_{\ell:\ell\neq k} \frac{1}{n^2\bar{N}^2} \sum_{i\in S_\ell}N_i\Omega_{ij}(1-\Omega_{ij})\cr
	&= (\mu_{kj}-\mu_j)^2 + \frac{1}{n_k^2\bar{N}_k^2}\Bigl( 1-\frac{n_k\bar{N}_k}{n\bar{N}} \Bigr)\sum_{i\in S_k}N_i\Omega_{ij}(1-\Omega_{ij})\cr
	&\qquad  + \frac{1}{n^2\bar{N}^2}\biggl[\Bigl( 1 - \frac{n\bar{N}}{n_k\bar{N}_k}\Bigr) \sum_{i\in S_k}N_i\Omega_{ij}(1-\Omega_{ij}) + \sum_{\ell:\ell\neq k} \sum_{i\in S_\ell}N_i\Omega_{ij}(1-\Omega_{ij}) \biggr]\cr
	&= (\mu_{kj}-\mu_j)^2 + \frac{1}{n_k^2\bar{N}_k^2}\Bigl( 1-\frac{n_k\bar{N}_k}{n\bar{N}} \Bigr)\sum_{i\in S_k}N_i\Omega_{ij}(1-\Omega_{ij})\cr
	&\qquad  - \frac{1}{n\bar{N} n_k\bar{N}_k}\underbrace{\biggl[ \sum_{i\in S_k}N_i\Omega_{ij}(1-\Omega_{ij}) - \frac{n_k\bar{N}_k}{n\bar{N}} \sum_{\ell=1}^K \sum_{i\in S_\ell}N_i\Omega_{ij}(1-\Omega_{ij}) \biggr]}_{\delta_{kj}}. 
\end{align}
Since $X_{ij}$ follows a binomial distribution, it is easy to see that $\mathbb{E}[X_{ij}]=N_i\Omega_{ij}$ and $\mathbb{E}[X_{ij}^2]=(\mathbb{E}[X_{ij}])^2+\mathrm{Var}(X_{ij})=N_i^2\Omega^2_{ij}+N_i\Omega_{ij}(1-\Omega_{ij})$. Combining them gives 
\beq \label{lem-unbias-2}
\mathbb{E}[X_{ij}(N_i-X_{ij})] = N_i (N_i-1)\Omega_{ij}(1-\Omega_{ij}). 
\eeq
Define
\[
\hat{\zeta}_{kj} = (\hat{\mu}_{kj}-\hat{\mu}_{j})^2 -    \frac{1}{n_k^2\bar{N}_k^2}\Bigl( 1-\frac{n_k\bar{N}_k}{n\bar{N}} \Bigr)\sum_{i\in S_k}\frac{X_{ij}(N_i-X_{ij})}{N_i-1}, 
\]
It follows from \eqref{lem-unbias-1}-\eqref{lem-unbias-2} that
\beq \label{lem-unbias-3}
\mathbb{E}[\hat{\zeta}_{kj}] = (\mu_{kj}-\mu_j)^2  - \frac{1}{n\bar{N}n_k\bar{N}_k}\delta_{kj}. 
\eeq
We are ready to compute $\mathbb{E}[T]$. By definition,  $T=\sum_{j=1}^p\sum_{k=1}^K n_k\bar{N}_k\hat{\zeta}_{kj}$ and $\rho^2=\sum_{j,k}(\mu_{kj}-\mu_j)^2$. Consequently, 
\beq \label{lem-unbias-4}
\mathbb{E}[T] = \sum_{j=1}^p\sum_{k=1}^K n_k\bar{N}_k\Bigl[  (\mu_{kj}-\mu_j)^2 - \frac{1}{n\bar{N}n_k\bar{N}_k}\delta_{kj}\Bigr] = \rho^2 - \frac{1}{n\bar{N}}\sum_{j=1}^p\sum_{k=1}^K\delta_{kj}. 
\eeq
We use the definition of $\delta_{kj}$ in \eqref{lem-unbias-1}. It is seen that for each $1\leq j\leq p$, 
\beq \label{lem-unbias-5}
\sum_{k=1}^K \delta_{kj} = \sum_{k=1}^K \sum_{i\in S_k}N_i\Omega_{ij}(1-\Omega_{ij}) - \Bigl(\sum_{k=1}^K \frac{n_k\bar{N}_k}{n\bar{N}}\Bigr)\sum_{\ell=1}^K \sum_{i\in S_\ell}N_i\Omega_{ij}(1-\Omega_{ij}) = 0.
\eeq
Combining \eqref{lem-unbias-4}-\eqref{lem-unbias-5} gives $\mathbb{E}[T]=\rho^2$. This proves the claim. \qed

\subsection{Proof of Theorem~\ref{thm:alt}} 
\label{sec:thm_alt_proof}

%
%
%

First we show that 
\begin{align}
	\var(T) \lesssim \Theta_n
	\label{eqn:varT_vs_Thetan}
\end{align}
Recall 
\begin{align*}
	\Theta_{n1}
	&=  4\sum_{k=1}^K\sum_{j=1}^p n_k\bar{N}_k (\mu_{kj}-\mu_j)^2\mu_{kj} 
	\\ \Theta_{n2} 
	&= 2\sum_{k=1}^K\sum_{i\in S_k}\sum_{j=1}^p  \Bigl(\frac{1}{n_k\bar{N}_k}-\frac{1}{n\bar{N}}\Bigr)^2\frac{N_i^3}{N_i-1} \Omega_{ij}^2
	\\
	\Theta_{n3} &=  \frac{2}{n^2\bar{N}^2}\sum_{1\leq k\neq \ell\leq K}\sum_{i\in S_k}\sum_{m\in S_\ell} \sum_{j=1}^pN_iN_m\Omega_{ij}\Omega_{mj} 
	\\ \Theta_{n4} &= 2\sum_{k=1}^K \sum_{\substack{i\in S_k, m\in S_k,\\ i\neq m}}\sum_{j=1}^p \Bigl(\frac{1}{n_k\bar{N}_k}-\frac{1}{n\bar{N}}\Bigr)^2 N_iN_m\Omega_{ij}\Omega_{mj}. 
\end{align*}
and that $\sum_{a = 1}^4 \Theta_{na} = \Theta_n$. 

By Lemma \ref{lem:var1}, we immediately have
\begin{align}
	\label{eqn:var1_vs_Theta1}
	\var( {\bf 1}_p' U_1 ) \leq \Theta_{n1}.  
\end{align}

For $U_2$, it is shown in the Proof of Lemma \ref{lem:var2}  that
\[
\mathrm{Var}({\bf 1}_p'U_2) = 4\sum_{k=1}^K\sum_{i\in S_k} \sum_{1\leq r<s\leq N_i}\frac{\theta_i}{N_i(N_i-1)}\bigl[\|\Omega_i\|^2 + O(\|\Omega_i\|_3^3)].
\]
Thus
\begin{align*}
	\mathrm{Var}({\bf 1}_p'U_2) &\lesssim 
	4 \sum_{k=1}^K\sum_{i\in S_k} \sum_{1\leq r<s\leq N_i}\frac{\theta_i}{N_i(N_i-1)} \|\Omega_i\|^2
	\\&= 2 \sum_{k=1}^K\sum_{i\in S_k} \theta_i  \|\Omega_i\|^2
	= \Theta_{n2} \num 	\label{eqn:var2_vs_Theta2}
\end{align*}

Next we study $U_3$. Using that $\Omega_{mj'} \leq 1$ and $\| \Omega_i \|_1 = 1$, we have 
\begin{align*}
	\sum_{k\neq \ell}\frac{n_kn_{\ell}\bar{N}_k\bar{N}_\ell}{n^2\bar{N}^2}{\bf 1}_p'(\Sigma_k\circ\Sigma_\ell){\bf 1}_p
	&= \frac{2}{n^2\bar{N}^2} \sum_{k\neq \ell }\sum_{i\in S_k}\sum_{m\in S_{\ell}}\sum_{j,j'}N_iN_m\Omega_{ij}\Omega_{ij'}\Omega_{mj}\Omega_{mj'}
	\\ 
	&\leq \frac{2}{n^2\bar{N}^2} \sum_{k\neq \ell }\sum_{i\in S_k}\sum_{m\in S_{\ell}}\sum_{j }N_iN_m\Omega_{ij} \Omega_{mj}  \sum_{j'} \Omega_{ij'} 
	\\&= \frac{2}{n^2\bar{N}^2} \sum_{k\neq \ell }\sum_{i\in S_k}\sum_{m\in S_{\ell}}\sum_{j }N_iN_m\Omega_{ij} \Omega_{mj}. 
\end{align*}
Therefore by Lemma \ref{lem:var3},
\begin{align}
	\var(\mathbf{1}_p' U_3 ) 
	\lesssim \frac{2}{n^2\bar{N}^2}\sum_{1\leq k\neq \ell\leq K}\sum_{i\in S_k}\sum_{m\in S_\ell} \sum_{j=1}^pN_iN_m\Omega_{ij}\Omega_{mj}
	= \Theta_{n3}.  
	\label{eqn:var3_vs_Theta3}
\end{align}

Similarly for $U_4$, we have by the Proof of Lemma \ref{lem:var4} that
\begin{align}  
	\label{eqn:var4_vs_Theta4}
	\mathrm{Var}({\bf 1}_p'U_4) &=4\sum_{k=1}^K\sum_{\substack{i\in S_k, m\in S_k\\i< m}}\kappa_{im}\Bigl(\sum_{j}\Omega_{ij}\Omega_{mj} +\delta_{im}\Bigr) \cr
	&\les \sum_{k=1}^K\sum_{\substack{i\in S_k, m\in S_k\\i< m}}\kappa_{im} \sum_{j}\Omega_{ij}\Omega_{mj}
	= \Theta_{n4}. 
\end{align}
Above we use that $|\delta_{im}| \leq \sum_j \Omega_{ij}\Omega_{mj}$ and recall that $\kappa_{im} = (\frac{1}{n_k\bar{N}_k}-\frac{1}{n\bar{N}})^2N_iN_m$.

Observe that by Lemma \ref{prop:unbiased}, 
\begin{align}
	\Theta_{n1} =
	4\sum_{k=1}^K\sum_{j=1}^p n_k\bar{N}_k (\mu_{kj}-\mu_j)^2\mu_{kj} 
	\les \max_{k} \| \mu_k \|_\infty \cdot \rho^2
	= \max_{k} \| \mu_k \|_\infty \cdot \E \, T. 
	\label{eqn:Thetan1_ubd}
\end{align}
Since \eqref{cond1-basic} holds,  Lemma \ref{lem:Theta_n2+n3+n4} applies and
\begin{align}
	\Theta_{n_2} + \Theta_{n3} + \Theta_{n4} 
	\asymp \sum_k \| \mu_k \|^2.
	\label{eqn:sum_Thetant_ubd}
\end{align}
Combining \eqref{eqn:varT_vs_Thetan}, \eqref{eqn:Thetan1_ubd}, and \eqref{eqn:sum_Thetant_ubd} proves the theorem. 
\qed

\subsection{Proof of Theorem~\ref{thm:alt2}}
\label{sec:detection_boundary_appendix} 

To prove Theorem~\ref{thm:alt2}, we must prove the following claims:
\begin{enumerate} \itemsep 1pt
	\item[(a)] Under the alternative hypothesis, $\psi\to\infty$ in probability. 
	\item[(b)] For any fixed $\kappa\in (0,1)$, the level-$\kappa$ DELVE test has an asymptotic level of $\kappa$ and an asymptotic power of $1$. 
	\item[(c)] If we choose $\kappa=\kappa_n$ such that $\kappa_n\to 0$ and $1 - \Phi(\mathrm{SNR}_n) = o(\kappa_n)$, where $\Phi$ is the CDF of $N(0,1)$, then the sum of type I and type II errors of the DELVE test converges to $0$.  
\end{enumerate}

We  show the first claim, that $\psi \to \infty$, under the alternative hypothesis and the conditions of Theorem \ref{thm:alt2}. In particular, recall we assume that
\begin{align}
	\label{eqn:detection_boundaryK}
	\frac{\rho^2}{ \sqrt{\sum_{k =1}^K  \| \mu_k \|^2  } }=	\frac{ n \bar{N} \| \mu \|^2	\omega_n^2 }{ \sqrt{\sum_{k =1}^K  \| \mu_k \|^2  } } \to \infty. 
\end{align}


%
Our first goal is to show that
\begin{align}
	\label{eqn:ideal_T_blowup}
	T/\sqrt{\var(T)} \stackrel{\mathbb{P}}{\to} \infty 
\end{align}
under the alternative.
By Chebyshev's inequality, it suffices to show that
\begin{align}
	\label{eqn:Chebyshev_powerK}
	\E \, T \gg \sqrt{ \var(T) }.
\end{align} 
By Theorem \ref{thm:alt}, 
\begin{align}
	\label{eqn:thm_alt21}
	\var(T) \lesssim \sum_k \| \mu_k \|^2 + \max_k \| \mu_k \|_\infty \cdot \E T =
	\sum_k \| \mu_k \|^2 + \max_k \| \mu_k \|_\infty \cdot \rho^2
\end{align}
By \eqref{eqn:detection_boundaryK},
\[
\E T = \rho^2 \gg \sqrt{ \sum_{k =1}^K  \| \mu_k \|^2  } 
\geq \max_{1\leq k\leq K} \| \mu_k \| _\infty. 
\]
Therefore,
\begin{align}
	\label{eqn:Theta_1}
	\sqrt{ \max_{1\leq k\leq K} \| \mu_k \|_\infty }  \cdot \rho \ll \rho^2 = \E T.
\end{align} 
Moreover, by \eqref{eqn:detection_boundaryK}, 
\begin{align}
	\label{eqn:thm_alt22}
	\sum_k \| \mu_k \|^2 \ll \rho^4 = (\E T)^2.
\end{align}
Combining \eqref{eqn:thm_alt21}, \eqref{eqn:Theta_1}, and \eqref{eqn:thm_alt22} implies \eqref{eqn:ideal_T_blowup}. 

Next we show that $V > 0$ with high probability (i.e., with probability tending to $1$ as $n \bar{N} \to \infty$). Recall that by Lemmas \ref{lem:Theta_n2+n3+n4}, \ref{lem:Vdecompose},  and \ref{lem:varV},
\begin{align}
	\E V &= \Theta_{n2} + \Theta_{n3} + \Theta_{n4} \gtrsim \sum_k \| \mu_k \|^2 > 0  , \, \text{ and} 
	\\  \var(V) &\lesssim \sum_k \frac{ \| \mu_k \|^2}{n_k^2 \bar{N}_k^2} 
	\vee \sum_k \frac{ \| \mu_k \|_3^3}{ n_k \bar{N}_k}. 
\end{align}
Using this, the Markov inequality, and \eqref{cond2-regular}, we have
\begin{align}
	\mathbb{P}\bigl(V<\mathbb{E}[V]/2\bigr) \leq \mathbb{P}\bigl( |V-\mathbb{E}[V]|\geq \mathbb{E}[V]/2 \bigr)\leq \frac{4\mathrm{Var}(V)}{(\mathbb{E}[V])^2} = o(1),
	\num 	\label{eqn:V_vs_0_whp}
\end{align}
which implies that $V >  0$ with high probability.

To finish the proof of the first claim, note that the assumptions of Proposition \ref{prop:var_estimation_alt} are satisfied and we have $V/\var(T) = O_{\mathbb{P}}(1)$. By this, \eqref{eqn:ideal_T_blowup}, and \eqref{eqn:V_vs_0_whp}, we have
\begin{align*}
	\psi = 	\frac{T {\bf 1}_{V> 0} }{\sqrt{V}}
	= \frac{ \sqrt{ \var(T)} }{ \sqrt{V}} \cdot \frac{T}{\sqrt{\var(T)}} \cdot {\bf 1}_{V> 0}
	\gtrsim  \frac{T}{\sqrt{\var(T)}} 
	\to \infty 
\end{align*}
in probability. 

The second claim follows directly from the first claim and Theorem \ref{thm:null2}. 

To prove the third claim, by Chebyshev's inequality and $T/ \sqrt{\var(T)} \to \infty$, it follows that $T > (1/2) \E T = (1/2)\rho^2$ with high probability as $n \bar{N} \to \infty$. By a similar Chebyshev argument as above, it also holds that $V < (3/2) \E V$ with high probability as $n \bar{N} \to \infty$. Recall that $\E V = \Theta_{n2}+ \Theta_{n3}+ \Theta_{n4} \les \sum_k \| \mu_k \|^2$ by Lemmas \ref{lem:Theta_n2+n3+n4} and \ref{lem:Vdecompose}. Thus, with high probability as $n \bar{N} \to \infty$, we have
$$\psi = T {\bf 1}_{V>0} / \sqrt{V} \gtrsim \rho^2 / \sqrt{ \E V } \gtrsim \frac{ n \bar{N} \| \mu \|^2 \omega_n^2 }{\sqrt{\sum_k \| \mu_k \|^2}} = \mathrm{SNR}_n.$$
Choosing $\alpha_n$ as specified yield the third claim. The proof is complete since all three claims are  established.

\qed

\subsection{Proof of Theorem~\ref{thm:LB}} 

%
%

Without loss of generality, 
we assume $p$ is even and write $m= p/2$. Let $\mu\in\mathbb{R}^m$ be a nonnegative vector with $\|\mu\|_1=1/2$
. Let $\tilde{\mu}=(\mu', \mu')'\in\mathbb{R}^p$. We consider the null hypothesis:
\beq \label{LBconstruct-null}
H_0: \qquad \Omega_i = \tilde{\mu}, \qquad 1\leq i\leq n. 
\eeq
We pair it with a random alternative hypothesis.  Let $b_1,b_2, \ldots,b_m$ be a collection of i.i.d.\ Rademacher variables. Let $z_1, z_2, \ldots,z_K$ denote an independent collection of i.i.d.\ Rademacher random variables conditioned on the event
$
|\sum_k z_k| \leq 100 \sqrt{K}. 
$
For a properly small sequence $\omega_n>0$ of positive numbers, let 
\beq \label{LBconstruct-alt}
H_1: \qquad \Omega_{ij}  = 
\begin{cases} 
	\mu_j\bigl(1 + \omega_n (n_k \bar{N}_k)^{-1} \big( \frac{1}{K} \sum_{k \in K} n_k \bar{N}_k \big) z_k b_j\bigr), & \mbox{if } 1\leq j\leq m, i \in S_k \cr
	\tilde{\mu}_{j} \bigl(1 - \omega_n (n_k \bar{N}_k)^{-1} \big( \frac{1}{K} \sum_{k \in K} n_k \bar{N}_k \big) z_k b_{j-m}\bigr), &\mbox{if }m+1\leq j\leq 2m, i \in S_k
\end{cases}
\eeq
In this section we slightly abuse notation, using $\omega_n$ to refer to the (deterministic) sequence above and reserving $\omega(\Omega)$ for the random quantity 
\beq
\label{eqn:omega_def}
\omega(\Omega) = \sqrt{  \frac{1}{n\bar{N}\|\mu\|^2} \sum_{k=1}^Kn_k\bar{N}_k \|\mu_{k}-\mu\|^2 }. 
\eeq
As long as 
\[\omega_n \leq \frac{ \min_k n_k \bar{N}_k}{  \frac{1}{K} \sum_{k \in [K] } n_k \bar{N}_k   } = \frac{\min_k n_k \bar{N}_k}{ n \bar{N}/K} ,\] 
then $\Omega_{ij} \geq 0$ for all $i \in [n], j \in [p]$. Furthermore, for each $1\leq i\leq n$, we have $\|\Omega_i\|_1=2\|\mu\|_1=1$. We suppose there exists a constant $c \in (0,1)$ such that 
\begin{align}
	c K^{-1} n \bar{N}  \leq n_k \bar{N}_k &\leq c^{-1} K^{-1} n \bar{N}
	\quad \text{for all } k \in [K] \label{eqn:lbd_assn1}
\end{align}
%
%
With \eqref{eqn:lbd_assn1} in hand, we may assume without loss of generality that 
\begin{align*}
	\num \label{eqn:omega_assn}
	\omega_n \leq c/2
\end{align*}
This assumption implies that \eqref{LBconstruct-alt} is well-defined and moreover $\Omega_{ij} \asymp \mu_j $. 

Next we characterize the random quantity $\omega(\Omega)$ in terms of $\omega_n$. 

\begin{lemma} \label{lem:LB-signal}
	Let $\omega^2(\Omega)$ be as in \eqref{eqn:omega_def}.  When $\Omega$ follows Model \eqref{LBconstruct-alt}, there exists a constant $c_1 \in (0,1)$ such that $c_1 \omega_n^2\leq \omega^2(\Omega)\leq c_1^{-1}\omega_n^2$ with probability $1$.
\end{lemma}
The proof of Lemma \ref{lem:LB-signal} is given in Section \ref{sec:LB-signal-proof}. By Lemma \ref{lem:LB-signal},  under the model \eqref{LBconstruct-alt} it holds with probability $1$ that
\begin{align*}
	\num \label{eqn:simpler_SNR}
	\frac{ n \bar{N} \| \mu \|^2 \omega^2(\Omega) }{\sqrt{ \sum_{k=1}^K \| \mu_k \|^2 }}
	\asymp K^{-1/2} n \bar{N} \| \mu \| \omega_n^2.
\end{align*}
Above we use that $\Omega_{ij} \asymp \mu_j$ , since we assume \eqref{eqn:omega_assn} 

We also require Proposition \ref{prop:LB} below, whose proof is given in Section \ref{sec:lbd_construction}. 

\begin{proposition} \label{prop:LB}
	Suppose that \eqref{eqn:lbd_assn1} and \eqref{eqn:omega_assn} hold. Consider the pair of hypotheses in \eqref{LBconstruct-null}-\eqref{LBconstruct-alt} and let $\mathbb{P}_0$, and $\mathbb{P}_1$ be the respective probability measures. If 
	$$\frac{ n \bar{N} \| \mu \|^2 \omega^2(\Omega) }{\sqrt{ \sum_{k=1}^K \| \mu_k \|^2 }} \asymp 	 K^{-1/2} n \bar{N} \| \mu \| \omega_n^2 \to 0,$$ 
	then the chi-square distance between $\mathbb{P}_0$ and $\mathbb{P}_1$ converges to $0$. 
\end{proposition}

Now we prove Theorem \ref{thm:LB}. Let $\delta_n$ denote an arbitrary sequence tending to $0$. Without loss of generality, we may assume that $\delta_n \leq c^*$ for a small absolute constant $c^* \in (0,1)$. Note that $K^{-1/2} n \bar{N} \geq 1 $ since $K \leq n$.  Thus for appropriate choice of sequences of $\mu = \mu_n$ and $\omega_n \leq c/2$ in models \eqref{LBconstruct-null}, \eqref{LBconstruct-alt} and applying  \eqref{eqn:simpler_SNR}, we obtain 
\begin{align*}
	\num \label{eqn:model_choice} 
	2 \delta_n \geq \frac{ n \bar{N} \| \mu \|^2 \omega^2(\Omega) }{\sqrt{ \sum_{k=1}^K \| \mu_k \|^2 }}  
	\geq \delta_n .
\end{align*}

Recall the definitions of $\mathcal{Q}^*_{0n}$ and $\mathcal{Q}^*_{1n}$ in \eqref{LB-param-class2}. Let $\Pi$ denote the distribution on $\xi = \{(N_i, \Omega_i, \ell_i)\} \in \mathcal{Q}^*_{1n}$ induced by \eqref{LBconstruct-alt}. Let $\xi_0$ denote the parameter associated to the simple null hypothesis in \eqref{LBconstruct-null} associated to our choice of $\mu$ and $\omega_n$ satisfying \eqref{eqn:model_choice}. We have by standard manipulations,
\begin{align*}
	\mc{R}( \mathcal{Q}^*_{0n}, \mathcal{Q}^*_{1n} )
	&:= 	\inf_{\Psi\in\{0,1\}} \Bigl\{ \sup_{\xi\in {\cal Q}_{0n}^*(c_0,\epsilon_n)} \mathbb{P}_{\xi}(\Psi=1)+\sup_{\xi\in {\cal Q}_{1n}^*(\delta_n; c_0,\epsilon_n)} \mathbb{P}_{\xi}(\Psi=0)\Bigr\}
	\\&=	\inf_{\Psi\in\{0,1\}} \Bigl\{ \sup_{\xi\in {\cal Q}_{0n}^*(c_0,\epsilon_n), \xi' \in {\cal Q}_{1n}^*(\delta_n; c_0,\epsilon_n)} \big[ \mathbb{P}_{\xi}(\Psi=1)
	+ \mathbb{P}_{\xi}(\Psi=0)\big] 
	\\&\geq 	\inf_{\Psi\in\{0,1\}} \Bigl\{ \sup_{\xi\in {\cal Q}_{0n}^*(c_0,\epsilon_n)}  \E_{\xi' \sim \Pi } \bigg[ \mathbb{P}_{\xi}(\Psi=1)
	+ \mathbb{P}_{\xi'}(\Psi=0)\bigg] \bigg\} 
	\\&\geq 	\inf_{\Psi\in\{0,1\}} \Bigl\{  \E_{\xi' \sim \Pi } \bigg[ \mathbb{P}_{\xi_0}(\Psi=1)
	+ \mathbb{P}_{\xi'}(\Psi=0)\bigg] \bigg\} 
	\\&= 	\inf_{\Psi\in\{0,1\}} \Bigl\{   \mathbb{P}_{0}(\Psi=1)
	+ \mathbb{P}_{1}(\Psi=0) \bigg\} .
\end{align*}
In the last line we recall the definition of $\pr_0$ and $\pr_1$ in \eqref{LBconstruct-null} and \eqref{LBconstruct-alt}, noting that for all events $E$, 
$$ \pr_1(E) = \E_{\xi' \sim \pi} \, \pr_{\xi'}(E).$$ 
Next, by the Neyman--Pearson lemma and the standard inequality $\text{TV}(P, Q) \leq \sqrt{ \chi^2 (P, Q)}$ (see e.g. Chapter 2 of \cite{tsybakov2009introduction}), 
\begin{align*}
	\mc{R}( \mathcal{Q}^*_{0n}, \mathcal{Q}^*_{1n} ) 
	&\geq \inf_{\Psi\in\{0,1\}} \Bigl\{   \mathbb{P}_{0}(\Psi=1)
	+ \mathbb{P}_{1}(\Psi=0) \bigg\}
	\\&=1 - \text{TV}\big( \pr_0, \pr_1 \big) 
	\geq 1-  \sqrt{ \chi^2( \pr_0, \pr_1 ) }.
\end{align*}
By Proposition \ref{prop:LB}, as $\delta_n \to 0$ we have $\chi^2(\pr_0, \pr_1) \to 0$ and thus $\mc{R}( \mathcal{Q}^*_{0n}, \mathcal{Q}^*_{1n} ) \to 1$, as desired.

\qed 

\subsubsection{Proof of Proposition~\ref{lem:LB-signal}}
\label{sec:LB-signal-proof}

Next, we perform a change of parameters that preserves the signal strength and chi-squared distance. The testing problem \eqref{LBconstruct-null} and \eqref{LBconstruct-alt} has parameters $\Omega_{ij}, N_i, \bar{N}_k, n_k, n,$ and $K$. Let $\pr_0$ and $\pr_1$ denote the distributions corresponding to the null and alternative hypotheses, respectively. For each $k \in [K]$, we combine all documents in sample $k$ to obtain new null and alternative distributions $\tilde \pr_0$ and $\tilde \pr_1$ with parameters $\tilde \Omega_{ij}, \tilde N_i, \bar{\tilde{N}}_i, \tilde n_i, \tilde n,$ and $\tilde K$ such that 
\begin{align*}
	\tilde K &= K = \tilde n \\
	\tilde N_i &=  n_i \bar{N}_i &&\quad  \text{ for $i \in [\tilde K]$  } \\
	\bar{\tilde{N}}_i &\equiv \tilde  N_i &&\quad  \text{ for $i \in [\tilde K]$  } \\
	\tilde n_i &= 1 &&\quad \text{ for $i \in [\tilde K]$  }. 
	\num \label{eqn:lbd-model-reparam}
\end{align*}
For notational ease, we define $\tilde N := \bar{\tilde{N}} =  \frac{1}{K} \sum_{k \in [K]} n_k \bar{N}_k $.  Furthermore, we have $\tilde \Omega_i = \mu$ for all $i \in [\tilde n]$ under the null  $\tilde \Omega_i = \mu_i$ for all $i \in [\tilde n]$ under the alternative. Explicitly, in the reparameterized model, we have the null hypothesis
\beq \label{LBconstruct-null-repar}
H_0: \qquad \Omega_i = \tilde{\mu}, \qquad 1\leq i\leq n. 
\eeq
and alternative hypothesis 
\beq \label{LBconstruct-alt_reparam}
H_1: \qquad \Omega_{ij}  = 
\begin{cases} 
	\mu_j\bigl(1 + \omega_n \tilde N_i^{-1} \tilde{N}z_ib_j\bigr), & \mbox{if } 1\leq j\leq m,    \cr
	\tilde{\mu}_{j} \bigl(1 - \omega_n \tilde N_i^{-1} \tilde{N}z_ib_{j-m}\bigr), &\mbox{if }m+1\leq j\leq 2m.  
\end{cases}
\eeq
for all $i \in [\tilde K] = [K] = [\tilde n] $. Observe that the likelihood ratio is preserved: $\frac{d\pr_0}{d\pr_1} = \frac{\tilde d\pr_0}{d\tilde \pr_1}$ and also $\omega(\Omega) = \omega(\widetilde{\Omega})$. For simplicity we work with this reparameterized model in this proof.

If $z_1, \ldots, z_{\tilde n}$ are independent Rademacher random variables then with probability at least $1/2$ it holds that
\begin{align}
	\label{eqn:z_conditioning}
	|\sum_i z_i| \leq 100 \sqrt{\tilde n} 
\end{align}
by Hoeffding's inequality. Recall that our random model is defined in \eqref{LBconstruct-alt} where (i) $z_1, \ldots, z_{\tilde n}$ are independent Rademacher random variables conditioned on the event $|\sum_i z_i| \leq 100 \sqrt{\tilde n}$, and (ii) $b_1, \ldots, b_m$ are independent Rademacher random variables.

Now we study $\omega^2(\widetilde{\Omega})$. For each $1\leq j\leq m$,  we have $\widetilde{\Omega}_{ij}=\mu_j(1+\omega_n \tilde N_i^{-1}\tilde{N}z_ib_j)$. Define $\eta_j = (\tilde n\tilde{N})^{-1}\sum_{i=1}^{\tilde n}\tilde N_i\widetilde{\Omega}_{ij}=\mu_j (1 + \omega_n\bar{z}b_j)$ for $1 \leq j \leq m$ and 
$\eta_j = (\tilde n\tilde{N})^{-1}\sum_{i=1}^{\tilde n}\tilde N_i\widetilde{\Omega}_{ij}=\tilde{\mu}_j (1 - \omega_n\bar{z}b_j)$ for $m < j \leq 2m$.
We have
\begin{align*}
	\sum_{i=1}^{\tilde n}\sum_{j=1}^p\tilde N_i(\widetilde{\Omega}_{ij}-\eta_j)^2 
	& = 2\sum_{i=1}^{\tilde n} \sum_{j=1}^m\tilde N_i\cdot \mu_j^2 \omega_n^2 \frac{\tilde{N}^2}{\tilde N_i^2}(z_i-\bar{z})^2b_j^2\cr
	&= 2 \omega_n^2\tilde{N}^2\|\mu\|^2 \sum_{i=1}^{\tilde n} \tilde N_i^{-1}(z_i-\bar{z})^2. 
\end{align*}
By \eqref{eqn:z_conditioning}, $| \bar{z}| \leq 100 \sqrt{ \tilde n}$. Thus $| z_i - \bar{z}| \asymp 1$. %
Write $\tilde{N}_*=( \tilde n^{-1}\sum_{i=1}^{\tilde n } \tilde N_i^{-1})$. It follows that 
\[
\sum_{i=1}^{\tilde n}\sum_{j=1}^p\tilde N_i(\widetilde{\Omega}_{ij}-\eta_j)^2  \asymp \omega_n^2 \tilde N^2 \| \mu \|^2 \cdot \tilde n \tilde N_*^{-1}. 
\]
Note that $\tilde{N}\geq \tilde{N}_*$. Additionally, by assumption \eqref{eqn:lbd_assn1}, $\tilde N_i \asymp \tilde{N}\leq c^{-1}\tilde{N}_*$. It follows that 
\beq \label{thm-LB-1}
\sum_{i=1}^{\tilde n}\sum_{j=1}^p\tilde N_i(\widetilde{\Omega}_{ij}-\eta_j)^2  \asymp    \tilde n \tilde N \| \mu \|^2 \omega_n^2.
\eeq
Moreover, $\|\eta\|^2=\sum_{j=1}^p \mu_j^2(1+\omega_n\bar{z}b_j)^2$. By our conditioning on the event in \eqref{eqn:z_conditioning}, 
\[|\omega_n\bar{z}b_j| \lesssim \omega_n {\tilde n }^{-1/2}.\] 
Since $\omega_n \leq 1$ and $\sum_j b_j = 0$, we have
\begin{align}
	\label{eqn:eta_norm_vs_mu_norm} 
	\|\eta\|^2 = \| \mu \|^2 +  \sum_{j = 1}^p \mu_j^2 \omega_n^2 \bar{z}^2 
	= \| \mu \|^2 [1 + O( \tilde n^{-1} ) ] \asymp \| \mu \|^2. 
\end{align}
Hence
\beq \label{thm-LB-2}
\omega^2(\widetilde{\Omega}) = 
\omega^2(\Omega)
\asymp  \omega_n^2, \qquad\mbox{where recall}\quad \omega(\widetilde{\Omega})=\frac{\sum_{i=1}^{\tilde n }\sum_{j=1}^p \tilde N_i(\widetilde{\Omega}_{ij}-\eta_j)^2}{{\tilde n }\tilde{N}\|\eta\|^2}. 
\eeq
This finishes the proof.
\qed 

\subsubsection{Proof of Proposition~\ref{prop:LB}}
\label{sec:lbd_construction} 

In this proof, we continue to employ the reparametrization in \eqref{eqn:lbd-model-reparam}. As discussed there, this reparametrization preserves the likelihood ratio and thus the chi-square distance. 

By definition, $\chi^2(\mathbb{P}_0,\mathbb{P}_1) = \int (\frac{d\mathbb{P}_1}{d\mathbb{P}_0})^2d\mathbb{P}_0 -1$. 
It suffices to show that 
\beq \label{LB-proof-0}
\int \Bigl(\frac{d\mathbb{P}_1}{d\mathbb{P}_0}\Bigr)^2d\mathbb{P}_0 = 1+o(1). 
\eeq
From the density of of multinomial distribution, $d\mathbb{P}_0 = \prod_{i,j} \tilde{\mu}_j^{X_{ij}}$, and $d\mathbb{P}_1 =\mathbb{E}_{b,z}[ \prod_{i,j}\widetilde{\Omega}_{ij}^{X_{ij}}]$. 
It follows that 
\[
\frac{d\mathbb{P}_1}{d\mathbb{P}_0}  = \mathbb{E}_{b,z}\biggl[\prod_{i=1}^{\tilde n }\prod_{j=1}^p \Bigl(\frac{\widetilde{\Omega}_{ij}}{\tilde{\mu}_j}\Bigr)^{X_{ij}} \biggr] . 
\]

Let $b\rp{0}=(b\rp{0}_1,\ldots,b\rp{0}_m)'$ and $z\rp{0}=(z\rp{0}_1,\cdots, z\rp{0}_{\tilde n })'$ be independent copies of $b$ and $z$.  We construct $\widetilde{\Omega}\rp{0}_{ij}$ similarly as in \eqref{LBconstruct-alt_reparam}. It is seen that 
\begin{align}  \label{thm-LB-3}
	\int \Bigl(\frac{d\mathbb{P}_1}{d\mathbb{P}_0}\Bigr)^2d\mathbb{P}_0 &= \mathbb{E}_{X} \mathbb{E}_{b,z, b\rp{0},z\rp{0}}\biggl[\prod_{i=1}^{\tilde n }\prod_{j=1}^p \Bigl(\frac{\widetilde{\Omega}_{ij}\widetilde{\Omega}\rp{0}_{ij}}{\tilde{\mu}^2_j}\Bigr)^{X_{ij}} \biggr] \cr
	&=   \mathbb{E}_{b,z, b\rp{0},z\rp{0}}\biggl\{ \prod_{i=1}^{\tilde n } \mathbb{E}_{X_i}\biggl[\prod_{j=1}^p \Bigl(\frac{\widetilde{\Omega}_{ij}\widetilde{\Omega}\rp{0}_{ij}}{\tilde{\mu}^2_j}\Bigr)^{X_{ij}}\biggr] \biggr\}  \cr
	&= \mathbb{E}_{b,z, b\rp{0},z\rp{0}}\biggl\{ \prod_{i=1}^{\tilde n } \biggl(\sum_{j=1}^p \tilde{\mu}_j\cdot\frac{\widetilde{\Omega}_{ij}\widetilde{\Omega}\rp{0}_{ij}}{\tilde{\mu}_j^2}\Bigr)^{\tilde N_i}\biggr]\biggr\}\cr
	&= \mathbb{E}[\exp(M)], \quad\mbox{with}\quad M:=\sum_{i=1}^{\tilde n } \tilde N_i \log\Bigl(\sum_{j=1}^p \tilde{\mu}_j^{-1}\widetilde{\Omega}_{ij}\widetilde{\Omega}\rp{0}_{ij}\Bigr). 
\end{align}
Here, the third line follows from the moment generating function of a multinomial distribution. We plug in the expression of $\widetilde{\Omega}_{ij}$ in \eqref{LBconstruct-alt}. By direct calculations, 
\begin{align*}
	\sum_{j=1}^p \tilde{\mu}_j^{-1}\widetilde{\Omega}_{ij}\widetilde{\Omega}\rp{0}_{ij} &= \sum_{j=1}^m\mu_j\bigl(1+\omega_n\tilde N_i^{-1}\tilde{N}z_ib_j\bigr)\bigl(1+\omega_n\tilde N_i^{-1}\tilde{N}z\rp{0}_ib\rp{0}_j\bigr) \cr
	&\qquad +\sum_{j=1}^m\mu_j\bigl(1-\omega_n\tilde N_i^{-1}\tilde{N}z_ib_j\bigr)\bigl(1-\omega_n\tilde N_i^{-1}\tilde{N}z\rp{0}_ib\rp{0}_j\bigr) \cr
	&=2\|\mu\|_1 + 2\sum_{j=1}^m \mu_j \omega_n^2 \tilde N_i^{-2}\tilde{N}^2z_iz\rp{0}_ib_jb\rp{0}_j \cr
	&= 1 +  2\sum_{j=1}^m\mu_j \omega_n^2 \tilde N_i^{-2}\tilde{N}^2z_iz\rp{0}_ib_jb\rp{0}_j. 
\end{align*}
We plug it into $M$ and notice that $\log(1+t)\leq t$ is always true. It follows that
\beq \label{thm-LB-4}
M\leq \sum_{i=1}^{\tilde n } \tilde N_i\cdot 2\sum_{j=1}^m \mu_j \omega_n^2 \frac{\tilde{N}^2}{\tilde N_i^2}z_iz\rp{0}_ib_jb\rp{0}_j = 2\tilde{N}\omega_n^2\Bigl( \sum_{i=1}^{\tilde n } \frac{\tilde{N}}{\tilde N_i}z_iz\rp{0}_i\Bigr)\Bigl(\sum_{j=1}^m \mu_j b_j b\rp{0}_j\Bigr)=:M^*. 
\eeq

We combine \eqref{thm-LB-4} with \eqref{thm-LB-3}. It is seen that to show \eqref{LB-proof-0}, it suffices to show that
\beq  \label{thm-LB-5}
\mathbb{E}[\exp(M^*)]=1+o(1). 
\eeq
We now show \eqref{thm-LB-5}. Write  $M_1 = \sum_{i=1}^{\tilde n } (\tilde N_i^{-1}\tilde{N})z_iz\rp{0}_i$ and $M_2=\sum_{j=1}^p \mu_j b_jb\rp{0}_j$.

Recall that we condition on the event \eqref{eqn:z_conditioning}. By Hoeffding's inequality, Bayes's rule, and \eqref{eqn:z_conditioning},
\begin{align*}
	\mathbb{P}(|M_1|>t) 
	&= \pr\bigg( | \sum_i \frac{\tilde N}{ \tilde N_i} z_i z_i\rp{0} \geq t \, \,\bigg| \, \,
	| \sum_i z_i | \leq 100 \sqrt{\tilde n}, 	| \sum_i z_i\rp{0} | \leq 100 \sqrt{\tilde n}
	\bigg)
	\\ &= \frac{ \pr\big( | \sum_i \frac{\tilde N}{ \tilde N_i} z_i z_i\rp{0} | \geq t  \big)  }{  \pr( | \sum_i z_i | \leq 100 \sqrt{\tilde n} ) \, \,
		\pr( | \sum_i z_i\rp{0} | \leq 100 \sqrt{\tilde n} ) }
	\\ &\leq 4 \cdot 2\exp\Bigl(-\frac{t^2}{8\sum_{i=1}^{\tilde n } (\tilde N_i^{-1}\tilde{N})^2}\Bigr)
	\\&= 8 \exp\Bigl(-\frac{t^2}{8{\tilde n }}\Bigr). 
\end{align*}
for all $t >0$. In the last line, we have used the assumption of $\tilde N_i \asymp \tilde N$.
By Hoeffding's inequality again, we also have 
\begin{align*}
	\mathbb{P}(|M_2|>t)\leq 2\exp\Bigl(-\frac{t^2}{8\sum_{j=1}^p\mu_j^2}\Bigr)=2\exp\Bigl(-\frac{t^2}{8\|\mu\|^2}\Bigr)
\end{align*}
for all $t >0$.
Write $s^2_{\tilde n }=\sqrt{{\tilde n }}\tilde{N}\omega_n^2\|\mu\|$. 
It follows that
\begin{align} \label{thm-LB-6}
	\mathbb{P}(M^*>t) &= \mathbb{P}\bigl(2\tilde{N}\omega_n^2M_1M_2>t\bigr)= \mathbb{P}\bigl(M_1M_2>t\cdot \sqrt{{\tilde n }}\|\mu\|s_{\tilde n }^{-2}\bigr)\cr
	&\leq \mathbb{P}\bigl(M_1>\sqrt{t}\cdot \sqrt{{\tilde n }}s_{\tilde n }^{-1}\bigr) +  \mathbb{P}\bigl(M_2>\sqrt{t}\cdot \|\mu\|s_{\tilde n }^{-1}\bigr)\cr
	&\leq 8\exp\Bigl(-\frac{t}{8s^{2}_{\tilde n }}\Bigr) + 2\exp\Bigl(-\frac{t}{8s_{\tilde n }^2}\Bigr)\cr
	&\leq 4\exp(-c_1t /s_{\tilde n }^2), 
\end{align}
for some constant $c_1>0$. Here, in the last line, we have used the assumption of $\tilde N_i \asymp \tilde N$.

Let $f(x)$ and $F(x)$ be the density and distribution function of $M^*$. Write $\bar{F}(x)=1-F(x)$.  
Using integration by part, we have $\mathbb{E}[\exp(M^*)]=\int_0^\infty \exp(x)f(x)dx =-\exp(x)\bar{F}(x)|_0^{\infty}+\int_0^\infty \exp(x)\bar{F}(x)dx= 1+\int_0^\infty \exp(x)\bar{F}(x)dx$, provided that the integral exists. As a result, when $s_{\tilde n }=o(1)$, 
\begin{align*}
	\mathbb{E}[\exp(M^*)] - 1&=  \int_0^\infty \exp(t)\cdot\mathbb{P}(M^*>t)\cr
	&\leq 4\int_0^\infty \exp\bigl(- [c_1s_{\tilde n }^{-2}-1]t\bigr)dt\cr
	&\leq 4(c_1s_{\tilde n }^{-1}-1)^{-1} = 4s_{\tilde n }/(c_1-s_{\tilde n }). 
\end{align*}
It implies $\mathbb{E}[\exp(M^*)]=1+o(1)$, which is exactly \eqref{thm-LB-5}. This completes the proof.
because
\[
s^2_{\tilde n }=\sqrt{{\tilde n }}\tilde{N}\omega_n^2\|\mu\|
= \frac{ n \bar{N} \| \mu \|	\omega_n^2 }{ \sqrt{K}  }
\asymp \frac{ n \bar{N} \| \mu \|	\omega_n^2 }{ \sqrt{ \sum_{k \in K} \| \mu_k \|^2 }. }
\]

\qed

\subsection{Proof of Theorem~\ref{thm:K=2}}

%
%
%
%

First we show that
\begin{align}
	&T/\sqrt{\var(T)} \Rightarrow N(0,1) \label{eqn:ideal_T_AN_K=2}, \, \, \text{ and}  \\
	&V/\var(T) \to 1. \label{eqn:var_consistency_K=2}
\end{align}
If \eqref{eqn:ideal_T_AN_K=2} and \eqref{eqn:var_consistency_K=2} hold, then by mimicking the proof of Theorem \ref{thm:null2}, we see that $\psi$ is asymptotically normal and the level-$\kappa$ DELVE test has asymptotic level $\kappa$. We omit the details as they are quite similar.

Recall the martingale decomposition of $T$ described in Section \ref{sec:asymptotic_normality_appendix}. Observe that, under our assumptions, Lemmas \ref{lem:condl_var_bias}--\ref{lem:S_bounds_on_bounds} are valid. Moreover, by Lemmas \ref{lem:Theta_n2+n3+n4-K=2} and  \ref{lem:var_lbd}
\begin{align*}
	\num \label{eqn:varT_lbd_K=2}
	\var(T) \gtrsim \Theta_{n2} + \Theta_{n3} + \Theta_{n4} 
	\gtrsim   \bigg \|  \frac{m \bar{M}}{ n\bar{N}+ m \bar{M}} \eta +  
	\frac{n \bar{N}}{ n\bar{N}+ m \bar{M}} \theta  \bigg \|^2.
\end{align*}
Combining \eqref{eqn:varT_lbd_K=2} with Lemmas \ref{lem:condl_var_bias}--\ref{lem:S_bounds_on_bounds} and mimicking the argument in Section \ref{sec:thm-null-proof} implies that $T/\sqrt{V} \Rightarrow N(0,1)$. Thus \eqref{eqn:ideal_T_AN_K=2} is established.

Moreover, \eqref{eqn:var_consistency_K=2} is a direct consequence of our assumptions and Proposition \ref{prop:var_estimation_null_K=2}. The claims of  Theorem \ref{thm:K=2} regarding the null hypothesis follow.

To prove the claims about the alternative hypothesis, it suffices to show
\begin{align}
	&T/\sqrt{\var(T)} \to \infty \label{eqn:ideal_T_blowup_K=2}, \\
	&V > 0 \, \, \text{ with high probability} \label{eqn:V_vs_0_K=2}, \, \, \text{and} \\ 
	&V = O_{\mathbb{P}}(\var(T)) \label{eqn:V_vs_varT_K=2}.
\end{align}
Once these claims are established, we prove that $\psi = T {\bf 1}_{V>0}/\sqrt{V} \to \infty$ under the alternative by mimicking the last step of the proof of Theorem \ref{thm:alt2} in Section \ref{sec:detection_boundary_appendix}. We omit the details as they are very similar. 

Note that \eqref{eqn:V_vs_varT_K=2} follows directly from our assumptions and Proposition \ref{prop:var_estimation_alt_K=2}. 

As in the proof of Theorem \ref{thm:alt2} in Section \ref{sec:detection_boundary_appendix}, to establish \eqref{eqn:ideal_T_blowup_K=2}, it suffices to prove that
\begin{align}
	\label{eqn:ideal_T_blowup_K=2_equiv}
	\E T = \rho^2 \gg \var(T). 
\end{align}
Our main assumption under the alternative when $K = 2$ is 
\beq \label{SNR(K=2)-supp}
\frac{ \| \eta - \theta \|^2 }{  \big( \frac{1}{n\bar{N}} + \frac{1}{m \bar{M}}  \big) 
	\max\{\| \eta \|,\, \| \theta \|\} } \to \infty. 
\eeq
As shown in Section \ref{sec:thm_alt_proof}, we have that
\begin{align}
	\var(T) \les  \Theta_n = \Theta_{n1} +  \sum_{t = 2}^4\Theta_{nt}.
\end{align}
Applying \eqref{eqn:Theta_1} to the first term and Lemma \ref{lem:Theta_n2+n3+n4-K=2} to the remaining terms, we have
\begin{align*}
	\var(T) &\les \max \{  \| \eta \|_\infty, \, \| \theta \|_\infty \} \cdot \rho^2 
	+ \bigg \|  \frac{m \bar{M}}{ n\bar{N}+ m \bar{M}} \eta +  
	\frac{n \bar{N}}{ n\bar{N}+ m \bar{M}} \theta  \bigg \|^2
	\\&\les \max \{  \| \eta \|, \, \| \theta \| \} \cdot \rho^2 
	+ \max \{  \| \eta \|^2, \, \| \theta \|^2 \}
	\num \label{eqn:varT_mainbd_K=2} 
\end{align*}
Next, note that
\begin{align*}
	\rho^2
	&= n \bar{N} \| \eta - \mu \|^2
	+  m\bar{M} \| \theta - \mu \|^2
	\\&= n\bar{N}  \bigg \| \eta - \big( \frac{n \bar{N}}{n\bar{N}+ m\bar{M}} \eta + \frac{m\bar{M}}{n\bar{N}+ m\bar{M}} \theta \big)  \bigg \|^2 
	\\&\quad + m\bar{M}
	\bigg \| \theta - \big( \frac{n \bar{N}}{n\bar{N}+ m\bar{M}} \eta + \frac{m\bar{M}}{n\bar{N}+ m\bar{M}} \theta \big)  \bigg \|^2 
	\\&= n\bar{N} \cdot \big(\frac{m\bar{M}}{n\bar{N}+ m\bar{M}}\big)^2
	\| \eta - \theta  \|^2 
	+ m\bar{M} \cdot \big(\frac{n\bar{N}}{n\bar{N}+ m\bar{M}}\big)^2
	\| \eta - \theta  \|^2 
	\\&= \frac{ n \bar{N} m \bar{M}}{ ( n\bar{N}+ m\bar{M})} \| \eta - \theta \|^2
	= \big( \frac{1}{n\bar{N}} + \frac{1}{m\bar{M}} \big)^{-1} \| \eta - \theta \|^2. 
	\num \label{eqn:rho_K=2} 
\end{align*}
By \eqref{SNR(K=2)-supp}, \eqref{eqn:varT_mainbd_K=2}, and \eqref{eqn:rho_K=2}, we have
\begin{align*}
	\frac{(\E T)^2}{\var(T)} 
	&\gtrsim  \frac{\rho^4}{ \max \{  \| \eta \|, \, \| \theta \| \} \cdot \rho^2 
		+ \max \{  \| \eta \|^2, \, \| \theta \|^2 \} } 
	\\&\gtrsim \frac{ \| \eta - \theta \|^2 }{ ( \frac{1}{n\bar{N}} + \frac{1}{m\bar{M}} ) \max \{  \| \eta \|, \, \| \theta \| \} }
	+ \big(\frac{ \| \eta - \theta \|^2 }{ ( \frac{1}{n\bar{N}} + \frac{1}{m\bar{M}} ) \max \{  \| \eta \|, \, \| \theta \| \} }\big)^2 \to \infty, 
\end{align*}
which proves \eqref{eqn:ideal_T_blowup_K=2_equiv} and thus \eqref{eqn:ideal_T_blowup_K=2}. 

To prove \eqref{eqn:V_vs_0_K=2}, we mimick the Markov argument in \eqref{eqn:V_vs_0_whp} and use that under our assumptions, $\var(V)/(\E V)^2 = o(1)$ . We omit the details as they are similar. Since we have established \eqref{eqn:ideal_T_blowup_K=2}, \eqref{eqn:V_vs_0_K=2}, and \eqref{eqn:V_vs_varT_K=2}, the proof is complete.
\qed

\subsection{Proof of Theorem~\ref{thm:K=n}}

Note that $T/\sqrt{\var(T)} \Rightarrow N(0,1)$ by our assumptions and Proposition \ref{prop:null}. In particular, using that $n \to \infty$ and the monotonicity of the $\ell_p$ norms we have
\[
\frac{ \| \mu \|_4^4}{K \| \mu \|^4 }
= \frac{ \| \mu \|_4^4}{n \| \mu \|^4 }
\leq \frac{1}{n} \cdot \frac{ \|\mu\|^4 }{ \| \mu \|^4 }
= \frac{1}{n} \to 0.
\]
Moreover, $V^*/\var(T) \to 1$ in probability by Proposition \ref{prop:var_estimation_null_K=n}. It follows by Slutsky's theorem that $\psi^* = T/\sqrt{V^*} \Rightarrow N(0,1)$ and that the level-$\kappa$ DELVE test has an asymptotic level $\kappa$. 

To conclude the proof, it suffices to show that $\psi^* \to \infty$ under the alternative. As in the proof of Theorem~\ref{thm:alt2}, this follows immediately if we can show
\begin{align}
	&T/\sqrt{\var(T)} \to \infty \label{eqn:ideal_T_blowup_K=n}, \\
	&V^* > 0 \, \, \text{ with high probability} \label{eqn:V_vs_0_K=n}, \, \, \text{and} \\ 
	&V^* = O_{\mathbb{P}}(\var(T)) \label{eqn:V_vs_varT_K=n}.
\end{align}
Note that \eqref{eqn:ideal_T_blowup_K=n} follows from \eqref{eqn:ideal_T_blowup}, and \eqref{eqn:V_vs_varT_K=n} is the content of Proposition \ref{prop:var_estimation_alt_K=n}. Since our assumptions imply that $\E V^* \gg \sqrt{\var(V^*)}$, \eqref{eqn:V_vs_0_K=n} follows by a Markov argument as in \eqref{eqn:V_vs_0_whp}. 

\qed

\subsection{Proof of Theorem~\ref{thm:boundary}}

We apply Theorem~\ref{thm:null2} to get the asymptotic null distribution. Since $N_i=N$ and $\mu=p^{-1}{\bf 1}_p$, it is easy to see that Condition \ref{cond:K=n} is satisfied under our assumption of $p=o(N^2n)$. Therefore, by Theorem~\ref{thm:null2}, $\psi^* \to N(0,1)$ under $H_0$. 

We now show the asymptotic alternative distribution. By direct calculations and using $\sum_{i=1}^n \delta_{ij}=0$ and $\sum_{j=1}^p\delta_{ij}=0$, we have 
\[ 
\sum_{i,j} N_i (\Omega_{ij}-\mu_j)^2 = \frac{nN\nu_n^2}{p}, \quad\;\; \sum_{i,j} N_i (\Omega_{ij}-\mu_j)^2\Omega_{ij}=\frac{nN\nu_n^2}{p^2}, \quad\;\; \sum_{i} \|\Omega_i\|^2 = \frac{n(1+\nu_n^2)}{p}. 
\]
We apply Lemmas~\ref{lem:decompose}-\ref{lem:var4} and plug in the above expressions. Let $S = \bo' U_2$. It follows that  
\begin{align}
	T = \frac{nN\nu_n^2}{p} + S+O_{\mathbb{P}}\biggl( \frac{\sqrt{nN}\nu_n}{p} + \frac{1}{\sqrt{p}}\biggr), \qquad\mbox{where}\;\; \mathrm{Var}(S)=2p^{-1}n[1+o(1)]. 
	\label{eqn:T_behavior_contig}
\end{align}
First, we plug in $\nu_n^2= a\sqrt{2p}/(N\sqrt{n})$. It gives $p^{-1}nN\nu_n^2=\sqrt{2n/p}$. Second, $p^{-1}\sqrt{nN}\nu_n\asymp (np)^{-1/4} \sqrt{n/p}=o(\sqrt{n/p})$. It follows that
\beq \label{thm-boundary-1}
T = a\sqrt{2n/p} + S+o_{\mathbb{P}}\bigl(\sqrt{n/p}\bigr), \qquad\mbox{where}\;\; \mathrm{Var}(S) = (2n/p)[1+o(1)]. 
\eeq

Recall the martingale decomposition $S = \sum_{(\ell, s)} E_{\ell, s}$ where $E_{\ell,s}$ is defined in \eqref{eqn:Eells_def}. Observe that Lemmas \ref{lem:S_condvar} and \ref{lem:S_fourth_mom} hold (even under the alternative). Define $\widetilde{E}_{\ell,s} = E_{\ell,s}/\sqrt{\var(S)}$. Using $\var(S) \gtrsim n \sum_i \| \Omega_i \|^2$ and these lemmas, it is straightforward to verify that the following conditions hold:
\begin{align}
	&\sum_{(\ell, s)} \var\big( \widetilde{E}_{\ell, s} \big| \mc{F}_{\prec (\ell, s)} \big) \stackrel{\mathbb{P}}{\to} 1 \label{eqn:mgclt1_contig} \\
	&	\sum_{(\ell, s)} \E \widetilde{E}_{\ell,s}^4 \stackrel{\mathbb{P}}{\to} 0. \label{eqn:mgclt2_contig}
\end{align}
As in Section \ref{sec:thm-null-proof}, the martingale CLT applies and we have 
\begin{align*}
	S/\sqrt{\var(S)} \Rightarrow N(0,1). 
\end{align*} 
By \ref{eqn:T_behavior_contig}, 
\beq \label{thm-boundary-2}
T/\sqrt{\mathrm{Var}(S)}\;\; \to\;\; N(a, 1). 
\eeq
By Lemma \ref{lem:var2} and \eqref{eqn:varT_asymp_K=n}, 
$$\var(S) = [1+o(1)] \Theta_{n2} = [1+o(1)] \var(T)$$ 
By Proposition \ref{prop:var_estimation_alt_K=n}, we have that $V^*/\var(T) \to 1$ in probability. As a result, 
\beq \label{thm-boundary-3}
V^*/\mathrm{Var}(S)\;\; \to\;\; 1, \qquad\mbox{in probability}. 
\eeq
We combine \eqref{thm-boundary-2} and \eqref{thm-boundary-3} to conclude that $\psi = T/\sqrt{V^*}\to N(a, 1)$. 

\qed 

%


\section{Proofs of the corollaries for text analysis}

\subsection{Proof of Corollary~\ref{cor:TM}}

Note that Corollary \ref{cor:TM} follows immediately from the slightly more general result stated below.

\begin{cor} 
	\label{cor:TM-general}
	Consider Model~\eqref{Mod1-data} and suppose that $\Omega=\mu{\bf 1}_n'$ under the null hypothesis and that $\Omega$ satisfies \eqref{TopicModel} under the alternative hypothesis. Define $\xi\in\mathbb{R}^n$ by $\xi_i=\bar{N}^{-1}N_i$ and let $\widetilde{\Omega} = \Omega[\diag(\xi)]^{1/2}$. Let $\lambda_1, \ldots, \lambda_M>0$ and $\widetilde{\lambda}_1, \ldots, \widetilde{\lambda}_M>0$ denote the singular values of $\Omega$ and $\widetilde{\Omega}$, respectively, arranged in decreasing order.We further assume that under the alternative hypothesis, 
	\beq 
	\label{SNR-TM-supp}
	\frac{	\bar{N}\cdot \sum_{k=2}^M \widetilde{\lambda}_k^2}{\sqrt{\sum_{k=1}^M\lambda_k^2}} \to\infty.
	\eeq
	For any fixed $\kappa\in (0,1)$, the level-$\kappa$ DELVE test has an asymptotic level $\kappa$ and an asymptotic power $1$.  Moreover if $N_i \asymp \bar{N}$ for all $i$, we may replace $\sum_{k=2}^M \widetilde{\lambda}_k^2$ with $\sum_{k=2}^M\lambda_k^2$ in the numerator of \eqref{SNR-TM-supp}. 
\end{cor}

\begin{proof}[Proof of Corollary \ref{cor:TM-general}]
	
	This is a special case of our testing problem with $K=n$. Moreover, $\mu=n^{-1}\Omega\xi$ matches with the definition of $\mu$ in \eqref{Mod2-groupMean}. Therefore, we can apply Theorem~\ref{thm:K=n} directly. It remains to verify that the condition 
	\beq \label{cor-TM-0}
	\frac{	\bar{N}\cdot \sum_{k=2}^M \widetilde{\lambda}_k^2}{\sqrt{\sum_{k=1}^M\lambda_k^2}} \to\infty
	\eeq
	is sufficient to lead to the condition 
	\beq 
	\label{SNR(K=n)-TM}
	\frac{n\bar{N}\|\mu\|^2\omega_n^2}{ \sqrt{\sum_i \|\Omega_i \|^2} } \to\infty. 
	\eeq
	If we show this then Theorem~\ref{thm:K=n} applies directly.
	We first calculate  $\omega_n^2$. Recall $\xi_i=N_i/\bar{N}$ for $1\leq i\leq n$. Write 
	\[
	\widetilde{\Omega} = \Omega [\diag(\xi)]^{1/2}, \qquad \widetilde{\xi} = [\diag(\xi)]^{1/2}{\bf 1}_n. 
	\]
	For $K=n$, by \eqref{omega(K=n)}, $\omega_n^2=\frac{1}{n\bar{N}\|\mu\|^2}\sum_{i=1}^n N_i\|\Omega_i-\mu\|^2$. It follows that 
	\beq \label{cor-TM-1}
	\omega_n^2=\frac{1}{n\|\mu\|^2}\Bigl\|(\Omega-\mu{\bf 1}'_n)[\diag(\xi)]^{1/2}\Bigr\|_F^2 =\frac{1}{n\|\mu\|^2}\bigl\|\widetilde{\Omega} - \mu \widetilde{\xi}'\bigr\|_F^2.  
	\eeq
	Recall that $\widetilde{\lambda_1},\ldots,\widetilde{\lambda}_M$ are the singular values of $\widetilde{\Omega}$.  We apply a well-known result in linear algebra \citep{HornJohnson}, namely Weyl's inequality: For any rank-1 matrix $\Delta$, 
	$\|\widetilde{\Omega}-\Delta\|_F^2\geq \sum_{k\neq 1}\widetilde{\lambda}_k^2$. 
	In \eqref{cor-TM-1}, $\mu\widetilde{\xi}'$ is a rank-1 matrix. 
	It follows that 
	\beq \label{cor-TM-2}
	\bigl\|\widetilde{\Omega} - \mu \widetilde{\xi}'\bigr\|_F^2 \geq \sum_{k=2}^M \widetilde{\lambda}_k^2. 
	\eeq
	Hence 
	\begin{align*}
		\frac{n\bar{N}\|\mu\|^2\omega_n^2}{ \sqrt{\sum_i \|\Omega_i \|^2} } 
		\geq \frac{ \bar{N} \cdot \sum_{k=2}^M \widetilde{\lambda}_k^2}{ \| \Omega \|_F}
		= \frac{ \bar{N} \cdot \sum_{k= 2}^M \widetilde{\lambda}_k^2}{ \sqrt{\sum_{k=1}^M \lambda_k^2 }},
	\end{align*}
	which implies \eqref{SNR(K=n)-TM} by our assumption. The first claim is proved. 
	
	Next we prove the second claim. Observe that  if $N_i \asymp \bar{N}$ , then by Weyl's inequality: 
	\begin{align*}
		\omega_n^2 &= \frac{1}{\| \mu \|^2 n\bar{N}} \sum_i N_i \| \Omega_i - \mu|^2
		\gtrsim \frac{1}{ \| \mu \|^2} \sum_i  \| \Omega_i - \mu\|^2 
		\\&=  \frac{1}{ \| \mu \|^2} \| \Omega - \mu {\bf 1}_n ' \|_F^2 \geq  \frac{1}{ \| \mu \|^2} \sum_{k=2}^M \lambda_k^2. 
	\end{align*}
	Thus
	\begin{align*}
		\frac{n\bar{N}\|\mu\|^2\omega_n^2}{ \sqrt{\sum_i \|\Omega_i \|^2} } 
		\geq \frac{ \bar{N} \cdot \sum_{k=2}^M \lambda_k^2}{ \| \Omega \|_F}
		= \frac{ \bar{N} \cdot \sum_{k= 2}^M \lambda_k^2}{ \sqrt{\sum_{k=1}^M \lambda_k^2 }}. 
	\end{align*}
	We see that the assumption 
	\begin{align}
		\frac{ \bar{N} \cdot \sum_{k= 2}^M \lambda_k^2}{ \sqrt{\sum_{k=1}^M \lambda_k^2 }} \to \infty 
	\end{align}
	implies \eqref{SNR(K=n)-TM}. The second claim is established and the proof is complete. 
\end{proof}

\subsection{Proof of Corollary~\ref{cor:TM-LB}}

%
%

Recall the construction of a simple null and simple (random) alternative model from Section \ref{sec:lbd_construction}, specialized below to the case of $K=n$ and $N_i \equiv N$: 
\beq \label{LBconstruct-null-topic}
H_0: \qquad \Omega_i = \tilde{\mu}, \qquad 1\leq i\leq n. 
\eeq
\beq \label{LBconstruct-alt-topic}
H_1: \qquad \Omega_{ij}  = 
\begin{cases} 
	\mu_j\bigl(1 + \omega_n z_i b_j\bigr), & \mbox{if } 1\leq j\leq m \cr
	\tilde{\mu}_{j} \bigl(1 - \omega_n z_i b_{j-m}\bigr), &\mbox{if }m+1\leq j\leq 2m
\end{cases}
\eeq
where $b_1, \ldots, b_m$ are i.i.d.\ Rademacher random variables and $z_1, \ldots, z_n$ are i.i.d\ Rademacher random variables conditioned to satisfy $|\sum_{i} z_i| \leq 100 \sqrt{n}$. Define 
$$\tilde b= (b_1, \ldots, b_m, b_1, \ldots, b_m)'.$$
To derive the lower bound of Corollary \ref{cor:TM-LB}, we assume without loss of generality that $\omega_n$ is a sufficiently small absolute constant. 

We claim that $H_1$ prescribes a  topic model with $M = 2$ topics. To see this,
under the alternative,
\begin{align*}
	\num 	\label{eqn:two_topic_Omega}
	\Omega_i = \begin{cases}
		\mu \circ (\bo + \omega_{n} \, \tilde b) &\quad \text{ if } z_i = 1 \\
		\mu \circ (\bo - \omega_{n} \, \tilde b) &\quad \text{ if } z_i = -1.
	\end{cases}
\end{align*} 
Moreover, we showed in Section \ref{sec:lbd_construction} that $\Omega_{ij} \geq 0$ for all $i, j$ and that $\| \Omega_{ij} \|_1 = 1$. From \eqref{eqn:two_topic_Omega}, we see that $\Omega = AW$ where $A \in \R^{p \times 2}$ and $W \in \R^{2 \times n}$ are defined as follows:
\begin{align*}
	&A_{:1} = 	\mu \circ  (\bo + \omega_{n} \, \tilde b), \quad 
	A_{:2} = 	\mu \circ (\bo -  \omega_{n} \, \tilde b)
	\\ &W_{:i} = \begin{cases}
		(1,0)' \quad \text{if } z_i = 1 \\
		(0,1)' \quad \text{if } z_i = -1. 
	\end{cases}
\end{align*}
Moreover, under the null hypothesis, $\Omega$ clearly prescribes a topic model with $K = 1$. Therefore $\Omega$ follows the topic model \eqref{TopicModel}. Moreover, since $N_i \equiv N$, we have   $\Omega[\diag(\xi)]^{1/2} = \Omega$.

By Proposition \ref{prop:LB} specialized to our setting, we know that the $\chi^2$ distance between the null and alternative goes to zero if 
\begin{align*}
	\sqrt{n} N \| \mu \|  \omega_n^2 \to 0.
\end{align*}
Thus to prove Corollary \ref{cor:TM-LB} it suffices to show that 
\begin{align*}
	\num 	\label{eqn:cor-TM-main}
	\frac{N \sum_{k \geq 2}^M \lambda_k^2}{\sqrt{ \sum_{k=1}^M \lambda_k^2} } = \frac{N \lambda_2^2}{ \sqrt{ \sum_{k=1}^M \lambda_k^2} }\gtrsim \sqrt{n} N \| \mu \|  \omega_n^2 
\end{align*}

Accordingly we study the second largest singular value of $\Omega$. First we have some preliminary calculations. Let $U = \{ i: z_i = 1\}$, and let $V =  \{ i: z_i = -1\}$. Define 
\begin{align*}
	u &= \mu \circ (\bo + \omega_{n} \, \tilde b), \, \, \text{ and }
	\\ v&= \mu \circ (\bo - \omega_{n} \, \tilde b). 
\end{align*}
Observe that
\begin{align*}
	\langle u, v \rangle 
	&= \| \mu \|^2 - \omega_n^2 \| \mu \circ \tilde b \|^2 = 
	\| \mu \|^2 (1 - \omega_n^2). 
\end{align*}
Also, since $\omega_n$ is a sufficiently small absolute constant,
\begin{align*}
	\| u \|^2 &= \| \mu \|^2 + 2 \omega_n \langle \mu, \mu \circ \tilde b \rangle+
	\omega_n^2 \| \mu \circ \tilde b \|^2
	= (1 + \omega_n^2) \| \mu \|^2 + 2\omega_n \sum_j \mu_j^2 \, \tilde b_j
	\gtrsim \| \mu \|^2
	, \, \, \, \text{ and } \\
	\| v\|^2 &= \| \mu \|^2 - 2 \omega_n \langle \mu, \mu \circ \tilde b \rangle+
	\omega_n^2 \| \mu \circ \tilde b \|^2
	= (1 + \omega_n^2) \| \mu \|^2 - 2 \omega_n \sum_j \mu_j^2 \, \tilde b_j
	\gtrsim \| \mu \|^2. 
	\num \label{eqn:u_and_v_norms}
\end{align*}
Again, since we assume that $\omega_n$ is a sufficiently small absolute constant, 
\begin{align*}
	\delta^2 := 	\frac{ \langle u, v \rangle^2}{ \|u\|^2 \| v \|^2 }
	&= \frac{ \| \mu \|^4 (1 - \omega_n^2)^2 }{
		(1  + \omega_n^2)^2 \| \mu \|^4 - 4 \omega_n^2 \langle \mu, \mu \circ b \rangle^2 
	}
	\leq  \frac{ \| \mu \|^4 (1 - \omega_n^2)^2 }{
		(1  + \omega_n^2)^2 \| \mu \|^4 - 4 \omega_n^2 \| \mu \|^4
	}
	\\&= \frac{ \| \mu \|^4 (1 - \omega_n^2)^2}{\| \mu \|^4(1 + 2 \omega_n^2 - 3 \omega_n^4)} 
	=\frac{ (1 - \omega_n^2)^2 }{ 1 + 2 \omega_n^2 -3 \omega_n^4 }
	\num \label{eqn:incoherence_param}
\end{align*}
Note that
\begin{align*}
	\| a u + b v \|^2 
	&= a^2 \| u \|^2 + 2 ab \langle u, v \rangle + b^2 \| v \|^2 
	\geq a^2 \| u\|^2 + b^2 \| v \|^2 - 2 a b \delta \|u \| \|v \|
	\\&\geq (1 - \delta) \big( a^2 \| u \|^2  + b^2 \| v \|^2 \big) 
	+ \| au - b v \|^2 \geq (1 - \delta) \big( a^2 \| u \|^2  + b^2 \| v \|^2 \big).
\end{align*}
By \eqref{eqn:incoherence_param}, we have for $\omega_n$ sufficiently small that 
\begin{align*}
	1 - \delta  &\geq 1 - \frac{ 1 - \omega_n^2 }{\sqrt{ 1 + 2 \omega_n^2 -3 \omega_n^4} } 
	= \frac{\sqrt{1 + 2 \omega_n^2 -3 \omega_n^4} - 1 +  \omega_n^2 }{\sqrt{1 + 2 \omega_n^2 -3 \omega_n^4} }
	\\&\geq  \frac{  \omega_n^2 }{\sqrt{1 + 2 \omega_n^2 -3 \omega_n^4}}  \gtrsim \omega_n^2.
\end{align*}
Thus
\begin{align*}
	\num \label{eqn:incoherence_lower_bound}
	\| au + b v \|^2 \geq \omega_n^2( a^2 \| u\|^2 + b^2 \| v \|^2)
	\gtrsim \omega_n^2 \| \mu \|^2 (a^2 + b^2)
\end{align*}
Recall that if $M$ is a rank $k$ matrix, then 
\begin{align*}
	\num \label{eqn:rankM_sing}
	\lambda_k(M) = \sup_{y: \| y \| = 1, \, y \in \text{Ker}(M)^{\perp}} \| My \| = \sup_{y: \| y \| = 1, \, y \in \text{Im}(M')} \| My \|. 
\end{align*}
We have
\begin{align*}
	\Omega \Omega' = \sum_{i \in U} uu' + \sum_{i \in V} v v'
	= |U| u u' + |V| v v'. 
\end{align*}
Let $y \in \R^n$ satisfy $\|y\| = 1$ and $y = \Omega' x$ for some $x$. We have
\begin{align*}
	\Omega y =	\Omega \Omega' x 
	= |U| \langle u, x \rangle u + |V| \langle v, x \rangle v.
\end{align*}
By the previous equation and  \eqref{eqn:incoherence_lower_bound}, 
\begin{align*}
	\| \Omega y \|^2 = 	\| \Omega \Omega' x \|^2
	= \bigg \| |U| \langle u, x \rangle u + |V| \langle v, x \rangle v \bigg  \|^2
	\gtrsim \omega_n^2 \| \mu \|^2 \big( |U|^2 \langle u, x \rangle^2
	+ |V|^2 \langle v, x \rangle^2 \big). 
\end{align*}
By our conditioning on $z$, we have $\min(|U|, |V|) \gtrsim n$. Moreover
$$ 1 = \|y \|^2 = \| \Omega' x \|^2 =|U|  \langle u, x \rangle^2 + |V| \langle v, x \rangle^2. $$
Applying these facts and \eqref{eqn:rankM_sing}, we obtain
\begin{align*}
	\lambda_2^2 \geq \| \Omega y \|^2 =	\| \Omega \Omega' x \|^2 \gtrsim   \omega_n^2 \| \mu \|^2 n \big( |U| \langle u, x \rangle^2
	+ |V| \langle v, x \rangle^2 \big) 
	= \omega_n^2 \| \mu \|^2 n. 
\end{align*}
Next, 
\begin{align}
	\sum_{k  =1}^M \lambda^2_k = \| \Omega \|_F^2 
	= \sum_{i \in U} \| u \|^2 + \sum_{i \in V} \| v \|^2 
	= |U| \cdot \|u \|^2 + |V| \cdot \| v \|^2
	\asymp n \| \mu \|^2 
\end{align}
We conclude that 
\begin{align*}
	\frac{N \sum_{k \geq 2}^M \lambda_k^2}{\sqrt{ \sum_{k=1}^M \lambda_k^2} } = \frac{N \lambda_2^2}{ \sqrt{ \sum_{k=1}^M \lambda_k^2} }
	\gtrsim \frac{ N \cdot \omega_n^2 \| \mu \|^2 n }{ \sqrt{n} \| \mu \| }
	= \sqrt{n } N \| \mu \| \omega_n^2
\end{align*}
which establishes  \eqref{eqn:cor-TM-main}. The proof is complete. 
\qed

\subsection{Proof of Corollary~\ref{cor:author}}
This is a special case of our testing problem with $K=2$, we can apply Theorem~\ref{thm:K=2} directly. It remains to verify that the condition 
\beq \label{cor-author-0}
\frac{ \zeta_n^2\cdot (\|\eta_S\|_1+\|\theta_S\|_1) }{  \big( \frac{1}{n\bar{N}} + \frac{1}{m \bar{M}}  \big) 
	\max\{\| \eta \|,\, \| \theta \|\} } \to \infty
\eeq
is sufficient to yield the condition \eqref{SNR(K=2)} in Theorem~\ref{thm:K=2}. This is done by calculating $\|\eta-\theta\|^2$ directly. 
By our sparse model \eqref{SparseModel}, for $j\in S$, $|\sqrt{\eta_j}-\sqrt{\theta_j}|\geq \zeta_n$. It follows that for $j\in S$, 
\[
|\eta_j-\theta_j|^2 =(\sqrt{\eta_j} + \sqrt{\theta_j})^2(\sqrt{\eta_j} - \sqrt{\theta_j})^2\geq \zeta_n^2(\sqrt{\eta_j} + \sqrt{\theta_j})^2\geq \zeta_n^2(\eta_j+\theta_j). 
\]
It follows that 
\beq \label{cor-author-1}
\|\eta-\theta\|^2\geq \zeta_n^2\sum_{j\in S}(\eta_j+\theta_j)\geq \zeta_n^2\bigl(\|\eta_S\|_1 + \|\theta_S\|_1\bigr). 
\eeq
We plug it into \eqref{SNR(K=2)} and see immediately that \eqref{cor-author-0} implies this condition. The claim follows directly from Theorem~\ref{thm:K=2}. \qed

\section{A modification of DELVE for finite $p$} \label{sec:Finite-p}

Below we write out the variance of the terms of the raw DELVE statistic under the null, using the proofs of Lemmas \ref{lem:var2}--\ref{lem:var4}.  \begin{align*}
	\num \label{eqn:exact_variance} 
	\mathrm{Var}({\bf 1}_p'U_2) &= 2\sum_{k=1}^K\sum_{i\in S_k} \sum_{1\leq r<s\leq N_i}
	(\frac{1}{n_k\bar{N}_k}-\frac{1}{n\bar{N}})^2\frac{N_i^2}{(N_i-1)^2}
	\big[ \|\Omega_i\|^2-2\|\Omega_i\|_3^3 + \|\Omega_i\|^4 \big] 
	\\ 	\mathrm{Var}({\bf 1}_p'U_3)  &=  \frac{2}{n^2\bar{N}^2}\sum_{k\neq \ell}\sum_{i\in S_k} \sum_{m \in S_\ell} N_i N_m \Bigl(\sum_j \Omega_{ij}\Omega_{mj} - 2\sum_{j}\Omega^2_{ij}\Omega^2_{mj} + \sum_{j,j'}\Omega_{ij}\Omega_{ij'}\Omega_{mj}\Omega_{mj'} \Bigr)
	\\ \mathrm{Var}({\bf 1}_p'U_4) &=2\sum_{k=1}^K\sum_{\substack{i\in S_k, m\in S_k\\i \neq m}}(\frac{1}{n_k\bar{N}_k}-\frac{1}{n\bar{N}})^2 N_i N_m \Bigl(\sum_{j}\Omega_{ij}\Omega_{mj} - 2\sum_{j}\Omega^2_{ij}\Omega^2_{mj} + \sum_{j,j'}\Omega_{ij}\Omega_{ij'}\Omega_{mj}\Omega_{mj'}\Bigr). 
\end{align*}
In this section we develop an unbiased estimator for each term above, which leads to an unbiased estimator of $\var(T)$ by taking their sum. We require some preliminary results proved later in this section. Recall that Lemma \ref{lem:Omegaij2_unbiased} was established in the proof of Lemma \ref{lem:decompose}.

\begin{lemma}
	\label{lem:Omegajj'_unbiased}
	If $j \neq j'$, an unbiased estimator of $\Omega_{ij} \Omega_{ij'}$ is
	\begin{align*}
		\widehat{	\Omega_{ij} \Omega_{ij'} }
		:= \frac{X_{ij} X_{ij'}}{ N_i (N_i-1) } 
	\end{align*}
\end{lemma}

\begin{lemma}
	\label{lem:Omegaij2_unbiased}
	An unbiased estimator of $\Omega_{ij}^2$ is 
	\begin{align}
		\widehat{\Omega_{ij}^2} :=	\frac{X_{ij}^2 - X_{ij} }{N_i(N_i-1)}. 
	\end{align}
\end{lemma}

\begin{lemma}
	\label{lem:Omegajj'2_unbiased} 
	If $j \neq j'$, an unbiased estimator for $ \Omega_{ij}^2 \Omega_{ij'}^2$ is
	\begin{align*}
		\widehat{  \Omega_{ij}^2 \Omega_{ij'}^2 } 
		= \frac{ (X_{ij}^2 - X_{ij} ) (X_{ij'}^2 - X_{ij'} ) }{ N_i(N_i-1)(N_i - 2)(N_i - 3) }
	\end{align*}
\end{lemma}

\begin{lemma}
	\label{lem:Omega3_unbiased}
	An unbiased estimator of $\Omega_{ij}^3$ is 
	\begin{align}
		\widehat{\Omega_{ij}^3} :=	\frac{X_{ij}^3 - 3 X_{ij}^2 + 2 X_{ij} }{N_i(N_i-1)(N_i-2)}. 
	\end{align}
\end{lemma}

\begin{lemma}
	\label{lem:Omega4_unbiased}
	An unbiased estimator of $\Omega_{ij}^4$ is 
	\begin{align}
		\widehat{\Omega_{ij}^4} := \frac{X_{ij}^4 - 3 X_{ij}^3 - X_{ij}^2 + 3 X_{ij} }{N_i(N_i-1)(N_i-2)(N_i-3)}.
	\end{align}
\end{lemma}

Define 
\begin{align*}
	\widehat{ \| \Omega_i \|^2 }
	&:= \sum_j \widehat{\Omega_{ij}^2}
	\\ 	\widehat{ \| \Omega_i \|_3^3 }
	&:= \sum_j \widehat{\Omega_{ij}^3}
	\\ 	\widehat{ \| \Omega_i \|^4 }
	&:= \sum_j \widehat{\Omega_{ij}^4}
	+ \sum_{j \neq j'} \widehat{  \Omega_{ij}^2 \Omega_{ij'}^2 }. 
	\num \label{eqn:norm_estimators}
\end{align*}
Using Lemmas \ref{lem:Omegajj'_unbiased}--\ref{lem:Omega4_unbiased} and \eqref{eqn:norm_estimators}, we define an unbiased estimator for each term of \eqref{eqn:exact_variance}. Let $\widehat{\Omega_{ij}} = X_{ij}/N_i$ and define
\begin{align*}
	\num \label{eqn:exact_variance_estimators} 
	\widehat{\mathrm{Var}({\bf 1}_p'U_2)} &= 2\sum_{k=1}^K\sum_{i\in S_k} \sum_{1\leq r<s\leq N_i}
	(\frac{1}{n_k\bar{N}_k}-\frac{1}{n\bar{N}})^2\frac{N_i^2}{(N_i-1)^2}
	\big[\widehat{ \|\Omega_i\|^2}-2 \widehat{\|\Omega_i\|_3^3} + \widehat{\|\Omega_i\|^4} \big] 
	\\ 	\widehat{\mathrm{Var}({\bf 1}_p'U_3)}  &=  \frac{2}{n^2\bar{N}^2}\sum_{k\neq \ell}\sum_{i\in S_k} \sum_{m \in S_\ell} N_i N_m \Bigl(\sum_j \widehat{\Omega_{ij}} \widehat{\Omega_{mj}} - 2\sum_{j} \widehat{\Omega^2_{ij}} \widehat{\Omega^2_{mj}} + \sum_{j,j'}\widehat{\Omega_{ij}\Omega_{ij'}} \widehat{  \Omega_{mj}\Omega_{mj'}} \Bigr)
	\\ \widehat{\mathrm{Var}({\bf 1}_p'U_4)} &=2\sum_{k=1}^K\sum_{\substack{i\in S_k, m\in S_k\\i \neq m}}(\frac{1}{n_k\bar{N}_k}-\frac{1}{n\bar{N}})^2 N_i N_m \Bigl(\sum_{j}  \widehat{\Omega_{ij}} \widehat{\Omega_{mj}} - 2\sum_{j}\widehat{\Omega^2_{ij}} \widehat{\Omega^2_{mj}} + \sum_{j,j'}\widehat{\Omega_{ij}\Omega_{ij'}} \widehat{  \Omega_{mj}\Omega_{mj'}} \Bigr). 
\end{align*}
Define 
\begin{align}
	\label{eqn:DELVE_exact_var_est}
	\widetilde{V} = 	\widehat{\mathrm{Var}({\bf 1}_p'U_2)} 
	+ 	\widehat{\mathrm{Var}({\bf 1}_p'U_3)} 
	+ 	\widehat{\mathrm{Var}({\bf 1}_p'U_4)}. 
\end{align}
We define \textit{exact DELVE}  as $\tilde \psi = T/ \widetilde{V}^{1/2}$. Combining our results above, we obtain the following.

\begin{proposition}
	Consider the statistic $\widetilde{V}$ defined in \eqref{eqn:DELVE_exact_var_est}. Under the null hypothesis,  $\widetilde{V}$ is an unbiased estimator for $\var(T)$. 
\end{proposition}

With this result in hand, it is possible to derive consistency of $\widetilde{V}$ as an estimator of $\var(T)$ under certain regularity conditions. We omit the details. 

\subsection{Proof of Lemma~\ref{lem:Omegajj'_unbiased}}
Recall that  $B_{ijr} $ is the Bernoulli random variable $B_{ijr} =Z_{ijr} +\Omega_{ij}$ and satisfies $X_{ijr} = \sum_{r=1}^{N_i} B_{ijr}$. 
Observe that 
\begin{align*}
	X_{ij} X_{ij'}
	= \sum_{r,s} B_{ijr} B_{ij's} 
	= \sum_r B_{ijr} B_{ij'r}
	+ \sum_{r \neq s} B_{ijr}B_{ij's} 
	= 0 + \sum_{r \neq s} B_{ijr}B_{ij's} 
\end{align*}
Thus
\begin{align*}
	\E 	X_{ij} X_{ij'}
	=  N_i (N_i-1) \Omega_{ij} \Omega_{ij'},
\end{align*}
and we obtain
\begin{align*}
	\widehat{	\Omega_{ij} \Omega_{ij'} }
	= \frac{X_{ij} X_{ij'}}{ N_i (N_i-1)} 
\end{align*}
is an unbiased estimator for $\Omega_{ij} \Omega_{ij'}$, as desired. 
\qed 

\subsection{Proof of Lemma~\ref{lem:Omegajj'2_unbiased}}

Note that 
\begin{align*}
	X_{ij}^2 X_{ij'}^2
	&= \big( \sum_r B_{ijr} + \sum_{r \neq s } B_{ijr} B_{ijs} \big)
	\big( \sum_r B_{ij'r} + \sum_{r \neq s } B_{ij'r} B_{ij's} \big)
	\\&= \sum_{r} B_{ijr} B_{ij'r} +  \sum_{r_1 \neq r_2} B_{ijr} B_{ij's}
	+  \sum_{r_1 \neq s} B_{ijr_1} B_{ijs} \sum_{r_2} B_{ij'r_2}
	+ \sum_{r_1 \neq s} B_{ij'r_1} B_{ij's} \sum_{r_2} B_{ijr_2}
	\\&\quad 	 + \big( \sum_{r \neq s} B_{ijr} B_{ijs} \big) 
	\big( \sum_{r \neq s} B_{ij'r} B_{ij's} \big) 
	\\&= \sum_{r_1 \neq r_2} B_{ijr} B_{ij's} + \sum_{r_1 \neq s} B_{ijr_1} B_{ijs} \sum_{r_2} B_{ij'r_2}
	+ \sum_{r_1 \neq s} B_{ij'r_1} B_{ij's} \sum_{r_2} B_{ijr_2}
	\\&\quad + \big( \sum_{r \neq s} B_{ijr} B_{ijs} \big) 
	\big( \sum_{r \neq s} B_{ij'r} B_{ij's} \big) 
\end{align*}

Since $B_{ijr} B_{ij'r} = 0$, note that
\begin{align*}
	(X_{ij}^2 - X_{ij} ) (X_{ij'}^2 - X_{ij'} )
	&= \sum_{r_1 \neq s_1} \sum_{r_2 \neq s_2} B_{ijr_1} B_{ijs_1} B_{ij'r_2} B_{ij's_2}
	\\&= \sum_{ r_1, s_1, r_2, s_2 \, dist.} B_{ijr_1} B_{ijs_1} B_{ij'r_2} B_{ij's_2}.
\end{align*}
Thus 
\begin{align*}
	\E 	(X_{ij}^2 - X_{ij} ) (X_{ij'}^2 - X_{ij'} )
	&= \sum_{ r_1, s_1, r_2, s_2 \, dist.}  \E \big[ B_{ijr_1} B_{ijs_1} B_{ij'r_2} B_{ij's_2} \big] 
	\\&= N_i(N_i-1)(N_i - 2)(N_i - 3) \cdot  \Omega_{ij}^2 \Omega_{ij'}^2. 
\end{align*}
It follows that
\begin{align*}
	\widehat{  \Omega_{ij}^2 \Omega_{ij'}^2 } 
	= \frac{ (X_{ij}^2 - X_{ij} ) (X_{ij'}^2 - X_{ij'} ) }{ N_i(N_i-1)(N_i - 2)(N_i - 3) }
\end{align*}
is an unbiased estimator for $ \Omega_{ij}^2 \Omega_{ij'}^2$. 


\qed 

\subsection{Proof of Lemma~\ref{lem:Omega3_unbiased}}

Recall that  $B_{ijr} $ is the Bernoulli random variable $B_{ijr} =Z_{ijr} +\Omega_{ij}$ and satisfies $X_{ijr} = \sum_{r=1}^{N_i} B_{ijr}$. 
Observe that
\begin{align*}
	X_{ij}^3 = \sum_r B_{ijr} + 
	3\sum_{r_1 \neq r_2} B_{ijr_1} B_{ijr_2}
	+ \sum_{r_1 \neq r_2 \neq r_3} B_{ijr_1} B_{ijr_2} B_{ijr_3}. 
\end{align*}
Thus
\begin{align*}
	\E X_{ij}^3
	&= N_i \Omega_{ij}+ 3 N_i(N_i-1) \Omega^2_{ij} + N_i(N_i-1)(N_i-2) \Omega_{ij}^3. 
\end{align*}
Unbiased estimators for $\Omega_{ij}$ and $\Omega_{ij}^2$ are
\begin{align*}
	& \frac{X_{ij}}{N_i}
	\\ & \frac{X_{ij}^2}{N_i^2} - \frac{X_{ij}(N_i -X_{ij} )}{N_i^2(N_i-1)}
	= \frac{1}{N_i(N_i -1 )} \big(X_{ij}^2 - X_{ij} \big),
\end{align*}
respectively. Hence
\begin{align*}
	X_{ij}^3 
	- X_{ij} - 3 (X_{ij}^2 - X_{ij}) = X_{ij}^3 - 3 X_{ij}^2 + 2 X_{ij} 
\end{align*}
is an unbiased estimator for $N_i(N_i-1)(N_i-2) \Omega_{ij}^3$, as desired. 

\qed 

\subsection{Proof of Lemma~\ref{lem:Omega4_unbiased}}

Observe that 
\begin{align*}
	X_{ij}^4
	&= \sum_r B_{ijr}^4 + 4 \sum_{r_1 \neq r_2} B_{ijr_1}^3 B_{ijr_2}
	+ 6 \sum_{r_1 \neq r_2} B_{ijr_1}^2 B_{ijr_2}^2
	\\&\quad +  3 \sum_{r_1 \neq r_2 \neq r_3}B_{ijr_1}^2 B_{ijr_2} B_{ijr_3}
	+ \sum_{r_1 \neq r_2 \neq r_3 \neq r_4}  B_{ijr_1} B_{ijr_2} B_{ijr_3} B_{ijr_4}
	\\&=  \sum_r B_{ijr} + 10 \sum_{r_1 \neq r_2} B_{ijr_1} B_{ijr_2} +  3 \sum_{r_1 \neq r_2 \neq r_3}B_{ijr_1} B_{ijr_2} B_{ijr_3}
	\\&\quad 
	+ \sum_{r_1 \neq r_2 \neq r_3 \neq r_4}  B_{ijr_1} B_{ijr_2} B_{ijr_3} B_{ijr_4}.
\end{align*}
Thus 
\begin{align*}
	\E X_{ij}^4
	&= N_i \Omega_{ij} 
	+ 10 N_i (N_i-1)  \Omega_{ij}^2
	+ 3 N_i (N_i-1) (N_i-2) \Omega_{ij}^3 
	\\&\quad + N_i (N_i-1)(N_i-2)(N_i-3) \Omega_{ij}^4.
\end{align*}
Plugging in unbiased estimators for the first three terms, we have
\begin{align*}
	X_{ij}^4
	- X_{ij} -10 (X_{ij}^2 - X_{ij}) 
	- 3 ( X_{ij}^3 - 3 X_{ij}^2 + 2 X_{ij})
	= X_{ij}^4 - 3 X_{ij}^3 - X_{ij}^2 + 3 X_{ij} 
\end{align*}
is an unbiased estimator for $N_i(N_i-1)(N_i-2)(N_i-3)$, as desired. 
\qed

\bibliographystyle{rss}
\bibliography{topic}
\end{document}